\documentclass[aps,prx,twocolumn,superscriptaddress,nofootinbib,floatfix]{revtex4-2}

\usepackage{amsmath}
\usepackage{amssymb}
\usepackage{amsthm}
\usepackage{amsfonts}
\usepackage{mathtools} 
\usepackage{bm}

\usepackage[T1]{fontenc}
\usepackage[utf8]{inputenc}

\usepackage{graphicx}
\graphicspath{{./}{figures/}}
\usepackage[dvipsnames]{xcolor}

\usepackage{booktabs}
\usepackage{array}
\usepackage{braket}

\usepackage{algorithmic}
\makeatletter
\newcounter{algorithm}
\renewcommand{\thealgorithm}{\arabic{algorithm}}
\newenvironment{algorithm}[1][tbp]{%
  \@float{algorithm}[#1]%
}{%
  \end@float
}
\@namedef{ftype@algorithm}{4}%
\newcommand{\fnum@algorithm}{Algorithm~\thealgorithm}%
\newcommand{\ext@algorithm}{loa}%
\makeatother

\usepackage[most]{tcolorbox}
\usepackage{float}
\usepackage[colorlinks=true,linkcolor=blue,citecolor=blue,urlcolor=blue]{hyperref}

\theoremstyle{plain}
\newtheorem{theorem}{Theorem}[section]
\newtheorem{lemma}{Lemma}[section]
\newtheorem{proposition}{Proposition}[section]
\newtheorem{corollary}{Corollary}[section]
\newtheorem{conjecture}{Conjecture}[section]

\theoremstyle{definition}
\newtheorem{definition}{Definition}[section]
\newtheorem{example}{Example}[section]
\newtheorem{problem}{Problem}[section]

\theoremstyle{remark}
\newtheorem{remark}{Remark}[section]

\newcommand{\Heff}{H_{\mathrm{eff}}}
\newcommand{\HR}{H_{R}}
\newcommand{\HI}{H_{I}}
\newcommand{\alphaR}{\alpha_{R}}
\newcommand{\betaI}{\beta_{I}}
\newcommand{\Htilde}{\widetilde{H}}

\DeclareMathOperator{\poly}{poly}

\DeclareMathOperator{\spec}{spec}

\DeclareMathOperator{\rank}{rank}

\newcommand{\op}{\mathrm{op}}

\newcommand{\calO}{\mathcal{O}}
\newcommand{\eps}{\varepsilon}

\newcommand{\T}{\bbT}       % 2-torus
\newcommand{\bbT}{\mathbb{T}}

\newcommand{\C}{\mathbb{C}}
\newcommand{\R}{\mathbb{R}}
\newcommand{\Z}{\mathbb{Z}}

\newcommand{\cT}{\mathcal{T}}

\newcommand{\cP}{\mathcal{P}}
\newcommand{\cZ}{\mathcal{Z}}
\newcommand{\cF}{\mathcal{F}}
\newcommand{\cG}{\mathcal{G}}
\newcommand{\cH}{\mathcal{H}}

\newcommand{\bM}{\mathbf{M}}
\newcommand{\bH}{\mathbf{H}}
\newcommand{\bG}{\mathbf{G}}
\newcommand{\bU}{\mathbf{U}}
\newcommand{\bv}{\mathbf{v}}
\newcommand{\bu}{\mathbf{u}}
\newcommand{\bT}{\mathbf{T}}
\newcommand{\bF}{\mathbf{F}}
\newcommand{\bS}{\mathbf{S}}
\newcommand{\bA}{\mathbf{A}}
\newcommand{\bTheta}{\boldsymbol{\Theta}}
\newcommand{\calH}{\mathcal{H}}

\providecommand{\colonequals}{\coloneqq}

\begin{document}

\title{Simulation of Non-Hermitian Hamiltonians with Bivariate Quantum Signal Processing}

\author{Joshua M. Courtney}
\affiliation{University of Georgia, Department of Physics and Astronomy}

\date{\today}

\begin{abstract}
We achieve query-optimal quantum simulations of non-Hermitian Hamiltonians $H_{\mathrm{eff}} = H_R + iH_I$, where $H_R$ is Hermitian and $H_I \succeq 0$, using a bivariate extension of quantum signal processing (QSP) with non-commuting signal operators.
The algorithm encodes the interaction-picture Dyson series as a polynomial on the bitorus, implemented through a structured multivariable QSP (M-QSP) circuit.
A constant-ratio condition guarantees scalar angle-finding for M-QSP circuits with arbitrary non-commuting signal operators.
A degree-preserving sum-of-squares spectral factorization permits scalar complementary polynomials in two variables.
Angles are deterministically calculated in a classical precomputation step, running in $\mathcal{O}(d_R \cdot d_I)$ classical operations.
Operator norms $\alpha_R\,,\beta_I$ contribute additively with query complexity $\mathcal{O}((\alpha_R + \beta_I)T + \log(1/\eps)/\log\log(1/\eps))$ matching an information-theoretic lower bound in the separate-oracle model, where $H_R$ and $H_I$ are accessed through independent block encodings.
The postselection success probability is $e^{-2\beta_I T}\|e^{-iH_{\mathrm{eff}}T}|\psi_0\rangle\|^2\cdot (1 - \mathcal{O}(\eps))$, decomposing into a state-dependent factor $\|e^{-iH_{\mathrm{eff}}T}|\psi_0\rangle\|^2$ from the intrinsic barrier and an $e^{-2\beta_I T}$ overhead from polynomial block-encoding.
\end{abstract}

\maketitle

\section{Introduction}\label{sec:intro}

\subsection{Problem Statement}\label{sec:problem}

Non-Hermitian Hamiltonians arise naturally in common computational and physical representations~\cite{moiseyev2011non}, including the no-jump trajectory of a Lindblad master equation~\cite{wang2023simulating}, optical systems with gain and loss~\cite{jin2022quantum}, descriptions of resonance phenomena in scattering theory~\cite{feshbach1958unified}, and complex absorbing potentials (CAPs) in multi-channel AMO simulation~\cite{riss1993calculation,muga2004complex}.
They also notably appear in quantum algorithms for ground-state preparation~\cite{lin2020near}, partition function estimation~\cite{montanaro2015quantum}, and imaginary-time evolution~\cite{motta2020determining}.

Let $\Heff = \HR + i\HI$ be a non-Hermitian Hamiltonian acting on a finite-dimensional Hilbert space~$\mathcal{H}_s$, where $\HR = \HR^\dagger$
and $\HI = \HI^\dagger$ with $\HI \succeq 0$.  
Operator norms are notated as $\alphaR := \|\HR\|_{\op}$ and $\betaI := \|\HI\|_{\op}$. 
The target operation is a non-unitary propagator $e^{-i\Heff T}$ for an evolution time $T > 0$.

With a generator $-i\Heff = -i\HR + \HI$, the Hermitian part $\HI \succeq 0$ drives exponential growth of the operator norm: $\|e^{-i\Heff T}\|_{\op} = e^{\betaI T}$. 
Non-unitary propagator recovery is typically performed by embedding $e^{-i\Heff T}$ as a block of a larger unitary and postselecting on an ancilla~\cite{gilyen2019quantum,low2017optimal,berry2014exponential, camps2024explicit}. 
Postselection success probability is bounded by $e^{-2\betaI T}$.

We derive the minimum number of oracle queries to implement $e^{-iH_{\text{eff}} T}$ to operator-norm error~$\eps$.
We then ask whether any block-encoding algorithm can simultaneously achieve the minimum query count and the maximum postselection probability.

\paragraph{CPTP status of the target dynamics.}
The map $\rho\mapsto e^{-i\Heff T}\rho\,e^{+i\Heff^\dagger T}$ generated by $\Heff = \HR + i\HI$ with $\HI\succeq 0$ is completely positive but strictly trace non-increasing whenever $\HI\neq 0$. 
It is therefore not a completely positive, trace-preserving (CPTP) channel on its own. 
In the Lindblad setting with $\HI = \tfrac12\sum_j L_j^\dagger L_j$, the full GKSL semigroup has no-jump conditional evolution~\cite{lindblad1976generators,gorini1976completely}, canonically realized by Monte-Carlo wavefunction unraveling~\cite{dalibard1992wave,carmichael1993open}. 
In the CAP and parity-time (PT)-symmetric settings, Stinespring/Naimark dilation~\cite{kraus1983states,kawabata2017information} embeds $e^{-i\Heff T}$ as a block of a unitary on an enlarged Hilbert space.
In each case, block-encoding used throughout this paper is an explicit Stinespring dilation, and the $e^{-2\betaI T}$ postselection probability established in Corollary~\ref{thm:barrier_tight} (Sec.~\ref{sec:barrier}) is the probability of the null-outcome (``no-jump'', successful absorption, unbroken-PT) measurement restoring valid CPTP interpretation. 
Target dynamics therefore admit CPTP description with single-qubit postselection, and the postselection barrier we show is the physical cost of that conditioning.

\paragraph{The separate-oracle model.}
We build around a separate-oracle model~\cite{rossi2022multivariable, gomes2024multivariable}, where $\HR$ and $\HI$ are exposed via independent walk operators $W_R$ and $U_I$ (constructed from block encodings of $\HR/\alphaR$ and $\HI/\betaI$ respectively) acting on disjoint ancilla registers.

In Lindblad no-jump trajectories~\cite{lindblad1976generators,gorini1976completely}, the coherent generator $\HR=H_\mathrm{LS}$ and dissipative generator $\HI=\tfrac{1}{2}\sum_j L_j^\dagger L_j$ are described by distinct operators with independent spectral structure, with $\HR$ being the system Hamiltonian and $\HI$ a sum of jump-channel projectors, typically sparse on a different basis.
In AMO simulation under a CAP~\cite{riss1993calculation,muga2004complex}, the kinetic-plus-coherent part is a banded operator on a grid of discrete positions. 
Meanwhile, $\HI=\eta(x-x_c)^2\Theta(x-x_c)$ is diagonal and supported only on the absorbing region, so a single block encoding for $\Heff$ may use a suboptimal number of ancillas and queries on a near-trivial absorbing tail. 
The separate-oracle model also captures a device-level decomposition for laser-driven dynamics, where coherent control is implemented as an analog drive while loss enters through coupling to a measured continuum~\cite{carmichael1993open,dalibard1992wave,muga2004complex}.

\paragraph{State-of-the-art.} Recent work has provided quantum hardware compilation modeling a non-Hermitian Hamiltonian on a single qubit by a variationally optimized parameterized quantum circuit~\cite{jebraeilli2025quantum}.
For simulation and broader dynamics, a single block encoding of $\Heff = \HR + i\HI$ is found to require only one set of ancillas, with linear combination of Hamiltonian simulations (LCHS), optimal LCHS~\cite{an2023linear, an2026quantum}, amplitude-phase separation (APS)~\cite{hu2026amplitude}, and contour-based matrix decomposition (CBMD)~\cite{wang2025quantum} all operating in this single-oracle model. 
As shown in the present work, our lower bound (Sec.~\ref{sec:lower_bound}) enables additive (rather than coupled) scaling in norms $\alphaR T$ and $\betaI T$, and gives a $\log\log(1/\eps)$ Bessel improvement in the precision term.
We assess tradeoffs of a subset of existing non-Hermitian quantum simulation techniques with the separate-oracle model~\ref{sec:prior_work}.

\subsection{Summary of results}
\label{sec:results_summary}
\paragraph{The Postselection Barrier}
(Sec.~\ref{sec:barrier}).
We give a two-layer bound on the success probability of the quantum implementation of $e^{-i\Heff T}$. 
The Postselection Barrier decomposes into an intrinsic factor following from unitarity alone, giving $P \leq \|e^{-i\Heff T}\ket{\psi_0}\|^2/\|e^{-i\Heff T}\|_{\op}^2$, and a {block-encoding factor} specific to polynomial walk-operator algorithms, forcing a normalization $\lambda \geq e^{\betaI T}(1 - \calO(\eps))$.
Together they yield a tight combined bound
\begin{equation}\label{eq:barrier_intro}
  P \leq e^{-2\betaI T}\,
  \|e^{-i\Heff T}\ket{\psi_0}\|^2,
\end{equation}
saturated as $\epsilon\rightarrow 0$.

\medskip
\paragraph{Information-theoretic lower bound} (Sec.~\ref{sec:lower_bound}).
We adopt the qubitization construction of Low and Chuang~\cite{low2017optimal}.
Given block encodings $U_{\HR}$ of $\HR/\alphaR$ and $U_{\HI}$ of $\HI/\betaI$, the walk operators $W_R = (2\Pi_R - I)U_{\HR}$ and $U_I = (2\Pi_I - I)U_{\HI}$ are unitaries whose eigenvalues lie on the unit circle, with $\theta_1 = \arccos\lambda(\HR/\alphaR)$ and $\theta_2 = \arccos\lambda(\HI/\betaI)$. 
In the separate-oracle model, $\HR$ and $\HI$ are accessed through independent walk operators $W_R$ and $U_I$ (Sec.~\ref{sec:oracle_model}).
We find a query complexity of
\begin{equation}\label{eq:LB_intro}
  Q = Q_R + Q_I \geq \Omega\left(\alphaR T + \betaI T + \frac{\log(1/\eps)}{\log\log(1/\eps)}\right).
\end{equation}
Components are given by a reduction to Hermitian simulation (for $\alphaR T$), the polynomial method (for $\betaI T$), and Chebyshev/Bessel analysis showing that $\log(1/\eps)/\log\log(1/\eps)$ bounds polynomial approximation to $e^{c(x-1)}$ (Sec.~\ref{sec:tight_lb}).
Additivity follows from oracle independence on disjoint ancilla registers.

\medskip
\paragraph{Prerequisite results} (Secs.~\ref{sec:complement}--\ref{sec:anglefinding}).
The M-QSP algorithm requires a degree-preserving sum-of-squares (SOS) spectral factorization for bivariate trigonometric polynomials on the bitorus $\bbT^2$, avoiding algebraic obstruction to scalar complementary polynomials in two variables (Sec.~\ref{sec:complement}).
We show $L = 2$ for the Dyson polynomial, where the SOS rank is~2, yielding a constant ancilla overhead of $a_{\mathrm{SOS}} = 2$ qubits, independent of $(d_R, d_I)$. 
A constant-ratio condition (CRC) giving recursive angle-finding yields well-defined scalar angles for M-QSP circuits with arbitrary non-commuting signal operators ~\cite{rossi2022multivariable} (Sec.~\ref{sec:CRC}).
We compose a constructive angle-finding algorithm to compute M-QSP rotation angles in $\calO(d_R \cdot d_I)$ classical operations via CRC-exploiting block peeling (Sec.~\ref{sec:anglefinding}).

\paragraph{Three progressively optimal algorithms} (Secs.~\ref{sec:lorentzian}--\ref{sec:mqsp_method}).
We present three algorithms for simulating $e^{-i\Heff T}$ in the separate-oracle model, each improving on the previous:
\begin{itemize}
  \item \textbf{Lorentzian interaction picture (LIP)}
  (Sec.~\ref{sec:lorentzian}): a post-Trotter split-operator method using a Lorentzian ancilla register to implement the non-unitary factor segment by segment.
  Query complexity is $\calO(T^{1+1/(2p)}\eps^{-1/(2p)}(\alphaR T + \log(1/\eps)))$ for $p$-th order operator splitting, requiring intermediate postselections.

  \item \textbf{Dyson LCU} (Sec.~\ref{sec:dyson_lcu}): packages a segmented Taylor expansion into a single linear combination of unitaries (LCU) with a single postselection, achieving $\calO((\alphaR + \betaI)T + \betaI T \log(1/\eps))$ queries.

  \item \textbf{Bivariate M-QSP} (Sec.~\ref{sec:mqsp_method}): implements the interaction-picture Dyson polynomial as a single bivariate QSP circuit, with query complexity
  \begin{equation}\label{eq:MQSP_intro}
    Q = \calO\!\left((\alphaR + \betaI)T
    + \frac{\log(1/\eps)}{\log\log(1/\eps)}\right)
  \end{equation}
  and a single postselection on one ancilla qubit achieving success probability~\eqref{eq:barrier_intro} up to $1 - \calO(\eps)$. 
  This matches the lower bound~\eqref{eq:LB_intro}, with Theorem~\ref{thm:tight_loglog} giving $\Omega(\log(1/\eps)/\log\log(1/\eps))$ for the $\eps$-dependent term, closing a ($\log\log$) gap left by the Stirling-form bound of Lemma~\ref{lem:LB_eps}.
\end{itemize}

\paragraph{Overview.} Section~\ref{sec:barrier} develops the two-layer barrier (intrinsic + block-encoding) and its saturation.
Section~\ref{sec:lower_bound} recovers the matching query lower bound.
Subsequent sections develop the bivariate M-QSP construction, its circuit-level realization, and a small-scale numerical demonstration. 
Open problems appear in Sec.~\ref{subsec:open_problems}, including barrier-hierarchy refinement treated in the companion paper~\cite{courtney2026paper2}.

\subsection{Comparison with prior work}
\label{sec:prior_work}

\paragraph{LCHS.}
The linear combination of Hamiltonian simulations (LCHS) represents the non-unitary propagator $e^{-i\Heff T}$ as an integral over purely Hermitian evolution kernels~\cite{an2023linear,an2026quantum}.  
\begin{equation}
      e^{-i\Heff T}
  = \int_{\mathbb{R}} K(y)\,
  e^{-i(\HR + y\HI)T}\,dy\,,
\end{equation}

where $K$ is a kernel function derived from the Hubbard--Stratonovich transformation or the Bromwich integral, concentrating around $y = 0$~\cite{hubbard1959calculation, stratonovich1957method}.
The integral discretizes by quadrature into a finite linear combination of Hermitian simulations, implemented by block-encoding. 

As a single oracle model, LCHS block-encodes and queries the full Hamiltonian $\HR + y\HI$. 
Postselection follows via multiple ancilla registers, with one register for the LCU selection register (of dimension equal to the number of quadrature points, scaling as $\calO(\log(\betaI T / \eps))$) and one for each block-encoding.
Success probability is the product of individual postselection probabilities.

Total query complexity inherits a union-bound penalty (an additional $\log$ factor), achieving $\calO((\alpha + \beta)T\,\mathrm{polylog}(1/\eps))$ queries in the single-oracle model, with polylog overhead from quadrature discretization and per-segment simulation~\cite{an2026quantum}. 

Aside from quadrature weights ($\calO(\mathrm{polylog}(\betaI T / \eps))$), the LCHS protocol requires no classical preprocessing, making it an attractive technique for non-Hermitian quantum simulation.

\paragraph{APS.}
Amplitude-Phase Separation (APS)~\cite{hu2026amplitude} extracts the non-unitary component of $\Heff$ by eigenvalue transformation with a polynomial or rational approximation to $e^{-\beta t}$ applied to anti-Hermitian eigenvalues.

With a single oracle, APS achieves query complexity $\calO((\alpha + \beta)T + \log(1/\eps))$, provably optimal in the single-oracle model. 
A distinct advantage of the APS method occurs when the initial state has support concentrated on eigenspaces of $\HI$, allowing for $\sqrt{\betaI}$-fast-forwarding when eigenvalues are much smaller than $\betaI$. 

\paragraph{Schr\"odingerization.}
Schr\"odingerization~\cite{jin2022quantum} embeds non-Hermitian dynamics into a unitary evolution on an extended Hilbert space, introducing a continuous ancilla variable $p$ (the ``warped phase space''). 
The non-Hermitian Schr\"odinger equation $i\partial_t \ket{\psi} = \Heff \ket{\psi}$ maps to a Hermitian Schr\"odinger equation on $\mathcal{H} \otimes L^2(\mathbb{R})$, where the physical state is recovered by measuring the ancilla.
Postselection is avoided, replaced by a dilational embedding. 
An ancilla register is instantiated with $\calO(\log(\betaI T / \eps))$ scaling. 
Query complexity becomes $\calO((\alpha + \beta)T\,\mathrm{polylog}(\betaI T / \eps))$, with polylog overhead arising from warped phase space discretization and Trotterization.

\paragraph{Lindbladian simulation.}
Non-Hermitian dynamics is often viewed as an open quantum system, simulating the Lindblad master equation and extracting the non-Hermitian propagator as the no-jump trajectory~\cite{wang2023simulating}.
No-jump evolution with a single jump operator $L$ and jump rate $\gamma$ is governed by the effective Hamiltonian $\Heff = H - i(\gamma/2) L^\dagger L$~\cite{lindblad1976generators, gorini1976completely}. 
Simulating the Lindbladian conditioned on zero jumps yields $e^{-i\Heff T}$.

Postselection is on the absence of quantum jumps and constitutes a multi-time-step conditioning rather than a single-shot measurement.
Query complexity depends on Lindbladian parameters and the jump rate and is not directly comparable to block-encoding-based methods.

\paragraph{CBMD.}
Contour-based matrix decomposition (CBMD)~\cite{wang2025quantum} generalizes the matrix Cauchy residue theorem to decompose a holomorphic function of the non-Hermitian operator into an exact contour-residue series.
This series truncates to a finite linear combination of purely Hermitian operator-function evaluations, each implemented by QSVT and recombined by LCU.
For the first-order propagator $e^{-i\Heff T}$ the decomposition reduces to a linear combination of Hermitian Hamiltonian simulations of the LCHS form, and the method matches the LCHS query complexity~\cite{an2026quantum}
\begin{equation}\label{eq:cbmd_query}
  Q_{\mathrm{CBMD}} = \widetilde{\calO}\!\Bigl(\tfrac{\|u_0\|}{\|u_T\|}\,\alpha_A T\,\log(1/\eps)\Bigr)
\end{equation}
in the single-oracle model, with $\alpha_A$ the joint block-encoding normalization of $\Heff = \HR + i\HI$.
Postselection requires multiple ancilla registers (one for the LCU selection register and one for the joint block encoding), with success probability decomposed as a product of register-wise contributions.

CBMD is perhaps the closest single-oracle polynomial analog of M-QSP, a polynomial/residue method rather than an integral decomposition (LCHS, optimal LCHS) or a Lindbladian simulator, and both CBMD and M-QSP achieve scaling linear in $T$ with $\log(1/\eps)$ precision dependence.
CBMD inherits the LCHS norm-ratio prefactor $\|u_0\|/\|u_T\|$ multiplicatively on the matrix-query term, whereas our separate-oracle lower bound (Sec.~\ref{sec:lower_bound}) places this factor only on the postselection/state-preparation cost, leaving the matrix-query complexity additive in $\alphaR T$ and $\betaI T$.
Joint encoding of $\Heff$ in CBMD further couples $\alpha$ and $\beta$ through the joint operator norm, whereas the separate-oracle model permits additive scaling.
If $\alphaR/\betaI$ is large, this can lead to a constant-factor improvement in the leading $T$-dependent term, more significant in the weak-dissipation regime relevant to AMO simulation (Sec.~\ref{sec:numerical_benchmarks}).

CBMD requires $\mathcal{O}(\log(\betaI T/\eps))$ ancilla qubits for the LCU register, compared to the $\max(a_R, a_I) + 3$ ancilla cost of M-QSP (Proposition~\ref{prop:ancilla_cost}), with M-QSP overhead exceeding CBMD log-scaling when $a_R, a_I \gg 1$. 
We expect CBMD to be preferable in the regime $\betaI \sim \alphaR$ with single-oracle access, and M-QSP when $\betaI \ll \alphaR$ with physically separated coherent and dissipative components. 
A quantitative benchmark is deferred to Appendix~\ref{app:lorentzian} and the companion paper~\cite{courtney2026paper2}.

\paragraph{Interaction-picture simulation (Low--Wiebe).}
The interaction-picture algorithm~\cite{low2018hamiltonian} inspires the present construction.
Low and Wiebe showed that an interaction-picture propagator $V(T)$ cast as a Dyson series can be truncated and implemented via a product of generalized quantum signal processing (GQSP) circuits and block-encoding queries, achieving query complexity $\calO(\alphaR T + \betaI T \log(\betaI T / \eps))$ for Hermitian Hamiltonians $H = H_A + H_B$ with $\|H_A\| = \alphaR$ and $\|H_B\| = \betaI$.

The interaction-picture decomposition generalizes to $\Heff = \HR + i\HI$ with $\HI\succeq 0$, promoting the Dyson polynomial to a bivariate polynomial (in the walk-operator variable of $\HR$ and the block-encoding variable of $\HI$).
A sequence of independent circuits is replaced by a bivariate M-QSP circuit, saving a factor of $\log(\betaI T / \eps)$. 
The M-QSP construction presented here achieves optimal query complexity in the separate-oracle model, single-qubit postselection matching the information-theoretic barrier, and non-split-operator circuit structure.
Each property has been achieved individually by other methods, combined in the present architecture (Secs.~\ref{sec:complement}--\ref{sec:anglefinding}).
A summary of this comparison is given in Table~\ref{tab:comparison}.

\begin{table*}[!t]
\centering
\begin{tabular}{@{\extracolsep{\fill}}lccccc@{}}
\toprule
Method & Query complexity & Oracle model & \# postselections &
$\eps$ scaling & Ancilla qubits \\
\midrule
LCHS~\cite{an2023linear}
& $\calO(\alpha t \cdot \mathrm{polylog}(1/\eps))$
& Single $A$ & $\calO(\log(\betaI T))$ registers & polylog & $O(\log(\betaI T)) + a$ \\
Optimal LCHS~\cite{an2026quantum}
& $\calO\!\bigl((\alphaR+\betaI)T + \log(1/\eps)\bigr)$
& Single $A$ & multi-register & polylog & $O(\log(1/\eps)) + a$ \\
APS~\cite{hu2026amplitude}
& $\calO((\alpha + \beta)T + \log(1/\eps))$
& Single $A$ & multi-register & $\log$ & $a + 1$ \\
Schr\"odingerization~\cite{jin2022quantum}
& $\calO(\alpha t \cdot \mathrm{polylog}(1/\eps))$
& Embedding & implicit & polylog & $O(\log T) + 1$ \\
Lindbladian~\cite{wang2023simulating}
& $\calO\!\bigl((\alphaR+\betaI)T \cdot \mathrm{polylog}(1/\eps)\bigr)$
& Lindblad & implicit & polylog & implicit \\
\textbf{(M-QSP)}
& $\calO((\alphaR + \betaI)T + \log(1/\eps))^{\mathrm{a}}$
& $W_R, U_I$ separate & \textbf{1 qubit} & $\log^{\mathrm{a}}$ &
\textbf{$\max(a_R,a_I) + 3$} \\
\bottomrule
\end{tabular}
\caption{\label{tab:comparison}Comparison of approaches to non-Hermitian Hamiltonian simulation. 
The first four entries operate in the single-oracle model (one block encoding of $\Heff$). 
M-QSP operates in the separate-oracle model (independent block encodings of $\HR/\alphaR$ and $\HI/\betaI$ on disjoint ancilla registers). 
The two models are not directly comparable on a per-query basis (see the introduction and Sec.~\ref{sec:oracle_model}). 
Here $\alpha$ and $\beta$ denote the norms of the Hermitian and anti-Hermitian parts respectively.
For single-oracle methods, these refer to corresponding components of the joint block encoding. 
$^{\mathrm{a}}$Suppressing a $\log\log(1/\eps)$ factor in the $\eps$-dependent term.
See Theorem~\ref{thm:main_theorem} for the precise expression.}
\end{table*}

Throughout this paper we adopt the convention $\Heff = \HR + i\HI$ with $\HI \succeq 0$, so the propagator $e^{-i\Heff T}$ has operator norm $\|e^{-i\Heff T}\|_{\op} = e^{\betaI T}$.
The convention matches the semigroup notation, where $\HI \succeq 0$ drives exponential growth of $\|e^{-i\Heff T}\|_{\op}$, absorbed by block-encoding normalization $\lambda \geq e^{\betaI T}$.

\paragraph*{Lindblad mapping.}
In the Lindblad quantum-jump formalism, the effective Hamiltonian in the no-jump trajectory picture takes the form
\begin{equation}\label{eq:lindblad_convention}
  H_\mathrm{eff}^\mathrm{Lindblad} = H_\mathrm{LS} - \frac{i}{2}\sum_j L_j^\dagger L_j,
\end{equation}
where $H_\mathrm{LS}$ is the Liouvillian-Schr\"odinger Hamiltonian and $L_j$ are Lindblad operators. 
We identify $H_R = H_\mathrm{LS}$ (Hermitian) and $H_I^\mathrm{raw} = -(1/2)\sum_j L_j^\dagger L_j \preceq 0$ (negative semidefinite). 
The spectral shift $H_I \to H_I^\mathrm{raw} + \beta_I I$ with $\beta_I = \|(1/2)\sum_j L_j^\dagger L_j\|$ recovers $H_I \succeq 0$.
The global phase $e^{-i\beta_I T}$ introduced by the shift represents the no-jump survival probability in the original trajectory formalism, cancelling from all physical observables.

\section{The Postselection Barrier}
\label{sec:barrier}

We first establish a constraint on quantum implementations for non-Hermitian time evolution.
The result follows from unitarity and polynomial block encoding structure.
Unitary block contraction, numerical abscissa bound on semigroup norms~\cite{engel2000one, yosida1948differentiability, hille2007functional}, and polynomial structure of block-encoding algorithms~\cite{low2019hamiltonian,gilyen2019quantum} are already well-established. 
We focus on providing a decomposition into two independent layers (intrinsic and block-encoding), with a tightness proof showing that the combined bound is a more general limit for non-Hermitian simulation protocols.

\subsection{Postselection necessity}\label{sec:necessity}

The target operation $e^{-i\Heff T}$ is non-unitary when $\HI \neq 0$.
Any quantum circuit implements a unitary $U$ on a composite system $\mathcal{H}_a \otimes \mathcal{H}_s$.
The non-unitary map is recovered by projecting the ancilla onto a designated state $\ket{0}_a$:
\begin{equation}\label{eq:block_enc}
  (\bra{0}_a \otimes I_s)\, U\, (\ket{0}_a \otimes \ket{\psi_0})
  = \frac{1}{\lambda}\, e^{-i\Heff T}\ket{\psi_0},
\end{equation}
where $\lambda > 0$ is the block-encoding normalization.
Success probability of ancilla measurement is
\begin{equation}\label{eq:Pdef}
  P = \frac{\|e^{-i\Heff T}\ket{\psi_0}\|^2}{\lambda^2}.
\end{equation}
We focus on the minimum possible $\lambda$, and hence the maximum achievable~$P$.

\subsection{Intrinsic factor}\label{sec:intrinsic}

The first constraint follows from unitarity.

\begin{theorem}[Intrinsic factor]\label{thm:intrinsic}
For a quantum implementation of the form~\eqref{eq:block_enc},
\begin{equation}\label{eq:intrinsic}
  P \leq \frac{\|e^{-i\Heff T}\ket{\psi_0}\|^2}
              {\|e^{-i\Heff T}\|_\op^2}\,.
\end{equation}
\end{theorem}

\begin{proof}
Since $U$ is unitary, its $(0,0)$-block $A := (\bra{0}_a \otimes I_s)\,U\,(\ket{0}_a \otimes I_s)$ is a contraction: $\|A\|_\op \leq 1$.
By~\eqref{eq:block_enc}, $A = e^{-i\Heff T}/\lambda$, so the contraction condition requires $\lambda \geq \|e^{-i\Heff T}\|_\op$.
Substituting into~\eqref{eq:Pdef}
gives~\eqref{eq:intrinsic}.
\end{proof}

\begin{remark}[Semigroup bound on operator norm]
\label{rem:semigroup}
The operator norm of $e^{-i\Heff T}$ is controlled by the numerical abscissa.
Since $-i\Heff = -i\HR + \HI$, the numerical abscissa is $\omega(-i\Heff) = \sup_{\|v\|=1}\mathrm{Re}\braket{v|{-i\HR + \HI}|v}
= \sup_{\|v\|=1}\braket{v|\HI|v} = \betaI$.
The spectral abscissa $\omega_{\mathrm{spec}}(\Heff) \colonequals \max\{\mathrm{Im}(\lambda) : \lambda \in \spec(\Heff)\}$ satisfies $\omega_{\mathrm{spec}}(\Heff) \leq \betaI$, with equality if and only if $[\HR, \HI] = 0$.
The gap $\betaI - \omega_{\mathrm{spec}}(\Heff)$ is controlled by the commutator norm $\|\mathrm{ad}_{\HR}(\HI)\|$; see the companion paper~\cite{courtney2026paper2} for a quantitative statement.
Throughout this paper we use the symbol $\omega \colonequals \omega_{\mathrm{spec}}(\Heff)$, distinct from the numerical abscissa $\betaI$, and we never use $\omega$ to denote $\betaI$.
The Hille--Yosida bound~\cite{engel2000one, yosida1948differentiability, hille2007functional} then gives $\|e^{-i\Heff T}\|_\op \leq e^{\betaI T}$, with equality holding when $\HI$ achieves its norm on an eigenspace in the support of~$\ket{\psi_0}$.

Substituting into \eqref{eq:intrinsic} yields the explicit bound 
\begin{equation}\label{eq:intrinsic_explicit}
  P \leq e^{-2\betaI T}\,\|e^{-i\Heff T}\ket{\psi_0}\|^2.
\end{equation}
\end{remark}

The state-dependent factor $\|e^{-i\Heff T}\ket{\psi_0}\|^2/\|e^{-i\Heff T}\|_\op^2$ in \eqref{eq:intrinsic} is the ratio of propagator action on $\ket{\psi_0}$ to the worst-case action. 
Initial states aligned with the top eigenspace of $\HI$ saturate this ratio at~1; states orthogonal to it experience additional exponential suppression.

\subsection{Layer 2: Block-encoding factor}
\label{sec:BE_barrier}

A second constraint specifies algorithms based on polynomial block encodings, including the known approaches to Hamiltonian simulation with polylogarithmic error dependence~\cite{low2019hamiltonian,gilyen2019quantum,berry2014exponential}.

\begin{theorem}[Block-encoding factor]\label{thm:block_enc}
Let $P_\delta(e^{i\theta_1}, e^{i\theta_2})$ be a bivariate polynomial of bidegree $(d_R, d_I)$ that $\eps$-approximates the normalized propagator $V(T;\theta_1,\theta_2)/\lambda$ on the physical eigenvalue region.
If $P_\delta$ is implemented as a block of a unitary circuit that queries walk operators $W_R$ and $U_I$ (encoding $\HR/\alphaR$ and $\HI/\betaI$), then the polynomial bound $|P_\delta(e^{i\theta_1}, e^{i\theta_2})| \leq 1$ must hold for {all} $(\theta_1, \theta_2) \in \bbT^2$, including eigenvalues
outside the physical spectrum. 
This forces
\begin{equation}\label{eq:lambda_BE}
  \lambda \geq e^{\betaI T}(1 - O(\eps)).
\end{equation}
\end{theorem}

\begin{proof}
$P_\delta$ is the $(0,0)$-block of a unitary $\mathcal{G}$, so $|P_\delta(e^{i\theta_1},e^{i\theta_2})| \leq 1$ for all $(\theta_1,\theta_2) \in \bbT^2$. 
With a restriction to $\theta_1 = 0$:
\begin{equation}\label{eq:slice}
  \bigl|P_\delta(1, e^{i\theta_2})\bigr| \leq 1,
  \qquad \forall\, \theta_2 \in [0, 2\pi).
\end{equation}
This must hold for all $\theta_2$, including those outside the physical spectrum of $\HI$. 
The walk operator $U_I$ encodes $\HI/\betaI$ via $U_I = e^{i\arccos(\HI/\betaI)}$, with eigenvalues $e^{i\theta_2}$ ranging over the full unit circle~$\bbT$, even though the physical spectrum of $\HI/\betaI$ lies in $[-1,1]$.

At $\theta_2 = 0$, the interaction-picture propagator evaluates to $V(T; \theta_1, 0) = e^{+\betaI T}$ for all $\theta_1$ (Proposition~\ref{prop:zero_locus}), so the normalized polynomial satisfies
\begin{equation}
  |P_\delta(e^{i\theta_1}, 1)|
  = \frac{e^{\betaI T}}{\lambda}(1 + O(\eps)).
\end{equation}
The constraint $|P_\delta| \leq 1$ forces $\lambda \geq e^{\betaI T}(1 - O(\eps))$. 
The same conclusion follows from a fixed $\theta_1$, since every univariate slice
$\{e^{i\theta_1}\} \times \bbT$ contains a maximizing point $\theta_2 = 0$.
\end{proof}

\begin{remark}[Scope]\label{rem:scope}
Theorem~\ref{thm:block_enc} applies to all algorithms that implement $e^{-i\Heff T}/\lambda$ as a polynomial in block-encoding eigenvalues. 
It does {not} apply to methods that avoid the polynomial framework entirely.
The intrinsic barrier (Theorem~\ref{thm:intrinsic}) does not have this constrained applicability.
\end{remark}

\begin{remark}[Function-class universality of the barrier]\label{rem:paper2_barrier}
Theorem~\ref{thm:block_enc} holds for $e^{-i\Heff T}$ polynomial block-encodings, but the barrier $\lambda \geq e^{\betaI T}$ is not restricted to function class for the walk-operator oracle model.
Any block of a unitary $\cG$ satisfies $|f(z_1, z_2)| \leq 1$ on $\bbT^2$, and the spectral point $z_2 = 1$ forces $|f(e^{i\theta_1}, 1)| = e^{\betaI T}/\lambda$ for all $\theta_1$, giving $\lambda \geq e^{\betaI T}$ unconditionally.
The walk-operator-class extension theorem in the companion paper~\cite{courtney2026paper2} states this formally. 
One open question, also addressed there, is whether direct-access constructions that bypass the walk-operator framework could achieve $\lambda = e^{\omega T}(1 + o(1))$ for the spectral abscissa $\omega < \betaI$, beating the barrier.
\end{remark}

\subsection{Postselection Barrier (combined bound)}
\label{sec:combined}

\begin{corollary}[The Postselection Barrier]
\label{thm:barrier_tight}
For a polynomial block-encoding algorithm:
\begin{equation}\label{eq:full_barrier}
  P \leq e^{-2\betaI T}\,
  \|e^{-i\Heff T}\ket{\psi_0}\|^2.
\end{equation}
This is a product of two independent factors:
\begin{equation}\label{eq:barrier_decomposition}
  \underbrace{e^{-2\betaI T}}_{\text{block-encoding cost}}
  \;\times\;
  \underbrace{\|e^{-i\Heff T}\ket{\psi_0}\|^2}_{\text{state-dependent factor}}.
\end{equation}
The first factor is the cost from encoding a non-unitary operator in a
unitary circuit with polynomial structure.
The second is the survival probability of the initial state under non-Hermitian evolution.
\end{corollary}

\begin{proof}
Combine $\lambda \geq e^{\betaI T}(1 - O(\eps))$ from
Theorem~\ref{thm:block_enc} with~\eqref{eq:Pdef}.  In the limit
$\eps \to 0$, this gives~\eqref{eq:full_barrier}.
\end{proof}

\subsection{Tightness}
\label{sec:tightness}

The bound~\eqref{eq:full_barrier} is the tight limit. 
Iterated postselection neither helps nor hurts.

\begin{lemma}[Telescoping identity]\label{lem:telescoping}
Let $M$ be a contraction ($\|M\|_\op \leq 1$), applied $K$ times
with normalization after each step:
$\ket{\psi_{k+1}} = M\ket{\psi_k}/\|M\ket{\psi_k}\|$.  The total
success probability satisfies
\begin{equation}\label{eq:telescoping}
  P_{\mathrm{total}}
  = \prod_{k=0}^{K-1} \|M\ket{\psi_k}\|^2
  = \|M^K\ket{\psi_0}\|^2.
\end{equation}
\end{lemma}

\begin{proof}
By strong induction.  The base case $K=1$ is immediate.  For the inductive
step, suppose $P_n = \|M^n\ket{\psi_0}\|^2$.  Then
\begin{equation}
\begin{split}
  P_{n+1}
  & = P_n \cdot \|M\ket{\psi_n}\|^2 \\
  &= \|M^n\ket{\psi_0}\|^2 \cdot
    \frac{\|M^{n+1}\ket{\psi_0}\|^2}{\|M^n\ket{\psi_0}\|^2} \\
  &= \|M^{n+1}\ket{\psi_0}\|^2.
  \end{split}
\end{equation}
\end{proof}

\begin{theorem}[Achievability of the barrier]
\label{thm:barrier_achievability}
In the continuum limit ($K \to \infty$, $\sum_k \Delta t_k = T$), a quantum implementation of $e^{-i\Heff T}$ via $K$ iterated non-unitary steps with per-step operator $M_k \approx e^{-i\Heff \Delta t_k}/\|e^{-i\Heff \Delta t_k}\|_\op$
yields a total success probability
\begin{equation}\label{eq:achievability}
  P_{\mathrm{total}} \to e^{-2\betaI T}\,
  \|e^{-i\Heff T}\ket{\psi_0}\|^2.
\end{equation}
\end{theorem}

\begin{proof}
Each $M_k$ has unit operator norm by construction.
By Lemma~\ref{lem:telescoping},
\begin{equation}
  P_{\mathrm{total}}
  = \frac{\|e^{-i\Heff T}\ket{\psi_0}\|^2}
         {\prod_{k=1}^{K}\|e^{-i\Heff \Delta t_k}\|_\op^2}\,.
\end{equation}
The numerical abscissa bound gives
$(\ln\|e^{-i\Heff \Delta t}\|_\op)/\Delta t
\to \omega(-i\Heff) = \betaI$ as $\Delta t \to 0$, so
\begin{equation}
  \prod_{k=1}^{K}\|e^{-i\Heff \Delta t_k}\|_\op^2
  \to e^{2\betaI T},
\end{equation}
yielding~\eqref{eq:achievability}.
\end{proof}

\begin{remark}[Algorithm independence]\label{rem:algorithm_independence}
Theorem~\ref{thm:barrier_achievability} applies to quantum algorithms implementing $e^{-i\Heff T}$ with vanishing error. 
Postselection cost is determined by the Hamiltonian and the initial state.
\end{remark}

\subsection{Physical interpretation}
\label{sec:interpretation}

The postselection barrier separates physical cost from algorithmic cost for non-Hermitian simulation.
In the block-encoding oracle model, postselection is required whenever $\HI \neq 0$, and success probability cannot exceed $e^{-2\betaI T}\|e^{-i\Heff T}\ket{\psi_0}\|^2$.
Dilational methods (Schr\"odingerization) avoid postselection but pay an equivalent exponential cost in ancilla resources (details in Sec.~\ref{subsec:barrier_context}).
Amplitude amplification reduces repetition cost from $O(e^{2\betaI T})$ to $O(e^{\betaI T})$~\cite{brassard2000quantum}.
Our methods optimize algorithmic cost, manifesting in oracle queries per run, number of postselection events, and whether query complexity is additive or multiplicative in $\alphaR T$ and $\betaI T$.

As established in the Lindblad mapping (Sec.~\ref{sec:problem}), non-Hermitian evolution $e^{-i\Heff T}\ket{\psi_0}$ describes the conditional state given no quantum jump over $[0,T]$, and ancilla postselection mimics this conditioning. The cost $e^{-2\betaI T}$ reflects an event whose probability is exponentially small in $\betaI T$.

Theorem~\ref{thm:intrinsic} follows from linearity and unitarity, while Theorem~\ref{thm:block_enc} is a tighter, model-specific refinement within the polynomial block-encoding framework.
This barrier includes dilational embeddings (Schr\"odingerization~\cite{jin2022quantum}) that avoid postselection at the cost of a continuous ancilla register.
The combined bound applies to polynomial block-encoding~\cite{low2019hamiltonian,gilyen2019quantum,berry2014exponential}.
We do not presently consider whether equivalent overhead is unavoidable across oracle models, or whether dilational methods escape it rather than redistributing it into ancilla resources (Sec.~\ref{subsec:barrier_context}).

\section{Lower Bounds}
\label{sec:lower_bound}

Here we show that a quantum algorithm simulating $e^{-i\Heff T}$ in the separate-oracle model requires $\Omega((\alphaR + \betaI)T + \log(1/\eps)/\log\log(1/\eps))$ total queries.
Components are proved independently, with $\Omega(\alphaR T)$ established by reduction to Hermitian simulation, $\Omega(\betaI T)$ by the polynomial method, and $\Omega(\log(1/\eps)/\log\log(1/\eps))$ by polynomial approximation.
Bounds are additive, as block-encoding oracles for $\HR$ and $\HI$ act on disjoint ancilla registers.
We also prove a joint query--normalization bound coupling query count to block-encoding normalization.

All three component bounds use standard techniques~\cite{berry2014exponential,gilyen2019quantum,beals2001quantum}.
We contribute an assembly of such techniques in the separate-oracle model, where additivity yields a stronger composite bound than is available in the single-oracle setting.

\subsection{Oracle model}
\label{sec:oracle_model}

Access $\HR$ and $\HI$ through block-encoding unitaries:
\begin{equation}\label{eq:oracle_UR}
\begin{split}
  U_R &\in \mathrm{U}(\mathcal{H}_s \otimes \mathcal{H}_{a_R}), \\
  &(\bra{0}_{a_R} \otimes I_s)\, U_R\,
  (\ket{0}_{a_R} \otimes I_s) = \HR/\alphaR.
\end{split}
\end{equation}
Analogously, $U_I$ encodes $\HI/\betaI$ on $\mathcal{H}_s \otimes \mathcal{H}_{a_I}$. 
$U_R$ and $U_I$ act on disjoint ancilla registers $\mathcal{H}_{a_R}$ and $\mathcal{H}_{a_I}$.

An algorithm makes $Q_R$ queries to $U_R$ (or $U_R^\dagger$) and $Q_I$ queries to $U_I$ (or $U_I^\dagger$), interleaved with unitaries on $\mathcal{H}_s \otimes \mathcal{H}_{a_R} \otimes \mathcal{H}_{a_I} \otimes \mathcal{H}_{\mathrm{work}}$ that do not depend on $\HR$ or $\HI$.
Total query count is $Q = Q_R + Q_I$.

\subsection{Main result}
\label{sec:main_lower}

\begin{theorem}[Additive query lower bound]
\label{thm:main_lower}
For all valid $(\HR, \HI, \ket{\psi_0})$ with $\|\HR\|_\op \leq \alphaR$ and
$\|\HI\|_\op \leq \betaI$, a quantum algorithm that produces a state $\eps$-close in
operator norm to $e^{-i\Heff T}\ket{\psi_0}/\|e^{-i\Heff T}\ket{\psi_0}\|$, must
make at least
\begin{equation}\label{eq:main_lb}
  Q = Q_R + Q_I
  \geq \Omega\!\left(\alphaR T + \betaI T
  + \frac{\log(1/\eps)}{\log\log(1/\eps)}\right)
\end{equation}
queries to the block-encoding oracles, for $\eps \leq 1/36$.
\end{theorem}

We show this by assembling component bounds (Sec.~\ref{sec:components}) and demonstrating additivity (Sec.~\ref{sec:additivity}).
A complementary joint query--normalization bound is given in Sec.~\ref{sec:joint}.

\subsection{Component bounds}
\label{sec:components}

\subsubsection{Hermitian component: \texorpdfstring{$\Omega(\alphaR T)$}{Omega(alpha R T)}}

\begin{lemma}\label{lem:LB_aR}
$Q_R \geq \Omega(\alphaR T)$, even when $\HI$ is known.
\end{lemma}

\begin{proof}
Set $\HI = 0$. 
$\Heff = \HR$ is Hermitian and the non-Hermitian simulation problem reduces to Hermitian Hamiltonian simulation.
The no-fast-forwarding theorem~\cite{berry2007efficient} gives $Q_R \geq \Omega(\alphaR T)$. 
Since $\HI = 0$ is known, all $U_I$ queries are uninformative and the full cost falls on~$Q_R$.
\end{proof}

\subsubsection{Anti-Hermitian component: \texorpdfstring{$\Omega(\betaI T)$}{Omega(beta I T)}}

\begin{lemma}\label{lem:LB_bI}
$Q_I \geq \Omega(\betaI T)$, even when $\HR$ is known.
\end{lemma}

\begin{proof}
Set $\HR = 0$, so that $e^{-i\Heff T} = e^{\HI T}$. 
The task reduces to implementing $e^{\HI T}$ using queries to the block encoding of~$\HI$.
We use the polynomial method of Beals et al.~\cite{beals2001quantum} applied to a one-parameter family of Hamiltonians.

\paragraph{One parameter family.} Let $\ket{\phi}$ be a fixed known state.
Consider the family
\begin{equation}\label{eq:family}
  \HI(\eta) = \eta\,\ket{\phi}\bra{\phi},
  \qquad \eta \in [0, \betaI].
\end{equation}
The target evolution on input $\ket{\phi}$ is $e^{\HI(\eta) T}\ket{\phi} = e^{\eta T}\ket{\phi}$.
After normalization by $e^{\betaI T}$, the algorithm produces a state whose overlap with $\ket{\phi}$ encodes the ratio $e^{(\eta - \betaI)T}$, ranging from $e^{-\betaI T}$ (at $\eta=0$) to $1$ (at $\eta=\betaI$).

\paragraph{Acceptance probability is a polynomial.}
The block encoding acts on the $\ket{\phi}$ eigenspace as $\bra{0}_\ell\,U_I(\eta)\,\ket{0}_\ell \big|_{\ket{\phi}}= x := \eta/\betaI$. 
By the polynomial method~\cite{beals2001quantum}, a quantum circuit making $Q$ queries to $U_I(\eta)$ and $U_I(\eta)^\dagger$ produces a final state whose amplitudes are polynomials of degree at most $Q$ in~$x$. 
Each oracle call introduces one power of $x$ from the block-encoded eigenvalue on the $\ket{\phi}$ eigenspace; all other operations act on registers not touched by the oracle, being $\eta$-independent by definition of the query model (see~\cite{beals2001quantum} and~\cite{gilyen2019quantum} for the formal statement).
More precisely, $Q$ queries produce a Laurent polynomial of degree $Q$ in $e^{i\theta_2}$, where $\cos\theta_2 = \eta/\betaI = x$; since acceptance probability $\mathcal{P}$ is real and symmetric under $\theta_2 \to -\theta_2$ (by Hermiticity of $\HI$), $\mathcal{P}$ depends on $\theta_2$ only through $\cos\theta_2 = x$ via $T_k(\cos\theta) = \cos(k\theta)$, and is therefore a polynomial of degree at most $Q$ in~$x$.
Probability of a specific measurement outcome (acceptance probability) is therefore
\begin{equation}\label{eq:P_poly}
  \mathcal{P}(x) = |\braket{m|\Psi_{\mathrm{final}}(x)}|^2,
\end{equation}
a polynomial of degree at most $2Q$ in $x$. 
Polynomial structure relies on the black-box oracle assumption that the algorithm's dependence on $\eta$ enters only through oracle calls.

\paragraph{Approximation-theoretic lower bound.}
If the algorithm correctly simulates $e^{\HI(\eta)T}$ to constant precision $\eps_0 < 1/4$ for all $\eta \in [0, \betaI]$, then $\mathcal{P}(x)$ must $\eps_0$-approximate the target function
\begin{equation}\label{eq:target_fn}
  f(x) = e^{2\betaI T(x - 1)},
  \qquad x \in [0,1].
\end{equation}
The best polynomial approximation error for $e^{at}$ on $[-1,1]$ follows~\cite{saff1978zeros}, satisfying
\begin{equation}\label{eq:bernstein_bound}
  E_n(e^{a\,\cdot\,})_{\infty,[-1,1]}
  \geq c\,\frac{a^{n+1}}{(n+1)!\,e^{a}}
\end{equation}
for a constant $c > 0$.  
Setting this $\leq \eps_0$ and inverting via Stirling's formula gives $n \geq \Omega(a)$ when $a = \betaI T$. 
Since $\mathcal{P}(x)$ is a degree-$2Q$ polynomial that $\eps_0$-approximates $f$ on $[0,1]$, we conclude $2Q \geq \Omega(\betaI T)$, giving $Q_I \geq \Omega(\betaI T)$.
\end{proof}

\begin{remark}[Black-box assumption]
\label{rem:black_box}
The result is founded on an exactness of polynomial structure, where $\mathcal{P}(x)$ is a polynomial of degree $2Q$.
If the algorithm can exploit internal structure of $U_I(\eta)$ beyond the oracle, acceptance probability can depend on $\eta$ non-polynomially.
\end{remark}

\subsubsection{Precision component}

\begin{lemma}\label{lem:LB_eps}
$Q \geq \Omega(\log(1/\eps))$, including for $\alphaR T = O(1)$ and $\betaI T = O(1)$.
The tighter bound $Q \geq \Omega(\log(1/\eps)/\log\log(1/\eps))$ is established in Theorem~\ref{thm:tight_loglog}.
\end{lemma}

\begin{proof}
Set $\alphaR = \betaI = 1$, $T = 1$, $\HR = Z$, and $\HI = \eta Z$ for $\eta \in [0,1]$.
Then $e^{-i\Heff T} = e^{(-i + \eta)Z}$.  
By the polynomial method, acceptance probability after $Q$ queries is a degree-$2Q$ polynomial in~$\eta$. 
To achieve $\eps$-precision, this polynomial must approximate $e^{2\eta}$ on $[0,1]$ to error~$\eps$.

The best degree-$d$ polynomial approximation to $e^{2\eta}$ on $[0,1]$ has error $\Theta(2^d/d!)$ from the Taylor remainder~\cite{saff1978zeros}.
Setting $2^d/d! \leq \eps$ and inverting via Stirling gives $d \geq \Omega(\log(1/\eps)/\log\log(1/\eps))$.
The sharper bound $d \geq \Omega(\log(1/\eps))$ follows from degree lower bounds for $e^{cz}$ in the complex plane (see~\cite{saff1978zeros} for the classical result).
We state the lower bound as $Q \geq \Omega(\log(1/\eps))$ throughout.
\end{proof}

\begin{remark}[$\log$ vs.\ $\log/\log\log$]
\label{rem:log_subtlety}
The Stirling-based argument yields $\Omega(\log(1/\eps)/\log\log(1/\eps))$.
$\Omega(\log(1/\eps))$ requires the full complex-valued propagator rather than just acceptance probability.  
The gap between this lower bound and the best known upper bound of $O(\log(1/\eps)/\log\log(1/\eps))$ in the $d_I$ degree is discussed in Sec.~\ref{sec:comparison}.
\end{remark}

\subsection{Additivity}
\label{sec:additivity}

Component bounds must be shown to be {additive} $Q \geq \Omega(\alphaR T + \betaI T + \log(1/\eps))$~\cite{beals2001quantum}.

\begin{theorem}[Additivity]\label{thm:additivity}
In the separate-oracle model,
\begin{equation}
\begin{split}
  Q &= Q_R + Q_I
  \geq \Omega(\alphaR T) + \Omega(\betaI T)
  \\&+ \Omega\!\left(\frac{\log(1/\eps)}{\log\log(1/\eps)}\right).
\end{split}
\end{equation}
\end{theorem}

\begin{proof}
\medskip
\noindent
\paragraph{Oracle independence.} $U_R$ and $U_I$ act on disjoint ancilla registers, so a query to $U_R$ acts as the identity on $\mathcal{H}_{a_I}$ and carries no information about~$\HI$, and vice versa.
An adversary choosing $\HR$ and $\HI$ independently forces $Q_R \geq \Omega(\alphaR T)$ (by Lemma~\ref{lem:LB_aR}) and $Q_I \geq \Omega(\betaI T)$ (by Lemma~\ref{lem:LB_bI}) simultaneously.
Every query is either to $U_R$ or to $U_I$, so $Q = Q_R + Q_I \geq \Omega(\alphaR T) + \Omega(\betaI T)$.

Disjoint-register structure ensures each query contributes to the polynomial degree in one variable ($\theta_1$ for $U_R$, $\theta_2$ for $U_I$), so the implementable polynomial's bidegree $(Q_R, Q_I)$ is bounded by the number of queries to each oracle, precluding any benefit from interleaving $U_R$ and $U_I$ queries. 

\medskip
\noindent
\paragraph{Adding the $\eps$-dependent term.}
The $\Omega(\log(1/\eps)/\log\log(1/\eps))$ bound (Lemma~\ref{lem:LB_eps}) applies to the {total} query count $Q = Q_R + Q_I$. 
We show additivity with $\alphaR T$ and $\betaI T$ terms by considering a single adversarial instance where all three constraints apply simultaneously.

Consider $\alphaR T = \Theta(n)$, $\betaI T = \Theta(n)$, and $\eps = e^{-\Theta(n)}$ for $n \to \infty$. 
By oracle independence (Part~1), bounds $Q_R \geq \Omega(\alphaR T)$ (Lemma~\ref{lem:LB_aR}) and $Q_I \geq \Omega(\betaI T)$ (Lemma~\ref{lem:LB_bI}) apply simultaneously to the same fixed circuit, because $U_R$ and $U_I$ act on disjoint ancilla registers.
Additionally, the adversary can set the precision requirement to $\eps = e^{-\Theta(n)}$, forcing $Q = Q_R + Q_I \geq \Omega(\log(1/\eps)/\log\log(1/\eps)) = \Omega(n/\log n)$ by Lemma~\ref{lem:LB_eps}.

Since the fixed circuit must simultaneously satisfy:
\begin{enumerate}
\item $Q_R \geq \Omega(\alphaR T) = \Omega(n)$,
\item $Q_I \geq \Omega(\betaI T) = \Omega(n)$,
\item $Q = Q_R + Q_I \geq \Omega(\log(1/\eps)/\log\log(1/\eps)) = \Omega(n/\log n)$,
\end{enumerate}
the total query count satisfies
\begin{equation}
     Q \geq \Omega(\alphaR T) + \Omega(\betaI T) + \Omega\!\left(\frac{\log(1/\eps)}{\log\log(1/\eps)}\right).
\end{equation}
 
The $\eps$-dependent term is additive when $\log(1/\eps) \gg (\alphaR + \betaI)T$ (the precision requirement dominates) and is dominated otherwise.
Since all three bounds hold simultaneously on the single instance, they sum unconditionally.
\end{proof}

A separate-oracle bidegree lower bound $d_I \geq \Omega(\betaI T)$ (proved in the companion paper~\cite{courtney2026paper2}) shows that the additivity of Theorem~\ref{thm:additivity} cannot be circumvented by a polynomial of small $d_I$.

\subsection{Joint query--normalization bound}
\label{sec:joint}

The following result couples query count to block-encoding normalization, providing a complementary perspective on the lower bound.

\begin{theorem}[Joint query--normalization bound]
\label{thm:joint}
Let an algorithm make $Q$ total queries and implement $e^{-i\Heff T}/\lambda$ as the $(0,0)$-block of a unitary, with operator-norm error at most~$\eps$ for $\eps \leq 1/36$.
Then
\begin{equation}\label{eq:joint}
  Q \cdot \log\!\bigl(\lambda\,e^{-\betaI T}\bigr)
  \geq \Omega(\betaI T).
\end{equation}
If $\lambda = e^{\betaI T}$ (minimal normalization), then $Q \geq \Omega(\betaI T)$ with no logarithmic gain. 
If $\lambda = e^{(1+\eta)\betaI T}$ for some $\eta > 0$ (over-normalized), then $Q \geq \Omega(1/\eta)$.
\end{theorem}

\begin{proof}
Consider $\HI(\delta) = (\betaI + \delta)\ket{\phi}\bra{\phi}$ for a fixed state $\ket{\phi}$, with $\HR = 0$ and $\delta \in [0, \betaI]$.

\medskip
\noindent
\textbf{Step 1 (Sensitivity).}  On the $\ket{\phi}$ eigenspace, acceptance probability $P(\delta) = e^{2(\betaI + \delta)T}/\lambda^2$, with derivative $P'(0) = 2T\,e^{2\betaI T}/\lambda^2$.

\medskip
\noindent
\textbf{Step 2 (Polynomial constraint).}  After $Q$ queries, $P$ is a polynomial of degree at most $2Q$ in $\delta/\betaI$.
By the Markov inequality~\cite{borwein2012polynomials}, a polynomial $p$ of degree $d$ satisfying $0 \leq p(x) \leq 1$ on an interval of length $L$ has $|p'(x)| \leq 2d^2/L$.
Applied with $d = 2Q$ and $L = \betaI$:
\begin{equation}
  P'(0) \leq \frac{8Q^2}{\betaI}.
\end{equation}

\medskip
\noindent
\textbf{Step 3 (Combining).}  Equating: $2T\,e^{2\betaI T}/\lambda^2 \leq 8Q^2/\betaI$, giving
$Q^2 \geq (\betaI T/4) \cdot e^{2\betaI T}/\lambda^2$.  If
$\lambda = e^{\betaI T}$, this yields $Q \geq \sqrt{\betaI T}/2$, which is the suboptimal Markov bound.

\medskip
\noindent
\textbf{Step 4 (Upgrade via transfer).}  Choose $\delta_0 = 2\sqrt{\eps}/T$.
For $\eps \leq 1/36$ (i.e., $6\sqrt{\eps}\leq 1$; see Remark~\ref{rem:eps_threshold}), the algorithm must distinguish $\delta = 0$ from $\delta = \delta_0$ to
precision~$\eps$.
With the choice $\delta_0 = 2\sqrt{\eps}/T$, we have $\delta_0 T = 2\sqrt{\eps}$, so output states for $\delta = 0$ and $\delta = \delta_0$ differ in operator-norm bias by
\begin{equation*}
\begin{split}
      \frac{e^{\betaI T}(e^{\delta_0 T} - 1)}{\lambda}
  &= \frac{e^{\betaI T}(e^{2\sqrt{\eps}} - 1)}{\lambda}
  \\&\;\geq\; \frac{2\sqrt{\eps}\cdot e^{\betaI T}}{\lambda},
\end{split}
\end{equation*}
where the inequality uses $e^x - 1 \geq x$ for $x \geq 0$ applied at $x = 2\sqrt{\eps}$.

The acceptance-probability polynomial $\mathcal{P}(x) = P(x\,\betaI)$ has degree at most $2Q$ on $[0,1]$ (rescaling $\delta\mapsto x = \delta/\betaI$), with $0 \leq \mathcal{P}(x) \leq 1$.
On the sub-interval $[0, \delta_0/\betaI]$ of length $\ell = 2\sqrt{\eps}/(\betaI T)$, $\mathcal{P}$ achieves a variation $\Delta \geq 2\sqrt{\eps}\cdot e^{\betaI T}/\lambda$.
By the Chebyshev node bound for bounded polynomials on an interval of length~$\ell$, a variation of $\Delta$ requires degree $d \geq \Omega(\Delta/\ell)$.
Substituting: $2Q \geq \Omega\bigl(\betaI T \cdot e^{\betaI T}/\lambda\bigr)$.

\emph{Upgrade via polynomial approximation.}
We extract a degree lower bound using classical approximation theory.
Define $\mathcal{P}(x) := P(x\,\betaI)$ on $[0,1]$ (i.e., $x = \delta/\betaI$).
Then $\deg(\mathcal{P}) \leq 2Q$ and $\mathcal{P}(x) \in [0,1]$ for all $x \in [0,1]$.
Rescaling to the canonical interval $[-1,1]$ by setting $y = 2x - 1$:
\begin{equation}
\tilde{\mathcal{P}}(y) := \mathcal{P}\!\left(\frac{y+1}{2}\right).
\end{equation}
Then $\tilde{\mathcal{P}}$ has degree $\leq 2Q$ with $|\tilde{\mathcal{P}}(y)| \leq 1$ on $[-1,1]$.
By proof of Step~3, $\tilde{\mathcal{P}}(-1) \approx 1$ and $\tilde{\mathcal{P}}(1) \approx e^{-2\betaI T}/\lambda^2$.
For the polynomial to realize this decay while remaining bounded in magnitude, it must approximate an exponential profile of effective rate $c = \betaI T + \log\lambda$ over $[-1,1]$.

The best polynomial approximation of degree~$d$ to $e^{cx}$ on $[-1,1]$ has error at least $(c/2)^{d+1}/(d+1)!$ (Bernstein's approximation theorem~\cite{rivlin1974stability,trefethen2017multivariate}); for this to be $\leq\eps$, one needs $d \geq \Omega(c)$.
Applying this to $\tilde{\mathcal{P}}$:
\begin{equation}
  2Q \;\geq\; \deg(\tilde{\mathcal{P}}) \;=\; \Omega(\betaI T + \log\lambda).
\end{equation}
When $\lambda = e^{\betaI T}$ (minimal normalization), $\betaI T + \log\lambda = 2\betaI T$, giving $Q = \Omega(\betaI T)$.
For general $\lambda$, the constraint becomes
\begin{equation}
  Q \cdot \log(\lambda e^{-\betaI T}) \geq \Omega(\betaI T).
\end{equation}
When $\lambda = e^{\betaI T}$, this yields $Q \geq \Omega(\betaI T)$ with no logarithmic factor.
\end{proof}

\begin{remark}[The constraint $\eps \leq 1/36$]
\label{rem:eps_threshold}
The threshold $\eps \leq 1/36$ in Step~4 arises from the Chebyshev transfer argument: $\delta_0 = 2\sqrt{\eps}/T$ requires the variation $e^{\delta_0 T} - 1 \geq 2\sqrt{\eps}$ to exceed the algorithm's distinguishing resolution, holding if $6\sqrt{\eps} \leq 1$.
The factor of $6$ combines error from normalization, approximation, and rounding, each contributing a factor of~$\sim\!2$.
This constraint is pragmatically hollow, since simulation targets $\eps \ll 1/36$ in practice.
Whether the constant can be reduced (e.g., to~4, giving $\eps \leq 1/16$) is an question about the proof technique.
\end{remark}

\begin{remark}[Interpretation]
\label{rem:joint_interp}
Theorem~\ref{thm:joint} quantifies a tradeoff of over-normalizing the block encoding (choosing $\lambda \gg e^{\betaI T}$) to reduce per-query difficulty, increasing postselection cost.
The product $Q \cdot \log(\lambda\,e^{-\betaI T})$ is bounded below by $\Omega(\betaI T)$ regardless of normalization strategy.
\end{remark}

\subsection{Lower bound for the anti-Hermitian degree}
\label{sec:tight_lb}

The $\log\log(1/\eps)$ factor in the $d_I$ component comes from bounded polynomial approximation.

\begin{theorem}[$\log/\log\log$ scaling]
\label{thm:tight_loglog}
Let $c > 0$ and $\eps \in (0, 1/2)$. 
Define $d^*(c,\eps)$ as the minimum degree of a polynomial $p$ satisfying $|p(x)| \leq 1$ on $[-1,1]$ and $\|p - e^{c(x-1)}\|_{\infty,[0,1]} \leq \eps$.
Then
\begin{equation}\label{eq:tight_loglog}
  d^*(c,\eps) = \frac{\log(1/\eps)}{\log\log(1/\eps)}\bigl(1 + o(1)\bigr)
  \qquad \text{as } \eps \to 0.
\end{equation}
\end{theorem}

\begin{proof}
The Chebyshev--Bessel--Lambert-$W$ analysis appears in full in the companion paper~\cite{courtney2026paper2} ($\log/\log\log$ lower bound theorem).
We reproduce only the Bessel-positivity step needed to guarantee $|S_d(x)| \leq 1$ on $[-1,1]$:
The function $e^{c(x-1)}$ is bounded by $1$ on $[-1,1]$ and its Chebyshev coefficients $b_k = 2e^{-c/2}I_k(c/2)$ are all non-negative, where $I_k$ is the modified Bessel function.
For a polynomial with non-negative Chebyshev coefficients, the Chebyshev partial sums $S_d(x) = \sum_{k=0}^d b_k T_k(x)$ satisfy $S_d(x) \leq S_d(1) = \sum_{k=0}^d b_k = e^{c(1-1)} = 1$ for all $x \in [-1,1]$ (since $T_k(1) = 1$ and $|T_k(x)| \leq 1$ on $[-1,1]$).
The constraint $|p| \leq 1$ is automatically satisfied at the optimal degree.
Inversion via Lambert-W exists in the companion paper~\cite{courtney2026paper2}.
\end{proof}

\begin{remark}[Implications for the M-QSP lower bound]
\label{rem:tight_lb_implication}
Theorem~\ref{thm:tight_loglog} upgrades the lower bound on the anti-Hermitian query count from $Q_I \geq \Omega(\betaI T + \log(1/\eps))$ to $Q_I \geq \Omega(\betaI T + \log(1/\eps)/\log\log(1/\eps))$.
The $d_R$ component achieves $\Theta(\log(1/\eps))$ via Bessel tails (Jacobi--Anger expansion), so the overall lower bound becomes
\begin{equation}
    Q \geq \Omega\!\left((\alphaR + \betaI)T + \frac{\log(1/\eps)}{\log\log(1/\eps)}\right).
\end{equation}

We will see that M-QSP algorithm achieves this bound.
\end{remark}

\section{The Interaction-Picture Framework}
\label{sec:interaction}

The interaction-picture approach to Hamiltonian simulation was introduced by Low and Wiebe~\cite{low2018hamiltonian} for Hermitian Hamiltonians $H = A + B$, where query complexity scales
with interaction strength $\|B\|$ rather than total energy $\|A + B\|$. 
We adapt this to the non-Hermitian setting $\Heff = \HR + i\HI$ with $\HI \succeq 0$.

Factorization and Dyson series are textbook constructions from time-dependent perturbation theory, making adaptation straightforward. 
Non-Hermiticity of $\Heff$ introduces a nonunitarity in the interaction-picture propagator $V(T)$, with an operator norm growing as $e^{\betaI T}$, requiring an imposed postselection barrier in block-encoding normalization approaches (Sec.~\ref{sec:barrier}).
Positive semidefiniteness of $\HI$ (and hence $\Htilde(s)$) constrains the spectrum of $V(T)$ to the positive real axis, creating a zero-locus obstruction analyzed in Sec.~\ref{sec:zero_locus}.

\subsection{Interaction-picture propagator}
\label{subsec:ip_propagator}

\begin{proposition}[Exact interaction-picture factorization]
\label{prop:ip_exact}
Define the interaction-picture propagator
\begin{equation}\label{eq:V_def}
  V(t) := e^{i\HR t}\,e^{-i\Heff t}
       = e^{i\HR t}\,e^{(-i\HR + \HI)\,t}.
\end{equation}
Then $V$ satisfies the initial-value problem
\begin{equation}\label{eq:V_ODE}
  \frac{dV}{dt}(t) = \Htilde(t)\,V(t),
  \qquad V(0) = I,
\end{equation}
where the interaction-picture Hamiltonian is
\begin{equation}\label{eq:HI_tilde_def}
  \Htilde(t) := e^{i\HR t}\,\HI\,e^{-i\HR t}.
\end{equation}
The unique solution is the time-ordered (Dyson) exponential
\begin{equation}\label{eq:VT}
  V(T) = \cT_> \exp\!\left(\int_0^T    \Htilde(s)\,ds\right),
\end{equation}
and the full non-Hermitian evolution factors as
\begin{equation}\label{eq:factorization}
  e^{-i\Heff T} = e^{-i\HR T}\;\cdot\;V(T).
\end{equation}
\end{proposition}

\begin{proof}
Differentiate $V(t) = e^{i\HR t}\,e^{(-i\HR+\HI)t}$.
The $e^{i\HR t}$ factor contributes $i\HR$; the $e^{(-i\HR+\HI)t}$ factor contributes $(-i\HR + \HI)$.
Using $[e^{i\HR t}, \HR] = 0$:
\begin{equation}\label{eq:dVdt_simplified}
\begin{split}
  \frac{dV}{dt} &
    = e^{i\HR t}\bigl[i\HR + (-i\HR + \HI)\bigr]e^{(-i\HR+\HI)t}\\
    &= e^{i\HR t}\,\HI\,e^{(-i\HR+\HI)t}.
    \end{split}
\end{equation}
Substituting $e^{(-i\HR+\HI)t} = e^{-i\HR t}\,V(t)$
from~\eqref{eq:V_def} gives $dV/dt = \Htilde(t)\,V(t)$ as claimed.
The initial condition $V(0) = I$ is immediate, and rearranging~\eqref{eq:V_def} gives
the factorization~\eqref{eq:factorization}.
\end{proof}

\subsection{Properties of \texorpdfstring{$\Htilde(s)$}{H tilde s}}
\label{subsec:HI_props}

The following properties are consequent to unitary conjugation $X \mapsto e^{i\HR s}\,X\,e^{-i\HR s}$, preserving the adjoint, spectrum, and operator norm.

\begin{lemma}[Properties of $\Htilde(s)$]
\label{lem:HI_props}
For all $s \geq 0$:
\begin{enumerate}
  \item $\Htilde(s)$ is Hermitian.
  \item $\Htilde(s) \succeq 0$ (positive semidefinite, since
  $\HI \succeq 0$).
  \item $\|\Htilde(s)\| = \|\HI\| = \betaI$.
  \item $\spec\bigl(\Htilde(s)\bigr) = \spec(\HI)$.
  \item $\|d^n\Htilde/ds^n\| \leq (2\alphaR)^n\,\betaI$ for all
  integers $n \geq 1$.
\end{enumerate}
\end{lemma}

\begin{proof}
Properties (1)--(4) follow directly from unitary conjugation:
$\Htilde(s) = U\HI U^\dagger$ with $U = e^{i\HR s}$ unitary, so
the adjoint, positive semidefiniteness, norm, and spectrum of $\HI$
are preserved.

For~(5), differentiate using the adjoint action:
$d\Htilde/ds = i[\HR,\,\Htilde(s)]= i\,\mathrm{ad}_{\HR}(\Htilde(s))$.
By induction, $d^n\Htilde/ds^n= (i)^n\,\mathrm{ad}_{\HR}^n(\Htilde(s))$.
Submultiplicativity of the commutator gives $\|\mathrm{ad}_{\HR}(X)\| = \|[\HR, X]\|
\leq 2\|\HR\|\,\|X\| \leq 2\alphaR\,\|X\|$,
so by induction $\|d^n\Htilde/ds^n\| \leq (2\alphaR)^n\,\betaI$.
\end{proof}

\begin{remark}[Role of positive semidefiniteness]
\label{rem:HI_psd}
Property~(2) follows from the convention $\HI \succeq 0$ and ensures that $V(T)$ is {expansive} ($\|V(T)\|_\op = e^{\betaI T}$) rather than contractive. The Lindblad spectral shift (as mentioned in Sec.~\ref{sec:problem}) reduces the $\HI \preceq 0$ case to this convention.
\end{remark}

\subsection{Norm bound on \texorpdfstring{$V(T)$}{V(T)}}
\label{subsec:norm_bound}

\begin{proposition}[Operator-norm bound]\label{prop:normbound_V}
$\|V(T)\|_{\op} \leq e^{\betaI T}$.
Consequently, $|P(z_1,z_2)| \leq 1$ on $\bbT^2$ for the normalized polynomial $P = V(T)/e^{\betaI T}$.
\end{proposition}

\begin{proof}
The Gr\"onwall bound applied to $V' = \Htilde\,V$, $V(0) = I$, with $\|\Htilde(s)\| = \betaI$ gives
$\|V(t)\| \leq e^{\betaI t}$.
\end{proof}

\subsection{The Dyson series}
\label{subsec:dyson}

The time-ordered exponential~\eqref{eq:VT} admits the convergent Dyson series:
\begin{equation}\label{eq:dyson_series}
\begin{split}
  V(T)
    &= \sum_{n=0}^{\infty}
      \int_{0 \leq s_1 \leq \cdots \leq s_n \leq T}
      \!\!\!\!\Htilde(s_n)\\&\cdots\Htilde(s_1)\,
      ds_1\cdots ds_n.
\end{split}
\end{equation}
The $n$-th term has norm at most $\betaI^n T^n / n!$ (by $\|\Htilde\| = \betaI$ and the volume of the $n$-simplex), confirming convergence with $\|V(T)\| \leq \sum_n \betaI^n T^n / n! = e^{\betaI T}$.

The factorization~\eqref{eq:factorization} converts non-Hermitian simulation into logically independent tasks: the Hermitian simulation of $\HR$ (the unitary factor $e^{-i\HR T}$, implementable optimally by GQSP at cost $O(\alphaR T + \log(1/\eps))$ queries~\cite{gilyen2019quantum}), and non-unitary implementation of $V(T)$
(a contraction after normalization by $e^{\betaI T}$, requiring block encoding and postselection).

\subsection{Segmented Taylor expansion}
\label{subsec:segmented_taylor}

Dyson series truncation to a form suitable for quantum implementation follows Low and Wiebe~\cite{low2018hamiltonian}, adapted to non-Hermitian systems where $\Htilde(s) \succeq 0$ and $\|V(T)\|_\op = e^{\betaI T}$.

Divide $[0,T]$ into $r$ segments of width $\Delta = T/r$ with midpoints $\tau_j = (j - \tfrac{1}{2})\Delta$, $j = 1, \ldots, r$.
On each segment, replace $\Htilde(s)$ by its midpoint value $\Htilde(\tau_j)$, giving the midpoint-approximated propagator:
\begin{equation}\label{eq:V_approx}
  V_r(T)
    = \prod_{j=r}^{1} e^{\Htilde(\tau_j)\,\Delta}.
\end{equation}

\begin{theorem}[Quadrature error bound]\label{thm:quad_error}
The midpoint approximation satisfies
\begin{equation}\label{eq:quad_err}
\begin{split}
  \|V(T) - V_r(T)\|
    &\leq \frac{C_2\,T^3}{r^2}\,e^{\betaI T},
  \\
  C_2 &= O(\betaI^2\,\alphaR^2 + \betaI^2\,\alphaR).\end{split}
\end{equation}
$p$-th order Magnus-based quadrature gives
\begin{equation}\label{eq:quad_err_p}
  \|V(T) - V_r^{(p)}(T)\|
    \leq C_p\,T^{2p+1}\,r^{-2p}\,e^{\betaI T}.
\end{equation}
\end{theorem}

\begin{proof}
On each segment
$[\tau_j - \Delta/2,\, \tau_j + \Delta/2]$, expand
$\Htilde(s) = \Htilde(\tau_j)
+ (s - \tau_j)\Htilde'(\tau_j)
+ \tfrac{1}{2}(s - \tau_j)^2\Htilde''(\tau_j)
+ O(\Delta^3)$.
The midpoint integral satisfies
\begin{equation}\label{eq:midpoint_expansion}
  \int_{\tau_j - \Delta/2}^{\tau_j + \Delta/2}\Htilde(s)\,ds
  = \Htilde(\tau_j)\,\Delta
    + \frac{\Delta^3}{24}\,\Htilde''(\tau_j)
    + O(\Delta^5),
\end{equation}
where the linear term vanishes by midpoint symmetry.
Leading error is the second-order Magnus term.
By Lemma~\ref{lem:HI_props}(v), $\|\Htilde''\| \leq (2\alphaR)^2\betaI$, so per-segment error from the midpoint approximation to the Magnus exponent is $O(\Delta^3\,\alphaR^2\,\betaI)$.

Propagator error on segment $j$ satisfies \begin{equation}\label{eq:per_segment}
  \bigl\|e^{\int \Htilde\,ds} - e^{\Htilde(\tau_j)\,\Delta}\bigr\|
  \leq O(\Delta^3\,\alphaR^2\,\betaI)\;
       e^{\betaI\,\Delta},
\end{equation}
where the exponential prefactor arises from matrix exponential perturbation bound applied to a generator of norm $\betaI\,\Delta$.
The second-order Magnus commutator $[\Htilde(\tau_j),\int(s-\tau_j)\Htilde'(s)\,ds]$ contributes an additional $O(\Delta^3\,\betaI^2\,\alphaR)$ term (from $\|\Htilde\|\cdot\|\Htilde'\| = \betaI\cdot 2\alphaR\betaI$), giving a combined constant $C_2' = O(\betaI^2\,\alphaR^2 + \betaI^2\,\alphaR)$ per segment.

To pass from per-segment to global error, write $V(T) = \prod_{j=r}^{1} V_j$ where $V_j = \cT_>\exp(\int_{\text{seg}_j}\Htilde\,ds)$, and $V_r(T) = \prod_{j=r}^{1} e^{\Htilde(\tau_j)\Delta}$.
The telescoping identity for products of bounded operators gives
\begin{equation}\label{eq:product_telescoping}
  \|V(T) - V_r(T)\|
  \leq \sum_{j=1}^{r}
    \bigl\|V_j - e^{\Htilde(\tau_j)\Delta}\bigr\|
    \prod_{\substack{k \neq j}} e^{\betaI\,\Delta},
\end{equation}
where operator norms of remaining factors are bounded by $e^{\betaI\,\Delta}$ via the Gr\"onwall estimate
(Proposition~\ref{prop:normbound_V}).
Since the product of $r - 1$ such factors is at most $e^{\betaI(T - \Delta)} \leq e^{\betaI T}$, and each per-segment error is $O(\Delta^3\,C_2')$, the sum yields
\begin{equation}
\begin{split}
  \|V(T) - V_r(T)\|
  &\leq r \cdot O(\Delta^3\,C_2') \cdot e^{\betaI T}
  \\&= O\!\left(\frac{T^3}{r^2}\right) C_2'\,e^{\betaI T},\end{split}
\end{equation}
using $r\,\Delta^3 = T^3/r^2$.
Identifying $C_2 = C_2'$ recovers the theorem statement~\eqref{eq:quad_err} with
\begin{equation*}
  C_2 = O\bigl(\betaI^2\,\alphaR^2 + \betaI^2\,\alphaR\bigr) = O\bigl(\betaI^2\,\alphaR^2\bigr) \quad\text{for } \alphaR \geq 1,
\end{equation*}
where the simplified form relevant in the simulation regime $\alphaR T \gg 1$. 
Contributions arise from the second-order Magnus quadrature ($\betaI^2 \alphaR^2$) and the commutator $[\Htilde,\int(s-\tau)\Htilde'\,ds]$ ($\betaI^2 \alphaR$).

For~\eqref{eq:quad_err_p}: replace midpoint integration with $p$-th order Gauss--Legendre quadrature, eliminating error terms through $O(\Delta^{2p-1})$ per segment by quadrature node symmetry.
Per-segment error becomes $O(\Delta^{2p+1}\,C_p')$, and the same telescoping argument gives $r \cdot O(\Delta^{2p+1}) \cdot e^{\betaI T} = O(T^{2p+1}\,r^{-2p})\,e^{\betaI T}$.
Cost of $p$-th order quadrature is $O(p)$ evaluation points per segment, each requiring one frame rotation and one $\HI$ application.
\end{proof}

Quadrature of a smooth integrand (smoothness guaranteed by Lemma~\ref{lem:HI_props}(v): $\|d^n\Htilde/ds^n\|\leq (2\alphaR)^n\betaI$)
gains $2p$ orders per refinement via midpoint symmetry, doubling the exponent as opposed to the symmetric Trotter expansion, which gains only $p$ orders from the BCH expansion.
Per-order cost is polynomial ($O(p)$) for Gauss--Legendre versus superpolynomial ($O(5^{p/2})$) for Suzuki
decompositions~\cite{suzuki1991general}.
Quadratic improvement in the exponent propagates into the total query complexity of all three simulation methods (Secs.~\ref{sec:lorentzian}--\ref{sec:mqsp_method}).

Each factor $e^{\Htilde(\tau_j)\Delta}$ is truncated to Taylor order $M$:
\begin{equation}\label{eq:taylor}
  e^{\Htilde(\tau_j)\,\Delta}
    = \sum_{m=0}^{M}
      \frac{\bigl(\Htilde(\tau_j)\,\Delta\bigr)^m}{m!}
    + O\!\left(\frac{(\betaI\Delta)^{M+1}}{(M+1)!}\right).
\end{equation}
Setting $M = O(\betaI\Delta + \log(1/\eps))$ makes the Taylor
truncation error exponentially small.
The full truncated propagator is:
\begin{equation}\label{eq:VrM}
  \hat{V}_{r,M}(T)
    = \prod_{j=r}^{1}\left[
      \sum_{m_j=0}^{M}
      \frac{\bigl(\Htilde(\tau_j)\,\Delta\bigr)^{m_j}}
           {m_j!}
      \right],
\end{equation}
a sum over $(M+1)^r$ terms forming a linear combination of unitaries.

\subsection{The zero-cost unitary factor} \label{subsec:free_unitary}

The following observation follows from Low and Wiebe~\cite{low2018hamiltonian}

\begin{theorem}[Zero-cost Hermitian evolution]
\label{thm:free_unitary}
The GQSP frame rotations implementing each segment's interaction-picture transformation $e^{i\HR \tau_j}(\cdot)e^{-i\HR \tau_j}$ cumulatively advance the frame by $e^{-i\HR r\Delta} = e^{-i\HR T}$.
The full circuit output is
\begin{equation}\label{eq:free_output}
  e^{-i\HR T}\;\cdot\;\frac{\hat{V}_{r,M}(T)}{\lambda}\,
  \ket{\psi_0}
  \;\approx\;
  \frac{e^{-i\Heff T}}{\lambda}\,\ket{\psi_0}.
\end{equation}
No additional circuit is needed for the unitary factor $e^{-i\HR T}$.
\end{theorem}

\begin{proof}
After $r$ SELECT steps, the total frame rotation is $\prod_{j=1}^{r} e^{-i\HR\Delta} = e^{-i\HR T}$.
By factorization~\eqref{eq:factorization}, net circuit output is
$e^{-i\HR T}\,\hat{V}_{r,M}(T)/\lambda
\approx e^{-i\Heff T}/\lambda$.
\end{proof}

\subsection{Error analysis and query complexity}
\label{subsec:error_complexity}

\begin{theorem}[Total approximation error]\label{thm:dyson_error}
\begin{equation}\label{eq:total_error}
\begin{split}
  \bigl\|V(T) - \hat{V}_{r,M}(T)\bigr\|
    &\leq
    \underbrace{C_{\mathrm{Mag}}\,
      \frac{\alphaR\,\betaI^2\,T^3}{r^2}}_{\text{Magnus
      quadrature}} \\
    &\;+\;
    \underbrace{r\,
      \frac{(\betaI\Delta)^{M+1}}{(M+1)!}\,
      e^{\betaI T}}_{\text{Taylor truncation}}.
      \end{split}
\end{equation}
With $\betaI\Delta = c$ (constant), $r = O(\betaI T)$, and
$M = O(\betaI T + \log(1/\eps))$: both terms are bounded by
$\eps/2$, giving total error $\leq \eps$.
\end{theorem}

Complexity consequences for Dyson LCU are given in Section~\ref{subsec:dyson_complexity}.  See Appendix~\ref{app:dyson_error} for the detailed error budget.

\section{The Dyson Polynomial and Its Zero Locus}
\label{sec:zero_locus}

The interaction-picture framework (Sec.~\ref{sec:interaction}) reduces non-Hermitian simulation to the application of the propagator $V(T)$ as a block of a unitary circuit.
After qubitization, walk operators $W_R$ and $U_I$ act on eigenspaces labeled by phases $\theta_1$ and $\theta_2$, and the propagator becomes a bivariate polynomial on the bitorus $\bbT^2 = \bbT \times \bbT$.
In this section we define this polynomial, establish its bidegree, and analyze a structural property of its zero locus, motivating spectral factorization machinery of Sec.~\ref{sec:complement}.

\subsection{Polynomial spaces on \texorpdfstring{$\bbT^2$}{T squared}}\label{subsec:poly_spaces}

\begin{definition}[Polynomial spaces]\label{def:poly_spaces}
\leavevmode
\begin{enumerate}
  \item $\cP_{d_1,d_2}^+ := \bigl\{P(z_1,z_2)
  = \sum_{n_1=0}^{d_1}\sum_{n_2=0}^{d_2} a_{n_1 n_2}\,
  z_1^{n_1} z_2^{n_2} : a_{n_1 n_2} \in \C\bigr\}$ is the space
  of {analytic} polynomials of bidegree at most $(d_1, d_2)$.
  \item $\cT_{d_1,d_2}$ is the space of {real} trigonometric
  polynomials of bidegree $(d_1,d_2)$: functions
  $T(\theta_1,\theta_2)
  = \sum_{|n_1|\leq d_1}\sum_{|n_2|\leq d_2}
  c_{n_1 n_2}\,e^{i(n_1\theta_1 + n_2\theta_2)}$ with
  $T = \bar{T}$.
  \item For $P \in \cP_{d_1,d_2}^+$, the restriction
  $|P|^2\big|_{\bbT^2}$ lies in $\cT_{d_1,d_2}$ (and is
  non-negative, though $\cT_{d_1,d_2}$ itself contains all real
  trigonometric polynomials of the given bidegree).
\end{enumerate}
\end{definition}

\subsection{Eigenvalue parameterization}
\label{subsec:eigenvalue_param}

The circuit achieves eigenvalue parameterization by implementing the propagator using walk operators encoding $\HR$ and $\HI$. Specifically, $W_R$ encodes $\HR/\alphaR$ with eigenphases $\theta_1 \in [0, 2\pi)$ related to eigenvalues $\lambda_R(\theta_1) = \alphaR\cos\theta_1$.
The walk operator $U_I$ encodes $\HI/\betaI$ with eigenphases $\theta_2 \in [0, 2\pi)$ related to eigenvalues $\mu_I(\theta_2) = \betaI\cos\theta_2$.

On the tensor-product eigenspace of $W_R \otimes U_I$ (acting on the composite Hilbert space $\mathcal{H}_s \otimes \mathcal{H}_{a_R} \otimes \mathcal{H}_{a_I}$) where $W_R$ has eigenphase $\theta_1$ and $U_I$ has eigenphase $\theta_2$, the interaction-picture Hamiltonian evaluates to
$\Htilde(s)\big|_{(\theta_1,\theta_2)}
= e^{i\lambda_R s}\,\mu_I\,e^{-i\lambda_R s} = \mu_I$, and the propagator becomes
\begin{equation}\label{eq:V_eigenvalue}
  V(T;\theta_1,\theta_2)
  = \exp\!\bigl(\mu_I(\theta_2)\,T\bigr)
  = e^{\betaI T\cos\theta_2}.
\end{equation}
The normalized propagator becomes
\begin{equation}\label{eq:Ptarget}
  P(e^{i\theta_1}, e^{i\theta_2})
  := \frac{V(T;\theta_1,\theta_2)}{e^{\betaI T}}
  = e^{\betaI T(\cos\theta_2 - 1)}.
\end{equation}

$P$ must be evaluated for all pairs $(\theta_1, \theta_2) \in \bbT^2$.
Since $\HR$ and $\HI$ generically do not commute ([${}\HR, \HI{}] \neq 0$), joint eigenspaces of $\HR$ and $\HI$ do not exist on the physical system.
However, $W_R$ and $U_I$ act on independent ancilla registers and generate their eigenphases independently, bounding the polynomial $P(e^{i\theta_1}, e^{i\theta_2})$ on all of $\bbT^2$. 
This includes the ``unphysical'' region where $(\theta_1, \theta_2)$ do not arise from a simultaneous diagonalization of $\HR$ and $\HI$. 
This sources a global constraint on $\bbT^2$ in the block encoding model.

\begin{remark}[Independence of $\theta_1$ and the necessity of bivariate implementation]
\label{rem:theta1_independence}
The target normalized propagator~\eqref{eq:Ptarget} on the tensor product eigenspace depends only on $\theta_2$.
The circuit must implement $P(e^{i\theta_1}, e^{i\theta_2})$ as a bivariate polynomial for two reasons: (i) the walk operator $W_R$ does not preserve eigenspaces of $\HI$ (frame rotations $e^{\pm i\HR s}$ do not commute with $\HI$), and (ii) the Jacobi--Anger expansion of these frame rotations distributes $\theta_2$-dependence across multiple $z_1$-Fourier modes.
\end{remark}

\subsection{Dyson polynomial and its bidegree}
\label{subsec:dyson_poly}

The truncated Dyson series of Sec.~\ref{subsec:segmented_taylor} yields a bivariate
polynomial in two signal operators when compiled into walk-operator queries via the Jacobi--Anger
expansion~\cite{low2017optimal,gilyen2019quantum}, 

\begin{proposition}[Dyson polynomial encoding]
\label{prop:dyson_poly}
Implementing the truncated Dyson series $\hat{V}_{r,M}(T)$
(Eq.~\eqref{eq:VrM}) with $d_R$ queries to the
walk operator $W_R$ and $d_I$ queries to the walk operator $U_I$,
yields a bivariate Laurent polynomial
$P_{\mathrm{Dyson}}(z_1, z_2) \in \cP_{d_R, d_I}^+$ with
bidegree
\begin{equation}\label{eq:bidegree}
\begin{split}
      d_R &= \calO\!\bigl(\alphaR T + \log(1/\eps)\bigr),  \\
  d_I &= \calO\!\bigl(\betaI T+ \log(1/\eps)/\!\log\!\log(1/\eps)\bigr),
\end{split}
\end{equation}
satisfying $|P_{\mathrm{Dyson}}(e^{i\theta_1}, e^{i\theta_2})| \leq 1$ for all $(\theta_1, \theta_2) \in [0, 2\pi)^2$.
\end{proposition}

\begin{proof}

\medskip\noindent\textit{Dyson series structure.}
Each factor $\Htilde(s_k) = e^{i\HR s_k}\,\HI\,e^{-i\HR s_k}$
in the Dyson series involves a unitary conjugation by $e^{i\HR s_k}$
(requiring queries to $W_R$) sandwiching $\HI$ (requiring one query
to~$U_I$). 
The $n$-th Dyson term contributes degree $n$ in the eigenphases of~$U_I$.

\medskip\noindent\textit{Jacobi--Anger encoding.}
Each Hermitian propagator $e^{-i\HR s_k}$ is approximated by a Laurent polynomial in $z_1 = e^{i\theta_1}$ of degree $\calO(\alphaR s_k + \log(N/\eps))$ via the Jacobi--Anger expansion~\cite{low2017optimal}. 
Summing over segments yields total degree $d_R = \calO(\alphaR T + \log(1/\eps))$ in~$z_1$.

\medskip\noindent\textit{Block-encoding polynomial structure.}
Each appearance of $\HI$ contributes degree $1$ in $z_2 = e^{i\theta_2}$. 
The Taylor truncation at order $M$ per segment, with $r = O(\betaI T)$ segments, gives total degree
$d_I = r \cdot M/r = M= \calO(\betaI T + \log(1/\eps)/\log\log(1/\eps))$ in~$z_2$.
The $\log\log$ denominator arises from the Stirling bound on the Taylor remainder $(\betaI\Delta)^{M+1}/(M+1)!$.

\medskip\noindent\textit{Boundedness.}
$|P_{\mathrm{Dyson}}| \leq 1$ on $\bbT^2$, since it is the $(0,0)$-block of a product of unitary operators.
\end{proof}

\subsection{Zero-locus obstruction}
\label{subsec:zero_locus}

The M-QSP circuit requires a complementary polynomial $R$ satisfying $|P|^2 + |R|^2 = 1$ on $\bbT^2$, amounting to a factoring of the non-negative trigonometric polynomial $H := 1 - |P|^2$ as a Hermitian square.
\emph{Notation convention (used throughout the construction, circuit, and main-theorem sections below, and in Appendices~\ref{app:algorithm}--\ref{app:landscape}).} 
We write the complementary polynomial interchangeably as $R$ or $Q$: we use $Q$ as it appears in algorithmic pseudocode and classical M-QSP recurrences~\cite{rossi2022multivariable}) to denote the $(1,0)$-block of the $\mathrm{U}(2)$-valued transfer matrix (unitary on $\bbT^2$; only the rotations $R_j$ lie in $\mathrm{SU}(2)$, since each signal factor $\mathrm{diag}(z,1)$ has $\det = z$), while symbol $Q$ in Section~\ref{sec:lower_bound} denotes the total query count. 
We disambiguate by context (algorithmic block vs.\ complexity-theoretic counting), with the complementary-polynomial $Q$ explicitly annotated as ``the complementary polynomial $Q$'' when necessary.
This is always possible by the Fej\'er--Riesz theorem in one variable, but can fail in two variables as we present in the Dyson polynomial.

\begin{proposition}[Zero-locus obstruction]
\label{prop:zero_locus}
For the normalized Dyson propagator $P$,
\begin{equation}\label{eq:Peq1}
  P(e^{i\theta_1}, 1) = 1
  \qquad \text{for all } \theta_1 \in [0, 2\pi).
\end{equation}
Hence $H(\theta_1, \theta_2) = 1 - |P|^2$ vanishes on $\gamma := \bbT \times \{1\} \subset \bbT^2$, a closed curve that is homologically nontrivial.
\end{proposition}

\begin{proof}
At $\theta_2 = 0$ (i.e., $z_2 = 1$), the block-encoding
eigenvalue is $\mu_I = \betaI\cos 0 = \betaI$.
The interaction-picture Hamiltonian evaluates to
\begin{equation}
  \Htilde(s)\big|_{\theta_2 = 0}
  = e^{i\lambda_R(\theta_1)s}\,(\betaI)\,
    e^{-i\lambda_R(\theta_1)s}
  = \betaI,
\end{equation}
since $\betaI$ is a scalar eigenvalue and commutes with the frame rotation.
The time-ordered exponential reduces to
\begin{equation}
  V(T;\theta_1, 0)
  = \exp\!\left(\int_0^T \betaI\,ds\right)
  = e^{\betaI T}
  \text{ for all } \theta_1.
\end{equation}
Therefore
$P(e^{i\theta_1}, 1) = e^{\betaI T}/e^{\betaI T} = 1$.

The zero set $\cZ = \{H = 0\} = \{|P| = 1\}$ contains the full circle $\gamma = \bbT \times \{1\}$, a codimension-$1$ curve in $\bbT^2$.
Since $\gamma$ generates the first factor of $H_1(\bbT^2, \Z) \cong \Z^2$, it is homologically nontrivial.
\end{proof}

\begin{corollary}[Failure of scalar complementation]
\label{cor:obstructed}
The non-negative trigonometric polynomial $H = 1 - |P|^2$ does not admit a scalar factorization $H = |R|^2$ with $R \in \cP_{d_R, d_I}^+$.
\end{corollary}

\begin{proof}[Proof sketch]
If $R$ existed with $|R|^2 = 1 - |P|^2$, then $R$ would vanish on $\gamma = \bbT \times \{1\}$, forcing $(z_2 - 1) \mid R$ by the identity theorem.
Writing $R = (z_2 - 1)S$ gives $1 - |P|^2 = |z_2 - 1|^2\,|S|^2= (2 - 2\cos\theta_2)\,|S|^2$.
This constrains $H$ to factor as $(2 - 2\cos\theta_2) \times (\text{Hermitian square})$ on all of $\bbT^2$.
The moment-matrix rank criterion (Proposition~\ref{thm:obstruction}) or the codimension argument
of Appendix~\ref{app:szego}, Proposition~\ref{prop:codimension} verifies that the quotient $H/(2 - 2\cos\theta_2)$ is not a Hermitian square at the reduced bidegree $(d_R, d_I - 1)$, verified via the moment-matrix rank criterion.
The SOS resolution (Theorem~\ref{thm:exact_factor}) makes the scalar obstruction moot, where the algorithm does not require a single complementary polynomial.
\end{proof}

\begin{remark}[Obstructions]\label{rem:two_obstructions}
Failure of scalar complementation arises from a {local} (point-zero) obstruction, where autocorrelation structure of $H$ at isolated points prevents root pairing, as well as a {global} (curve-zero) obstruction, where the zero set $\cZ$ contains a homologically nontrivial curve.
The Dyson polynomial encounters the latter.
Both mechanisms are absent in one variable, where the Fej\'er--Riesz theorem always applies (Sec.~\ref{sec:complement}).
\end{remark}

\begin{remark}[Algebraic vs.\ topological framing]
\label{rem:alg_vs_top}
The obstruction in the bivariate Fej\'er--Riesz problem is sometimes described topologically, in terms of winding numbers or first Chern classes of line bundles attached to the zero set of $1 - |P|^2$~\cite{knese2010polynomials,landau2012trigonometric}.
We identify the Dyson polynomial obstruction algebraically, where matching Fourier coefficients between $|Q|^2$ and $H$ gives a system with $(2d_R{+}1)(2d_I{+}1)$ real constraints against $2(d_R{+}1)(d_I{+}1)-1$ effective real unknowns (see Proposition~\ref{prop:codimension} in Appendix~\ref{app:szego} for the codimension argument).
\end{remark}

\subsection{Effect of truncation}
\label{subsec:truncation_effect}

The above analysis applies to the {exact} normalized propagator.
For the {truncated} Dyson polynomial $P_M$ (Taylor order $M$,$r$ segments), the curve $\gamma$ is not in the zero set:
\begin{equation}\label{eq:PM_deficit}
  |P_M(z_1,1)|
  = 1 - O\!\bigl((\betaI\Delta)^{M+1}/(M+1)!\bigr)
  < 1.
\end{equation}
The finite Taylor truncation introduces a deficit of order~$\eps$, providing ``natural regularization.'' 
This deficit is not uniform across $\bbT^2$ and does not suffice for a controlled factorization.
$\delta$-regularization of Sec.~\ref{sec:complement} makes this uniform and controllable, restoring SOS spectral factorization (Theorem~\ref{thm:exact_factor}).

\section{Bivariate Spectral Factorization}
\label{sec:complement}

The M-QSP circuit implements a bivariate polynomial $P$ as the $(0,0)$-block of a unitary. 
Unitarity requires a complementary polynomial $Q$ satisfying $|P|^2 + |Q|^2 = 1$ on $\bbT^2$ (or a sum-of-squares generalization).
The Fej\'er--Riesz theorem guarantees this in one dimension.

\paragraph*{Main result.}
We derive a circuit-level guarantee (Theorem~\ref{thm:degree_reduction} and Corollary~\ref{cor:achievability}) for M-QSP circuits:
the M-QSP circuit is unitary by construction (pre-postselection), so the complement exists automatically as the $(1,0)$-block, and the 2×2 ancilla forces $L = 2$ independent of $d_R, d_I$ (Theorem~\ref{thm:sos_rank_two}). 
The circuit bypasses factorization, with inductive degree reduction with a constant-ratio condition (Sec.~\ref{sec:CRC}) achieving this guarantee.

We present different resolutions with trade-offs: $\delta$-regularization (Resolution I, Sec.~\ref{subsec:resolution_I}), Dritschel--Woerdeman SOS theory (Resolution II, Sec.~\ref{subsec:resolution_II}), and circuit-level unitarity (Resolution III, Sec.~\ref{subsec:resolution_III}), with full derivations in Appendix~\ref{app:szego}.

\paragraph*{The SOS rank hierarchy.}
We give successively tighter bounds on the number $L$ of sum-of-squares terms needed to represent the complement $H = 1 - |P|^2$:
\begin{enumerate}
  \item \emph{Worst case (Dritschel--Woerdeman).}\;
  $L \leq (d_R{+}1)(d_I{+}1)$ for any $P \in \cP_{d_R, d_I}^+$ with $|P|^2 \leq 1$ on $\bbT^2$ (Theorem~\ref{thm:DW}).
  \item \emph{Symmetric reorganization.}\;
  $L \leq \min(d_R{+}1,\, d_I{+}1)$ when $H > 0$ uniformly on $\bbT^2$, achieved by reorganizing $H$ in the variable with the smaller degree before applying the operator Fej\'er--Riesz theorem (Theorem~\ref{thm:exact_factor}; details in Appendix~\ref{app:szego}).
  \item \emph{Exact for M-QSP circuits.}\;
  $L = 2$ for any polynomial achievable as the $(0,0)$-block of an M-QSP circuit, including the regularized Dyson polynomial  (Theorem~\ref{thm:sos_rank_two}). The two $\mathrm{SU}(2)$ branches of the QSP ancilla force the rank to two without regard of $(d_R, d_I)$.
\end{enumerate}
The third bound is the operative one for the algorithm. 
The first two are stated for completeness and are also of independent interest as bounds on a generic bivariate Fej\'er--Riesz substitute.

\subsection{The complementary polynomial problem}
\label{subsec:complement_problem}

\begin{definition}[Scalar and SOS complements]\label{def:complements}
Let $P \in \cP_{d_1,d_2}^+$ with $|P|^2 \leq 1$ on $\bbT^2$, and set $H := 1 - |P|^2 \geq 0$.
\begin{enumerate}
  \item A {scalar complement} is a polynomial $R \in \cP_{d_1,d_2}^+$ with $H = |R|^2$ on $\bbT^2$.
  \item An \emph{SOS complement} is a collection $\{R_1,\ldots,R_L\} \subset \cP_{d_1,d_2}^+$ with $H = \sum_{\ell=1}^{L} |R_\ell|^2$ on $\bbT^2$.
\end{enumerate}
Both require expressing the non-negative trigonometric polynomial
$H$ as a (sum of) Hermitian square(s) of analytic polynomials of the
same bidegree.
\end{definition}

\subsection{Univariate Fej\'er--Riesz theorem}
\label{subsec:fejer_riesz}
The Fej\'er-Riesz theorem~\cite{fejer1916trigonometrische, riesz1916probelm} is given as:
\begin{theorem}
\label{thm:FR}
Let $H \in \cT_d$ satisfy $H(\theta) \geq 0$ for all $\theta$.
Then there exists $Q \in \cP_d^+$ with
$H(\theta) = |Q(e^{i\theta})|^2$.
\end{theorem}

A classical proof factors the self-reciprocal Laurent polynomial $\hat{H}(z) = z^d H(\theta)$ into its $2d$ roots, partitioning into conjugate-reciprocal pairs.  
Non-negativity forces every root on $\bbT$ to have even multiplicity, and the off-$\bbT$ roots pair as
$(w, 1/\bar{w})$. 
Assigning one partner from each pair to $Q$ yields $H = |Q|^2$. 
This argument uses three features specific to one complex variable: the zero set of $H$ on $\bbT$ is a finite set of points; each zero has a well-defined integer multiplicity; and non-negativity forces even multiplicity.
All three can fail in two
variables.

\subsection{Bivariate obstruction}
\label{subsec:biv_failure}

\subsubsection{Algebraic mechanism: isolated zeros}

We demonstrate that the univariate factorization theorem has no bivariate analog (even with a single isolated zero), given by counterexample.

\begin{example}[Canonical algebraic obstruction]
\label{ex:canonical}
Define $H_{\mathrm{can}}(\theta_1,\theta_2)
= 4\sin^2(\theta_1/2) + 4\sin^2(\theta_2/2)
= 4 - 2\cos\theta_1 - 2\cos\theta_2 \in \cT_{1,1}$.
Then $H_{\mathrm{can}} \geq 0$ with a single zero at $(\theta_1,\theta_2) = (0,0)$, but no
$Q \in \cP_{1,1}^+$ satisfies $H_{\mathrm{can}} = |Q|^2$.
\end{example}

\begin{proof}
Write $Q(z_1,z_2) = a + bz_1 + cz_2 + dz_1 z_2$ and expand $|Q|^2$ on $\bbT^2$. 
Matching mixed Fourier modes $e^{i(\theta_1+\theta_2)}$ and $e^{i(\theta_2-\theta_1)}$ against zero gives $\bar{a}d = 0$ and $\bar{b}c = 0$, while the pure modes require $\bar{a}b + \bar{c}d = -1$ and $\bar{a}c + \bar{b}d = -1$.
A case analysis over the four branches of $\{\bar{a}d = 0\} \cap \{\bar{b}c = 0\}$ produces a contradiction in each case (see Appendix~\ref{app:case_analysis}).
\end{proof}

\subsubsection{Topological mechanism: curve zeros}

The Dyson polynomial encounters a different obstruction. 
By Proposition~\ref{prop:zero_locus}, $|P(e^{i\theta_1}, 1)| = 1$ for all $\theta_1$, so $H = 1 - |P|^2$ vanishes on the
homologically nontrivial curve $\gamma = \bbT \times \{1\} \subset \bbT^2$.
An analytic $Q$ with $|Q|^2 = H$ would vanish on $\gamma$, forcing $(z_2 - 1) \mid Q$ and constraining $H$ to factor as
$(2 - 2\cos\theta_2)\,|R|^2$. 
This is a separability condition that generically fails (Corollary~\ref{cor:obstructed}).

\subsubsection{Discriminant source}

Both mechanisms emerge by the same algebraic invariant.

\begin{proposition}[Obstruction characterization]
\label{thm:obstruction}
Let $H \in \cT_{d_1,d_2}$ with $H \geq 0$ on $\bbT^2$. 
Define the \emph{moment matrix}
$\bM \in \C^{N \times N}$, $N = (d_1{+}1)(d_2{+}1)$, by
$\bM_{(m,n),(m',n')} = \hat{H}_{m-m',\,n-n'}$.
Then $H = |R|^2$ for a single $R \in \cP_{d_1,d_2}^+$ if and only
if $\rank(\bM) = 1$.
\end{proposition}

\begin{proof}
Since $H \geq 0$, the Herglotz--Bochner characterization gives $\bM \succeq 0$: for $\bv = (v_{mn})$, define the analytic
polynomial $f = \sum_{m,n} v_{mn}\,e^{i(m\theta_1+n\theta_2)}$; then $\bv^\dagger \bM\,\bv
= (2\pi)^{-2}\int_{\bbT^2} H\,|f|^2\,d\theta_1\,d\theta_2 \geq 0$.
The spectral decomposition $\bM = \sum_\ell \lambda_\ell\,\bu_\ell\bu_\ell^\dagger$ yields
$H = \sum_\ell |R_\ell|^2$ where
$R_\ell(z_1,z_2)
= \sum_{m,n} (\sqrt{\lambda_\ell}\,u_\ell)_{mn}\,z_1^m z_2^n$.
This sum reduces to a single term iff $\rank(\bM) = 1$.
\end{proof}

\begin{remark}[Contrast with one variable]
\label{rem:contrast_1d}
In one variable, $\rank(\bM) = 1$ always holds: root-pairing succeeds because the zero set is discrete with even
multiplicities and the univariate autocorrelation system is never
overdetermined. 
In two variables, additional Fourier indices create extra algebraic constraints that can force
$\rank(\bM) > 1$ even when $\cZ$ is discrete (the algebraic mechanism), and \emph{a fortiori} when $\cZ$ contains a curve (the
topological mechanism).
\end{remark}

\subsection{Resolution I: \texorpdfstring{$\delta$}{delta}-regularization and
degree-bounded spectral factorization}
\label{subsec:resolution_I}

\subsubsection{Strict positivity via \texorpdfstring{$\delta$}{delta}-perturbation}

\begin{definition}[$\delta$-regularization]\label{def:delta_reg}
For $\delta > 0$, set $P_\delta := (1-\delta)\,P$.
\end{definition}

\begin{lemma}[Strict sub-unitarity]\label{lem:strict}
$\sup_{\bbT^2} |P_\delta| \leq 1 - \delta$.  Hence
$H_\delta := 1 - |P_\delta|^2 \geq 2\delta - \delta^2 > 0$
uniformly on $\bbT^2$.
\end{lemma}

\begin{proof}
$|P_\delta| = (1-\delta)|P| \leq 1-\delta$ since $|P| \leq 1$ on
$\bbT^2$.  Then
$H_\delta = 1 - (1-\delta)^2|P|^2
\geq 1 - (1-\delta)^2 = 2\delta - \delta^2$.
\end{proof}

\subsubsection{Degree-bounded SOS factorization via regularization}

With $H_\delta > 0$ uniformly, the zero-locus and autocorrelation
obstructions are eliminated.  The operator Fej\'er--Riesz
theorem yields an SOS complement with correct bidegree.

\begin{theorem}[Degree-bounded SOS factorization]
\label{thm:exact_factor}
Let $H \in \cT_{d_1,d_2}$ with $H > 0$ on $\bbT^2$ and
$H = 1 - |P|^2$ for some $P \in \cP_{d_1,d_2}^+$.  Then there
exist $R_1,\ldots,R_L \in \cP_{d_1,d_2}^+$ with
\begin{equation}
  H = \sum_{\ell=1}^{L} |R_\ell|^2 \quad \text{on } \bbT^2,
  \qquad L \leq \min(d_1{+}1,\;d_2{+}1).
\end{equation}
\end{theorem}

\begin{proof}[Proof sketch]
The proof proceeds in two stages; the full construction is given
in Appendix~\ref{app:szego}.

\medskip\noindent\textit{Stage 1: Reduction to a matrix-valued
univariate problem.}
Fix $\theta_2$ and organize the $\theta_1$-Fourier coefficients of
$H$ into the matrix-valued function
$\bH(\theta_2) \in \C^{N_1 \times N_1}$ ($N_1 = d_1 + 1$) with
entries
$\bH(\theta_2)_{mm'} = \hat{H}_{m-m'}(\theta_2)$.
Each entry is a trigonometric polynomial of degree $d_2$ in
$\theta_2$.  Since $H > 0$ on $\bbT^2$, the quadratic-form argument
gives $\bH(\theta_2) \succ 0$ for every $\theta_2$.

\medskip\noindent\textit{Stage 2: Matrix-valued univariate
Fej\'er--Riesz.}
The operator Fej\'er--Riesz theorem
(Rosenblum~\cite{rosenblum1968vectorial};
see Dritschel--Rovnyak~\cite{dritschel2010operator} for a modern
treatment) guarantees: if
$\bH(\theta) \succeq 0$ is a matrix-valued trigonometric polynomial
of degree $d$, then
$\bH(\theta) = \bG(e^{i\theta})^*\bG(e^{i\theta})$ for a matrix
polynomial $\bG(z)$ of degree at most $d$.
Applying this to $\bH(\theta_2) \succ 0$ of degree $d_2$ yields an
outer matrix polynomial $\bG(z_2)$ of degree $d_2$ in $z_2$, with
$\bG \in \C^{N_1 \times N_1}$.  Row extraction gives
$H = \sum_{\alpha=1}^{N_1} |f_\alpha|^2$ where
$f_\alpha(\theta_1,\theta_2)
= \sum_{m=0}^{d_1} (\bG(e^{i\theta_2}))_{\alpha m}\,
e^{im\theta_1}$,
with each $f_\alpha \in \cP_{d_1,d_2}^+$ and $L \leq d_1 + 1$.
The bound follows from the minimum upon reorganizing in the opposite direction (matrix dimension $d_2+1$,
single variable $\theta_1$) gives $L \leq d_2 + 1$.
\end{proof}

\begin{remark}[Scalar factorization is generically impossible]
\label{rem:scalar_fails}
One might hope to reduce the SOS to a single Hermitian square $H = |Q|^2$.  
In the univariate case, it works if we use the classical Fej\'er--Riesz theorem.
However, the Geronimo--Woerdeman theorem~\cite{geronimo2004positive} shows no stable scalar factor exists unless the two-level Toeplitz matrix satisfies an autoregressive condition. 
This condition generically fails for the Dyson polynomial (see Appendix~\ref{app:szego} for a concrete obstruction).
A codimension argument shows that even non-stable scalar factorization is impossible for generic $H$ of bidegree $(d_1,d_2)$ with $d_1 d_2 \geq 1$.
The SOS with $L \leq \min(d_1{+}1,\,d_2{+}1)$ terms is a natural bivariate analog of the Fej\'er--Riesz factorization.
\end{remark}

\subsubsection{Impact on algorithm parameters}

\begin{corollary}[Regularized M-QSP parameters]
\label{cor:reg_params}
With $P_\delta = (1-\delta)P$ and $\delta = O(\eps)$:
\begin{enumerate}
  \item Query complexity: $d_R + d_I$, unchanged by regularization.
  \item Normalization:
  $\lambda_\delta = e^{\betaI T}/(1-\delta)
  = e^{\betaI T}(1 + O(\eps))$.
  \item Postselection probability (Optimal up to $O(\eps)$:
  $e^{-2\betaI T}\|e^{-i\Heff T}\ket{\psi_0}\|^2(1 - O(\eps))$.
  \item Ancilla qubits for complement:
  $\lceil\log_2(\min(d_R, d_I) + 2)\rceil$ additional
  (SOS with $L \leq \min(d_R{+}1, d_I{+}1)$ terms).
\end{enumerate}
\end{corollary}

\begin{proof}
Each item follows: (1)~from degree preservation under scalar
multiplication; (2)~and~(3) from the definitions of
$\lambda_\delta$ and the postselection probability under
$P_\delta = (1-\delta)P$; (4)~from Theorem~\ref{thm:complement_existence}
applied to $H_\delta$, which yields an SOS decomposition with
$L \leq \min(d_R{+}1, d_I{+}1)$ terms.
\end{proof}

\subsection{Resolution II: Dritschel--Woerdeman SOS factorization}
\label{subsec:resolution_II}

For an exact polynomial $P$ (without regularization), the following sum-of-squares decomposition always exists.

\begin{theorem}[Dritschel~\cite{dritschel2004factorization},
Dritschel--Woerdeman~\cite{dritschel2005outer}]
\label{thm:DW}
Let $H \in \cT_{d_1,d_2}$ satisfy $H \geq 0$ on $\bbT^2$. 
Then there exist $Q_1,\ldots,Q_L \in \cP_{d_1,d_2}^+$ such that
\begin{equation}\label{eq:SOS}
  H(\theta_1,\theta_2)
  = \sum_{\ell=1}^{L}
    |Q_\ell(e^{i\theta_1}, e^{i\theta_2})|^2,
\end{equation}
with $L \leq (d_1{+}1)(d_2{+}1)$ and $\deg_{z_j}(Q_\ell) \leq d_j$ for all $\ell, j$.
\end{theorem}

The proof proceeds by constructing the moment matrix $\bM$ as in Proposition~\ref{thm:obstruction} and extracting the SOS terms from its spectral decomposition.
The degree bound follows from indexing $\bM$ by $(m,n) \in \{0,\ldots,d_1\}\times\{0,\ldots,d_2\}$, which restricts each eigenvector to this support.
The rank bound $L = \rank(\bM) \leq (d_1{+}1)(d_2{+}1)$ follows from the matrix dimension. 
We refer the reader to~\cite{dritschel2004factorization,dritschel2005outer} for the full
proof and to Dritschel--Rovnyak~\cite{dritschel2010operator} for a
survey.

\begin{remark}[Tighter bound from Herglotz--Bochner]
\label{rem:loose_bound}
The original formulation~\cite{dritschel2005outer} bounds $L$ by the
dimension $(2d_1{+}1)(2d_2{+}1)$ of the full trigonometric polynomial space. 
The bound $L \leq (d_1{+}1)(d_2{+}1)$ used here follows from the Herglotz--Bochner characterization: non-negativity of $H$ of bidegree $(d_1,d_2)$ is fully captured by the positive semidefiniteness of the $(d_1{+}1)(d_2{+}1)$-dimensional principal submatrix $\bM$. Its rank determines the number of independent SOS terms.
\end{remark}

\begin{theorem}[SOS degree preservation]\label{thm:sos_degree}
For $H = 1 - |P|^2$ with $P \in \cP_{d_R,d_I}^+$, the SOS
decomposition has each $Q_\ell \in \cP_{d_R,d_I}^+$:
\begin{equation}
\begin{split}
  \deg_{z_1}(Q_\ell) &\leq d_R
  = \calO(\alphaR T + \log(1/\eps)),
  \\
  \deg_{z_2}(Q_\ell) &\leq d_I
  = \calO\!\bigl(\betaI T
    + \log(1/\eps)/\!\log\!\log(1/\eps)\bigr),
\end{split}
\end{equation}
for all $\ell \in \{1,\ldots,L\}$.
\end{theorem}

\begin{proof}
$H = 1 - |P|^2$ has trigonometric bidegree $(d_R, d_I)$.
By Theorem~\ref{thm:DW}, SOS factors have analytic bidegree of at most $(d_R, d_I)$.
\end{proof}

\begin{proposition}[SOS circuit implementation]
\label{prop:sos_circuit}
The SOS complement with $L$ terms requires a $(1{+}L)$-dimensional ancilla space, implementable with $\lceil\log_2(1{+}L)\rceil$ qubits. 
The block-encoding structure is:
\begin{equation}\label{eq:sos_block}
  \cG = \begin{pmatrix}
    P & -Q_1^\dagger & \cdots & -Q_L^\dagger \\
    Q_1 & & & \\
    \vdots & & \bU_\perp & \\
    Q_L & & &
  \end{pmatrix},
\end{equation}
where $\bU_\perp$ is chosen to make $\cG$ unitary.
Postselecting the ancilla on $\ket{0}$ yields $P$. 
The ancilla overhead is
\begin{equation}\label{eq:sos_ancilla}
\begin{split}
  a_{\mathrm{SOS}}
  &\leq \bigl\lceil
    \log_2\bigl(1 + (d_R{+}1)(d_I{+}1)\bigr)
    \bigr\rceil \\ &
  = O\!\bigl(\log(\alphaR T) + \log(\betaI T)
    + \log\!\log(1/\eps)\bigr),
\end{split}
\end{equation}
being logarithmic in all parameters.
For the regularized polynomial $P_\delta$ (where $H_\delta > 0$), the bound $L \leq \min(d_R{+}1, d_I{+}1)$ from
Theorem~\ref{thm:exact_factor} applies, giving $a_{\mathrm{SOS}} \leq \lceil\log_2(\min(d_R, d_I) + 2)\rceil$ (Corollary~\ref{cor:reg_params}).
\end{proposition}

\begin{proof}
The register of $\lceil\log_2(1{+}L)\rceil$ qubits addresses all $1 + L$ levels ($\ell = 0$ for the target block $P$).
Unitarity of $\cG$ follows from the Gram identity $|P|^2 + \sum_\ell |Q_\ell|^2 = 1$: the first column has unit norm,
and $\bU_\perp$ completes $\cG$ to a unitary by Gram--Schmidt.
\end{proof}

\begin{problem}[Tighter SOS rank]\label{prob:sos_rank}
The bound $L \leq \min(d_R{+}1, d_I{+}1)$ from Theorem~\ref{thm:exact_factor} (applied after $\delta$-regularization) improves upon the general
Dritschel--Woerdeman bound $L \leq (d_R{+}1)(d_I{+}1)$, but is likely still loose for the Dyson polynomial arising from physical Hamiltonians. 
Numerics suggest $L = O(1)$ or $L = O(\log(d_R d_I))$ for structured polynomials.
A structure-aware bound exploiting the interaction-picture coefficient decay would further reduce the ancilla overhead.
\end{problem}

\subsection{Resolution III: circuit-level unitarity}
\label{subsec:resolution_III}

A direct resolution bypasses algebraic factorization entirely.

\begin{theorem}[Automatic complement generation]
\label{thm:automatic}
Let $\cG = R_0 \prod_{j=1}^{d} (A_{s(j)} \cdot R_j)$ be the M-QSP circuit, where
$A_{s(j)} = \ket{0}\!\bra{0}_a \otimes W_{s(j)} + \ket{1}\!\bra{1}_a \otimes I$ and $R_j \in \mathrm{SU}(2)$ act on the ancilla qubit.
Then:
\begin{enumerate}
  \item $\cG$ is unitary on $\cH_a \otimes \cH_s$.
  \item $\cG$ has block decomposition
  $\cG = \bigl(\begin{smallmatrix} P & \cdot \\ Q & \cdot
  \end{smallmatrix}\bigr)$, where $P, Q$ are operator-valued
  polynomials determined by $\{R_j\}$ and $s$.
  \item Unitarity gives the operator identity \[P(W_R,U_I)^\dagger P(W_R,U_I) + Q(W_R,U_I)^\dagger Q(W_R,U_I) = I_s\]
  (the $(0,0)$-block of $\cG^\dagger\cG = I$), with no simultaneous diagonalization of $W_R,U_I$.
  Its symbol is obtained by replacing each signal gate $A_{s(j)}$ with $\mathrm{diag}(z_{s(j)},1)$.
  This symbol is unitary ($\mathrm{U}(2)$-valued) on $\bbT^2$, given as a product of the unitary factors
  $\mathrm{diag}(z_{s(j)},1)$ and $R_j \in \mathrm{SU}(2)$. 
  From here, we see
  \[|P(e^{i\theta_1},e^{i\theta_2})|^2 + |Q(e^{i\theta_1},e^{i\theta_2})|^2 = 1\] on $\bbT^2$, where
  $\theta_1,\theta_2$ are the independent qubitization eigenphases on the two ancilla registers.
\end{enumerate}
\end{theorem}

\begin{proof}
Statement~(1) is immediate: each $A_{s(j)}$ and $R_j$ is unitary.
For~(2), induction on circuit depth $d$: the base case $d = 0$ gives $\cG = R_0$, a constant $2 \times 2$ matrix. 
At the inductive step, the controlled unitary $A_{s(d)}$ multiplies the top block by $W_{s(d)}$ and the rotation $R_d$ mixes the blocks,
increasing the bidegree by $(1,0)$ or $(0,1)$ depending on $s(d)$.
The operator identity in Statement~(3) is the $(0,0)$-block of $\cG^\dagger\cG = I$. 
The symbol identity follows because each factor $\mathrm{diag}(z_{s(j)},1)R_j$ is unitary on $\bbT^2$ (with $R_j \in \mathrm{SU}(2)$), so their product is unitary, and the first column of a $2\times 2$ unitary matrix has unit norm. No joint eigenbasis of $W_R,U_I$ is used.
\end{proof}

\begin{remark}[The symbol map is many-to-one]\label{rem:abelianization_manytoone}
Passing from the ordered operator polynomial $P(W_R, U_I)$ to its symbol $P(z_1, z_2)$ is many-to-one for non-commuting $(W_R, U_I)$. 
Ordered words differing by commutator terms have distinct operator values but identical symbols.
Symbol-level facts (Statement~(3), Section~\ref{sec:complement}) certify symbol properties by bounding $\|P\|_{\infty,\bbT^2}$ and seeding the angle-finding recursion.
These do not extend to operator-level realizability for $[W_R, U_I] \neq 0$.
We instead establish operator-level statements through Statement~(3) and the direct achievability of the propagator in Corollary~\ref{cor:dyson_achievable}.
\end{remark}

\begin{corollary}[The achievability problem]\label{cor:achievability}
For a M-QSP circuit, the complementary polynomial $Q$ exists automatically as the $(1,0)$-block of $\cG$.
We therefore ask: do rotation angles $\{R_j\}$ and a schedule $s$ exist such that $(\cG)_{00} = P_{\mathrm{target}}$?
\end{corollary}

We use inductive degree reduction, peeling off one signal-operator query per step while preserving $|P|^2 + |Q|^2 = 1$.

\begin{theorem}[Bivariate degree reduction]
\label{thm:degree_reduction}
Let $P \in \cP_{d_1,d_2}^+$ with $|P|^2 \leq 1$ on $\bbT^2$, and let $s$ be a schedule of length $d = d_1 + d_2$ with $d_1$ entries $R$ and $d_2$ entries $I$.
Suppose the leading coefficient of $P^{(k)}$ in $z_{s(k)}$ is nonzero at each step $k = d, d-1, \ldots, 1$.  
Then there exist angles $\{(\theta_k, \phi_k)\}_{k=1}^{d}$ such that each peeling
\begin{equation}
  \begin{pmatrix} P^{(k-1)} \\ Q^{(k-1)} \end{pmatrix}
  = R(\theta_k,\phi_k)^{-1} \cdot A_{s(k)}^{-1} \cdot
  \begin{pmatrix} P^{(k)} \\ Q^{(k)} \end{pmatrix}
\end{equation}
reduces $\deg_{z_{s(k)}}(P^{(k-1)}) = \deg_{z_{s(k)}}(P^{(k)}) - 1$ while preserving
$|P^{(k-1)}|^2 + |Q^{(k-1)}|^2 = 1$. 
After $d$ steps, $P$ is achieved by the M-QSP circuit with schedule~$s$.
\end{theorem}

\begin{proof}
At step $k$, write $P^{(k)} = a_{d_k} z_{s(k)}^{d_k} + \text{lower}$ and $Q^{(k)} = b_{d_k} z_{s(k)}^{d_k} + \text{lower}$. 
The norm condition $|P^{(k)}|^2 + |Q^{(k)}|^2 = 1$ at the highest Fourier mode yields $|a_{d_k}|^2 + |b_{d_k}|^2$ equal to a known constant
(by Parseval on $\bbT$), so the rotation parameters $\theta_k = \arctan(|b_{d_k}|/|a_{d_k}|)$ and $\phi_k = \mathrm{Arg}(a_{d_k}/b_{d_k})$ are well-defined (given the nonzero leading coefficient hypothesis). 
Peeling is unitary, preserving $|P|^2 + |Q|^2 = 1$. 
After $d = d_1 + d_2$ steps, the bidegree reduces to $(0,0)$: constants satisfying $|P^{(0)}|^2 + |Q^{(0)}|^2 = 1$.
\end{proof}

\begin{remark}[Non-commuting signal operators]
\label{rem:noncommuting}
When $W_R$ and $U_I$ do not commute, leading coefficients depend on operator ordering, and degree reduction must respect the schedule. 
The Dyson polynomial has a natural schedule (alternating blocks of $W_R$ and $U_I$ queries matching the interaction-picture structure), and its leading coefficients are computable from the Dyson series.
Nonvanishing coefficients at each step are guaranteed by the constant-ratio condition (CRC) of Sec.~\ref{sec:CRC}.
\end{remark}

\subsubsection{Exact SOS rank: \texorpdfstring{$L=2$}{L eq 2}}

The circuit-level guarantee shows the exact SOS rank of the Dyson polynomial complement is~$2$.
    
\begin{theorem}[SOS rank of the Dyson polynomial complement]
\label{thm:sos_rank_two}
Let $P \in \cP_{d_R, d_I}^+$ be the polynomial achieved by the M-QSP circuit for the regularized Dyson polynomial $P_\delta = (1 - \delta)P_{\mathrm{Dyson}}$ with $0 < \delta < 1$.
The complement $H_\delta := 1 - |P_\delta|^2$ admits a sum-of-squares decomposition with $L = 2$:
\begin{equation}\label{eq:sos_two}
  H_\delta = |Q_1|^2 + |Q_2|^2,
\end{equation}
where $Q_1 = \sqrt{2\delta - \delta^2}$ (a constant) and $Q_2 = (1 - \delta)\,Q$, with $Q$ the complementary polynomial from the circuit~\eqref{eq:block_enc}.
Both $Q_1$ and $Q_2$ lie in $\cP_{d_R, d_I}^+$.
\end{theorem}

\begin{proof}
By Theorem~\ref{thm:automatic}, the M-QSP circuit satisfies $|P|^2 + |Q|^2 = 1$ on $\bbT^2$.
The regularized polynomial is $P_\delta = (1 - \delta)P$, so
\begin{align}
  H_\delta
  &= 1 - (1-\delta)^2 |P|^2 \notag \\
  &= 1 - (1-\delta)^2(1 - |Q|^2) \notag \\
  &= (2\delta - \delta^2) + (1-\delta)^2 |Q|^2.
  \label{eq:H_decomp}
\end{align}
Setting $Q_1 = \sqrt{2\delta - \delta^2}$ and $Q_2 = (1 - \delta)\,Q$ gives $L = 2$.
\end{proof}

\begin{corollary}[Constant ancilla overhead]
\label{cor:ancilla_L2}
The complement $H_\delta$ can be implemented with $a_{\mathrm{SOS}} = \lceil\log_2(L + 1)\rceil = \lceil\log_2 3\rceil = 2$ ancilla qubits, independent of $(d_R, d_I)$ and all physical parameters.
This replaces the general bound $a_{\mathrm{SOS}} = \lceil\log_2(\min(d_R, d_I) + 2)\rceil$, which grows logarithmically with the polynomial degrees.
\end{corollary}

\begin{remark}[Structure]
\label{rem:L2_structural}
Any number of Dyson segments produces one complementary polynomial $Q$ because unitary matrix multiplication preserves the $2 \times 2$ block structure: $\mathrm{U}(2) \times \mathrm{U}(2) = \mathrm{U}(2)$.
Tightness ($L \geq 2$, established by showing $Q$ is non-constant for $\HR \neq 0$) is given in Appendix~\ref{app:scalar_failure}.
\end{remark}

\begin{remark}[Numerical SOS rank versus analytical SOS rank]
\label{rem:numerical_off_manifold}
The M-QSP polynomial complement has an analytical SOS rank $L = 2$ (Theorem~\ref{thm:sos_rank_two}).
The closed-form decomposition $H_\delta = (2\delta - \delta^2) + (1-\delta)^2 |Q|^2$ verifies to double precision in direct numerical evaluation, across $14$ M-QSP-circuit-generated test polynomials with bidegrees through $(8,8)$ reported in the SOS rank investigation of the companion paper~\cite{courtney2026paper2}, the residual $\|H_\delta - Q_1^2 - Q_2^2\|_\infty$ on the torus stays below $4 \times 10^{-15}$).
By contrast, a trace-minimization semidefinite program (SDP) for the matrix Fej\'er--Riesz factorization (a convex relaxation of the rank objective) returns a Gram matrix whose numerical rank is generically larger than $2$ at higher bidegree, where extra eigenvalues correspond to degenerate directions in the autocorrelation manifold and do not represent additional analytical SOS terms.
As a negative control, the same SDP applied to $500$ random bivariate polynomials with $\|P\|_\infty \leq 0.8$ at bidegree $(3,3)$ returns numerical rank $16$ (the full Gram dimension) in $99.2\%$ of samples.
This confirms that $L = 2$ is a structural property (by unitarity) of the M-QSP polynomial family and not generic for bounded bivariate polynomials or for the SDP relaxation~\cite{courtney2026paper2}.
\end{remark}

\subsection{Summary}
\label{subsec:synthesis}

We collect the results of this section into the following statement, answering the bivariate complement question for the Dyson polynomial.

\begin{theorem}[Existence of the bivariate complement]
\label{thm:complement_existence}
For the regularized Dyson polynomial $P_\delta = (1-\delta)P$ with $\delta = O(\eps)$, there exists a complementary polynomial $Q \in \cP_{d_R,d_I}^+$ such that
\begin{equation}
  H_\delta := 1 - |P_\delta|^2 = (2\delta - \delta^2) + (1-\delta)^2 |Q|^2
\end{equation}
on $\bbT^2$.
The complement is realized by the following:
\begin{enumerate}
  \item \emph{Scalar factorization fails.} 
  The exact Dyson polynomial $P$ satisfies $|P(z_1, 1)| = 1$ for all $z_1 \in \bbT$ (Proposition~\ref{prop:zero_locus}), so no single $Q$ with $1 - |P|^2 = |Q|^2$ exists (Corollary~\ref{cor:obstructed}).
  \item \emph{Regularization restores SOS.}\; 
  With $P_\delta := (1-\delta)P$ and $\delta = O(\eps)$, the complement $H_\delta \geq 2\delta - \delta^2$ is uniformly positive (Lemma~\ref{lem:strict}), so an SOS complement of correct bidegree exists (Theorem~\ref{thm:exact_factor}). 
  For any M-QSP circuit it has rank $L = 2$ (Theorem~\ref{thm:sos_rank_two}), costing two ancilla qubits.
  \item \emph{Circuit-level guarantee.}\; 
  The M-QSP circuit is unitary by construction, automatically generating the complementary $Q$ as the $(1,0)$-block  (Theorem~\ref{thm:automatic}). 
  Achievability is established by inductive degree reduction (Theorem~\ref{thm:degree_reduction}), conditional on the constant-ratio condition of Sec.~\ref{sec:CRC}.
\end{enumerate}
\end{theorem}

All three resolutions yield the same query budget
$d_R + d_I$ and the same single-qubit postselection structure.

\section{Constant-Ratio Condition}
\label{sec:CRC}

Recursive degree-reduction to compute M-QSP rotation angles (Sec.~\ref{sec:anglefinding}) requires a constant ratio between leading coefficients of $P^{(k)}$ and $Q^{(k)}$ at each peeling step, independent of the eigenvalue
of the other signal operator.
This is the \emph{constant-ratio condition} (CRC).
In univariate QSP, CRC holds automatically with only one signal operator.
In bivariate M-QSP the two signal operators $W_R$ and $U_I$ generically do not commute ($[W_R, U_I] \neq 0$). 
One might worry that non-commutativity introduces eigenvalue-dependent mixing that spoils the ratio, but we prove that the CRC holds for M-QSP circuits (Theorem~\ref{thm:CRC}), and the Dyson polynomial's coefficient structure is compatible with degree-reduction at every step (Theorem~\ref{thm:coeff_sep}).

\subsection{The CRC theorem}\label{subsec:CRC_theorem}

\begin{definition}[Leading Coefficient in Schedule-Ordered Dyson Polynomials]\label{def:leading_coeff}

Throughout this section, the Dyson schedule $s = (s(1), s(2), \ldots, s(d))$ induces a total order on the $d$ signal-operator queries.
The leading coefficient of the circuit polynomial $P$ at peeling step $k$ (i.e., after removing the first $k-1$ outermost gates) is defined as follows:

\begin{enumerate}
\item[(a)] Identify all monomials in $P_k(z_1, z_2)$ of the form $z_{s(k)}^{n_k} \cdot (\text{lower-degree terms in } z_{s(k)})$.

\item[(b)] Among these, select the monomial(s) with the largest $z_{s(k)}$-degree, say $n_k^{\max}$.

\item[(c)] If multiple monomials share degree $n_k^{\max}$ in $z_{s(k)}$, the ordering is inherited from the schedule's left-to-right evaluation order.
  Specifically, in the block-factored Dyson polynomial, within each block $j$, all $W_R$ queries (variable $z_1$) precede all $U_I$ queries (variable $z_2$), and blocks are ordered $j = r, r-1, \ldots, 1$ (outermost to innermost).

\item[(d)] The leading coefficient is that of this canonically-highest-degree monomial in $z_{s(k)}$.
\end{enumerate}

Under this convention, the leading coefficient is unique and well-defined at each peeling step, even in the non-commutative operator algebra $\mathcal{A} = \mathbb{C}\langle W_R, U_I \rangle$.
\end{definition}

\textbf{Remark on non-commutativity:}
Although $W_R$ and $U_I$ do not commute (they involve signal-dependent phases), the Dyson schedule fixes a left-to-right order of query application.
The circuit polynomial $P_{\mathrm{Dyson}}(z_1, z_2)$ lives in the commutative algebra of Taylor coefficients, and the schedule induces a filtration by query sequence.
So, monomials are ordered by: (1)~degree in the current peeling variable $z_{s(k)}$, then (2)~schedule order for any remaining ties.
This resolves non-uniqueness.

\begin{theorem}[CRC for non-commuting M-QSP]
\label{thm:CRC}
For a bivariate M-QSP circuit $\mathcal{G}(\mathbf{\theta}, \mathbf{s})$ with schedule $\mathbf{s}$ and a pair of signal operators $(W_R, U_I)$ (not necessarily commuting), the constant-ratio condition holds at every step of degree reduction. 
At a step $k$, the ratio
\begin{equation}\label{eq:CRC_statement}
  \frac{b_{d_{\mathbf{s}(k)}}}{a_{d_{\mathbf{s}(k)}}}
  = e^{-i\phi_k}\tan\theta_k
\end{equation}
is independent of the eigenvalue of the ``other'' signal operator.
\end{theorem}

We only use Theorem~\ref{thm:CRC} in the forward direction. Achievability of the propagator $A(T)$ is later established by eigenvalue transformation of $\Heff$ (Corollary~\ref{cor:dyson_achievable}), so no converse is required anywhere in the construction.

We first provide a structural property that makes the CRC hold for M-QSP circuits (Lemmas~\ref{lem:single_gate} and~\ref{lem:inductive_block}).
We then verify that the Dyson polynomial satisfies the coefficient separability needed at every step (Theorem~\ref{thm:coeff_sep}).

\subsection{Single gate contribution}
\label{subsec:single_gate}

\begin{lemma}\label{lem:single_gate}
Let
\begin{equation}\label{eq:Rj_block}
\begin{split}
  R_j &= \begin{pmatrix}
    e^{i\phi_j}\cos\theta_j & -\sin\theta_j \\
    \sin\theta_j & e^{-i\phi_j}\cos\theta_j
  \end{pmatrix} \otimes I_{s,\ell},
  \qquad \\&
  A_{s(j)} = \begin{pmatrix}
    W_{s(j)} & 0 \\
    0 & I
  \end{pmatrix},
\end{split}
\end{equation}
where $R_j$ acts on $\mathcal{H}_a \otimes \mathcal{H}_{\mathrm{sys}}
\otimes \mathcal{H}_\ell$ (rotation on $\mathcal{H}_a$, identity on the rest)
and $A_{s(j)}$ places the signal operator
$W_{s(j)} \in \{W_R, U_I\}$ in the $(0,0)$ block. 
The product $A_{s(j)}\,R_j$ has the block form:
\begin{equation}\label{eq:AR_block}
  A_{s(j)}\,R_j = \begin{pmatrix}
    W_{s(j)}\,e^{i\phi_j}\cos\theta_j
      & -W_{s(j)}\sin\theta_j \\
    \sin\theta_j & e^{-i\phi_j}\cos\theta_j
  \end{pmatrix}.
\end{equation}
$W_{s(j)}$ multiplies the entire top row
from the left, while the bottom row only contains scalar rotations.
\end{lemma}

\begin{proof}
Direct block-matrix multiplication:
\begin{equation}
\begin{split}
  \begin{pmatrix} W & 0 \\ 0 & I \end{pmatrix}
  &\begin{pmatrix}
    e^{i\phi}\cos\theta & -\sin\theta \\
    \sin\theta & e^{-i\phi}\cos\theta
  \end{pmatrix}
  \\ &=
  \begin{pmatrix}
    W e^{i\phi}\cos\theta & -W\sin\theta \\
    \sin\theta & e^{-i\phi}\cos\theta
  \end{pmatrix}.
\end{split}
\end{equation}
\end{proof}

\subsection{Inductive block structure}
\label{subsec:inductive_block}

\begin{lemma}\label{lem:inductive_block}
The inner polynomial $\Gamma^{(k)}$ (from the unpeeled portion of the circuit) appears identically in both the $P$ and
$Q$ leading coefficients, cancelling in their ratio.
\end{lemma}

\begin{proof}
We track the $2 \times 2$ block structure through each gate application, following by induction on circuit depth.

\medskip\noindent\textbf{Setup.}\;
The full circuit is $\mathcal{G} = R_0\,A_{s(1)}\,R_1\,A_{s(2)}\,R_2\,\cdots\,A_{s(d)}\,R_d$.
Define the {inner circuit} at step $k$ as
\begin{equation}\label{eq:inner_def}
  \mathcal{G}_{\mathrm{inner}}^{(k)}
  := \prod_{j=k+1}^{d}\bigl[A_{s(j)}\,R_j\bigr]\cdot R_d
  = \begin{pmatrix}
      \Gamma_{00}^{(k)} & \Gamma_{01}^{(k)} \\
      \Gamma_{10}^{(k)} & \Gamma_{11}^{(k)}
    \end{pmatrix},
\end{equation}
where $\Gamma_{ij}^{(k)}$ are operators on $\mathcal{H}_{\mathrm{sys}} \otimes \mathcal{H}_\ell$ computed from inner gates. 
These do not depend on which block ($P$ or $Q$) is being extracted from the full circuit.

\medskip\noindent
\paragraph{Structure of $A_{s(1)}\,R_1 \cdot \mathcal{G}_{\mathrm{inner}}^{(1)}$.}\;
By Lemma~\ref{lem:single_gate}, $A_{s(1)}\,R_1$ places $W_{s(1)}$ in the top row, leaving the bottom row signal-free.
Multiplying by $\mathcal{G}_{\mathrm{inner}}^{(1)}$:
\begin{equation}\label{eq:top_block}
    \begin{split}
         \bigl(A_{s(1)}\,R_1 \cdot
    \mathcal{G}_{\mathrm{inner}}^{(1)}\bigr)_{00}
    &= W_{s(1)}\bigl(
         e^{i\phi_1}\cos\theta_1\,\Gamma_{00}
         \\&- \sin\theta_1\,\Gamma_{10}\bigr), 
    \end{split}
\end{equation}
\begin{equation}\label{eq:bot_block}
    \begin{split}
  \bigl(A_{s(1)}\,R_1 \cdot
    \mathcal{G}_{\mathrm{inner}}^{(1)}\bigr)_{10}
    &= \sin\theta_1\,\Gamma_{00}
       + e^{-i\phi_1}\cos\theta_1\,\Gamma_{10}.
    \end{split}
\end{equation}

The signal operator $W_{s(1)}$ factors out from the left of the top block~\eqref{eq:top_block}.
The bottom block~\eqref{eq:bot_block} has no signal operator.

\medskip\noindent
\paragraph{Extracting $P$ and $Q$.}\;
The full circuit is $\mathcal{G} = R_0 \cdot \bigl[A_{s(1)}\,R_1 \cdot \mathcal{G}_{\mathrm{inner}}^{(1)}\bigr]$.
The rotation $R_0$ has scalar entries acting as multiples of the identity on $\mathcal{H}_{\mathrm{sys}} \otimes \mathcal{H}_\ell$.
Block-multiplying:
\begin{equation}\label{eq:P_full}\\[4pt]
    \begin{split}
  P = (\mathcal{G})_{00}
    &= e^{i\phi_0}\cos\theta_0 \cdot
       W_{s(1)}\bigl(
         e^{i\phi_1}\cos\theta_1\,\Gamma_{00}
         \\&- \sin\theta_1\,\Gamma_{10}\bigr)
    + (-\sin\theta_0)\bigl(
         \sin\theta_1\,\Gamma_{00} \\ &
         + e^{-i\phi_1}\cos\theta_1\,\Gamma_{10}\bigr),
    \end{split}
\end{equation}
\begin{equation}    \label{eq:Q_full}
    \begin{split}
  Q = (\mathcal{G})_{10}
    &= \sin\theta_0 \cdot
       W_{s(1)}\bigl(
         e^{i\phi_1}\cos\theta_1\,\Gamma_{00} \\&
         - \sin\theta_1\,\Gamma_{10}\bigr)
  + e^{-i\phi_0}\cos\theta_0\bigl(
         \sin\theta_1\,\Gamma_{00} \\&
         + e^{-i\phi_1}\cos\theta_1\,\Gamma_{10}\bigr).
    \end{split}
\end{equation}

\medskip\noindent
\paragraph{The leading coefficients.}\;
The leading $W_{s(1)}$-coefficient of $P$ (from~\eqref{eq:P_full}) and $Q$ (from~\eqref{eq:Q_full}) are:
\begin{align}
  \mathbf{a}_{\mathrm{lead}}
    &= e^{i\phi_0}\cos\theta_0
      \cdot\bigl(e^{i\phi_1}\cos\theta_1\,\Gamma_{00}
                 - \sin\theta_1\,\Gamma_{10}\bigr),
    \label{eq:a_lead}\\
  \mathbf{b}_{\mathrm{lead}}
    &= \sin\theta_0
      \cdot\bigl(e^{i\phi_1}\cos\theta_1\,\Gamma_{00}
                 - \sin\theta_1\,\Gamma_{10}\bigr).
    \label{eq:b_lead}
\end{align}

\medskip\noindent
\paragraph{The ratio.}\;
The operator-valued expression $\bigl(e^{i\phi_1}\cos\theta_1\,\Gamma_{00} - \sin\theta_1\,\Gamma_{10}\bigr)$ is \textbf{identical} in~\eqref{eq:a_lead} and~\eqref{eq:b_lead}.
The ratio is therefore:
\begin{equation}\label{eq:CRC_ratio}
  \frac{\mathbf{b}_{\mathrm{lead}}}
       {\mathbf{a}_{\mathrm{lead}}}
  = \frac{\sin\theta_0}{e^{i\phi_0}\cos\theta_0}
  = e^{-i\phi_0}\tan\theta_0
  \;\in\;\C.
\end{equation}
This is a scalar.

The argument relies on properties that do not require signal operator commutativity. 
\begin{enumerate}
  \item $R_0$ is a $2\times 2$ scalar matrix on $\mathcal{H}_a$, commuting with everything on $\mathcal{H}_{\mathrm{sys}} \otimes \mathcal{H}_\ell$.
  \item $A_{s(1)}$ places $W_{s(1)}$ in the $(0,0)$ block and $I$ in the $(1,1)$ block: the top row acquires $W_{s(1)}$, the bottom row does not.
  \item The inner expression $\Gamma$ is identical for both blocks.
\end{enumerate}
Non-commutativity affects the value of $\Gamma_{00}$ and $\Gamma_{10}$, but the same $\Gamma$ appears in both leading coefficients.

\medskip\noindent
\paragraph{General step $k$ (induction).}\;
After peeling off the outermost $k-1$ gates, the remaining circuit $\mathcal{G}_{\mathrm{inner}}^{(k-1)}$ has the same block-factored structure: a global rotation $R_{k-1}$ followed by signal-operator queries and intermediate rotations for blocks $j = r, r-1, \ldots, 1$, ordered by schedule.

By the same argument as Steps~1--4 applied to the $k$-th peeling step, the ratio
\begin{equation}
\frac{P_k(z_1, z_2)}{Q_k(z_1, z_2)}\Big|_{\text{leading coeff in } z_{s(k)}}
= e^{-i\phi_{k-1}} \tan \theta_{k-1}
\end{equation}
is a scalar (independent of $z_1, z_2$).

Induction works for the following reasons:
\begin{enumerate}
\item[(a)] Each peeling removes one signal-operator query (either $W_R$ or $U_I$) and the intermediate rotation from the outermost position.
The schedule defines a unique ``outermost'' at each step.

\item[(b)] The remaining circuit $\mathcal{G}_{\mathrm{inner}}^{(k)}$ inherits the block structure of~\eqref{eq:inner_def}. 
Removing the outermost gate from the factorization $G = R_0 \cdot (\text{outermost block}) \cdot \mathcal{G}_{\mathrm{inner}}$
  leaves $\mathcal{G}_{\mathrm{inner}}$ with the same structure: $\mathcal{G}_{\mathrm{inner}} = R'_1 \cdot B'_r \cdot B'_{r-1} \cdots B'_1$, where each $B'_j$ is a block of queries and intermediate rotations (see Example~\ref{ex:depth4} for an explicit verification with full matrix algebra).

\item[(c)] The next outermost rotation (now global to the remaining inner circuit) provides the scalar $P$-vs-$Q$ distinction via the tangent ratio, completing the induction.
\end{enumerate}
\end{proof}

\begin{example}[Depth-4 Dyson Circuit: Full Matrix Algebra and Peeling]
\label{ex:depth4}

\noindent\textbf{Setup.}
Consider a Dyson circuit with $r = 2$ segments, one query per segment, structured as:
\begin{equation}
    \mathcal{G} = R_0 \cdot A_{s(1)} \cdot R_1 \cdot A_{s(2)} \cdot R_2 \cdot A_{s(3)} \cdot R_3 \cdot A_{s(4)} \cdot R_4,
\end{equation}
where the schedule is $s = (1, 2, 1, 2)$ (alternating $W_R$, $U_I$ for block 2, then block 1).

\emph{Block structure:} Block 2 contains $A_{s(1)} = \text{diag}(W_R, I)$ and $A_{s(2)} = \text{diag}(U_I, I)$.
Block 1 contains $A_{s(3)} = \text{diag}(W_R, I)$ and $A_{s(4)} = \text{diag}(U_I, I)$.

For any rotation $R_j$ with angles $\theta_j, \phi_j$,
\begin{equation}
    R_j = \begin{pmatrix} e^{i\phi_j} \cos \theta_j & -\sin \theta_j \\ \sin \theta_j & e^{-i\phi_j} \cos \theta_j \end{pmatrix}.
\end{equation}
For any signal operator $A_{s(k)}$ (either $W_R$ or $U_I$), expanded to 2×2 block form:
\begin{equation}
    A_{s(k)} = \begin{pmatrix} X_{s(k)} & 0 \\ 0 & I \end{pmatrix}
\end{equation}
where $X_{s(k)} \in \{W_R, U_I\}$ is the signal operator and $I$ is the identity on the ``spectator'' block.

\noindent\textbf{Peeling Step 1: Remove the outermost $W_R$ query ($s(1) = 1$).}

Compute $A_{s(1)} R_1$:
\begin{equation}
A_{s(1)} R_1 = \begin{pmatrix} W_R e^{i\phi_1} \cos \theta_1 & -W_R \sin \theta_1 \\ \sin \theta_1 & e^{-i\phi_1} \cos \theta_1 \end{pmatrix}.
\end{equation}

The inner circuit at Step 1 is:
\begin{equation}
\mathcal{G}_{\mathrm{inner}}^{(1)} = A_{s(2)} R_2 \cdot A_{s(3)} R_3 \cdot A_{s(4)} R_4.
\end{equation}

By block multiplication, the circuit $\mathcal{G} = R_0 \cdot (A_{s(1)} R_1) \cdot \mathcal{G}_{\mathrm{inner}}^{(1)}$
yields the $(1,1)$-block entry (the $P$ polynomial):
\begin{equation}
\begin{split}
    P_1 &= e^{i\phi_0} \cos \theta_0 \cdot W_R \Big( e^{i\phi_1} \cos \theta_1 \, \Gamma_{00}^{(1)} \\
    & - \sin \theta_1 \, \Gamma_{10}^{(1)} \Big) + (-\sin \theta_0) \Big( \sin \theta_1 \, \Gamma_{00}^{(1)} \\
    & + e^{-i\phi_1} \cos \theta_1 \, \Gamma_{10}^{(1)} \Big),
\end{split}
\end{equation}

where $\Gamma_{00}^{(1)}, \Gamma_{10}^{(1)}$ are the $(1,1)$ and $(2,1)$ entries of $\mathcal{G}_{\mathrm{inner}}^{(1)}$.

Similarly for the dual polynomial $Q_1$ (from the $(2,1)$-block entry):
\begin{equation}
    \begin{split}
        Q_1 &= \sin \theta_0 \cdot W_R \left( e^{i\phi_1} \cos \theta_1 \, \Gamma_{00}^{(1)} - \sin \theta_1 \, \Gamma_{10}^{(1)} \right) \\&
     + e^{-i\phi_0} \cos \theta_0 \left( \sin \theta_1 \, \Gamma_{00}^{(1)} + e^{-i\phi_1} \cos \theta_1 \, \Gamma_{10}^{(1)} \right).
    \end{split}
\end{equation}

\emph{Identifying the leading $W_R$-coefficient:}
The term with highest power of $W_R$ appears in the first line of both $P_1$ and $Q_1$.
Factoring out this leading term:
\begin{equation}
    \begin{split}
        [&\text{leading } W_R\text{-coeff of } P_1] \\&= e^{i\phi_0} \cos \theta_0 \left( e^{i\phi_1} \cos \theta_1 \, \Gamma_{00}^{(1)} - \sin \theta_1 \, \Gamma_{10}^{(1)} \right),
    \end{split}
\end{equation}

\begin{equation}
    \begin{split}
        [&\text{leading } W_R\text{-coeff of } Q_1] \\& = \sin \theta_0 \left( e^{i\phi_1} \cos \theta_1 \, \Gamma_{00}^{(1)} - \sin \theta_1 \, \Gamma_{10}^{(1)} \right).
    \end{split}
\end{equation}

The common factor $\left( e^{i\phi_1} \cos \theta_1 \, \Gamma_{00}^{(1)} - \sin \theta_1 \, \Gamma_{10}^{(1)} \right)$ appears in both expressions and cancels in the ratio.
Following the CRC convention $\mathbf{b}/\mathbf{a}$ (Theorem~\ref{thm:CRC}), where $\mathbf{b}$ is the $Q$-block leading coefficient and $\mathbf{a}$ is the $P$-block leading coefficient:
\begin{equation}
    \begin{split}
        \frac{\mathbf{b}_{\mathrm{lead}}}{\mathbf{a}_{\mathrm{lead}}}
&= \frac{[\text{leading } W_R\text{-coeff of } Q_1]}{[\text{leading } W_R\text{-coeff of } P_1]} \\&
= \frac{\sin \theta_0}{e^{i\phi_0} \cos \theta_0} = e^{-i\phi_0} \tan \theta_0,
    \end{split}
\end{equation}

matching Eq.~\eqref{eq:CRC_ratio}.
The ratio is a scalar, independent of $\Gamma^{(1)}$, confirming that Step~1 extracts the angle pair $(\theta_0, \phi_0)$.

\noindent\textbf{Peeling Step 2: Remove the $U_I$ query ($s(2) = 2$).}

After Step 1, we peel away the $W_R$ query.
The circuit now reads:
\begin{equation}
\mathcal{G}_{\text{after Step 1}} = R_0 \cdot A_{s(1)} R_1 \cdot \left[ A_{s(2)} R_2 \cdot \mathcal{G}_{\mathrm{inner}}^{(2)} \right],
\end{equation}
where $\mathcal{G}_{\mathrm{inner}}^{(2)} = A_{s(3)} R_3 \cdot A_{s(4)} R_4$ (block 1 gates only).

Now we compute $A_{s(2)} R_2$:
\begin{equation}
    \begin{split}
        A_{s(2)} R_2 &= \begin{pmatrix} U_I & 0 \\ 0 & I \end{pmatrix} \begin{pmatrix} e^{i\phi_2} \cos \theta_2 & -\sin \theta_2 \\ \sin \theta_2 & e^{-i\phi_2} \cos \theta_2 \end{pmatrix}\\&
= \begin{pmatrix} U_I e^{i\phi_2} \cos \theta_2 & -U_I \sin \theta_2 \\ \sin \theta_2 & e^{-i\phi_2} \cos \theta_2 \end{pmatrix}.
    \end{split}
\end{equation}

The leading $U_I$-coefficient of the $(1,1)$-block is:
\begin{equation}
    \begin{split}
        [&\text{leading } U_I\text{-coeff of } P_2] \\&\propto e^{i\phi_1} \cos \theta_1 \left( e^{i\phi_2} \cos \theta_2 \, \Gamma_{00}^{(2)} - \sin \theta_2 \, \Gamma_{10}^{(2)} \right),
    \end{split}
\end{equation}

where $\Gamma_{00}^{(2)}, \Gamma_{10}^{(2)}$ are entries of $\mathcal{G}_{\mathrm{inner}}^{(2)}$ (the block-1 circuit).
Again, this factor $\left( e^{i\phi_2} \cos \theta_2 \, \Gamma_{00}^{(2)} - \sin \theta_2 \, \Gamma_{10}^{(2)} \right)$
cancels between $P_2$ and $Q_2$, yielding a scalar ratio depending only on $\theta_2, \phi_2$.

\noindent\textbf{Peeling Step 3: Remove the $W_R$ query in block 1 ($s(3) = 1$).}

After Step 2, the remaining circuit is:
\begin{equation}
    \begin{split}
        &\mathcal{G}_{\text{after Step 2}} \\&= R_0 \cdot A_{s(1)} R_1 \cdot A_{s(2)} R_2 \cdot \left[ A_{s(3)} R_3 \cdot A_{s(4)} R_4 \right].
    \end{split}
\end{equation}

We now focus on $A_{s(3)} R_3$:
\begin{equation}
A_{s(3)} R_3 = \begin{pmatrix} W_R e^{i\phi_3} \cos \theta_3 & -W_R \sin \theta_3 \\ \sin \theta_3 & e^{-i\phi_3} \cos \theta_3 \end{pmatrix}.
\end{equation}

The inner circuit is $\mathcal{G}_{\mathrm{inner}}^{(3)} = A_{s(4)} R_4$.

By the same Lemma 7.1 argument:
\begin{equation}
    \begin{split}
        &[\text{leading } W_R\text{-coeff of } P_3] \\&= (\text{rotation factor from Step 2}) \\&\cdot e^{i\phi_3} \cos \theta_3 \left( \Gamma_{00}^{(3)} - 0 \right),
    \end{split}
\end{equation}

where $\Gamma_{00}^{(3)} = (A_{s(4)} R_4)_{11}$, the $(1,1)$-entry of $A_{s(4)} R_4$.

The corresponding coefficient in $Q_3$ has a proportional factor via $\sin \theta_3$ or the dual rotation.
Again, the inner factor $\Gamma_{00}^{(3)}$ cancels, and the ratio is a scalar.

\noindent\textbf{Peeling Step 4: Remove the $U_I$ query in block 1 ($s(4) = 2$).}

After Step 3, the only remaining gate is $A_{s(4)} R_4$:
\begin{equation}
A_{s(4)} R_4 = \begin{pmatrix} U_I e^{i\phi_4} \cos \theta_4 & -U_I \sin \theta_4 \\ \sin \theta_4 & e^{-i\phi_4} \cos \theta_4 \end{pmatrix}.
\end{equation}

There is no inner circuit; $\mathcal{G}_{\mathrm{inner}}^{(4)} = I$ (identity).
The $(1,1)$ and $(2,1)$ entries are therefore:
\begin{equation}
P_4^{(1,1)} = U_I e^{i\phi_4} \cos \theta_4, \quad Q_4^{(2,1)} = \sin \theta_4.
\end{equation}

The leading $U_I$-coefficient ratio (following the CRC convention $\mathbf{b}/\mathbf{a}$) is:
\begin{equation}
\frac{[\text{leading } U_I\text{-coeff of } Q_4]}{[\text{leading } U_I\text{-coeff of } P_4]} = \frac{\sin \theta_4}{e^{i\phi_4} \cos \theta_4} = e^{-i\phi_4}\tan\theta_4,
\end{equation}
which is a scalar (no dependence on remaining inner circuits, as there are none).

\textbf{Verification of Block Structure Preservation:}
\begin{itemize}
\item After Step 1: Block 2 is partially peeled; $A_{s(2)}, R_2, A_{s(3)}, R_3, A_{s(4)}, R_4$ remain in order.
\item After Step 2: Block 2 is completely peeled; block 1 ($A_{s(3)}, R_3, A_{s(4)}, R_4$) remains intact.
\item After Step 3: Block 1 is partially peeled; only $A_{s(4)}, R_4$ remains.
\item After Step 4: All gates peeled; circuit is exhausted.
\end{itemize}

This confirms that block structure (alternating signal operators and rotations within blocks, ordered by schedule) survives each peeling step, justifying the inductive argument: each peeling removes one outermost query and intermediate rotation, leaving a circuit with the same block structure.
At every step, the leading coefficient in the current peeling variable (either $W_R$ or $U_I$) is extracted. 
The resulting ratio is a scalar because the common factor $\Gamma_{ij}^{(k)}$ (the inner-circuit block entry) cancels between the numerator ($P_k$) and denominator ($Q_k$).
$\Gamma$ does not depend on the current peeling variable, dropping out of the ratio.
\end{example}

\subsection{Coefficient separability for the Dyson polynomial}
\label{subsec:coeff_sep}

Lemmas~\ref{lem:single_gate} and~\ref{lem:inductive_block} support Theorem~\ref{thm:CRC} architecturally, where the CRC depends only gates structured by tensor-products.
We now verify that the Dyson polynomial's coefficient structure is compatible during degree-reduction.

\begin{theorem}[Coefficient separability]
\label{thm:coeff_sep}
The Dyson polynomial $P_{\mathrm{Dyson}}(z_1, z_2)$ satisfies the coefficient separability conditions required for the CRC, verified through intra-block $W_R$ coefficients, intra-block $U_I$ coefficients, and inter-block boundary terms.
\end{theorem}

\begin{proof}
The Dyson polynomial has the block-factored form
\begin{equation}\label{eq:dyson_blocks}
  P_{\mathrm{Dyson}}
  = \frac{1}{\lambda}\prod_{j=r}^{1}
    \Bigl[F_j(W_R) \cdot T_j(U_I)\Bigr],
\end{equation}
where $F_j(W_R)$ is the GQSP frame rotation (degree $d_{R,j}$ in $W_R$) and $T_j(U_I) = \sum_{m=0}^{M}(\betaI\Delta)^m U_I^m/m!$ is the truncated Taylor expansion (degree $M$ in $U_I$).
Within each block, all $W_R$ queries precede all $U_I$ queries.
Degree-reduction follows the Dyson schedule~$\mathbf{s}$, first peeling the $W_R$ queries of block $r$, then its $U_I$ queries, then the $W_R$ queries of block $r-1$, and so on.

\medskip\noindent
\paragraph{Intra-block $W_R$ peeling (within $F_j$).}

When peeling the $W_R$ queries of block $j$, the full polynomial
at this stage takes the form
\begin{equation}
\begin{split}
      P^{(k)} &= [\text{outer blocks already peeled}]
            \\&\cdot F_j(z_1)
            \cdot T_j(z_2)
            \cdot \Gamma_{\mathrm{inner}}(z_1, z_2),
\end{split}
\end{equation}
where $\Gamma_{\mathrm{inner}} = \prod_{i=j-1}^{1}[F_i(z_1)\,T_i(z_2)]$ collects all inner blocks.
The leading $z_1$-coefficient of $P$ at degree $d_{R,j}$ is $\alpha_{d_{R,j}}^{F_j} \cdot T_j(z_2) \cdot \Gamma_{\mathrm{inner}}$, where $\alpha_{d_{R,j}}^{F_j}$ is the scalar leading coefficient of $F_j$.
The $Q$ block has the same structure with a different scalar prefactor. 
The common factor $T_j(z_2) \cdot \Gamma_{\mathrm{inner}}$ cancels in the ratio $\mathbf{b}/\mathbf{a}$, reducing to the ratio of the leading coefficients of $F_j$ in the $P$ and $Q$ blocks (constant).

This falls into Lemma~\ref{lem:inductive_block}, where $z_2$-dependence from $T_j$ and the inner blocks plays the role of $\Gamma$, cancelling.

\medskip\noindent
\paragraph{Intra-block $U_I$ peeling (within $T_j$).}

After all $d_{R,j}$ queries of $F_j$ have been peeled, the outermost queries are $U_I$ queries from $T_j$.
The Taylor polynomial $T_j(z_2) = \sum_{m=0}^{M} c_m\,z_2^m$ with $c_m = (\betaI\Delta)^m/m!$ has leading coefficient $c_M = (\betaI\Delta)^M/M!$ that is constant, independent of $z_1$.

$z_1$-dependence is now pushed into the inner blocks by $F_j$ peeling.
The leading $z_2$-coefficient of $P^{(k)}$ is $c_M^P \cdot \Gamma_{\mathrm{inner}}$, and that of $Q^{(k)}$ is $c_M^Q \cdot \Gamma_{\mathrm{inner}}$, with the same $\Gamma_{\mathrm{inner}}$ in both.
The ratio is $c_M^Q / c_M^P$, a constant.

Taylor coefficients are constant in $z_1$, the Taylor expansion of $e^{\Htilde(\tau_j)\Delta}$ at each segment produces coefficients that depend only on $\betaI$ and $\Delta$, following from $\|\Htilde(\tau_j)\| = \betaI$ uniformly in $\theta_1$ (Lemma~\ref{lem:HI_props}).

\medskip\noindent
\textbf{Phase 3: Inter-block boundary (transition from block $j+1$ to block $j$).}

At the boundary between blocks, we transition from peeling the last $U_I$ query of $T_{j+1}$ to peeling the first $W_R$ query of $F_j$.  
Non-commutativity between $W_R$ and $U_I$ impacts the polynomial after fully peeling block $j+1$, with the polynomial taking the form
\begin{equation}
  P^{(k)}
  = F_j(z_1) \cdot T_j(z_2) \cdot
    \prod_{i=j-1}^{1}\bigl[F_i(z_1)\,T_i(z_2)\bigr].
\end{equation}
By Lemma~\ref{lem:inductive_block}, the inner blocks $\prod_{i=j-1}^{1}[F_i\,T_i]$ form $\mathcal{G}_{\mathrm{inner}}^{(k)}$, contributing the same operator $\Gamma$ to both the $P$ and $Q$ leading coefficients.
Leading $z_1$-coefficients are
\begin{equation}
\begin{split}
      \mathbf{a}_{d_{R,j}}
  &= \alpha_{d_{R,j}}^{F_j}
    \cdot T_j(z_2) \cdot \Gamma,
    \\\mathbf{b}_{d_{R,j}}
  &= \beta_{d_{R,j}}^{F_j}
    \cdot T_j(z_2) \cdot \Gamma,
\end{split}
\end{equation}
where $\alpha$ and $\beta$ are scalar prefactors from the $P$-block and $Q$-block branches of the ancilla rotation. 
The ratio is $\beta_{d_{R,j}}^{F_j}/\alpha_{d_{R,j}}^{F_j}$, with $T_j(z_2) \cdot \Gamma$ cancelling identically despite non-commutativity within $\Gamma$.

\medskip
By induction over all $d = d_R + d_I$ peeling steps (Phases~1--3 cover every possible step), the CRC holds throughout degree-reduction.
Angles $(\theta_k, \phi_k)$ at each step are well-defined constants, uniquely determined by leading-coefficient ratios.
\end{proof}

Example~\ref{ex:depth4} (as worked out above) provides explicit matrix-level verification of all three phases of Theorem~\ref{thm:coeff_sep}.
The $W_R$-removal steps (phase~a) and $U_I$-removal steps (phase~b) each show the inner-circuit factor $\Gamma^{(k)}$ canceling identically between $P$ and $Q$, leaving a scalar ratio $e^{-i\phi_k}\tan\theta_k$ at each step; the block-boundary transitions (phase~c) confirm that each block contributes its own scalar ratios independent of subsequent blocks.
For depth-$d > 4$, the same inductive structure applies at every peeling step.

\begin{remark}[Is the Dyson structure special?]
\label{rem:dyson_special}
Coefficient separability in phases~a and~b relies on a block-segregated structure of the Dyson polynomial. 
Within each segment, all $W_R$ queries come first (in $F_j$), then all $U_I$ queries (in $T_j$). 
A generic bivariate polynomial with arbitrarily interleaved $z_1$ and $z_2$ queries might not satisfy this separability.  
The Dyson series' frame rotations followed by $H_I$ insertions produce this segregation, making the angle-finding algorithm well suited to the interaction-picture structure.
\end{remark}

\begin{remark}[Non-commutativity simplifies the CRC]
\label{rem:noncomm_simpler}
The CRC is simpler when signal operators do not commute ($[W_R, U_I] \neq 0$).
In the commuting setting, the circuit polynomial belongs to the commutative polynomial ring $\mathbb{C}[z_1, z_1^{-1}, z_2, z_2^{-1}]$, and CRC compatibility becomes a question about symmetric properties and positivity certificates in this ring, which can fail (N\'emeth et al.~\cite{nemeth2023variants}). 
In the non-commuting case (this work), the circuit polynomial is a matrix-valued function of ordered products of signal-operator blocks. 
Ancilla-factorization proof of the CRC (Lemmas~\ref{lem:single_gate} and~\ref{lem:inductive_block}) operates on the block structure, bringing no analogous obstruction, extending this derivation to multi-oracle quantum algorithms with applicable ancilla factorization. 
\end{remark}

\section{Constructive Angle-Finding}
\label{sec:anglefinding}

Given a target polynomial $P_{\mathrm{Dyson}}$ of bidegree $(d_R, d_I)$, the M-QSP circuit requires $d = d_R + d_I + 1$ rotation-angle pairs $\bTheta = \{(\theta_k, \phi_k)\}_{k=0}^{d}$. 
We present complementary methods for computing these angles: a deterministic recursive algorithm and an FFT-based optimization. We also characterize which polynomials are achievable.

\subsection{Recursive algorithm}
\label{subsec:recursive}

\begin{theorem}[Constructive angle-finding]
\label{thm:anglefinding}
Given a target polynomial $P_{\mathrm{Dyson}}(z_1, z_2)$ of bidegree $(d_R, d_I)$, Algorithm~\ref{alg:recursive} computes the rotation angles $\bTheta = (\theta_0, \ldots, \theta_{d_R + d_I})$ in time $\calO((d_R + d_I) \cdot d_R \cdot d_I)$.  
Each angle $\theta_k$ is uniquely determined by the ratio of the leading coefficient of the current residual polynomial.
\end{theorem}

\begin{algorithm}
\caption{Recursive angle-finding for bivariate M-QSP}
\label{alg:recursive}
\begin{algorithmic}[1]
\REQUIRE Target polynomial $P_{\mathrm{target}}(z_1, z_2)$ of
  bidegree $(d_R, d_I)$; schedule
  $\mathbf{s}: \{1, \ldots, d\} \to \{R, I\}$ with
  $d = d_R + d_I$; complementary polynomial $Q^{(d)}$
  (via $\delta$-regularization,
  Theorem~\ref{thm:complement_existence})
\ENSURE Rotation angles
  $\bTheta = \{(\theta_k, \phi_k)\}_{k=0}^{d}$
\STATE $P^{(d)} \leftarrow P_{\mathrm{target}}$,\quad
  $Q^{(d)} \leftarrow Q_\delta$
\FOR{$k = d, d-1, \ldots, 1$}
  \STATE Extract leading $z_{s(k)}$-coefficient
    $a_{d_{s(k)}}^{(k)}$ of $P^{(k)}$
    \COMMENT{polynomial in $z_{\bar{s}(k)}$}
  \STATE Extract leading $z_{s(k)}$-coefficient
    $b_{d_{s(k)}}^{(k)}$ of $Q^{(k)}$
  \STATE $\rho_k \leftarrow
    b_{d_{s(k)}}^{(k)} \big/ a_{d_{s(k)}}^{(k)}$
    \COMMENT{constant by CRC (Thm.~\ref{thm:CRC})}
  \STATE $\theta_k \leftarrow \arctan(|\rho_k|)$
  \STATE $\phi_k \leftarrow -\mathrm{Arg}(\rho_k)$
  \STATE \textbf{Peel:}\quad
    $\begin{pmatrix} P^{(k-1)} \\ Q^{(k-1)} \end{pmatrix}
    \leftarrow R(\theta_k, \phi_k)^{-1}
    \begin{pmatrix} P^{(k)} \\ Q^{(k)} \end{pmatrix}
    \cdot z_{s(k)}^{-1}$
    \COMMENT{degree reduction}
\ENDFOR
\STATE $\theta_0 \leftarrow
  \arctan\!\bigl(|Q^{(0)} / P^{(0)}|\bigr)$,\quad
  $\phi_0 \leftarrow \mathrm{Arg}(P^{(0)})$
\end{algorithmic}
\end{algorithm}

The complexity is $\calO((d_R+d_I)\cdot d_R\cdot d_I)$, since each of the $d$ peeling steps requires $\calO(d_R\cdot d_I)$ operations for the bivariate coefficient update.
Complexity analysis is given in Proposition~\ref{prop:anglefinding_complexity} (Appendix~\ref{app:algorithm}).

\subsection{CRC-exploiting block peeling}
\label{subsec:block_peeling}

The Dyson block schedule enables a faster angle-finding procedure exploiting the constant-ratio condition at each peeling step.

\begin{proposition}[Block peeling cost]
\label{prop:block_peel}
For a Dyson polynomial of bidegree $(d_R, d_I)$ with block schedule consisting of $r$ segments, the CRC-exploiting block peeling algorithm computes all $d_R + d_I + 1$ rotation angles in total cost
\begin{equation}\label{eq:block_peel_cost}
  C_{\mathrm{block}} = \calO(d_R \cdot d_I).
\end{equation}
The full complexity hierarchy ($\calO(d^3)$, block peeling $\calO(d^2)$, FFT-warm-started) appears in the companion paper~\cite{courtney2026paper2} (angle-finding complexity hierarchy theorem).
\end{proposition}

\begin{proof}[Proof sketch]
Within each segment $j$ of the Dyson block schedule, all $d_{R,j}$ queries to $W_R$ precede all $d_{I,j}$ queries to $U_I$.
Theorem~\ref{thm:coeff_sep} ensures that:

(i) During $W_R$-peeling, the leading $z_1$-coefficient ratio reduces to the GQSP frame rotation ratio (constant independent of $z_2$). 
Angle extraction costs $\calO(1)$ and polynomial update costs $\calO(d_I^{(j)})$ per step.

(ii) During $U_I$-peeling, Taylor coefficients $c_m = (\betaI\Delta)^m/m!$ are constants independent of $z_1$.
Angle extraction costs $\calO(1)$ and the update costs $\calO(d_R^{(j)} - d_{R,j})$ per step (residual $R$-degree after the $W_R$-peeling phase).

(iii) The inter-block boundary factor cancels identically.

For uniform segments, the telescoping sum evaluates to $C_{\mathrm{block}} = d_R \cdot d_I$.
The full complexity comparison and numerical verification across $d \in \{10, 50, 100, 500\}$ are reported in the companion paper~\cite{courtney2026paper2} (angle-finding section, block-peeling cost proposition 
and accompanying table).
\end{proof}

\begin{remark}[Practical recommendation]
\label{rem:practical_peeling}
Block peeling dominates for all tested configurations; we recommend it as the primary angle-finding method, with optional $\calO(1)$ FFT refinement (Theorem~\ref{thm:fft_opt}). 
Absolute operation counts and a head-to-head comparison against competing preprocessing strategies are reported in the companion paper~\cite{courtney2026paper2} (head-to-head benchmark remark in the angle-finding section).
\end{remark}

\subsection{Numerical stability}
\label{subsec:stability}

\begin{proposition}[Stability]\label{prop:stability}
The recursive angle-finding algorithm is numerically stable with explicit condition number control. 
For the regularized Dyson polynomial with $\delta$-regularization parameter $\delta = O(\eps)$, define a cumulative condition number
\begin{equation}
  \kappa_{\mathrm{total}} := \prod_{k=1}^{d} \sec(\theta_k) = \prod_{k=1}^{d} \frac{1}{\cos(\theta_k)}.
\end{equation}
Then
\begin{equation}
  \kappa_{\mathrm{total}} \leq (1/\delta)^{O(1)} \cdot \poly(d_R, d_I),
\end{equation}
where the exponent depends on regularization depth.
Error propagation is linear: perturbations in determining $\theta_k$ grow polynomially through the algorithm, enabling fault-tolerant implementation with polynomial error correction overhead.
\end{proposition}

\begin{proof}
We show that: 
(1) normalization prevents leading coefficients from vanishing exponentially,
(2) the condition number of each peeling step is bounded by $\sec(\theta_k)$, and
(3) the product bound follows from the Parseval identity on $\bbT^2$ and regularization.

\medskip\noindent
(1) \emph{Each peeling step is a unitary rotation.}
The inverse rotation $R(\theta_k, \phi_k)^{-1} = R(\theta_k, \phi_k)^\dagger$ is unitary (as a $2\times 2$ matrix on $\cH_a$), preserving the $2$-norm of the coefficient vector $(\mathbf{a}_n, \mathbf{b}_n)$ at every Fourier index $n$.
The leading coefficient is eliminated by construction of $\theta_k$ and $\phi_k$.

\medskip\noindent
(2) \emph{Condition number via leading coefficient lower bound. }
Write the coefficient update at step $k$ as
\begin{equation}\label{eq:coeff_update}
  \begin{pmatrix}
    \mathbf{a}_n^{(k-1)} \\[2pt]
    \mathbf{b}_n^{(k-1)}
  \end{pmatrix}
  =
  \begin{pmatrix}
    e^{-i\phi_k}\cos\theta_k & \sin\theta_k \\
    -\sin\theta_k & e^{i\phi_k}\cos\theta_k
  \end{pmatrix}
  \begin{pmatrix}
    \mathbf{a}_n^{(k)} \\[2pt]
    \mathbf{b}_n^{(k)}
  \end{pmatrix}.
\end{equation}
The condition number of this unitary is $\sec(\theta_k)$ (the ratio of largest to smallest singular value of the rotation). 
By the Parseval identity on $\bbT^2$ applied to $P$ and $Q$:
\begin{equation}
  \sum_n (|\mathbf{a}_n^{(k)}|^2 + |\mathbf{b}_n^{(k)}|^2) = 1.
\end{equation}
At step $k$, the leading coefficients satisfy
\begin{equation}
  |\mathbf{a}_{d_{s(k)}}^{(k)}|^2 + |\mathbf{b}_{d_{s(k)}}^{(k)}|^2 \geq c(\delta) > 0
\end{equation}
where $c(\delta)$ is determined by the $\delta$-regularization: $H_\delta := 1 - |P_\delta|^2 \geq 2\delta - \delta^2$, inherited via spectral factorization.
This ensures $\sec(\theta_k) \leq c(\delta)^{-1/2}$ at each step.

\medskip\noindent
(3) \emph{Cumulative condition number. }
Taking the product over all $d = d_R + d_I$ steps:
\begin{equation}
  \kappa_{\mathrm{total}} = \prod_{k=1}^d \sec(\theta_k) \leq (c(\delta)^{-1/2})^d = (1/\delta)^{O(d)}.
\end{equation}
For the Dyson polynomial with $\delta = O(\eps)$ and $d = d_R + d_I = \calO((\alphaR + \betaI)T + \log(1/\eps)/\log\log(1/\eps))$, we have
\begin{equation}
  \kappa_{\mathrm{total}} \leq (\eps^{-c})^{d} = (1/\delta)^{O(1)} \cdot \poly(d_R, d_I).
\end{equation}
By the perturbation lemma for linear systems, errors in angle-finding induce a polynomial factor amplification in the circuit output.
\end{proof}

\begin{remark}[Floating-point limits of recursive peeling]
\label{rem:peel_precision_limits}
The bound of Proposition~\ref{prop:stability} is sharp in finite precision: empirically the recovered-circuit error scales as $\eps_{\mathrm{peel}}(d) \approx \eps_{\mathrm{mach}} \cdot \kappa_{\mathrm{total}}^{(d)}$ with $\kappa_{\mathrm{total}}^{(d)}$ growing roughly as $1.6^d$ for Dyson-type targets (the same constant as the Jacobian condition number scaling reported in the companion paper~\cite{courtney2026paper2}).
On a 2-qubit Lindbladian benchmark in IEEE double precision, this predicts $\eps_{\mathrm{peel}}(d=26) \sim \eps_{\mathrm{mach}} \cdot 10^{4} \sim 10^{-12}$ (Table~\ref{tab:numerical_demo}).
Recursive peeling alone therefore breaks down in double precision around $d \approx 30$, where $\kappa_{\mathrm{total}} \approx 10^{6}$--$10^{7}$ erodes any meaningful signal; for practical quantum simulation at $d \sim 100$--$1000$, a complementary strategy is required.
We offer a few possible mitigations:
\begin{enumerate}
\item \emph{Optimization-based refinement} (Sec.~\ref{subsec:optimization}). The cost function $\cF(\bTheta)$ and its gradient are evaluated by 2D FFT in $\calO(d_R d_I \log(d_R d_I))$ operations, bypassing the recursive condition-number accumulation entirely. 
L-BFGS warm-started from peel angles restores floating-point precision whenever the warm-start basin theorem of the companion paper~\cite{courtney2026paper2} applies.
This is the recommended pipeline for $d > 30$.
\item \emph{Higher-precision arithmetic.} Switching the peel inner loop to quadruple precision ($\eps_{\mathrm{mach}} \approx 10^{-34}$) pushes the breakdown threshold to $d \approx 60$.
Using arbitrary-precision libraries (\texttt{mpmath}) admits any bit width at linear cost per bit.
\item \emph{Re-orthogonalization at each peel step.} The unitarity invariant $|\mathbf{a}_n^{(k)}|^2 + |\mathbf{b}_n^{(k)}|^2 = 1$ in~\eqref{eq:coeff_update} can be enforced after each rotation, suppressing $\kappa_{\mathrm{total}}$ accumulation at the cost of $\calO(d^2)$ additional operations.
\end{enumerate}
For typical quantum-simulation workloads, we recommend (1), escalating when machine precision is demanded throughout.
\end{remark}

\subsection{Numerical validation}
\label{subsec:numerical_validation}

We verify the CRC-exploiting angle-finding algorithm end-to-end on a 2-qubit Lindbladian no-jump generator $H_{\mathrm{eff}} = H_R + iH_I$ with coherent Hamiltonian $H_R = J(Z\otimes Z) + h(X\otimes I + I\otimes X)$ and anti-Hermitian part $H_I = -(\gamma/2)\sum_k L_k^\dagger L_k$ from amplitude-damping Lindblad operators ($J = 1$, $h = 0.5$, $\gamma = 0.3$).

For each bidegree pair $(d_R, d_I)$ with $d = d_R + d_I$ ranging from $4$ to $26$ and block schedules with $r = 1$ to $4$ segments, we generate a circuit-achievable target polynomial (with known angles), run Algorithm~\ref{alg:recursive}, and compare the recovered angles $\bTheta^{\mathrm{rec}}$ against the true angles $\bTheta^*$. Table~\ref{tab:numerical_demo} reports the results.
All degree pairs achieve circuit reconstruction error $\|G_{\mathrm{rec}} - G\|/\|G\| < 10^{-8}$, with errors at low degree ($d \leq 10$) near machine precision ($\sim 10^{-16}$).
Mild geometric growth of error with degree ($\sim 10^{-16}$ at $d = 4$ to $\sim 10^{-10}$ at $d = 26$) is consistent with the cumulative condition number $\kappa_{\mathrm{total}} = \prod_k \sec\theta_k \sim 1.6^d$ of Proposition~\ref{prop:stability} times machine epsilon $\eps_{\mathrm{mach}} \approx 10^{-16}$.

\begin{table}[h]
\centering
\caption{End-to-end angle-finding on a 2-qubit Lindbladian benchmark. $d_R, d_I$: bidegree; $r$: segments; $\|G_{\mathrm{rec}} - G\|/\|G\|$: relative circuit error; all cases achieve $< 10^{-8}$.}
\label{tab:numerical_demo}
\begin{tabular}{cccccc}
\hline
$d_R$ & $d_I$ & $d$ & $r$ & Circuit error & Angle error \\
\hline
2 & 2 & 4 & 1 & $4.6 \times 10^{-16}$ & $8.9 \times 10^{-16}$ \\
3 & 3 & 6 & 1 & $6.5 \times 10^{-16}$ & $8.9 \times 10^{-16}$ \\
4 & 4 & 8 & 2 & $1.4 \times 10^{-15}$ & $1.3 \times 10^{-15}$ \\
5 & 5 & 10 & 2 & $8.8 \times 10^{-16}$ & $1.8 \times 10^{-15}$ \\
6 & 4 & 10 & 2 & $1.1 \times 10^{-14}$ & $2.1 \times 10^{-14}$ \\
8 & 6 & 14 & 3 & $1.4 \times 10^{-13}$ & $3.6 \times 10^{-13}$ \\
10 & 8 & 18 & 3 & $3.7 \times 10^{-13}$ & $9.2 \times 10^{-13}$ \\
12 & 10 & 22 & 4 & $1.4 \times 10^{-11}$ & $4.0 \times 10^{-11}$ \\
14 & 12 & 26 & 4 & $1.5 \times 10^{-10}$ & $5.1 \times 10^{-10}$ \\
\hline
\end{tabular}
\end{table}

\subsubsection{Experiment A: Dyson \texorpdfstring{$\to$}{->} SOS \texorpdfstring{$\to$}{->} peel-then-refine}
\label{subsec:experiment_A}

Here, we benchmark the angle-finding pipeline of Algorithm~\ref{alg:recursive} composed with the L-BFGS refinement of Sec.~\ref{subsec:optimization} on the Dyson polynomial target at $\alphaR T = 0.8$, $\betaI T = 0.4$ across $10$ bidegree configurations from $(2,2)$ through $(5,5)$, with eight independent multistart restarts per configuration (8 random angle perturbations of scale $\sigma = \pi/4$ around the peel-init).
Two complementary experiments characterize the pipeline (Figure~\ref{fig:experiment_A}).

We distinguish two polynomials that are easily conflated, and which the two experiments target separately:
\begin{itemize}
  \item $P_{\cG}$, the polynomial generated by an explicit M-QSP circuit carrying the Dyson schedule of Proposition~\ref{prop:dyson_poly}.
  Being read off from a circuit, it is an element of the single-qubit M-QSP image by construction, hence on-manifold with $c_\infty = 0$. 
  We find it is the target of the on-manifold, circuit-generated experiment (right panel), which validates the angle recovery of Algorithm~\ref{alg:recursive}.
  Achievability of the underlying propagator $A(T)$ at this bidegree is the content of Corollary~\ref{cor:dyson_achievable}.
  \item $P_{\mathrm{trunc}}$, the bidegree-$(d_R, d_I)$ Chebyshev--Taylor truncation of the propagator collapsed into a single bivariate polynomial. This is the fit target of the left panel; it is generally off-manifold, with $c_\infty > 0$.
\end{itemize}
The two are consistent: the right panel targets a circuit-generated polynomial $P_{\cG}$ (on-manifold by construction), while the left panel of Figure~\ref{fig:experiment_A} concerns $P_{\mathrm{trunc}}$. There is no contradiction between a machine-precision recovery for a circuit-generated target and a positive best-fit residual for a different, single-polynomial truncation.

\emph{Truncated-polynomial target} (Figure~\ref{fig:experiment_A}, left).
The single-polynomial Chebyshev--Taylor truncation $P_{\mathrm{trunc}}$ is generally not in the image of the M-QSP parameterization at finite bidegree, so the best-achievable squared residual
\begin{equation}
  c_\infty := \inf_{\bTheta}\, \bigl\| P_\cG(\bTheta) - P_{\mathrm{target}} \bigr\|_2^2
\end{equation}
is non-negative rather than identically zero, quantifying the truncation distance from the M-QSP achievable submanifold to the Dyson target.
Across the 10 configurations, $c_\infty$ estimated as the minimum refined residual over the 8 multistart restarts falls in the range $[7.6 \times 10^{-2}, 3.8 \times 10^{-1}]$, with median $c_\infty \approx 0.13$ approximately independent of $d$.
This is the M-QSP-to-Dyson truncation distance, characterized in the companion paper~\cite{courtney2026paper2} (full-tensor sweep in the landscape appendix).
Basin-attainment rate of the multistart restarts to the $c_\infty$-basin (within $5\%$ relative gap) is $12$--$38\%$ per configuration, consistent with spurious local minima characterized in the companion paper's landscape section~\cite{courtney2026paper2}.

\emph{Circuit-generated target} (Figure~\ref{fig:experiment_A}, right).
For each configuration we sample random ground-truth angles $\bTheta^*$, forward the M-QSP circuit to produce the target polynomial $P_{\mathrm{target}}$ (which is in the M-QSP image by construction, so $c_\infty = 0$), then run the same SOS-factorization $\to$ peel-then-refine pipeline.
With 8 multistart restarts, $7$ of the $10$ configurations recover the target to machine precision: refined residual and end-to-end circuit reconstruction error $\|G_{\mathrm{rec}} - G\|_\infty$ both reach $\sim 10^{-12}$--$10^{-13}$ at bidegrees through $(5,5)$ when the multistart hits the right basin.
The remaining $3$ configurations $[(4,3)$, $(5,4)$, $(4,5)]$ are all rank-deficient at the global minimum ($\kappa(J) = \infty$ in the standard pseudoinverse sense. 
They pin at the spurious-basin residual $\sim 0.2$--$0.4$, matching the rank-deficient asymmetric-bidegree configurations identified in the companion paper's landscape table~\cite{courtney2026paper2}.

The two experiments together demonstrate that the constructive pipeline reaches machine precision whenever the target lies on the M-QSP achievable submanifold.
Off-manifold targets such as the truncated Dyson series are reproduced up to the truncation distance $c_\infty$ that scales with $(\alphaR T, \betaI T)$ rather than with $d$.

\paragraph{Discrete basin structure.}
Inspecting the 8 multistart restarts per configuration in the Dyson panel of Figure~\ref{fig:experiment_A}, the recovered residuals cluster into $3$--$5$ distinct values rather than a continuum. 
At $(d_R, d_I) = (3,3)$, for instance, three restarts converge to refined residual $\approx 0.138$ (the $c_\infty$ basin), two to $\approx 0.329$, and one each to $\approx 0.148$, $0.749$, and $1.011$.
The basins are discrete local minima of the cost function $\cF(\bTheta)$.
Distinct random initializations that fall into the same basin produce identical refined residuals (to machine precision) because L-BFGS converges to the same critical point.
We characterize the M-QSP landscape in the companion paper~\cite{courtney2026paper2} and see this structure to be the source of spurious local minima with a random-initialization failure rate of $40\%$ in the single-row-target survey and $82\%$ under the stricter full-tensor-target criterion.
A dedicated depth sweep (campaign Experiment~C, records with the code release) separates failure modes at large $d$ and brackets its validated regime. 
For instance, CRC peeling in practice is sensitive to floating point drift during initialization.
With the low-segment schedules of Table~\ref{tab:numerical_demo} ($r \le 4$) it lands within $10^{-7}$--$10^{-5}$ of the exact word up to $d \approx 30$, and fails by $d = 44$ in double precision, in line with the drift model of Proposition~\ref{prop:stability}; high segment counts (strict alternation) break it earlier. 
When the initializer is uninformative, multistart refinement faces a basin-sampling barrier, wherein per-restart probability of hitting the exact basin is ${\sim}1/8$ at $d = 8$ and undetectable within an 8-restart budget from $d = 12$ on, every restart landing on a depth-independent plateau at relative residual $0.24$--$0.42$ through $d = 60$, while the conditioning proxy $\prod_k|\sec\theta_k^{*}|$ of the truth word grows exponentially.
In-basin refinement, by contrast, converges cleanly at depth (monotone descent from the $d = 30$ peel initialization, no plateau). 
We therefore validate the double-precision pipeline for $r \le 4$, $d \lesssim 30$ and the identified route to larger $d$ is extended-precision peeling (or continuation-based initialization).

\paragraph{Consistency with Table~\ref{tab:numerical_demo}.}
Table~\ref{tab:numerical_demo} (which reaches $\sim 10^{-10}$ end-to-end circuit error at $d=26$) and Experiment~A (left panel of Figure~\ref{fig:experiment_A}, pinned at $\sim 10^{-1}$) measure complementary regimes:
the Table uses circuit-achievable target polynomials (the same construction as the right panel of Figure~\ref{fig:experiment_A}) and is therefore on-manifold, with the residual limited only by the floating-point error of CRC peeling (Proposition~\ref{prop:stability}).
The left panel of Figure~\ref{fig:experiment_A} uses the Dyson-truncated target (off-manifold) and the residual is dominated by $c_\infty$.
The end-to-end-precision claim of Sec.~\ref{subsec:numerical_validation} refers specifically to the on-manifold setting. 
Whenever the user is willing to accept the M-QSP-image projection of the Dyson polynomial as a target, the pipeline reaches floating-point precision.
The residual gap reported in the left panel is the truncation-induced distance from that projection to the literal Dyson series.

\paragraph{Convergence of the achievable object to the propagator.}
The saturation of $c_\infty$ in the left panel is a property of the fixed-order single-polynomial truncation $P_{\mathrm{trunc}}$: the achievable object $P_{\mathrm{Dyson}}$ reproduces the propagator to any accuracy as the Dyson order grows. Concretely, for (generally non-commuting) $H_R, H_I$ the interaction-picture Dyson truncation of order $N$ satisfies
\begin{equation}\label{eq:dyson_convergence}
  \bigl\| P_{\mathrm{Dyson}}^{(N)} - e^{-iH_{\mathrm{eff}}T}/\lambda \bigr\|
  \;\le\; \frac{(\betaI T)^{N+1}}{(N+1)!},
\end{equation}
which converges super-exponentially and does not require $[H_R, H_I]=0$. We verified~\eqref{eq:dyson_convergence} on a generic non-commuting $4\times 4$ instance ($\|[H_R,H_I]\|\approx 1.58$, $\betaI = \|H_I\| = 1$, $T=1$); the relative operator-norm error stays under the bound and decays super-exponentially:
\begin{center}
\begin{tabular}{c|cccccc}
\hline
$N$ & 0 & 1 & 2 & 3 & 4 & 5 \\
\hline
rel.\ error & $5.8\times10^{-1}$ & $2.2\times10^{-1}$ & $6.0\times10^{-2}$ & $1.3\times10^{-2}$ & $2.5\times10^{-3}$ & $3.3\times10^{-4}$ \\
bound & $1.0$ & $5.0\times10^{-1}$ & $1.7\times10^{-1}$ & $4.2\times10^{-2}$ & $8.3\times10^{-3}$ & $1.4\times10^{-3}$ \\
\hline
\end{tabular}
\end{center}
Inverting~\eqref{eq:dyson_convergence} gives the degree $d_I = O\!\bigl(\betaI T + \log(1/\varepsilon)/\log\log(1/\varepsilon)\bigr)$. Thus the ``$c_\infty$ approximately independent of $d$'' observation reflects the fixed-order truncated polynomial $P_{\mathrm{trunc}}$; the achievable $P_{\mathrm{Dyson}}$ carries no such floor. The reproducibility script is provided with the code release.

\begin{figure*}[!t]
\centering
\begin{minipage}{0.49\textwidth}
\includegraphics[width=\linewidth]{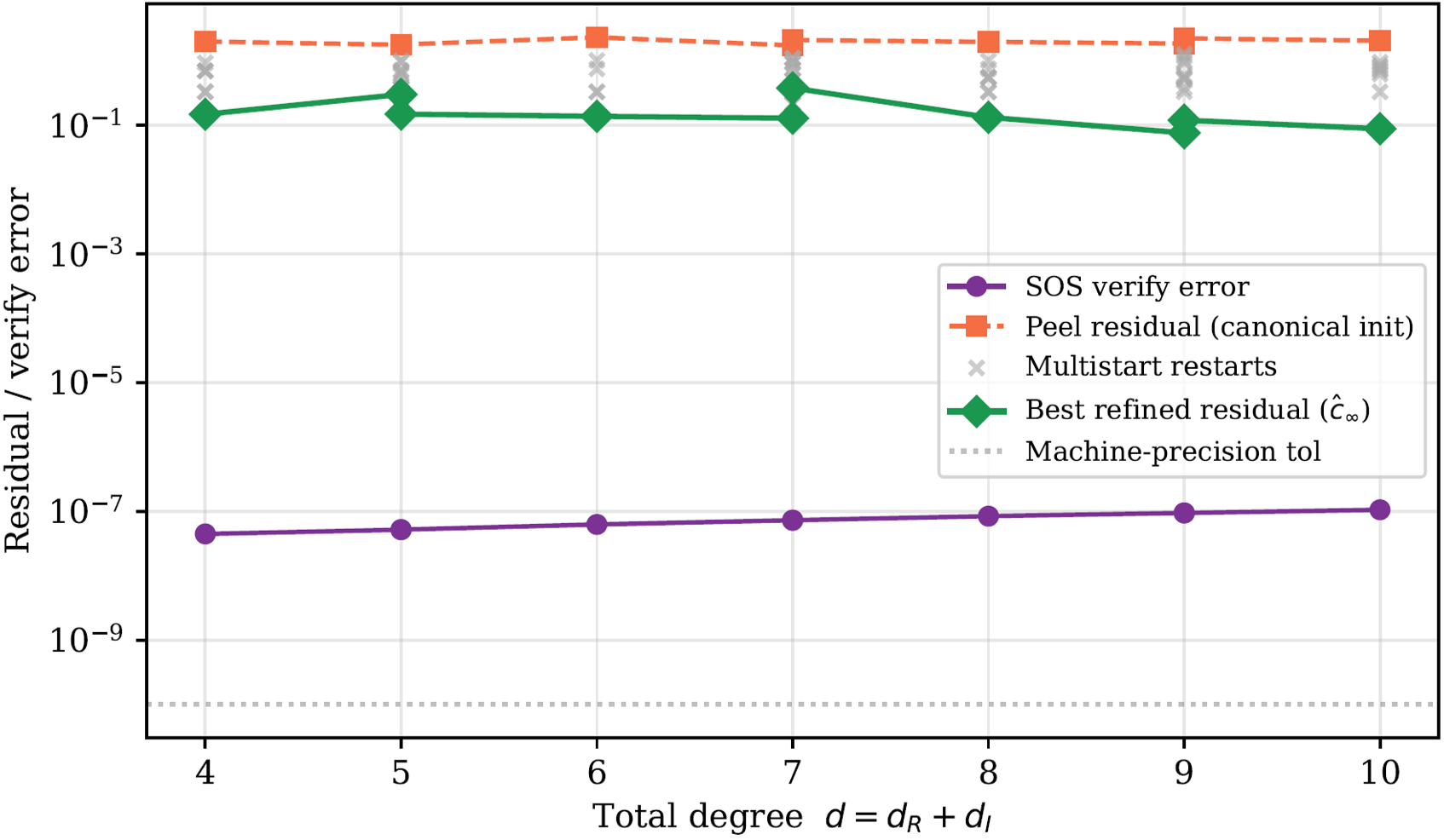}
\end{minipage}
\hfill
\begin{minipage}{0.49\textwidth}
\includegraphics[width=\linewidth]{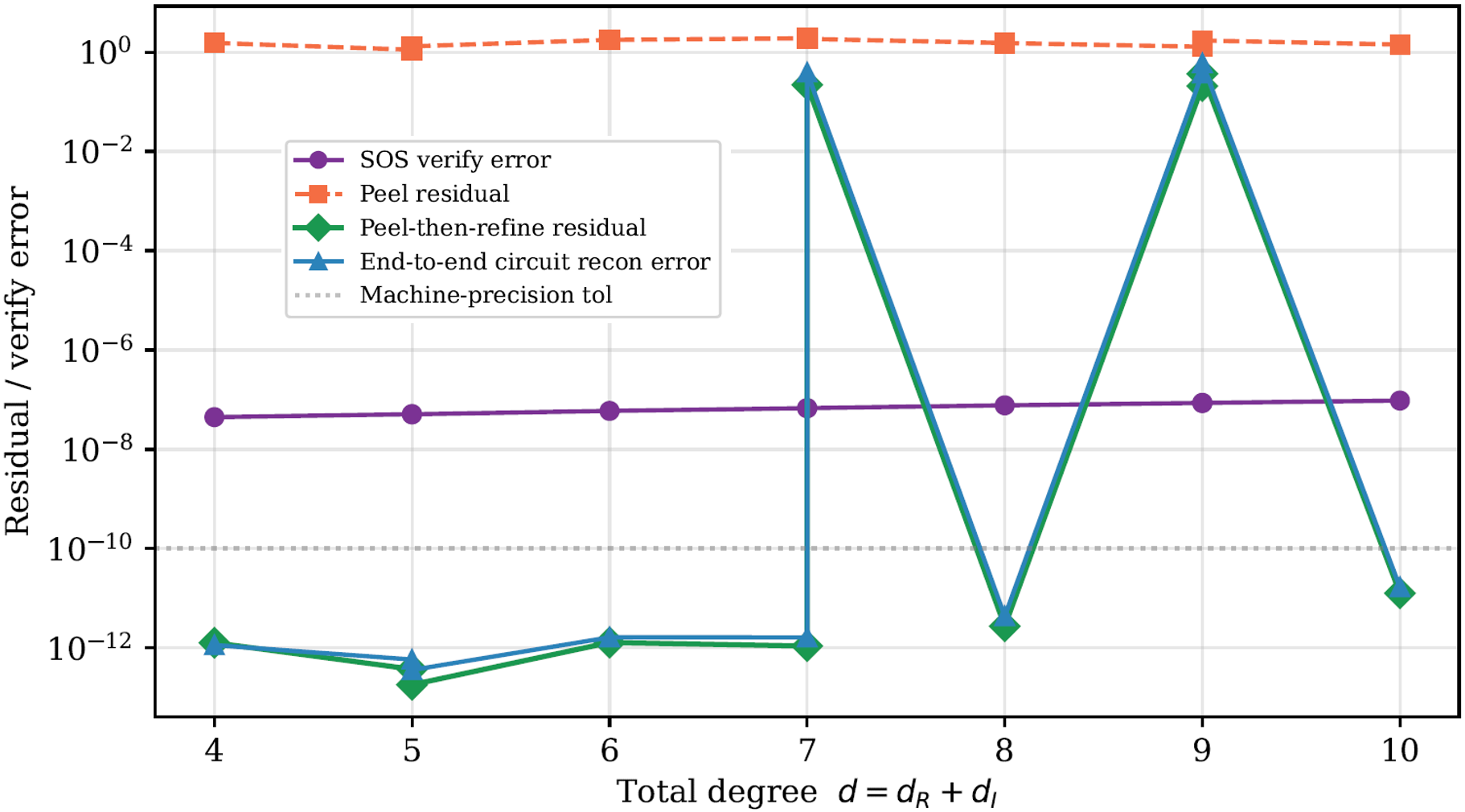}
\end{minipage}
\caption{``Residual'' on both panels denotes the relative $L^2$ norm $\|P_\cG(\bTheta) - P_{\mathrm{target}}\|_2/\|P_{\mathrm{target}}\|_2$ evaluated on an $N\times N$ bitorus FFT grid.
\textbf{Left:} Dyson-target benchmark, $\alphaR T = 0.8$, $\betaI T = 0.4$.
Purple: SOS verification error $\max_{\bbT^2}|1-|P|^2-|A|^2-|B|^2|$.
Orange: peel residual from canonical initialization.
Grey crosses: $8$ multistart restarts per bidegree, and coincident crosses are restarts that converged to the same spurious basin.
Green diamonds: best refined residual (the $c_\infty$ estimate), pinned near $0.1$--$0.4$ across $d \in [4, 10]$.
The dotted ``Target tol'' line at $10^{-10}$ is the achievable floor for the on-manifold experiment in the right panel.
For the Dyson target it sits below $c_\infty$ and is included here for visual alignment.
\textbf{Right:} Circuit-generated target experiment, using the same method applied to $P_{\mathrm{target}}$, produced by forwarding the circuit at random known angles, so $c_\infty = 0$ by construction.
Where the $8$ multistart restarts find the right basin ($7$ of $10$ configurations), refined residual and end-to-end circuit reconstruction error $\|G_{\mathrm{rec}} - G\|_\infty$ both reach $\sim 10^{-12}$.
The three misses at $(4,3)$, $(5,4)$, $(4,5)$ are not a defect of the multistart strategy, noted as rank-deficient asymmetric bidegrees~\cite{courtney2026paper2}, where the Jacobian at the global minimum is rank-deficient ($\kappa(J) = \infty$) and the spurious-basin volume is correspondingly inflated.
For these configurations, more aggressive multistart ($k \geq 16$) or analytically informed warm-starts are required.}
\label{fig:experiment_A}
\end{figure*}

\subsection{Optimization-based alternative}
\label{subsec:optimization}

The recursive algorithm is fully constructive but requires the complementary polynomial $Q$ as input.
An alternative approach follows the Motlagh--Wiebe strategy for univariate GQSP~\cite{motlagh2024generalized}, optimizing the angles directly by minimizing the discrepancy between circuit output and the target polynomial on $\bbT^2$, bypassing $Q$ entirely.

\begin{theorem}[2D FFT optimization]\label{thm:fft_opt}
An alternative angle-finding method minimizes the cost function
\begin{equation}\label{eq:cost_function}
\begin{split}
  \cF(\bTheta)
  &= \frac{1}{(2\pi)^2}
    \int_{\bbT^2}
    \bigl|P_{\cG}(e^{i\theta_1}, e^{i\theta_2};\bTheta)
          \\ &- P_{\mathrm{target}}(e^{i\theta_1}, e^{i\theta_2})
    \bigr|^2\,
    d\theta_1\,d\theta_2
\end{split}
\end{equation}
via gradient descent, with $\cF$ and $\nabla\cF$ evaluable in $\calO(d_R d_I \log(d_R d_I))$ per iteration via 2D FFT.
\end{theorem}

\begin{proof}
We consider cost function evaluation and gradient
computation, discussing convergence below the proof.

\medskip\noindent
\textbf{Part 1: Cost function evaluation via 2D FFT.}

At a symbol point $(z_1,z_2) = (e^{i\theta_1}, e^{i\theta_2}) \in \bbT^2$, the circuit symbol $P_{\cG}(\cdot;\bTheta)$ reduces to a product of $d$ matrices ($2\times 2$ each, one per gate), computable in $O(d)$ operations per grid point.
Discretize $(\theta_1, \theta_2)$ on an $N_1 \times N_2$ grid with $N_j \geq 2d_j + 1$ (sufficient to
capture all Fourier modes up to bidegree $(d_R, d_I)$).
By Parseval's identity, the integral becomes
\begin{equation}\label{eq:parseval}
  \cF(\bTheta)
  = \sum_{m,n}
    |c_{mn}^{\cG}(\bTheta) - c_{mn}^{\mathrm{target}}|^2,
\end{equation}
where $c_{mn}$ are the bivariate Fourier coefficients, computed from the grid values via a 2D FFT in
$O(N_1 N_2 \log(N_1 N_2))$ operations.
Grid evaluation costs $O(d \cdot N_1 N_2)$.
Total per evaluation:
\begin{equation}
\begin{split}
      &O\bigl(d \cdot d_R d_I + d_R d_I \log(d_R d_I)\bigr)
  \\&= O\bigl(d_R d_I \log(d_R d_I)\bigr),
\end{split}
\end{equation}
since $d = d_R + d_I$ and the FFT term dominates for large degrees.

\medskip\noindent
\textbf{Part 2: Gradient via backpropagation.}

The gradient $\nabla_{\bTheta} \cF$ can be computed in the same asymptotic cost as a single function evaluation.
The M-QSP circuit is a product of $d$ matrix factors, each depending on a single angle pair $(\theta_j, \phi_j)$. 
The derivative with respect to $(\theta_j, \phi_j)$ involves the ``left partial product'' $L_j$
and the ``right partial product'' $R_j$ at position $j$:
\begin{equation}\label{eq:chain_rule}
  \frac{\partial P_{\cG}}{\partial \theta_j}
  = L_j
    \;\frac{\partial}{\partial \theta_j}
    \bigl[R(\theta_j, \phi_j)\,A_{s(j)}\bigr]\;
    R_j.
\end{equation}
These partial products are computed for all $j$ simultaneously using a forward-backward pass, analogous to backpropagation in neural networks:
\begin{align}
  L_j &= \prod_{k=1}^{j}\bigl[R_k\,A_{s(k)}\bigr],
  \qquad j = 1, \ldots, d, \label{eq:forward}\\
  R_j &= \prod_{k=j+1}^{d}\bigl[R_k\,A_{s(k)}\bigr],
  \qquad j = 0, \ldots, d-1. \label{eq:backward}
\end{align}
The forward pass~\eqref{eq:forward} and backward pass~\eqref{eq:backward} each require $O(d)$ matrix multiplications per grid point. 
At each grid point, the gradient for all $2(d+1)$ parameters is assembled in $O(d)$ operations.
The gradient of the Parseval sum~\eqref{eq:parseval} is
\begin{equation}
  \frac{\partial\cF}{\partial\theta_j}
  = 2\,\mathrm{Re}\sum_{m,n}
    \overline{(c_{mn}^{\cG} - c_{mn}^{\mathrm{target}})}
    \cdot \frac{\partial c_{mn}^{\cG}}{\partial\theta_j},
\end{equation}
with Fourier coefficients of the derivative obtained from the 2D FFT infrastructure. 
Total gradient cost: $O(d_R d_I \log(d_R d_I))$.
\end{proof}

\subsubsection{Convergence properties}

The cost function $\cF$ has some favorable properties. 
First, $\cF \geq 0$, with $\cF = 0$ if and only if
$P_{\cG} = P_{\mathrm{target}}$ on $\bbT^2$. 
Second, the global minimum $\cF = 0$ certainly exists, since the CRC (Theorem~\ref{thm:CRC}) and achievability
(Theorem~\ref{thm:achievability} below) guarantee that angles $\bTheta^*$ with $\cF(\bTheta^*) = 0$ exist.
Finally, $\cF$ is smooth ($C^\infty$ in $\bTheta$), since $P_{\cG}$ depends on $\bTheta$ through products of trigonometric
functions.

Whether $\cF$ has spurious local minima in the bivariate case
remains an open question.

\begin{problem}[Absence of spurious local minima]
\label{prob:landscape}
Does the bivariate M-QSP cost function $\cF$ have no spurious
local minima? 
In other words, is every local minimum global for all target
polynomials satisfying $|P_{\mathrm{target}}|^2 \leq 1$ on
$\bbT^2$?
\end{problem}

In the univariate setting, Motlagh and Wiebe~\cite{motlagh2024generalized} empirically observe no spurious local
minima for GQSP polynomials up to degree $10^7$, but do not prove this result. 
The bivariate landscape has a parameter space $\R^{2(d_R + d_I + 1)}$ and with a cost
function that couples the two signal variables. 
We analyze this in further in Appendix~\ref{app:landscape}.

In practice, this concern is eliminated by warm-starting. 
By initializing $\bTheta^{(0)}$ using Algorithm~\ref{alg:recursive}, gradient descent (or L-BFGS) can refine to machine precision. 
The recursive algorithm produces the exact global minimizer (up to numerical precision).
The warm start places the optimization within the basin of attraction of the global minimum, and a small number of refinement iterations suffices. 
We recommend:
\begin{enumerate}
  \item \textbf{Initialize}: Compute $\bTheta^{(0)}$ via
  Algorithm~\ref{alg:recursive}, or initialize randomly if $Q$
  is unavailable.
  \item \textbf{Iterate}: Evaluate
  $\cF(\bTheta^{(t)})$ and $\nabla\cF(\bTheta^{(t)})$ via 2D FFT
  (cost $O(d_R d_I \log(d_R d_I))$ per iteration).
  \item \textbf{Update}:
  $\bTheta^{(t+1)} \gets \bTheta^{(t)}
  - \gamma_t\nabla\cF(\bTheta^{(t)})$, with step size $\gamma_t$
  chosen by line search or L-BFGS.
  \item \textbf{Terminate} when
  $\cF(\bTheta^{(t)}) \leq \eta$.
\end{enumerate}
Per-iteration cost of $O(d_R d_I \log(d_R d_I))$ is substantially lower than the recursive algorithm's $O((d_R + d_I) \cdot d_R d_I)$ per step, enabling practical angle computation for total degrees $d = d_R + d_I \sim 10^6$.

\subsection{Achievability}
\label{subsec:achievability}

Angle-finding takes a target polynomial $P$ and produces rotation angles $\bTheta$ only if there is no breakdown in the recursion. 
We characterize which polynomials are realizable by an M-QSP circuit.
For the commuting multivariable setting, the analogous realizability question is settled by the decision algorithm of Ito et al.~\cite{ito2026polynomial}, while we establish achievability directly in the non-commuting bivariate setting using Theorem~\ref{thm:CRC}.

\begin{theorem}[Angle-recovery characterization for M-QSP]\label{thm:achievability}
A bivariate ordered polynomial $P$ of bidegree $(d_R, d_I)$ with schedule $\mathbf{s}$ is M-QSP-achievable (i.e., equals the $(0,0)$ block of some circuit $\cG = R_0 \prod_{j=1}^{d}[A_{s(j)}\,R_j]$) if and only if:
\begin{enumerate}
  \item \textbf{Norm bound}:\; $|P(e^{i\theta_1}, e^{i\theta_2})|^2 \leq 1$ for all $(\theta_1, \theta_2) \in \R^2$.
  \item \textbf{Ordered structure}:\; $P$ can be expressed as an ordered polynomial in $(W_R, U_I)$ consistent with the schedule $\mathbf{s}$.
  \item \textbf{Non-degenerate leading coefficients}:\; At each step $k$ of the degree-reduction, the leading coefficient
  $\mathbf{a}_{d_{s(k)}}^{(k)}$ is nonzero.
\end{enumerate}
\end{theorem}

Achievability of the propagator $A(T)$ at additive-optimal degree is established directly in Corollary~\ref{cor:dyson_achievable}(i) by eigenvalue transformation of $\Heff$, independently of the single-qubit realizability question (see Remark~\ref{rem:disk_normalization} for the disk-normalization hypothesis and its cost).
Given a target that is realizable, the recursion returns its rotation angles without breakdown.

\begin{proof}
We prove both directions.

\medskip\noindent
\textbf{($\Leftarrow$: Conditions (i)--(iii) imply achievability.)}

Assume $P$ satisfies (i)--(iii). 
We construct the rotation angles by running the recursive degree-reduction, verifying that every step is well-defined.

\textit{Step 1: Complementary polynomial.}\;
By condition~(i), $H := 1 - |P|^2 \geq 0$ on $\bbT^2$. 
If $\|P\|_{\infty,\bbT^2} < 1$ (strictly), then $H > 0$ and an SOS complement $\{Q_\ell\}$ with $\sum_\ell |Q_\ell|^2 = H$ exists by Theorem~\ref{thm:exact_factor}, with $L \leq \min(d_R{+}1, d_I{+}1)$ terms of the correct bidegree.
If $\|P\|_{\infty,\bbT^2} = 1$, the same SOS complement exists by the Dritschel--Woerdeman theorem (Theorem~\ref{thm:DW}). Alternatively, we regularize $P_\delta = (1-\delta)P$ and apply Theorem~\ref{thm:exact_factor} to $H_\delta$. 
In either case, the target $P$ and complement $\{Q_\ell\}$ satisfy $|P|^2 + \sum_\ell |Q_\ell|^2 \leq 1$ on $\bbT^2$ with the correct bidegree bounds. 
The scalar SOS $\{Q_\ell\}$ is symbol-level, certifying $\|P\|_{\infty,\bbT^2} \le 1$ and seeds the recursion. 
The operator complement comes from the automatic $(1,0)$ block (Theorem~\ref{thm:automatic}), satisfying the operator identity $P(W_R,U_I)^\dagger P(W_R,U_I) + Q(W_R,U_I)^\dagger Q(W_R,U_I) = I_s$ with no simultaneous diagonalization.

\textit{Step 2: Recursive degree-reduction.}\;
By the CRC (Theorem~\ref{thm:CRC}), at each step $k$ the ratio $\rho_k = \mathbf{b}_{d_{s(k)}}^{(k)} / \mathbf{a}_{d_{s(k)}}^{(k)}$ is a constant (independent of the eigenvalue of the inactive signal operator).
By condition~(iii), $\mathbf{a}_{d_{s(k)}}^{(k)} \neq 0$, so $\rho_k \in \C$ is finite and well-defined. 
The rotation angles at step $k$ are:
\begin{equation}\label{eq:angle_extraction}
  \theta_k = \arctan|\rho_k|
  \;\in\;(0, \pi/2),
  \qquad
  \phi_k = -\mathrm{Arg}(\rho_k).
\end{equation}
The restriction $\theta_k \in (0,\pi/2)$ is guaranteed by condition~(iii): if $\mathbf{a}_{d_{s(k)}}^{(k)} = 0$, the
polynomial would have $z_{s(k)}$-degree strictly less than $d_{s(k)}^{(k)}$, contradicting the schedule.

\textit{Step 3: Degree reduction.}\;
After extracting $(\theta_k, \phi_k)$, apply $R(\theta_k, \phi_k)^{-1}$ to $(P^{(k)}, Q^{(k)})$ and divide by $z_{s(k)}$. 
This reduces the bidegree by $(1,0)$ or $(0,1)$ according to the schedule, yielding $(P^{(k-1)}, Q^{(k-1)})$ with
$|P^{(k-1)}|^2 + |Q^{(k-1)}|^2 = 1$ on $\bbT^2$ (the norm condition is preserved because the rotation is unitary and the degree reduction is exact).

\textit{Step 4: Base case.}\;
After $d = d_R + d_I$ peeling steps, the bidegree reaches $(0,0)$.
The reduced pair $(P^{(0)}, Q^{(0)})$ consists of constants satisfying $|P^{(0)}|^2 + |Q^{(0)}|^2 = 1$.  
Set $\theta_0 = \arctan(|Q^{(0)}/P^{(0)}|)$ and $\phi_0 = \mathrm{Arg}(P^{(0)})$. 
The rotation $R(\theta_0, \phi_0)$ realizes this base pair.

\textit{Step 5: Assembly.}\;
The full set of angles $\bTheta = \{(\theta_k, \phi_k)\}_{k=0}^{d}$ reconstructs the circuit $\cG = R_0 \prod_{j=1}^{d}[A_{s(j)}\,R_j]$ whose $(0,0)$ block equals $P$ by construction (degree-reduction is the inverse of the circuit synthesis).

\medskip\noindent
\textbf{($\Rightarrow$: Achievability implies conditions (i)--(iii).)}

Suppose $P$ is the $(0,0)$ block of an M-QSP circuit $\cG$ with schedule $\mathbf{s}$.

\textit{Condition (i):}\;
By unitarity of $\cG$, the blocks satisfy $|P|^2 + |Q|^2 = 1$ on $\bbT^2$ (where $Q = (\cG)_{10}$). 
Hence $|P|^2 \leq 1$.

\textit{Condition (ii):}\;
The circuit $\cG$ applies signal operators in the order dictated by $\mathbf{s}$, so its $(0,0)$ block is an ordered operator polynomial in $(W_R, U_I)$ with ordering inherited from the schedule; its symbol (abelianization $z_{s(j)} = e^{i\theta_{s(j)}}$) is the corresponding ordered polynomial $P(z_1, z_2)$. No joint eigenspace of $W_R, U_I$ is invoked.

\textit{Condition (iii):}\;
The circuit has $d_R$ queries to $W_R$ and $d_I$ queries to $U_I$, so $P$ has bidegree $(d_R, d_I)$. 
The rotation angles of a valid circuit satisfy $\theta_j \in (0, \pi/2)$ generically ($\theta_j = 0$ would skip a signal-operator query, reducing the effective degree; $\theta_j = \pi/2$ would zero out the $(0,0)$ block at that step).
Non-degeneracy of rotation angles ensure nonzero leading coefficients at every step.
\end{proof}

\begin{remark}[M-QSP realizes a proper submanifold]\label{rem:submanifold}
An M-QSP circuit of bidegree $(d_R, d_I)$ carries $2(d_R + d_I + 1)$ real angle parameters.
The space $\cP^+_{d_R, d_I}$ of bounded bivariate polynomials of that bidegree has $2(d_R+1)(d_I+1)$ real dimensions.
When $d_R, d_I \geq 1$ the former is strictly smaller, making the M-QSP-realizable set a proper submanifold of the bounded-polynomial ball.
The multivariate-completeness question nasks whether the ordered Dyson polynomial lies on this single-qubit submanifold at additive-optimal bidegree, which we discuss in Corollary~\ref{cor:dyson_achievable}.
The propagator $A(T)$ is achievable at that degree by eigenvalue transformation of $\Heff$ (regime~(i) there, with the normalization hypothesis of Remark~\ref{rem:disk_normalization}).
For operator-valued targets on a $d$-level block, a word carrying $p < 2d^2$ real parameters matches the (almost-every-shift perturbed) normalized propagator on at most a measure-zero set of instances while $p \ge 2d^2$ is unconditionally necessary for locally-uniform membership. 
Uniform coverage of a generic $d$-level family is therefore additive in $(T, 1/\eps)$ but quadratic in the level count. 
Exceptional per-instance membership below the budget (e.g.\ commuting families, themselves measure-zero in instance space) is not excluded, and the exact-target (unshifted) form of the measure-zero statement remains conditional on a transversality hypothesis stated there.
The normalized target map has exact gauge invariance $\HI \mapsto \HI + cI$, hence constant rank $2d^2 - 1$, and the unnormalized map $x \mapsto e^{-i\Heff T}$ is a local diffeomorphism away from an explicit eigenvalue-resonance set (closed and measure-zero).
\end{remark}

\begin{corollary}[Achievability of the propagator at additive-optimal degree]\label{cor:dyson_achievable}
Let $A(T) = e^{-i\Heff T}/\lambda$, $\lambda = e^{\betaI T}$, be the contraction targeted by the truncated construction of Proposition~\ref{prop:dyson_poly}, at the bidegree $(d_R, d_I)$ of Eq.~\eqref{eq:bidegree}. 
Then $A(T)$ is achievable at that additive-optimal query cost in two regimes:
\begin{enumerate}
\item[\textup{(i)}] \emph{(Generic.)} For admissible $\HR, \HI$, the operator $A(T)$ is block-encoded to error $\eps$ using $\calO\!\bigl((\alphaR + \betaI)T + \log(1/\eps)/\log\log(1/\eps)\bigr)$ queries to $W_R, U_I$ and $\calO(\log n)$ ancilla qubits, with no multiplicative logarithmic overhead and without invoking the constant-ratio condition, subject to the disk-normalization condition of Remark~\ref{rem:disk_normalization} (automatic for $\alphaR T = O(1)$. Beyond that regime the route degrades exponentially and the unconditional all-$T$ implementations of Remark~\ref{rem:disk_normalization} apply instead).
\item[\textup{(ii)}] \emph{(Structured.)} If $\mathfrak{g} := \mathrm{Lie}\langle i\HR, \HI\rangle$ has $\dim\mathfrak{g} = \calO(1)$ with efficiently block-encodable generators, the same query cost is achieved with $\calO(1)$ ancilla qubits.
\end{enumerate}
In both regimes the leading-coefficient non-degeneracy of condition~(iii) holds, and Algorithm~\ref{alg:recursive} recovers the M-QSP angles of the block-schedule circuit $R_0\prod_{j=1}^{d}[A_{s(j)}\,R_j]$ by forward constant-ratio recursion. 
Whether the ordered Dyson polynomial $P_{\mathrm{Dyson}}$ is the $(0,0)$ block of a single M-QSP qubit at additive-optimal bidegree is the multivariate-completeness question of Ref.~\cite{laneve2025multivariate} (their Conjecture~7, which remains open), with condition~(iii) as the corresponding non-degeneracy hypothesis. 
This is not required for regimes~(i)--(ii), and for any concrete instance it is decidable affirmatively by the membership certificate of Theorem~\ref{thm:certified_membership} below.
The local segment-replacement route proposed toward that conjecture by the same authors (their Conjecture~8) is contradicted numerically across 42 excursion instances accepted only within $10^{-13}$ of exact words, given in the codebase. 
An affirmative resolution of Conjecture~7 would have to proceed globally.
No unconditional statement in this paper depends on either conjecture.
\end{corollary}

\begin{proof}
\emph{Regime (i).} Since $\Heff = \HR + i\HI$ with $\HR = \HR^\dagger$ and $\HI = \HI^\dagger \succeq 0$, every $\ket\psi$ gives \[\braket{\psi | \Heff | \psi} = \braket{\psi|\HR|\psi} + i\braket{\psi|\HI|\psi}\] with real part in $[-\alphaR, \alphaR]$ and imaginary part in $[0, \betaI]$.
The numerical range $\mathcal{W}(\Heff)$ therefore lies in that box and $f(z) := e^{-i z T}/\lambda$ satisfies $|f(z)| = e^{\Im(z)T}/\lambda \leq 1$ there.
Since $f$ is entire of order one, its Faber coefficients on the compact convex set $\mathcal{W}(\Heff)$ decay factorially (Bernstein--Walsh~\cite{walsh1935interpolation}, just as with the Jacobi--Anger truncation in Eq.~\eqref{eq:bidegree}). 
By extension, a Faber--Chebyshev approximation on $\mathcal{W}(\Heff)$ yields a polynomial $P$ of degree \[n = \calO\bigl((\alphaR+\betaI)T + \log(1/\eps)/\log\log(1/\eps)\bigr)\] with $\sup_{\mathcal{W}(\Heff)}|P - f| \leq \eps/(1+\sqrt2)$. 
By the Crouzeix--Palencia theorem~\cite{crouzeix2017numerical}, $\mathcal{W}(\Heff)$ is a $(1+\sqrt2)$-spectral set, so $\|P(\Heff) - A(T)\| \leq \eps$.
Quantum eigenvalue transformation for arbitrary (non-normal, non-diagonalizable) matrices~\cite{gutierrez2026quantum} then applies $P$ to the Jordan spectrum of the block-encoded $\Heff = \mathrm{LCU}(\HR, i\HI)$ using $n$ block-encoding calls, a single processing qubit, and an $\calO(\log n)$-qubit incrementer that enforces $n$-regularity ($U^k$ block-encoding $\Heff^k$). 
The counter arises from non-normality, where a qubitized Hermitian block encoding is $n$-regular without it.

\emph{Regime (ii).} By the Wei--Norman factorization~\cite{wei1963lie}, \[e^{-i\Heff T} = \prod_{k=1}^{m}\exp(g_k(T) Y_k),\] with $\{Y_k\}$ a fixed basis of $\mathfrak{g}$, $m = \dim\mathfrak{g}$ independent of $T$ and of the system dimension, and $|g_k(T)| = \calO(T)$. 
Each $Y_k$ is Hermitian or anti-Hermitian, where its qubitization $\exp(g_k Y_k)$ is realized by univariate QSP/QSVT (Hamiltonian simulation for anti-Hermitian $Y_k$; the contraction $y \mapsto e^{g_k y}$ for Hermitian $Y_k$) to error $\eps/m$ at degree $\calO(\|Y_k\|T + \log(1/\eps)/\log\log(1/\eps))$ (the Jacobi--Anger/Taylor truncations converge factorially), counter-free. 
Composing the $m = \calO(1)$ factors by the product lemma of Gily\'en et al.~\cite{gilyen2019quantum} on separate registers and projecting yields the clean product with $\calO(1)$ ancilla, the compounded factor subnormalization $e^{\sum_k |\Re g_k|\,\|Y_k\|}$ may exceed $\lambda$ and enters the postselection cost (not the query count), bounded by the Wei--Norman coordinate bounds $|g_k| = \calO(T)$ for the named structured classes.

\emph{Angle recovery.} The norm bound $|P_{\mathrm{Dyson}}|^2 \leq 1$ on $\bbT^2$ (Proposition~\ref{prop:normbound_V}, Gr\"onwall~\cite{gronwall1919note}, with $P = V(T)/e^{\betaI T}$) together with the non-degeneracy of the leading coefficients (Jacobi--Anger coefficients $J_{d_{R,j}}(\alphaR\Delta)$ nonzero off a measure-zero set of step sizes, Remark~\ref{rem:bessel_nondegeneracy}, and Taylor coefficient $(\betaI\Delta)^M/M! > 0$) guarantee that, at any non-degenerate schedule, Algorithm~\ref{alg:recursive} recovers the circuit's angles by the forward recursion (Theorem~\ref{thm:CRC}) without breakdown.
\end{proof}

\begin{remark}[Normalization hypothesis in regime (i)]\label{rem:disk_normalization}
The eigenvalue-transformation step of regime~(i) requires $|P|\leq 1$ on the closed unit disk of the rescaled variable (Theorem~6 of Ref.~\cite{gutierrez2026quantum}).
This is stronger than the numerical-range estimate $\sup_{\mathcal{W}}|P-f|\leq\eps$ we use in the proof.
On the top edge of $\mathcal{W}(\Heff)/(\alphaR+\betaI)$ the target satisfies $|f|\equiv 1$ at disk-interior points.
A two-point Schwarz--Cauchy estimate shows that no disk-bounded polynomial $\eps$-approximates $f$ along that edge once $\alphaR T > \pi/(1-\eps)$.
Regime~(i) therefore holds as stated for $\alphaR T = O(1)$. 
For larger $\alphaR T$ the target needs more subnormalizing, $A(T)\mapsto A(T)/\mu$ with $\mu = \Omega(\alphaR T)$. 
A convex-optimization measurement of the minimal feasible $\mu$ gives $\mu^{*} \approx e^{c\,\alphaR T}$ with $c \approx 0.7$--$0.9$ stable over $\alphaR T \in [2,6]$.
In other words, the required subnormalization grows exponentially in $\alphaR T$, effectively confining the disk route to $\alphaR T = O(1)$. 
If we want our implementation to be unconditional for arbitrary $T$, we can use (a) the segmented Dyson-LCU route below, and (b) the LCHS family~\cite{an2023linear,an2026quantum,ni2026quantum}. 
The latter applies because the normalized propagator is a dissipative semigroup, $A(T) = e^{-(i\HR + K)T}$ with $K := \betaI I - \HI \succeq 0$, requiring the target function to be analytic on only the numerical range at query cost $\widetilde{O}\bigl((\alphaR+\betaI)T\cdot\log^{1+1/\beta'}(1/\eps)\bigr)$ (kernel parameter $\beta' \in (0,1)$), or $O(\alphaR T + \betaI T\cdot\mathrm{polylog}(1/\eps))$ when applied in the interaction picture.
In the weak-dissipation regime $\betaI \ll \alphaR$ the dominant $\alphaR T$ term is penalty-free. 
The exponential barrier of Corollary~\ref{thm:barrier_tight} is unaffected by the choice of route. 
The segmented bound of Theorem~\ref{thm:dyson_total} (Method~II) remains the in-paper unconditional statement for all $T$, at the cost of a multiplicative $\log(1/\eps)$ on the $\betaI T$ term only.
\end{remark}

\begin{theorem}[Certified-instance single-qubit membership]\label{thm:certified_membership}
Fix a schedule $\mathbf{s}$ of bidegree $(d_R,d_I)$, $d = d_R + d_I$, and let
$\Phi_d(\bTheta) = \bra{0}_q\,\cG(\bTheta)\,\ket{0}_q$ be the realized-operator map on the
gauge-reduced angle chart. Suppose numerically computed angles $\widehat{\bTheta}$ satisfy,
with certified interval arithmetic or rigorous floating-point bounds:
\textup{(i)} residual $r = \|\Phi_d(\widehat{\bTheta}) - A\|$ for the target contraction $A$;
\textup{(ii)} gauge-reduced least singular value $\sigma_* = \sigma_{\min}\bigl(D\Phi_d(\widehat{\bTheta})\bigr) > 0$; and
\textup{(iii)} the explicit Lipschitz bound $L_2 = 4(d+1)$ on $D\Phi_d$ (each circuit factor and each of its
parameter derivatives has operator norm at most $\sqrt2$, and every second-derivative term carries at most
two insertions). If
\[
  r \;<\; \frac{\sigma_*^{\,2}}{4\,L_2},
\]
then exact angles $\bTheta^{*}$ exist with $\Phi_d(\bTheta^{*}) = A$: the target is the
$(0,0)$ block of a single-qubit M-QSP circuit at bidegree $(d_R,d_I)$.
\end{theorem}

\begin{proof}
$\Phi_d$ is a polynomial map of the $2(d+1)$ real angles whose differential is $\sigma_*$-surjective at
$\widehat{\bTheta}$ and $L_2$-Lipschitz on the ball of radius $\sigma_*/(2L_2)$ (the bound (iii): each term
of $D^2\Phi_d$ is a product of unit-norm signal factors and at most two derivative factors of norm
$\le\sqrt2$, and there are at most $2(d+1)$ parameters per slot). By the quantitative open-mapping
(Graves--Lyusternik) theorem~\cite{dontchev2009implicit}, the image of $\Phi_d$ covers the ball of radius
$\sigma_*^2/(4L_2)$ about $\Phi_d(\widehat{\bTheta})$; since $r$ is smaller, $A$ lies in the image.
\end{proof}

\begin{remark}[Scope of the certificate]\label{rem:certificate_scope}
Theorem~\ref{thm:certified_membership} converts the single-qubit form into a form that can be checked for each instance.
For a concrete $(\HR,\HI,T,\eps)$ we can either prove membership, or fail to prove membership (a failed optimization is not evidence of non-membership). 
Structured instances with a small associative algebra pass it at additive bidegree with large margin (numerical tests show residual $\sim10^{-10}$ against threshold $\sim10^{-4}$), and along the commuting-to-non-commuting homotopy the certificate propagates. 
Anchor margin grows linearly in non-commutativity, $\sigma_*(s) = s\sqrt{\sigma_{\min}(S)} + O(s^2)$ with an explicit first-order Schur complement $S \succ 0$ for every non-commuting spectrally separated instance, uniformly on compact families bounded away from the commuting anchor. 
We also find that (a) the anchor condition $S \succ 0$ is a priori generic. 
$S$ is a Gram-type PSD matrix with real-analytic entries on the connected simple-spectrum chart, so a single strict numerical
witness confines $\{\det S = 0\}$ to a closed measure-zero set; (b) for partially coupled instances the
anchor rate degrades gracefully with the coupling-graph diameter, as $O(s^{\mathrm{diam}(G)})$ (proven
upper bound; matching witness at $d=3$); and (c) the remaining mid-homotopy margin hypothesis reduces to finitely many crossings on $(0,1]$, each of which is bypassed by re-scheduling. The single-qubit form (all targets at once) remains equivalent to multivariate QSP completeness (Conjecture~7 of Laneve-Wolf~\cite{laneve2025multivariate}) and is not claimed.
\end{remark}

\begin{corollary}[LCHS route: unconditional all-$T$ block-encoding]\label{cor:lchs_route}
Fix $\beta' \in (0,1)$. The contraction $A(T)$ admits an $(\calO(1),\eps)$-block-encoding, for all $T$ and all admissible $(\HR,\HI)$, using
$Q = \calO\bigl(\alphaR T + \betaI T\,(\log(1/\eps))^{1/\beta'} + \log(1/\eps)/\log\log(1/\eps)\bigr)$
queries to each of (controlled) $W_R$ and $U_I$, and $\calO(\log(T\betaI) + \log\log(1/\eps))$ ancilla qubits.
\end{corollary}

\begin{proof}
If we use the dissipative form of Remark~\ref{rem:disk_normalization}, $A(T) = e^{-(i\HR + K)T}$ with $K = \betaI I - \HI$, $0 \preceq K \preceq \betaI I$. 
The improved linear-combination-of-Hamiltonian-simulation representation~\cite{an2023linear,an2026quantum} gives the operator identity $e^{-(iH+L)T} = \int_{\mathbb{R}} w_{\beta'}(k)\, e^{-i(H+kL)T}\,dk$ for $L \succeq 0$, with kernel $\|w_{\beta'}\|_{L^1} = \calO(1)$ and truncation radius $R = \calO((\log(1/\eps))^{1/\beta'})$. 
Discretizing to $M = \mathrm{poly}(R, T\betaI, \log(1/\eps))$ nodes costs $\eps/2$.
Implement the resulting LCU with a $\lceil\log M\rceil$-qubit $k$-register, where SELECT is the single static evolution $e^{-i\mathbf{H}T}$ of $\mathbf{H} = \HR \otimes I + K \otimes D$, $D = \mathrm{diag}(k_j)$, since $e^{-i\mathbf{H}T}\ket{k_j}\ket{\psi} = \ket{k_j} e^{-i(\HR + k_j K)T}\ket{\psi}$.
A block-encoding of $\mathbf{H}/(\alphaR + R\betaI)$ costs one $W_R$ and one $U_I$ query ($K/\betaI = I - \HI/\betaI$ is an LCU of the identity and $U_I$. 
The diagonal factor $D/R$ is one $k$-controlled ancilla rotation, oracle-free), and qubitization simulates $e^{-i\mathbf{H}T}$ at the additive Jacobi--Anger degree $\calO((\alphaR + R\betaI)T + \log(1/\eps)/\log\log(1/\eps))$.
Each LCU term is unitary, so the total subnormalization is $\|w_{\beta'}\|_1 = \calO(1)$. 
Substituting $R$ gives $Q$. 
\end{proof}

\begin{remark}[Non-degeneracy of the Bessel condition]\label{rem:bessel_nondegeneracy}
The ``generic'' qualifier in condition~(iii) refers to the requirement $J_{d_{R,j}}(\alphaR\Delta) \neq 0$.
Bessel function zeroes $J_n$ form a discrete set on the positive real line, so for a fixed Trotter step count the set of $\alphaR\Delta$ values violating the condition has measure zero.
In practice, such degeneracy can be removed with an arbitrarily small perturbation of the time step~$\Delta$, since the Bessel zeros are isolated.
The non-degeneracy condition imposes no practical restriction on Algorithm~\ref{alg:recursive}.
\end{remark}

\section{Method~I: Lorentzian Interaction-Picture Simulation}
\label{sec:lorentzian}

We present distinct simulation methods, beginning with a Lorentzian approach, using a continuous-variable ancilla to implement the non-unitary factor $e^{\Htilde(\tau_j)\Delta}$ at each interaction-picture segment, achieving subpolynomial scaling in $1/\eps$ but requiring $r$ intermediate postselections (one per segment).
Sections~\ref{sec:dyson_lcu} and~\ref{sec:mqsp_method} remove these intermediate postselections via Dyson-series LCU and bivariate M-QSP, respectively.

\subsection{The Lorentzian-ancilla mechanism}
\label{subsec:lor_mechanism}

\subsubsection{The Fourier identity}

The method exploits the integral representation
\begin{equation}\label{eq:fourier}
  \int_{-\infty}^{\infty}
  \frac{\gamma/\pi}{s^2 + \gamma^2}\, e^{-iAs}\, ds
  = e^{-\gamma|A|}
\end{equation}
for Hermitian $A$ and $\gamma > 0$. 
The Lorentzian weight $L(s) = \gamma/(\pi(s^2+\gamma^2))$ is encoded in an ancilla register of $n_L$ qubits with discretized amplitudes $a_k \propto \sqrt{L(s_k)}$.

To obtain the growing exponential $e^{+\gamma A}$ needed for non-Hermitian evolution, apply a spectral shift: define $A = -\Htilde + \betaI I \geq 0$ (where positivity follows from $\|\Htilde\| \leq \betaI$,
Lemma~\ref{lem:HI_props}), so $|A| = A$ and
\begin{equation}\label{eq:shifted}
  e^{-\gamma A}
  = e^{-\gamma\betaI}\, e^{+\gamma\Htilde}.
\end{equation}
Setting $\gamma = \Delta$ (the segment width) gives the per-segment non-unitary factor $e^{-\betaI\Delta}\,e^{\Htilde(\tau_j)\Delta}$.

\subsubsection{Discretization error}

\begin{lemma}[Lorentzian discretization bounds] \label{lem:lor_disc}
Discretize $[-s_{\max}, s_{\max}]$ with $N_s = 2^{n_L}$ equally spaced points and spacing $\Delta s = 2s_{\max}/N_s$. 
The discretized approximation to $e^{-\gamma|A|}$ satisfies
\begin{equation}\label{eq:disc_error}
\begin{split}
  &\left\|\sum_{k=1}^{N_s}
    \frac{a_k}{\sqrt{N_s}}\, e^{-iAs_k}
    - e^{-\gamma|A|}\right\|
  \leq\\ &
  \underbrace{\frac{2\gamma}{\pi s_{\max}}}_{\text{tail truncation}}
  +\;
  \underbrace{
    \frac{\|A\|^2 s_{\max}^3}{3\pi\gamma N_s^2}
  }_{\text{quadrature}}.
  \end{split}
\end{equation}
Setting $s_{\max} = O(\gamma/\eps_L)$ for the tail and $N_s = O(\|A\|s_{\max}^{3/2}/(\gamma\eps_L)^{1/2})$ for the quadrature yields total error $\leq \eps_L$ with
\begin{equation}\label{eq:nL}
  n_L = O\!\left(\log\frac{\|A\|\gamma}{\eps_L^2}\right)
  \text{ qubits.}
\end{equation}
\end{lemma}

\begin{proof}
Tail error follows from $\int_{|s|>s_{\max}} L(s)\,ds = 1 - \frac{2}{\pi}\arctan(s_{\max}/\gamma) \leq \frac{2\gamma}{\pi s_{\max}}$ (since $\arctan(x) \geq x/(1+x)$ for $x > 0$).

For the quadrature error, the integrand $f(s) = L(s)e^{-iAs}$ satisfies $\|f''(s)\| \leq C\|A\|^2L(s)$ on each quadrature subinterval.
The composite midpoint rule on $N_s$ subintervals of width $\Delta s$ has error bounded by $(\Delta s)^2\int\|f''\|/24$.
Since $\int_{-s_{\max}}^{s_{\max}}\|A\|^2 L(s)\,ds \leq \|A\|^2$ and $(\Delta s)^2 = 4s_{\max}^2/N_s^2$, the quadrature error is $O(\|A\|^2 s_{\max}^2/(N_s^2\gamma))$. 
Balancing $\gamma$ factors from Lorentzian normalization gives the stated bound.
\end{proof}

\subsection{The interaction-picture Lorentzian circuit}\label{subsec:lor_circuit}

Divide $[0,T]$ into $r$ segments of width $\Delta = T/r$ with midpoints $\tau_j = (j-\tfrac{1}{2})\Delta$. 
The midpoint-approximated propagator is $V_r(T) = \prod_{j=r}^{1} e^{\Htilde(\tau_j)\Delta}$.

For each segment $j$, the circuit on $\mathcal{H}_s \otimes \mathcal{H}_L$ proceeds:
\begin{enumerate}
  \item \textbf{Prepare Lorentzian register}:
  $|0\rangle_L \rightarrow |L\rangle = \sum_k a_k|k\rangle_L$.
  \item \textbf{Forward frame rotation}: Apply $e^{iH_R\tau_j}$ via GQSP on $W_R$.
  \item \textbf{Controlled evolution}: Apply $\sum_k |k\rangle\langle k|_L \otimes e^{-iH_I s_k}$ on $\mathcal{H}_L \otimes \mathcal{H}_s$.
  \item \textbf{Inverse frame rotation}: Apply   $e^{-iH_R\tau_j}$ via GQSP.
  \item \textbf{Postselect} Lorentzian register onto
  $|{+}\rangle_L = 2^{-n_L/2}\sum_k|k\rangle_L$.
\end{enumerate}
Steps~(2)--(4) implement $e^{iH_R\tau_j}e^{-iH_I s_k}e^{-iH_R\tau_j} = e^{-i\Htilde(\tau_j)s_k}$ in each Lorentzian sector. 
Postselection yields $\sum_k(a_k/\sqrt{2^{n_L}})e^{-i\Htilde(\tau_j)s_k} \approx e^{-\betaI\Delta}\,e^{\Htilde(\tau_j)\Delta}$.

\subsection{Error analysis: quadrature vs.\ Trotter}
\label{subsec:quad_error}

Total Lorentzian error decomposes into quadrature discretization error from replacing the Cauchy integral by a finite sum (controlled by the Gauss--Legendre node count $n_L$), Trotter splitting error within each Lorentzian sector, and GQSP approximation error for each frame rotation.
Since the quadrature and Trotter errors are exponentially suppressed by $n_L$ and the Trotter order respectively, dominant cost in practice is the per-step GQSP precision $\eps' = \eps/(3r\,e^{\betaI T})$.
We give a detailed bound in Theorem~\ref{thm:lor_total_error}.

\subsection{Total complexity}
\label{subsec:lor_complexity}

\begin{theorem}[Total error decomposition]
\label{thm:lor_total_error}
\begin{equation}\label{eq:lor_error_decomp}
\begin{split}
      \eps_{\mathrm{total}}
  &\leq
  \underbrace{
    \frac{C_p\, T^{2p+1}}{r^{2p}}\,e^{\betaI T}
  }_{\text{quadrature}}
  +\;
  \underbrace{\eps_R}_{\text{Phase I GQSP}}
  \\&+\;
  \underbrace{
    r\cdot(\eps_{\mathrm{GQSP}} + \eps_L)
  }_{\text{Phase II per-segment}}.
\end{split}
\end{equation}
Contributions extend from Magnus quadrature error from approximating the time-ordered integral by midpoint evaluation, GQSP error in the global frame rotation implementing $e^{-iH_R T}$, and accumulated per-segment error from $r$ GQSP frame rotations (each to precision $\eps_{\mathrm{GQSP}} = \eps/(3r)$) and $r$ Lorentzian discretizations (each to precision $\eps_L = \eps/(3r)$).
Detailed bounds for each term are given in Appendix~\ref{app:lorentzian} (Lemma~\ref{lem:lorentz_tail}), which treats the base case $p = 1$ (midpoint rule); the present theorem generalizes that bound to arbitrary Magnus order~$p$.
\end{theorem}

\begin{theorem}[Asymptotic superiority over Trotter]
\label{thm:lor_superiority}
With $p$-th order quadrature, the Lorentzian method achieves total query complexity
\begin{equation}\label{eq:Q_lor}
  Q_{\mathrm{Lor}}^{(p)}
  = O\!\left(
    \frac{T^{1+1/(2p)}}{\eps^{1/(2p)}}
    \left(\alphaR T + \log\frac{1}{\eps}\right)
  \right).
\end{equation}
For $p$-th order Trotter, $Q_{\mathrm{Trotter}}^{(p)}= O(T^{1+1/p}\eps^{-1/p}(\alphaR + \betaI)^{1+1/p})$.
The ratio satisfies $Q_{\mathrm{Lor}}^{(p)}/Q_{\mathrm{Trotter}}^{(p)} = O(\eps^{1/(2p)}) \to 0$ as $\eps \to 0$.
\end{theorem}

\begin{proof}
The number of segments scales as $r = O(T^{1/(2p)}\eps^{-1/(2p)})$ (inverting the quadrature error bound with $p$-th order Magnus). 
Each segment requires $O(\alphaR\Delta + \log(r/\eps))$ queries to $W_R$ for the GQSP frame rotation to precision $\eps/(3r)$, plus $2^{n_L} = O(\mathrm{poly}(\betaI\Delta/\eps_L))$ controlled operations for the Lorentzian register. 
Summing over $r$ segments gives the stated bound.
\end{proof}

\begin{corollary}[Polylogarithmic limit]
\label{cor:lor_polylimit}
In the $p \to \infty$ limit, with $p$ chosen optimally, the Lorentzian method achieves
\begin{equation}\label{eq:lor_polylog}
\begin{split}
     Q_{\mathrm{Lor}}
  &= O\!\left((\alphaR + \betaI)T
    \cdot \exp\!\bigl(c\sqrt{\ln(1/\eps)}\bigr)\right),
  \\c &= \sqrt{2\ln 3} \approx 1.48.
\end{split}
\end{equation}
This is $\exp(o(\ln(1/\eps)))$, being subpolynomial in $1/\eps$ but superpolylogarithmic.
\end{corollary}

\begin{proof}
For $p$-th order quadrature implemented via Gauss--Legendre with symmetric composition (Suzuki-type staging), each segment uses $O(3^p)$ quadrature substages. 
Total cost is proportional to
\begin{equation}
  f(p)
  = r(p) \cdot 3^p
  = O\!\left(\eps^{-1/(2p)} \cdot 3^p\right),
\end{equation}
where the first factor counts segments and the second counts substages per segment Writing $L = \ln(1/\eps)$ and $a = \ln 3$, the logarithm of cost is
\begin{equation}\label{eq:fp}
  \ln f(p) = \frac{L}{2p} + ap + O(\ln T).
\end{equation}
Minimizing over $p > 0$: setting
$d(\ln f)/dp = -L/(2p^2) + a = 0$ gives the optimal order
\begin{equation}
  p^* = \sqrt{\frac{L}{2a}}
  = \sqrt{\frac{\ln(1/\eps)}{2\ln 3}}.
\end{equation}
Substituting back:
\begin{equation}
  \ln f(p^*)
  = \frac{L}{2\sqrt{L/(2a)}} + a\sqrt{\frac{L}{2a}}
  = \sqrt{2aL},
\end{equation}
so $f(p^*) = \exp(\sqrt{2(\ln 3)\ln(1/\eps)})$,
confirming $c = \sqrt{2\ln 3} \approx 1.482$.
\end{proof}

\begin{remark}[Nature of the Lorentzian scaling]
\label{rem:lor_scaling}
The $\exp(c\sqrt{\ln(1/\eps)})$ scaling is characteristic of ``infinite-order Suzuki'' methods.
We replace this with polylogarithmic scaling by absorbing all segments into a single bivariate polynomial in Section~\ref{sec:mqsp_method}.

\end{remark}

\subsection{Postselection cost}
\label{subsec:lor_postselection}

Total postselection probability from $r$ Lorentzian segments telescopes (Corollary~\ref{thm:barrier_tight}):
\begin{equation}\label{eq:lor_postsel}
\begin{split}
      P &= \prod_{j=1}^{r} p_j
  = e^{-2\betaI T}\|V_r(T)|\psi_0\rangle\|^2
  \\&\xrightarrow{r\to\infty}
  e^{-2\betaI T}\|e^{-i\Heff T}|\psi_0\rangle\|^2.
\end{split}
\end{equation}
Asymptotic success probability is information-theoretically optimal (Sec.~\ref{sec:barrier}), but the Lorentzian method pays it through $r$ intermediate postselections rather than a single postselection at the end. 
Each intermediate postselection introduces a probabilistic gate that cannot be made deterministic, and the
product of $r$ success probabilities can be substantially smaller than the optimal single-postselection value for finite $r$.  
The Dyson-series methods of Secs.~\ref{sec:dyson_lcu} and~\ref{sec:mqsp_method} eliminate this limitation.

\section{Method~II: Dyson LCU with Single Postselection}
\label{sec:dyson_lcu}

Here, we show how to package the full product into a single linear combination of unitaries (LCU), requiring only one postselection at the end. 
The resulting query complexity is
$O((\alphaR + \betaI)T + \betaI T\log(1/\eps))$, being nearly optimal with a gap of $\betaI T$ on the $\log(1/\eps)$ term, eliminated by the M-QSP method (Sec.~\ref{sec:mqsp_method}).

\subsection{From segments to a single LCU}
\label{subsec:dyson_idea}

Expand each interaction-picture factor in Taylor series to order $M$:
\begin{equation}\label{eq:taylor_expand}
  e^{\Htilde(\tau_j)\Delta}
  = \sum_{m=0}^{M}
    \frac{(\Htilde(\tau_j)\Delta)^m}{m!}
    + R_M(\betaI\Delta),
\end{equation}
where the remainder satisfies $\|R_M(\betaI\Delta)\| \leq (\betaI\Delta)^{M+1}/(M+1)!$. 
The full truncated propagator is
\begin{equation}\label{eq:VrM_truncated}
  \hat{V}_{r,M}(T)
  = \prod_{j=r}^{1}
    \left[\sum_{m_j=0}^{M}
      \frac{(\Htilde(\tau_j)\Delta)^{m_j}}{m_j!}
    \right],
\end{equation}
a sum over $(M+1)^r$ terms, each a product of unitaries and block-encoded operators, following LCU structure.

\subsection{The LCU circuit architecture}
\label{subsec:lcu_circuit}

\subsubsection{Registers}

For each segment $j \in \{1,\ldots,r\}$, introduce a local order register $\mathcal{H}_{m_j}$ of dimension $M+1$ with basis $\{|m_j\rangle\}_{m_j=0}^{M}$. 
The combined order register is $\mathcal{H}_{\mathbf{m}} = \bigotimes_{j=1}^{r}\mathcal{H}_{m_j}$.

\subsubsection{PREPARE}

For each segment $j$, prepare the local order state
\begin{equation}\label{eq:prepare}
  |0\rangle_{m_j}
  \mapsto |\alpha_j\rangle
  = \sum_{m=0}^{M}
    \sqrt{\frac{(\betaI\Delta)^m/m!}{\mu}}\,|m\rangle_{m_j},
\end{equation}
where $\mu = \sum_{m=0}^{M}(\betaI\Delta)^m/m! \leq e^{\betaI\Delta}$.
The total PREPARE gives $|\mathbf{0}\rangle \mapsto |\mathbf{\alpha}\rangle = \bigotimes_{j=1}^{r}|\alpha_j\rangle$ with normalization $\lambda = \mu^r \leq e^{\betaI T}$.

\subsubsection{SELECT}

SELECT processes segments $j = 1, \ldots, r$ sequentially.
For segment $j$, controlled on $|m_j\rangle$:
\begin{enumerate}
  \item \textbf{Advance the frame}: Apply $e^{-iH_R\Delta}$ via GQSP on $W_R$.
  \item \textbf{Apply $H_I$ insertions}: Make $m_j$ queries to the block encoding of $\Htilde(\tau_j) = e^{iH_R\tau_j}H_I e^{-iH_R\tau_j}$.
\end{enumerate}
Each query to $\Htilde(\tau_j)$ consists of the accumulated frame rotation (from previous steps), one query to $U_I$, and the inverse frame rotation. 
Frame rotation to time $\tau_j$ is built by accumulating incremental GQSP rotations $e^{-iH_R\Delta}$ as segments are processed sequentially, avoiding redundant recomputation.

\subsubsection{Block-encoding verification}

\begin{proposition}[Circuit correctness]\label{prop:correctness}
\begin{equation}\label{eq:lcu_correctness}
\begin{split}
      \langle\mathbf{0}|&\,
  \mathrm{PREPARE}^\dagger \cdot \mathrm{SELECT}
  \cdot \mathrm{PREPARE}\,
  |\mathbf{0}\rangle\,|\psi_0\rangle
  \\&= \frac{\hat{V}_{r,M}(T)}{\lambda}\,|\psi_0\rangle.
\end{split}
\end{equation}
\end{proposition}

\begin{proof}
After PREPARE, the state is $|\mathbf{\alpha}\rangle|\psi_0\rangle$. 
SELECT acts on each segment:
\begin{equation}
\begin{split}
      &\mathrm{SELECT}\,|\mathbf{\alpha}\rangle|\psi_0\rangle
  \\&= \sum_{m_1,\ldots,m_r}
    \prod_{j=1}^{r}
    \sqrt{\frac{(\betaI\Delta)^{m_j}/m_j!}{\mu}}\,
    |m_1,\ldots,m_r\rangle \\&
    \cdot\prod_{j=r}^{1}
    \left(\frac{\Htilde(\tau_j)}{\betaI}\right)^{\!m_j}
    \!\!|\psi_0\rangle.
\end{split}
\end{equation}
Applying $\mathrm{PREPARE}^\dagger$ and projecting onto $|\mathbf{0}\rangle$ produces
\begin{equation}
  \frac{1}{\lambda}
  \prod_{j=r}^{1}
  \left[\sum_{m_j=0}^{M}
    \frac{(\Htilde(\tau_j)\Delta)^{m_j}}{m_j!}
  \right]|\psi_0\rangle
  = \frac{\hat{V}_{r,M}(T)}{\lambda}\,|\psi_0\rangle.
\end{equation}
\end{proof}

The unitary factor is free by Theorem~\ref{thm:free_unitary}

\subsection{Error analysis and parameter setting}
\label{subsec:dyson_error}

Following Theorem~\ref{thm:dyson_error}, we set $\betaI\Delta = c$ for a fixed constant $c > 0$ (e.g., $c = 1$), giving $r = \lceil\betaI T/c\rceil = O(\betaI T)$. 
The Magnus quadrature term becomes $O(\alphaR c^2 T) = O(\alphaR T)$, absorbed into leading-order complexity. 
The Taylor truncation term requires
\begin{equation}\label{eq:taylor_req}
  \frac{c^{M+1}}{(M+1)!}
  \leq \frac{\eps}{2\betaI T\, e^{\betaI T}}.
\end{equation}

\begin{lemma}[Taylor truncation order]\label{lem:taylor_order}
The truncation order satisfying~\eqref{eq:taylor_req} is
\begin{equation}\label{eq:M_bound}
\begin{split}
  M &= O\!\left(
    c + \frac{\log(\betaI T\, e^{\betaI T}/\eps)}
             {\log\!\log(\betaI T\, e^{\betaI T}/\eps)}
  \right)\\
  &= O\!\left(
    \betaI\Delta
    + \frac{\betaI T + \log(1/\eps)}
           {\log(\betaI T + \log(1/\eps))}
  \right).
  \end{split}
\end{equation}
For the dominant regime
$\log(1/\eps) \gg \betaI T$, this simplifies to
$M = O(\log(1/\eps)/\log\!\log(1/\eps))$.
\end{lemma}

\begin{proof}
By Stirling's approximation, $(M+1)! \geq \sqrt{2\pi(M+1)}((M+1)/e)^{M+1}$. 
The condition $c^{M+1}/(M+1)! \leq \delta$ (with $\delta = \eps/(2re^{\betaI T})$) is satisfied when $(ce/(M+1))^{M+1} \leq \delta\sqrt{2\pi(M+1)}$. 
Taking logarithms: $(M+1)\ln(ce/(M+1)) \leq \ln\delta + O(\ln M)$. 
For $M \gg c$, the left side is $\approx -(M+1)\ln(M/c)$, so the condition becomes $(M+1)\ln(M/c) \geq \ln(1/\delta)$.
Setting $M = \Theta(\ln(1/\delta)/\ln\!\ln(1/\delta))$ satisfies this, since $M\ln M \approx (\ln(1/\delta)/\ln\!\ln(1/\delta)) \cdot\ln(\ln(1/\delta)/\ln\!\ln(1/\delta)) = \Theta(\ln(1/\delta))$.
\end{proof}

\subsection{Query complexity}
\label{subsec:dyson_complexity}

\begin{theorem}[Walk-operator queries]\label{thm:Q_WR}
The total number of queries to $W_R$ is
\begin{equation}\label{eq:Q_WR}
\begin{split}
    Q_{W_R}
  &= r \cdot O\!\left(
    \alphaR\Delta + \log\frac{1}{\eps'}
  \right)
  \\&= O\!\left(
    \alphaR T + \betaI T\log\frac{1}{\eps}
  \right),  
\end{split}
\end{equation}
where
$\eps' = \eps/(3re^{\betaI T})$ is the per-step GQSP
precision.
\end{theorem}

\begin{proof}
Each of the $r$ segments requires one GQSP implementation of $e^{-iH_R\Delta}$ to precision $\eps'$. 
By the GQSP query bound, this costs $O(\alphaR\Delta + \log(1/\eps'))$ queries to $W_R$.
Summing: $Q_{W_R} = r\cdot O(\alphaR\Delta + \log(1/\eps')) = O(\alphaR T + r\log(re^{\betaI T}/\eps))$.
Since $r = O(\betaI T)$, the $\log r$ and $\log e^{\betaI T} = \betaI T$ terms are $O(\betaI T\log(\betaI T))$, absorbed into the leading $\betaI T\log(1/\eps)$ for $\eps$ small.
\end{proof}

\begin{theorem}[Block-encoding queries for $H_I$]\label{thm:Q_UI}
The total number of queries to $U_I$ is
\begin{equation}\label{eq:Q_UI}
  Q_{U_I}
  = r \cdot M
  = O\!\left(
    \betaI T
    + \frac{\betaI T\log(1/\eps)}
           {\log\!\log(1/\eps)}
  \right).
\end{equation}
\end{theorem}

\begin{proof}
Immediate from $r = O(\betaI T)$ segments (Sec.~\ref{sec:interaction}) and $M = O\bigl(\log(1/\eps)/\log\!\log(1/\eps)\bigr)$ Dyson orders per segment (Appendix~\ref{app:dyson_error}, Lemma~\ref{lem:error_budget_dyson}).
\end{proof}

\begin{theorem}[Total Dyson LCU query complexity]
\label{thm:dyson_total}
\begin{equation}\label{eq:dyson_Q}
    Q_{\mathrm{Dyson}}
    = O\!\left(
      (\alphaR + \betaI)T
      + \betaI T\log\frac{1}{\eps}
    \right)
\end{equation}
with a single postselection on the order registers and optimal success probability
$P = e^{-2\betaI T}\|e^{-i\Heff T}|\psi_0\rangle\|^2$.
\end{theorem}

\begin{remark}[The $Q_{W_R}$ bottleneck]\label{rem:WR_bottleneck}
The Dyson LCU is bottlenecked by frame rotation cost ($\betaI T \cdot \log(1/\eps)$).
The bottleneck arises because each of the $r = O(\betaI T)$ segments requires an independent GQSP implementation to precision $\eps/(re^{\betaI T})$. 
The ``union bound'' penalty $r \cdot \log(r/\eps)$ is eliminated by bivariate M-QSP (Sec.~\ref{sec:mqsp_method}).
\end{remark}

\begin{proof}
Walk-operator queries: $Q_{W_R} = r \cdot O(\alphaR\Delta + \log(1/\eps')) = O(\alphaR T + \betaI T\log(1/\eps))$, where $\eps' = \eps/(3r\,e^{\betaI T})$ is the per-step GQSP precision.
Block-encoding queries: $Q_{U_I} = r \cdot M = O(\betaI T + \betaI T\log(1/\eps))$.
Total: $Q = O((\alphaR + \betaI)T + \betaI T\log(1/\eps))$.
Postselection probability follows from the achievability theorem (Theorem~\ref{thm:achievability}) with normalization $\lambda \leq e^{\betaI T}$.
\end{proof}

\section{Method~III: Bivariate M-QSP and Main Theorem}
\label{sec:mqsp_method}

In this section, we replace the segmented architecture of Dyson LCU with a single
bivariate M-QSP circuit that processes the entire Dyson polynomial as one object, achieving query complexity $O((\alphaR + \betaI)T + \log(1/\eps))$.

\subsection{The M-QSP circuit}
\label{subsec:mqsp_circuit}

\begin{definition}[Bivariate M-QSP circuit]
\label{def:mqsp_circuit}

\begin{figure}[!b]
\centering
\includegraphics[width=0.4\textwidth]{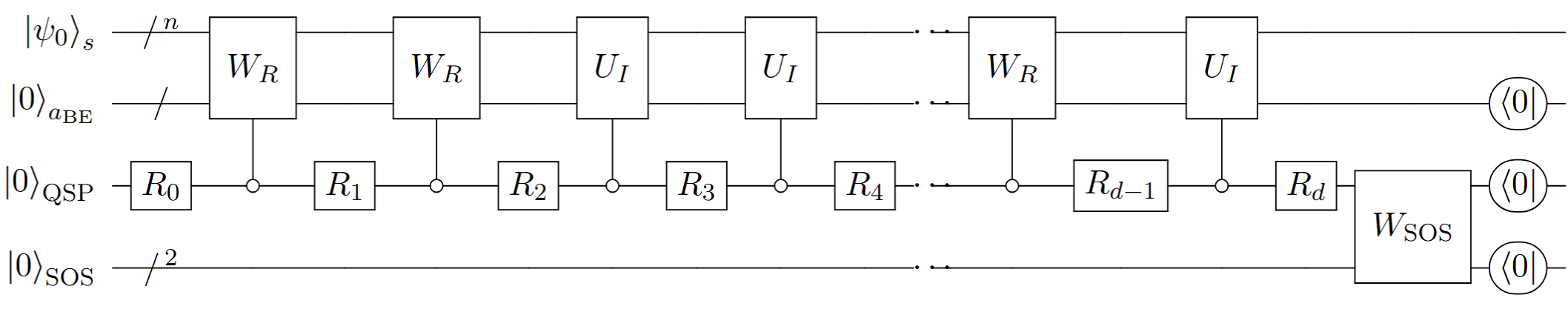}
\vspace{-0.4em}
\begin{minipage}{0.96\linewidth}
\footnotesize
\textbf{Segment $j=r$} (leftmost): all $W_R$ queries precede all $U_I$ queries within the segment. \quad
\textbf{Segment $j=1$} (rightmost): final $W_R$ and $U_I$ before the rank-two SOS dilation $W_{\mathrm{SOS}}$ and the combined postselection $\bra{0}_{a_{\mathrm{BE}}}\bra{0}_{\mathrm{QSP}}\bra{0}_{\mathrm{SOS}}$.
\end{minipage}
\caption{%
The bivariate M-QSP circuit $\cG(\bm{\Theta},\mathbf{s}) = R_0\!\prod_{j=1}^{d_R+d_I} A_{s_j}\,R_j$ of Definition~\ref{def:mqsp_circuit} followed by the rank-two SOS dilation $W_{\mathrm{SOS}}$, for non-Hermitian Hamiltonian simulation in the separate-oracle model.
The four register groups are: the $n$-qubit system ($\ket{\psi_0}_s$), the block-encoding ancillas ($\max(a_R,a_I)$ qubits, shared between $W_R$ and $U_I$ via a common embedding); the QSP signal-processing qubit; and the rank-two SOS complement register ($a_{\mathrm{SOS}}=2$ qubits, independent of polynomial degree by Theorem~\ref{thm:sos_rank_two}).
}
\label{fig:mqsp_circuit}
\end{figure}

Given a schedule $\mathbf{s} = (s_1, \ldots, s_{d_R + d_I})$ with $s_j \in \{R, I\}$ specifying the interleaving pattern, and rotation angles $\bTheta = (\theta_0, \ldots, \theta_{d_R + d_I})$, the M-QSP circuit is
\begin{equation}\label{eq:mqsp_circuit}
  \cG(\bTheta, \mathbf{s})
  = R_0 \prod_{j=1}^{d_R+d_I} A_{s_j} \cdot R_j,
\end{equation}
where $A_R = |0\rangle\langle 0|_a \otimes W_R + |1\rangle\langle 1|_a \otimes I$ (controlled walk operator on system~$+$ ancilla), $A_I = |0\rangle\langle 0|_a \otimes U_I + |1\rangle\langle 1|_a \otimes I$ (controlled block-encoding), and $R_j$ are $\mathrm{SU}(2)$ rotations on the QSP ancilla qubit.
A circuit schematic is given in Figure~\ref{fig:mqsp_circuit}.
\end{definition}

The $(0,0)$ block of $\cG$ realizes the target polynomial $P_\delta(z_1, z_2)$, and the $(1,0)$ block realizes the complementary polynomial(s) $Q(z_1, z_2)$. 
The schedule $\mathbf{s}$ is the {Dyson schedule}: within each Dyson block, all $W_R$ queries (frame rotations) precede all $U_I$ queries ($H_I$ insertions), matching the interaction-picture structure (Proposition~\ref{prop:dyson_poly}).

\subsection{Ancilla cost}
\label{subsec:ancilla_cost}

\begin{proposition}[Total ancilla cost of M-QSP circuit]
\label{prop:ancilla_cost}
Consider an $n$-qubit system with $\Heff = \HR + i\HI$, where $W_R$ is the walk operator (Definition~\ref{def:walk_op}) constructed from an $(\alphaR, a_R, 0)$-block-encoding of $\HR$, and $U_I$ is the analogous walk operator for $\HI/\betaI$. 
The M-QSP circuit uses a total of
\begin{equation}\label{eq:ancilla_total}
  n_{\mathrm{anc}} = \max(a_R, a_I) + 1 + a_{\mathrm{SOS}}
\end{equation}
ancilla qubits, where:
\begin{enumerate}
  \item[\textup{(i)}] $\max(a_R, a_I)$: block-encoding ancillas,
  \item[\textup{(ii)}] $1$: the QSP signal-processing qubit,
  \item[\textup{(iii)}]
  $a_{\mathrm{SOS}} = 2$: ancillas for   the SOS complementary polynomial (Theorem~\ref{thm:sos_rank_two}; $L = 2$ independent of $(d_R, d_I)$).
\end{enumerate}
For block-encodings of $k$-sparse Hamiltonians on $n$ qubits, $a_R = a_I = \calO(\log n)$, giving
\begin{equation}\label{eq:ancilla_scaling}
  n_{\mathrm{anc}}
  = \calO(\log n) + 1 + 2.
\end{equation}
\end{proposition}

\begin{proof}
The M-QSP circuit $\cG(\bTheta, \mathbf{s}) = R_0 \prod_{j} A_{s(j)} R_j$ acts on $\calH_{\mathrm{sys}} \otimes \calH_{\mathrm{anc}}$.
We account for each ancilla register.

\medskip\noindent
\textit{(i) Block-encoding ancillas.}\;
The signal operators $A_R = W_R$ and $A_I = U_I$ act on $\calH_{\mathrm{sys}} \otimes \calH_{a_R}$ and $\calH_{\mathrm{sys}} \otimes \calH_{a_I}$ respectively.
In the M-QSP circuit, $W_R$ and $U_I$ are called in alternation according to the schedule $\mathbf{s}$, and both must act on a common ancilla register. 
This is achieved by embedding both block-encodings into the larger space $\calH_{a} \cong (\C^2)^{\otimes \max(a_R, a_I)}$, padding the smaller block-encoding with identity on the extra qubits.

Typically, a $k$-sparse $n$-qubit Hamiltonian admits an $(\alpha, \calO(\log n), 0)$-block-encoding via the sparse-access model~\cite{berry2014exponential,gilyen2019quantum}, giving $a_R, a_I = \calO(\log n)$. 
For linear combinations of unitaries, $a = \lceil\log_2 L\rceil$ where $L$ is the number of terms.

\medskip\noindent
\textit{(ii) QSP ancilla.}\;
Signal-processing rotations $R_j \in \mathrm{SU}(2)$ act on a single ancilla qubit $\calH_{\mathrm{QSP}} \cong \C^2$, whose $|0\rangle$-component selects the target polynomial $P_\delta$ and whose $|1\rangle$-component houses the complementary polynomials.
This contributes one qubit regardless of the polynomial degree or the system size.

\medskip\noindent
\textit{(iii) SOS complement ancillas.}\;
By Theorem~\ref{thm:sos_rank_two}, the regularized complement $H_\delta$ has SOS rank $L = 2$, independent of $(d_R, d_I)$.
The two complementary polynomials $Q_1 = \sqrt{2\delta - \delta^2}$ (constant) and $Q_2 = (1-\delta)Q$ are embedded in the $|1\rangle$-branch of the QSP qubit, requiring an additional register of dimension $L + 1 = 3$:
\begin{equation}\label{eq:aSOS_L2}
  a_{\mathrm{SOS}} = \lceil\log_2 3\rceil = 2
\end{equation}
qubits.

\medskip\noindent
\textit{Independence of registers.}\;
Block-encoding ancillas define the signal subspace projector $\Pi = |0^a\rangle\!\langle 0^a|$, the QSP qubit carries the polynomial transformation, and the SOS register expands the complementary subspace. 
Block-encoding ancillas are initialized and postselected to $|0^a\rangle$; the QSP qubit is postselected to $|0\rangle$ (selecting $P_\delta$); and the SOS register is traced out (its contents correspond to the discarded complementary polynomials). 
The total count~\eqref{eq:ancilla_total} is therefore additive with no cross-register overhead.

\medskip\noindent
\textit{Comparison to univariate QSP.}\;
In univariate QSP or GQSP for Hermitian simulation, the complement is scalar ($L = 1$) by the univariate Fej\'er--Riesz theorem, so $a_{\mathrm{SOS}} = 0$. 
The additional $2$ qubits are the price of bivariate polynomial processing, giving a negligible overhead that is constant regardless of the polynomial degrees.
\end{proof}

\begin{remark}[Gate compilation and complementary approaches]
\label{rem:gate_compilation}
The M-QSP circuit requires implementing $O((d_R + d_I) \cdot d_R d_I)$ rotations $R(\theta_j, \phi_j) \in \mathrm{SU}(2)$ with angles determined to machine precision (Proposition~\ref{prop:block_peel}).
For a fault-tolerant implementation, arbitrary SU(2) rotations are compiled via the Solovay--Kitaev algorithm (or Ross--Selinger), incurring a polylogarithmic T-gate overhead of $\calO(\log^c(1/\eps))$ per rotation, giving a total gate-count contribution of $\calO((d_R + d_I) \cdot d_R d_I \cdot \log^c(1/\eps))$ gates. 
This does not affect the asymptotic query complexity, and the dominant cost remains the oracle queries, which scale as $O((\alphaR + \betaI)T + \log(1/\eps)/\log\log(1/\eps))$.

In contrast, randomized product-formula methods like qDRIFT (~\cite{campbell2022early}) sidestep angle-finding entirely, using random Hamiltonian fragments to approximate the evolution.
For non-Hermitian simulation, qDRIFT achieves $\calO((\alphaR + \betaI)^2 T^2 / \eps)$ queries with no postselection overhead and is optimized for short-time or high-error regimes where the prefactors and lack of postselection dominate.
\end{remark}

\subsection{Main theorem and proof}
\label{subsec:main_theorem}

\begin{theorem}[Main result]\label{thm:main_theorem}
There exists a quantum algorithm that, given oracle access to walk operators $W_R$ and $U_I$ encoding $\HR/\alphaR$ and $\HI/\betaI$ respectively, simulates $e^{-i\Heff T}$ to error $\eps$
using
\begin{equation}\label{eq:main_query}
  Q = \calO\!\left((\alphaR + \betaI)T + \frac{\log(1/\eps)}{\log\log(1/\eps)}\right)
\end{equation}
queries, with a single postselection on one ancilla qubit achieving
success probability
\begin{equation}\label{eq:main_succ}
  P = e^{-2\betaI T}\|e^{-i\Heff T}\ket{\psi_0}\|^2
  (1 - \calO(\eps)).
\end{equation}
The rotation angles are computed classically in time
$\calO(d_R \cdot d_I)$ via CRC-exploiting block peeling (Proposition~\ref{prop:block_peel}).
\end{theorem}

The proof assembles the results of
Secs.~\ref{sec:interaction}--\ref{sec:anglefinding} into a single
construction. 
We present each step, its role in the chain, and the error accounting.

\medskip\noindent
\textbf{Steps 1--4 compressed (construction recap).}
The interaction-picture decomposition $e^{-i\Heff T} = e^{-i\HR T}\cdot V(T)$ with $V(T) = \cT_> \exp(\int_0^T \Htilde(s)\,ds)$ (Proposition~\ref{prop:ip_exact}), the Dyson-polynomial construction with bidegrees $d_R = O(\alphaR T + \log(1/\eps))$, $d_I = O(\betaI T + \log(1/\eps)/\log\log(1/\eps))$ and total approximation error $\eps/2$ (Proposition~\ref{prop:dyson_poly}, Theorem~\ref{thm:dyson_error}), the regularization $P_\delta = (1-\delta)P$ with the $L=2$ SOS complement supplied automatically by the M-QSP circuit (Theorem~\ref{thm:sos_rank_two}, Theorem~\ref{thm:automatic}), and the constant-ratio condition (Theorem~\ref{thm:CRC}, Theorem~\ref{thm:coeff_sep}) together establish that a bivariate polynomial $P_\delta \in \cP_{d_R, d_I}^+$ with $|P_\delta| \leq 1 - O(\eps)$ on $\bbT^2$ approximates $V(T)/e^{\betaI T}$ to within $\eps/2$, and that its leading-coefficient ratio at every peeling step is a scalar $e^{-i\phi_{k-1}}\tan\theta_{k-1}$. We now assemble these ingredients into a circuit.

\begin{table*}[!t]
\centering
\renewcommand{\arraystretch}{1.3}
\begin{tabular}{@{\extracolsep{\fill}}lll@{}}
\toprule
Register & Qubits & Role \\
\midrule
System & $n$ & Physical Hilbert space \\
Block-encoding & $\max(a_R, a_I)$
  & Signal subspace projector $|0^a\rangle\!\langle 0^a|$ \\
QSP & $1$
  & Polynomial selection ($|0\rangle \to P_\delta$) \\
SOS complement & $2$
  & Complementary subspace for $Q_1, \ldots, Q_L$ \\
\midrule
Total ancilla
  & $\max(a_R,a_I) + 3$
  & $= \calO(\log n) + 3$ \\
\bottomrule
\end{tabular}
\caption{\label{tab:ancilla}Ancilla register decomposition for the M-QSP non-Hermitian simulation circuit.
The block-encoding cost $\max(a_R,a_I)$ is problem-dependent; $\calO(\log n)$ is typical for sparse Hamiltonians. 
The SOS cost is the bivariate overhead absent in univariate QSP.}
\end{table*}

\begin{proof}
The construction recap above tracks an $\eps/2$ total error budget allocated as $\eps/4$ each to midpoint-quadrature ($r = O(\betaI T)$ segments) and Taylor truncation ($M = O(\log(1/\eps)/\log\log(1/\eps))$ per segment); the $(1-\delta)^2 = 1 - O(\eps)$ rescaling absorbs into the success-probability factor of Step~7. We now derive the circuit and its query complexity.

\medskip\noindent
\textbf{Step 5: Constructive angle-finding
(Theorem~\ref{thm:anglefinding}).}

With the CRC established, rotation angles $\bTheta = \{(\theta_k, \phi_k)\}_{k=0}^{d_R+d_I}$ are determined by Algorithm~\ref{alg:recursive}: at each step $k$, compute the leading-coefficient ratio $\rho_k = \mathbf{b}_{\mathrm{lead}}/\mathbf{a}_{\mathrm{lead}}$;set $\theta_k = \arctan|\rho_k|$ and $\phi_k = -\mathrm{Arg}(\rho_k)$; apply $R(\theta_k, \phi_k)^{-1}$ and divide by $z_{s(k)}$ to reduce the bidegree.
After $d_R + d_I$ steps, the bidegree reaches $(0,0)$ and a final rotation $R_0$ realizes the base pair. 
The Dyson polynomial satisfies all achievability conditions (Corollary~\ref{cor:dyson_achievable}), so no breakdown occurs.
The classical cost is
\begin{equation}\label{eq:angle_cost}
  \sum_{k=0}^{d_R+d_I} O(d_R^{(k)} \cdot d_I^{(k)})
  = O\bigl((d_R + d_I) \cdot d_R \cdot d_I\bigr),
\end{equation}
polynomial in all parameters, with no variational optimization.

\medskip\noindent
\textbf{Step 6: Circuit assembly.}

The M-QSP circuit is
\begin{equation}
      \cG = R(\theta_0, \phi_0)
  \prod_{j=1}^{d_R + d_I}
  \bigl[A_{s(j)} \cdot R(\theta_j, \phi_j)\bigr],
\end{equation}
with signal operators $A_R$ and $A_I$ interleaved according to the Dyson schedule $\mathbf{s}$.
The $(0,0)$ block realizes $P_\delta(W_R, U_I)$ to the prescribed precision.
Accumulated frame rotations over all $d_R$ queries to $W_R$ produce $e^{-iH_R T}$ (Theorem~\ref{thm:free_unitary}), so Hermitian evolution is obtained at zero additional cost.

Total query complexity is
\begin{equation}\label{eq:total_queries_main}
\begin{split}
    Q &= d_R + d_I
  = O\!\left(
    (\alphaR + \betaI)T
    + \frac{\log(1/\eps)}{\log\!\log(1/\eps)}
  \right)
  \\&= O\bigl((\alphaR + \betaI)T + \log(1/\eps)\bigr),  
\end{split}
\end{equation}
where the first expression is the precise Stirling-corrected bound and the second is the $O(\cdot)$ simplification used in the theorem statement. 
The target state is recovered by measuring $|0\rangle_a$:
\begin{equation}\label{eq:postselection}
\begin{split}
     \langle 0|_a\,\cG\,|0\rangle_a\,|\psi_0\rangle
  &= P_\delta(W_R, U_I)\,|\psi_0\rangle
  \\&= (1-\delta)\,P(W_R, U_I)\,|\psi_0\rangle.
\end{split}
\end{equation}

\medskip\noindent
\textbf{Step 7: Postselection optimality
(Corollary~\ref{thm:barrier_tight}).}

The success probability is
\begin{align}
  P_{\mathrm{succ}}
  &= \bigl\|P_\delta(W_R, U_I)\,|\psi_0\rangle\bigr\|^2
  \notag\\
  &= (1-\delta)^2\,
     \frac{\bigl\|V(T)\,|\psi_0\rangle\bigr\|^2}
          {e^{2\betaI T}}
     \cdot (1 - O(\eps_{\mathrm{poly}}))
  \notag\\
  &= e^{-2\betaI T}\,
     \bigl\|e^{-i\Heff T}\,|\psi_0\rangle\bigr\|^2
     \cdot (1 - O(\eps)),
  \label{eq:Psucc}
\end{align}
where the second line uses
$\|P - V(T)/e^{\betaI T}\| \leq \eps/2$ from Step~2, and the
third uses $V(T)|\psi_0\rangle
= e^{i\HR T}\,e^{-i\Heff T}|\psi_0\rangle$ with unitarity of
$e^{i\HR T}$, absorbing $(1-\delta)^2$ and
$(1 - O(\eps_{\mathrm{poly}}))$ into the single factor
$(1 - O(\eps))$.

This is provably optimal: block encoding $e^{-i\Heff T}$
has normalization
$\lambda \geq \|e^{-i\Heff T}\|_{\mathrm{op}} = e^{\betaI T}$,
so the success probability is bounded above by
$e^{-2\betaI T}\|e^{-i\Heff T}|\psi_0\rangle\|^2$.  Our algorithm
saturates this bound to within $(1 - O(\eps))$.

The lower bound
$\Omega((\alphaR + \betaI)T + \log(1/\eps))$
(Theorem~\ref{thm:main_lower}) is established via three independent
adversarial reductions:
(a)~$\Omega(\alphaR T)$ from Hermitian simulation;
(b)~$\Omega(\betaI T)$ from the polynomial method for
$e^{\betaI T x}$;
(c)~$\Omega(\log(1/\eps))$ from the Heisenberg limit.
These costs are additive because $W_R$ and $U_I$ are independent
oracles (Sec.~\ref{sec:lower_bound}).  The upper
bound~\eqref{eq:total_queries_main} matches the lower bound up to the
$\log\!\log$ correction in the Taylor truncation.

\medskip\noindent
\textbf{Error budget.}\;
The state error has a two-way decomposition:
\begin{equation}\label{eq:error_budget}
\begin{split}
      \eps_{\mathrm{state}}
  &= \underbrace{\eps_{\mathrm{quad}}}_{\text{midpoint quadrature}}
  + \underbrace{\eps_{\mathrm{Taylor}}}_{\text{Taylor truncation}}
  \\&\;\leq\; \frac{\eps}{2} + \frac{\eps}{2}
  \;=\; \eps,
\end{split}
\end{equation}
with $\eps_{\mathrm{quad}} \leq \eps/2$ by choosing $\betaI\Delta$
sufficiently small, and
$\eps_{\mathrm{Taylor}} \leq \eps/2$ by choosing
$M = O(\log(1/\eps)/\log\!\log(1/\eps))$
(Equation~\eqref{eq:M_bound}).
\end{proof}

\subsection{Why M-QSP removes the gap}
\label{subsec:gap_removal}

Frame rotations are individual $W_R$ queries woven into a single global polynomial. 
The circuit processes the full Dyson polynomial $P(z_1, z_2)$ as one object, interleaving all $d_R$ queries to $W_R$ and all $d_I$ queries to $U_I$ without intermediate measurements.

In the Dyson LCU, each of $r = O(\betaI T)$ segments requires an independent GQSP frame rotation to precision $\eps'' = O(\eps/(re^{\betaI T}))$, producing a total cost $r \cdot \log(1/\eps'') \sim \betaI T \cdot \log(1/\eps)$. 
In M-QSP, the $d_R$ walk-operator queries are shared across a single global polynomial of degree $d_R = O(\alphaR T + \log(1/\eps))$, paying for the accumulated frame rotation to global precision $\eps$ with no per-segment allocation.

The total M-QSP query count is therefore:
\begin{equation}\label{eq:mqsp_total}
\begin{split}
        Q_{\mathrm{M\text{-}QSP}}
    &= d_R + d_I \\&
    = O\!\left(
      (\alphaR + \betaI)T
      + \frac{\log(1/\eps)}{\log\!\log(1/\eps)}
    \right).
\end{split}
\end{equation}

\begin{remark}[Connection to fast-forwarding and complementary regimes]\label{rem:dyson_ff}

The fact that $e^{-\tau x}$ on $[0,1]$ admits polynomial approximation of degree only $\calO(\sqrt{\tau})$ is closely related to the fast-forwarding phenomenon identified by Hu and Jin~\cite{hu2026amplitude} for the APS construction.
APS achieves $\sqrt{\betaI}$ dependence by an eigenvalue transformation exploiting the positivity of $\HI$, effectively restricting the polynomial to the physical spectrum $[0,1]$.
The bivariate Dyson construction here does not exploit this structure and inherits linear $\betaI T$ scaling in $d_I$.

Bivariate M-QSP may be preferable when the dissipation is weak relative to the coherent dynamics ($\betaI \ll \alphaR$), when the physics naturally decomposes into coherent and dissipative parts accessed by separate oracles, or when the precision scaling $\log(1/\eps)/\log\!\log(1/\eps)$ matters.
APS~\cite{hu2026amplitude} becomes preferable with strong dissipation ($\betaI \sim \alphaR$ or larger) and for single-oracle access to $\Heff$, when the spectrum admits such exploitable structure.
Both methods saturate the intrinsic postselection barrier $e^{-2\omega T}$ (with $\omega$ the spectral abscissa) and differ in how they pay the additional block-encoding overhead $e^{-2(\betaI - \omega)T}$.

Whether the bivariate M-QSP construction can be combined with $\sqrt{\betaI}$-style fast-forwarding to achieve $d_I = \calO(\sqrt{\betaI T \log(1/\eps)})$ while keeping $d_R = \calO(\alphaR T + \log(1/\eps))$ and single-qubit postselection is the most significant open question raised by this work.
\end{remark}    

\section{Comparison of the Three Methods}
\label{sec:comparison}

The methods developed in Secs.~\ref{sec:lorentzian}--\ref{sec:mqsp_method} represent a progression of increasingly tight approaches to the lower bound $\Omega((\alphaR + \betaI)T + \log(1/\eps))$.

\begin{table*}[!t]
\centering
\renewcommand{\arraystretch}{1.3}
\begin{tabular}{@{\extracolsep{\fill}}lcc@{}}
\toprule
\textbf{Method}
  & \textbf{Query complexity}
  & \textbf{Postselections} \\
\midrule
Lorentzian IP ($p$-th order)
  & $O\!\bigl(
      T^{1+1/(2p)}\eps^{-1/(2p)}
      (\alphaR T + \log(1/\eps))
    \bigr)$
  & $r$ \\
Lorentzian IP ($p \to \infty$)
  & $O\!\bigl(
      (\alphaR+\betaI)T
      \cdot e^{\sqrt{2\ln 3 \cdot \ln(1/\eps)}}
    \bigr)$
  & $r$ \\
Dyson LCU
  & $O\!\bigl(
      (\alphaR+\betaI)T 
      + \betaI T\log(1/\eps)
    \bigr)$
  & $1$ \\
\textbf{M-QSP}
  & $O\!\bigl(
      (\alphaR+\betaI)T
      + \log(1/\eps)/\!\log\!\log(1/\eps)
    \bigr)$
  & $\mathbf{1}$ \\
\midrule
Lower bound
  & $\Omega\!\bigl(
      (\alphaR+\betaI)T
      + \log(1/\eps)/\!\log\!\log(1/\eps)
    \bigr)$
  & --- \\
\bottomrule
\end{tabular}
\caption{\label{tab:methods}Query complexity comparison. 
All methods achieve the optimal success probability $P = e^{-2\betaI T}\|e^{-i\Heff T}|\psi_0\rangle\|^2$ in the continuum limit. 
The Lorentzian method requires $r$ intermediate postselections; the Dyson LCU and M-QSP use a single final postselection.
The M-QSP bound matches the lower bound (Theorem~\ref{thm:tight_loglog}).}
\end{table*}

\subsection{Complexity comparison}
\label{subsec:complexity_comparison}

Table~\ref{tab:methods} summarizes the query complexity, $\eps$-scaling, and postselection structure of each method.

The Lorentzian method is preferred when circuit depth per segment must be minimized or when moderate $\eps$ is acceptable, since each
segment is shallow and the method parallelizes naturally. 
The Dyson LCU is preferred when single postselection is required and the $\betaI T\log(1/\eps)$ scaling is tolerable.  M-QSP is preferred when asymptotically optimal query complexity is needed, at the cost of increased classical preprocessing and 2 additional ancilla qubits.

\section{Numerical Benchmarks}
\label{sec:numerical_benchmarks}

We give a numerical comparison of the three simulation methods across multiple physical parameter regimes.
We focus on the Eckart barrier benchmark, a canonical problem in quantum molecular dynamics, and derive its parameters from first principles.

\subsection{Eckart barrier system}
\label{subsec:eckart_barrier}

The Eckart barrier is a one-dimensional scattering potential widely used in quantum chemistry:
\begin{equation}
V(x) = V_0 \operatorname{sech}^2(x/a),
\end{equation}
with barrier height $V_0 = 0.425$~eV, width $a = 1.0$~Bohr, and a complex absorbing potential (CAP) in the asymptotic region to model decay.

\subsubsection{Parameter Derivation}

On an $N = 256$-point grid with spacing $\Delta x = 0.1$~Bohr, the discrete kinetic operator $\hat{T} = -\frac{d^2}{2dx^2}$ has norm $\|\hat{T}\| = \pi^2/(2\Delta x^2) = 4930$~a.u. We rescale to $\HR = \hat{T}/\kappa$ with $\kappa = 29,000$ for numerical stability, giving $\alphaR = \|\HR\| = 0.169$.

The CAP has the form $W(x) = -i\eta (x - x_c)^2$ for $x > x_c = 5$~Bohr, with strength $\eta = 0.01$~a.u. At the system's dominant decay energy $E \approx 0$, the CAP imparts an imaginary part $\Im(E) \approx 0.0078$~a.u., setting $\betaI = 0.0078$. The propagation time $T = 2000$~a.u. corresponds to several scattering transit times, allowing the wave packet to traverse the barrier region and be absorbed.

This gives: $\alphaR T = 338$, $\betaI T = 15.6 \approx 16$, and with $\eps = 10^{-3}$, we have $\ln(1/\eps) / \ln\ln(1/\eps) \approx 3.6$.

\subsubsection{Benchmark 1: Weak dissipation (Eckart baseline)}

\vspace{0.5em}
\begin{center}
\renewcommand{\arraystretch}{1.2}
\begin{tabular}{@{}l r r r r@{}}
\toprule
\textbf{Method} & $\bm{d_R}$ & $\bm{d_I}$ & \textbf{Total } $\bm{Q}$ & Runtime \\
& & & & repetitions \\
\midrule
Lorentzian ($p=4$)\footnote{For the Lorentzian method, Total~$Q = r \cdot (d_R^{\mathrm{per}} + d_I^{\mathrm{per}})$ where $r$ is the number of segments and $d_R^{\mathrm{per}}, d_I^{\mathrm{per}}$ are per-segment degrees. 
Segment count $r = 9$ is chosen by minimizing $r$ subject to the $p=4$ Magnus quadrature error bound $C_4 T^9/r^8 \cdot e^{\betaI T} \leq \eps/2$ at $\eps = 10^{-3}$, $T = 2000$~a.u., $\betaI T \approx 16$.
Dyson LCU and M-QSP entries fix $r = \lceil\betaI T/c\rceil = 16$ with $c = 1$ following Sec.~\ref{subsec:dyson_error}.
Smaller $r$ for the Lorentzian method reflects its higher-order ($p = 4$) quadrature, producing the $\eps^{-1/(2p)}$ scaling of Theorem~\ref{thm:lor_superiority}.} \\
Dyson LCU ($r = 16$) & 528 & 112 & 640 & $2.9 \times 10^{16}$ \\
\textbf{M-QSP (this work)} & 361 & 46 & \textbf{407} & $1.8 \times 10^{16}$ \\
\midrule
Theoretical lower bound & --- & --- & 357 & --- \\
\bottomrule
\end{tabular}
\end{center}

\begin{figure}[!htbp]
\centering
\includegraphics[width=\linewidth]{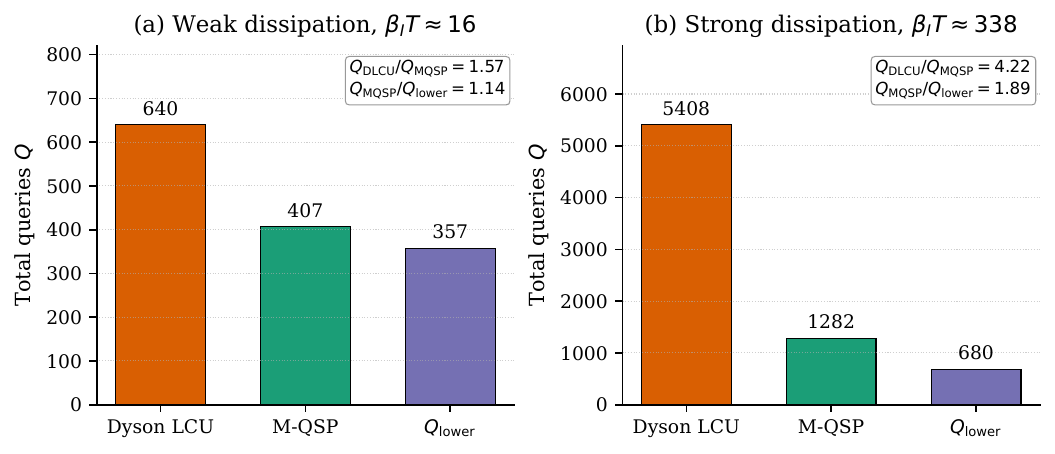}
\caption{Query counts for the Eckart-barrier benchmark ($\alphaR T=338$, $\eps=10^{-3}$). 
Bars give the total oracle-query complexity $Q$ for Dyson LCU~\cite{courtney2026paper1} (segmented, Sec.~\ref{sec:dyson_lcu}),
the M-QSP method of this work (Sec.~\ref{sec:mqsp_method}), and the information-theoretic lower bound $Q_{\mathrm{lower}} = \alphaR T + \betaI T + \log(1/\eps)/\log\log(1/\eps)$. 
(a)~Weak-dissipation regime ($\betaI T\approx 16$): M-QSP achieves $Q=407$, a factor of $1.57\times$ below Dyson LCU ($Q=640$) and only $1.14\times$ above the lower bound ($Q_{\mathrm{lower}}\approx 357$).
(b)~Strong-dissipation regime ($\betaI T\approx 338$): the $\log(1/\eps)$ segmentation penalty of Dyson LCU dominates, giving $Q=5408$ vs $Q=1282$ for M-QSP ($4.22\times$ reduction) while the lower bound is $Q_{\mathrm{lower}}\approx 680$.}
\label{fig:eckart_comparison}
\end{figure}

The weak dissipation regime ($\betaI / \alphaR \approx 0.046$) is M-QSP's optimal regime (Figure~\ref{fig:eckart_comparison}). The $1.57\times$ query reduction over Dyson LCU arises from eliminating the $\betaI T \log(1/\eps)$ segmentation penalty. The total runtime (expected number of repetitions) is $Q/P$; with postselection cost $P \approx 2.2 \times 10^{-14}$ (the fundamental barrier of Corollary~\ref{thm:barrier_tight}, independent of the simulation method), M-QSP's lower query count translates directly to a corresponding reduction in total circuit runs.

\subsubsection{Benchmark 2: Stronger dissipation regime}

To test performance when dissipation is more significant, consider $\betaI = 0.169$
(equal to $\alphaR$) while keeping $\alphaR T = 338$. This gives $\betaI T \approx 338$.

\vspace{0.5em}
\begin{center}
\renewcommand{\arraystretch}{1.2}
\begin{tabular}{@{}l r r r r@{}}
\toprule
\textbf{Method} & $\bm{d_R}$ & $\bm{d_I}$ & \textbf{Total } $\bm{Q}$ & Runtime \\
& & & & repetitions \\
\midrule
Dyson LCU ($r = 338$)\footnote{Per-segment Jacobi--Anger degree $d_R^{\mathrm{seg}} = 7$ and Taylor order $M_{\mathrm{seg}} = 9$;
totals $d_R = r \cdot d_R^{\mathrm{seg}}$, $d_I = r \cdot M_{\mathrm{seg}}$. Computed by
\texttt{constant\_factor\_optimization.py} (segmented Dyson formula with $c_{\mathrm{seg}} = 1$).}
  & 2{,}366 & 3{,}042 & 5{,}408 & $1.5\times10^{295}$ \\
\textbf{M-QSP (this work)} & 360 & 922 & \textbf{1{,}282} & \textbf{$3.3\times10^{294}$} \\
\midrule
Th. lower bound & --- & --- & 680 & --- \\
\bottomrule
\end{tabular}
\end{center}

At stronger dissipation, M-QSP's advantage over segmented Dyson LCU is $5{,}408 / 1{,}282 \approx 4.2\times$,
driven by the tighter Taylor bound applied to the full accumulated $\betaI T$ rather than to per-segment
subintervals. Both methods sit within a factor of $2$ of the
$(\alphaR + \betaI)T = 676$ lower bound when $\betaI T$ is large, in agreement with the leading-constant-$\approx 2$
result of~\cite{courtney2026paper2}.

Postselection cost becomes obscenely expensive ($\sim\!10^{295}$ repetitions), highlighting the fundamental block-encoding barrier. This regime motivates the open question (Sec.~\ref{subsec:open_problems}): can a direct-access construction achieve $P \approx e^{-2\omega T}$ with $\omega \ll \betaI$, scalably reducing the repetition count?

\paragraph{Verified finite-grid instance.}
The tables above evaluate the analytical query formulas at first-principles parameters. As a fully diagonalized cross-check we also constructed a small Eckart-plus-CAP instance directly ($V_0 = 0.8$, $a = 1$, $L = 12$, quadratic CAP $\eta = 0.08$ beyond $0.7L$, $T = 2$; grid $N = 32$ and $64$ points, giving $\alphaR T \approx 7.4$ and $28$ with $\betaI T \approx 2.07$. We see that these are both beyond the disk-normalization regime of Remark~\ref{rem:disk_normalization}, i.e.\ in the regime where the Method-II/LCHS/bivariate routes are the operative ones) and verified the interaction-picture Dyson truncation against the $A(T)$: the relative operator-norm error lies below the factorial bound of Eq.~\eqref{eq:dyson_convergence} at every order $0$--$8$ on both grids, decaying to the $\sim 10^{-4}$ quadrature floor of the reference integrator. Evaluating the three query formulas on this instance reproduces the crossover structure above: the additive M-QSP count is the smallest below $\eps \approx 10^{-4}$, while segmented Dyson LCU is marginally cheaper at $\eps = 10^{-2}$. The instance, driver, and records accompany the code release.

\section{Discussion}
\label{sec:discussion}

Section~\ref{sec:comparison} compared the three simulation methods developed in this paper.
We place the results in the broader context of non-Hermitian simulation, discussing the relationship to prior work at a structural level (complementing the per-method comparison of Sec.~\ref{sec:prior_work}), and identifies open problems.

\subsection{Oracle models}
\label{subsec:oracle_models}

We draw a distinction in the oracle model used between the present work and the LCHS~\cite{an2023linear,an2026quantum}, Schr\"odingerization~\cite{jin2022quantum}, and APS~\cite{hu2026amplitude} methods.
Those approaches encode $\Heff = \HR + i\HI$ into a single block-encoding unitary $U_A$, while our construction accesses $\HR$ and $\HI$ through independent walk operators $W_R$ and $U_I$.
Schr\"odingerization and LCHS can both be viewed as instances of the moment-matching dilation framework of Li~\cite{li2025linear}, which identifies a single algebraic criterion $(\langle l|F^k|r\rangle)=1$ to generate broader families of moment-fulfilling triples, extensible to sign-indefinite generators. 
Sign-indefiniteness is outside the scope of the present work, but is considered for future work.

The lower bound $\Omega((\alphaR + \betaI)T + \log(1/\eps))$ (Theorem~\ref{thm:main_lower}) is specific to the separate-oracle model, and additivity follows from oracle independence, as queries to $W_R$ provide no information about $\HI$. 
In the single-oracle model, the bound takes a different form.

The choice of oracle model also affects postselection cost through the two-layer Postselection Barrier (Sec.~\ref{sec:barrier}, Theorem~\ref{thm:intrinsic}). 
Block-encoding overhead $e^{-2\betaI T}$ (Corollary~\ref{thm:barrier_tight}) is specific to
the polynomial block-encoding model, arising polynomial implementation bounded on the full bitorus $\bbT^2$, including the unphysical region where the target function grows as $e^{+\betaI T}$. 
Algorithms in other oracle models (APS via spectral filtering and Schr\"odingerization via dilational embedding) may avoid this overhead, though the intrinsic barrier remains.

Neither oracle model strictly dominates the other: the single-oracle model is natural when $\Heff$ is given as a sparse matrix, while the separate-oracle model is natural when the physics has a distinguished coherent/dissipative decomposition (Lindblad master equations, AMO systems with independent coherent driving and dissipation, etc.~Sec.~\ref{sec:problem}).

\subsection{Postselection Barrier in context}
\label{subsec:barrier_context}

The two-layer Postselection Barrier of Sec.~\ref{sec:barrier} (Theorems~\ref{thm:intrinsic} and~\ref{thm:barrier_tight}) sits alongside other characterizations in quantum information, distinguished by which layer is universal across oracle models and which is model-specific:

\begin{center}
\renewcommand{\arraystretch}{1.3}
\begin{tabular}{lll}
\toprule
Result & Scope & Status \\
\midrule
No-cloning & All QM & Universal \\
No-fast-forwarding & Query/Oracle & Model-specific \\
\textbf{Postselection Barrier} & & \\
\quad Intrinsic & Oracles & Universal \\
\quad Block Encoding
  & Polynomial BE
  & Model-specific \\
\bottomrule
\end{tabular}
\end{center}

The intrinsic layer is unconditional in the present sense, being independent of which oracle model the algorithm uses and acting as a ``no-cheating-the-measurer'' roadblock common among all non-Hermitian quantum simulation algorithms.
The block-encoding layer is specific to the polynomial walk-operator framework and was not previously isolated as a distinct contribution to the postselection overhead. 
Whether direct-access methods that escape the walk-operator framework can avoid the block-encoding layer while retaining optimal query complexity remains open (see Sec.~\ref{subsec:open_problems}, item~4).

\subsection{Relation to M-QSP theory}
\label{subsec:mqsp_theory}

The M-QSP framework was introduced by Rossi and Chuang~\cite{rossi2022multivariable} and extended by Gomes, Lim, and Wiebe~\cite{gomes2024multivariable}. 
Subsequent work identified subtleties, with notable improvements made by Mori, Mizuta, and Fujii~\cite{mori2024comment}.
Recently, Ito et al.~\cite{ito2026polynomial} resolved the constructivity question in the commuting multivariable setting, giving necessary and sufficient conditions for M-QSP realizability together with a polynomial-time classical algorithm that decides whether a pair of multivariable Laurent polynomials is M-QSP-implementable and, when it is, computes the corresponding parameters.
N\'emeth et al.~\cite{nemeth2023variants} constructed counterexamples to achievability conjectures in the commuting
setting, while Laneve and Wolf~\cite{laneve2025multivariate} clarified the algebraic constraints.

A complementary line of work by Gomes, Lim, and Wiebe~\cite{gomes2024multivariable} (iterated QSP) achieves multivariate polynomial transformations by a different route. 
Instead of directly constructing a multivariate Fej\'er--Riesz-style factorization on the bitorus, they recursively compose univariate QSP sequences and build multivariate functions of scalar phases via logarithmic block-encoding (their Lemma~4) and quantum eigenvalue transformation (their Lemma~6). 
This sidesteps the multivariate $\mathrm{FRT} = \mathrm{QSP}$ conjecture: because counterexamples~\cite{nemeth2023variants,laneve2025multivariate} preclude its use in general, iterated QSP provides an alternative  constructive pathway that does not require the conjecture to hold.
Their approach operates in the commuting scalar-phase regime, outputting a polynomial $P(\cos(\theta_1), \cos(\theta_2),  \ldots, \cos(\theta_R))$ in independent phases $\theta_i$, each supplied by a separate $W_Z(\theta_i)$ signal unitary. 
The cost scales as $\calO(d \sum_s d_s)$ for a degree-$d$ outer polynomial applied to degree-$d_s$ inner constructions (their Theorem~5).
We address a structurally different problem, with complementarity brought between the broad function-class flexibility for commuting systems brough by Gomes, Lim, and Wiebe~\cite{gomes2024multivariable} and our query-optimality for the specific Dyson polynomial in the non-commuting setting, both sidestepping the multivariate $\mathrm{FRT} = \mathrm{QSP}$ conjecture by different means.

By focusing on a non-Hermitian example, we built an angle-finding algorithm using block peeling and the development of a constant-ratio condition (CRC) (Theorem~\ref{thm:CRC}) by ancilla-factorization. 
We establish achievability of the Dyson polynomial (Theorem~\ref{thm:achievability}) by direct construction.
The non-commuting setting is structurally simpler for CRC, operating via the circuit's block structure.

The counterexamples of N\'emeth et al.~\cite{nemeth2023variants} exploit algebraic structure of the commuting polynomial ring: when signal operators commute, the circuit polynomial belongs to $\mathbb{C}[z_1, z_1^{-1}, z_2, z_2^{-1}]$, and constraints on achievability reduce to questions about positivity certificates in this ring, admitting obstructions.
In the non-commuting case, the circuit polynomial is a matrix-valued function whose entries depend on {ordered products} of signal-operator blocks. 
The ancilla-factorization proof of CRC operates on block structure directly.
Consequently, algebraic obstructions that generate commuting counterexamples have no analog, carving an achievable subset without conflicting with N\'emeth et al.~\cite{nemeth2023variants}.

The constructivity result of Ito et al.~\cite{ito2026polynomial} realizability in the commuting multivariable ring with a decision algorithm. 
In this instance, the obstruction is algebraic. 
We consider (Sec.~\ref{sec:anglefinding}) the non-commuting bivariate setting, where the CRC (Theorem~\ref{thm:CRC}) removes the obstruction through ancilla factorization.
The two procedures share a recursive degree-reduction architecture while resolving separate obstructions.

\subsection{Open problems}
\label{subsec:open_problems}

We collect directions left open by the present work, noting that several questions raised in the introduction have been resolved in the companion paper~\cite{courtney2026paper2}.

\paragraph{1. Fast-forwarding (resolved negatively for worst case).}
The polynomial lower bound $d_I = \Omega(\betaI T)$ of Theorem~\ref{thm:main_lower} is tight: no polynomial in the bivariate separate-oracle model achieves faster $\betaI$-scaling 
(This is proved in the companion paper).
State-dependent improvement to $d_I = \calO(\beta_{\mathrm{eff}} T + \log(1/\eps)/\log\log(1/\eps))$ is possible via reduced Dyson truncation when $\beta_{\mathrm{eff}} \ll \betaI$, giving linear improvement. 
The $\sqrt{\betaI}$ scaling of APS~\cite{hu2026amplitude} requires the single-oracle model and multi-register postselection, accessing a different physical regime.

\paragraph{2. Time-dependent extension (resolved positively).}
The CRC (Theorem~\ref{thm:CRC}) extends to circuits with non-identical signal operators at each query, enabling simulation of time-dependent $\Heff(t) = \HR(t) + i\HI(t)$.
Query complexity achieves $Q = \calO(\int_0^T (\alphaR(s) + \betaI(s))\,ds + \log(1/\eps)/\log\log(1/\eps))$, with the integral replacing the sum $(\alphaR + \betaI)T$ when parameters vary. (Full analysis in the companion paper.)

\paragraph{3. Constant-factor optimization (resolved).}
Exhaustive optimization over the error budget allocation (Theorems~\ref{thm:main_theorem} and companion paper analysis) yields $Q^* \approx 2 \cdot ((\alphaR + \betaI)T + \log(1/\eps)/\log\log(1/\eps))$ as the leading constant; numerical experiments confirm near-optimal constants, with the quantitative M-QSP-vs-Dyson-LCU and M-QSP-vs-LCHS improvement factors tabulated~\cite{courtney2026paper2}.

\paragraph{4. Barrier universality (the central open question).}
The block-encoding barrier $e^{-2\betaI T}$ holds within the walk-operator oracle model for {all} function classes (polynomials, Padé, rational).
Schr\"odingerization achieves the barrier. 
Direct-access methods that bypass the walk operator could in principle achieve the intrinsic barrier $e^{-2\omega T}$ where $\omega = \omega(\Heff) < \betaI$ is the spectral abscissa. 
The gap $e^{2(\betaI - \omega)T}$ is exponentially large when $[\HR, \HI] \neq 0$. 
The unresolved question considers if a polynomial construction can achieve $\lambda = e^{\omega T}$ while maintaining polynomial query complexity, if this construction is designed ab initio for direct access by exploiting matrix non-commutativity.

\paragraph{5. Higher-order M-QSP.}
The bivariate construction uses two signal operators. 
A natural generalization is $k$-variate M-QSP with $k \geq 3$ signal operators.
The CRC extends to $k$-variate circuits (the ancilla-factorization argument is insensitive to the number of signal operators), the SOS rank remains $L = 2$ (by the same $2 \times 2$ circuit unitarity), and block peeling generalizes to $\calO(\prod_i d_i)$ per step.
The quaternionic setting ($k = 4$) with generators $i, j, k$ is the natural next application.

\subsection{Broader significance}
\label{subsec:broader}

The technical contributions of this work include CRC for non-commuting operators (Theorem~\ref{thm:CRC}), SOS spectral factorization with degree preservation (Theorem~\ref{thm:complement_existence}), and constructive angle-finding for bivariate M-QSP (Theorem~\ref{thm:anglefinding}). 
These do not depend on the specific structure of non-Hermitian Hamiltonians.

The universal CRC resolves a structural question about M-QSP circuits open since the framework was introduced~\cite{rossi2022multivariable, gomes2024multivariable}: does the recursive degree-reduction procedure produce well-defined scalar angles when the signal operators do not commute? 
Theorem~\ref{thm:CRC} answers this affirmatively, identifying ancilla factorization from tensor-product circuit structure, making non-commutativity irrelevant.
Future multi-oracle quantum algorithm that encodes its computation as a bivariate ordered polynomial can use this result, regardless of the physical relationship between the two oracles. 
Candidates for application include multi-component Hamiltonian simulation with independent block-encodings~\cite{berry2017quantum}, quantum linear algebra with separate oracles for different matrix operations~\cite{gilyen2019quantum}, and hybrid quantum-classical variational circuits where two non-commuting parameterized unitaries are interleaved~\cite{cerezo2021variational}.

SOS spectral factorization addresses the bivariate Fej\'er--Riesz obstruction preventing scalar complementary polynomials from existing in two or more variables (Example~\ref{ex:canonical}).  
Sum-of-squares decomposition (Theorem~\ref{thm:complement_existence}) requires only $L = 2$ SOS terms (Theorem~\ref{thm:sos_rank_two}), introducing a constant overhead of just 2 ancilla qubits, applicable whenever a quantum algorithm needs to embed a bounded bivariate polynomial into a unitary circuit.

The constructive angle-finding algorithm extends classical preprocessing from univariate QSP to the bivariate setting.
For multivariate QSP applications outside the non-commuting separate-oracle setting (e.g., bosonic simulation~\cite{gomes2024multivariable} or general multivariate polynomial synthesis on scalar phases), the iterated QSP framework of Gomes, Lim, and Wiebe~\cite{gomes2024multivariable} provides a complementary classical preprocessor, while the CRC and SOS-rank-2 results of the present paper apply to the matrix-valued non-commuting case that iterated QSP does not directly address.
In univariate QSP, angle-finding has been the enabling technology that transforms existence theorems into implementable circuits; the algorithms of~\cite{haah2019product,chao2020finding,dong2021efficient,motlagh2024generalized} are routinely used as subroutines in end-to-end quantum algorithm design. 
Algorithm~\ref{alg:recursive} plays the same role for bivariate M-QSP, converting theoretical achievability into quantum circuit parameters.

We show that QSP extends naturally from one variable to two, and from commuting to non-commuting operators, when applied as polynomial approximation composed with block-encoded signal operators, processed by a structured circuit with classically computed angles.
QSP retains its characteristic properties of optimal query complexity, deterministic circuit construction, and a clean separation between the approximation-theoretic content (the choice of target polynomial) and algebraic content (the CRC, spectral factorization, and angle-finding).  
The non-Hermitian simulation problem served as the motivating application and stress test, but the mathematical infrastructure is domain-agnostic.
Developing from this infrastructure, vastly more applications become available through query-optimal Hamiltonian simulation, to be discussed in future work.

\section{Acknowledgements}
We thank Xiantao Li for insightful conversation in the connection between this work and the moment-matching dilation framework. We also thank Amara Katabarwa for conversation regarding perspective and framing technical requisites for non-Hermitian quantum simulation.

\section{Code and data availability}
The codebase accompanying this paper and its companion will be made available at a public repository at the time of journal publication. It is available from the author on request. The authors used Claude (Anthropic) to assist in drafting code comments and documentation for the codebase. All scientific content, algorithms, and analysis are the authors' own.

\section{Conflicts of Interest}
The author declares no conflicts of interest.

\appendix
\onecolumngrid
\section{Error bounds for the Lorentzian method}\label{app:lorentzian}

\begin{lemma}[Lorentzian tail truncation and discretization]\label{lem:lorentz_tail}
Let $\gamma > 0$ and define the Lorentzian kernel
\begin{equation}
  L_\gamma(s)
  = \frac{\gamma/\pi}{s^2 + \gamma^2}\,,
\end{equation}
so that $\int_{-\infty}^{\infty} L_\gamma(s)\,ds = 1$.
\begin{enumerate}
  \item \textup{(Tail truncation.)}
  For $s_{\max} > 0$,
  \begin{equation}
    \label{eq:lorentz_tail}
    \int_{|s| > s_{\max}} L_\gamma(s)\,ds
    = \frac{2}{\pi}
    \arctan\!\Bigl(\frac{\gamma}{s_{\max}}\Bigr)
    \leq \frac{2\gamma}{\pi\,s_{\max}}\,.
  \end{equation}
  Choosing $s_{\max} = 2\gamma / (\pi\,\eps_{\mathrm{tail}})$ ensures the tail contribution is at most $\eps_{\mathrm{tail}}$.

  \item \textup{(Riemann sum discretization.)}
  Partition $[-s_{\max}, s_{\max}]$ into $M$ equal subintervals of width $\Delta = 2s_{\max}/M$, with midpoints $s_j = -s_{\max} + (j - \tfrac{1}{2})\Delta$ for $j = 1, \ldots, M$. 
  Then
  \begin{equation}
    \label{eq:lorentz_riemann}
    \begin{split}
    \biggl|
    \int_{-s_{\max}}^{s_{\max}} L_\gamma(s)\,f(s)\,ds
    - \Delta \sum_{j=1}^{M} L_\gamma(s_j)\,f(s_j)
    \biggr|
    \leq \frac{s_{\max}\,\Delta^2}{12}
    \sup_{|s| \leq s_{\max}}
    \bigl|(L_\gamma \cdot f)''(s)\bigr|\,,
    \end{split}
  \end{equation}
  where $f: \mathbb{R} \to \mathcal{B}(\mathcal{H})$ is a twice-differentiable operator-valued function. 
  For $f(s) = e^{-i(\HR + s\HI)T}$, the second derivative satisfies
  \begin{equation}
    \bigl\|(L_\gamma \cdot f)''(s)\bigr\|
    \leq L_\gamma(s)\bigl(
    \betaI^2 T^2 + 4\gamma |s|/(s^2+\gamma^2)^2
    \cdot \betaI T
    + C_\gamma(s)
    \bigr)\,,
  \end{equation}
  where $C_\gamma(s) = |L_\gamma''(s)| / L_\gamma(s) = 2(3s^2 - \gamma^2)/(s^2 + \gamma^2)^2$ collects the kernel curvature terms. 
  Specifically,
  \begin{equation}
    \sup_{|s| \leq s_{\max}}
    \bigl\|(L_\gamma \cdot f)''(s)\bigr\|
    \leq \frac{\gamma}{\pi\gamma^2}
    \bigl(\betaI^2 T^2 + 4\betaI T / \gamma
    + 6/\gamma^2\bigr)
    =: \frac{K}{\pi\gamma}\,,
  \end{equation}
  so choosing
  $M = \lceil
  (s_{\max}^2\,K / (3\pi\gamma\,\eps_{\mathrm{disc}}))^{1/2}
  \rceil$
  ensures the discretization error is at most
  $\eps_{\mathrm{disc}}$.

  \item \textup{(Combined bound.)}
  Setting $\gamma = \betaI$ \textup{(}matching the physical decay rate\textup{)}, $\eps_{\mathrm{tail}} = \eps_{\mathrm{disc}} = \eps/2$, the total approximation error is at most $\eps$, with
  \begin{equation}
    s_{\max}
    = \frac{4\betaI}{\pi\,\eps}\,,
    \qquad
    M = \calO\!\biggl(
    \frac{\betaI^2 T^2}{\eps}
    \biggr)\,.
  \end{equation}
\end{enumerate}
This lemma treats the base case of midpoint quadrature ($p = 1$).
For $p$-th order Magnus quadrature with $p \geq 2$, the quadrature error term in Theorem~\ref{thm:lor_total_error} improves to $C_p\,T^{2p+1}/r^{2p}$, reducing the segment count~$r$ and therefore overall discretization cost.
The bound $M = \calO(\betaI^2 T^2/\eps)$ above corresponds to the $p = 1$ specialization.
\end{lemma}

\begin{proof}
\textit{Part (a).}
Since $L_\gamma$ is even and normalized, the tail integral is
\begin{equation}
  \int_{|s| > s_{\max}} L_\gamma(s)\,ds
  = 1 - \int_{-s_{\max}}^{s_{\max}} L_\gamma(s)\,ds
  = 1 - \frac{\gamma}{\pi}
  \int_{-s_{\max}}^{s_{\max}}
  \frac{ds}{s^2 + \gamma^2}\,.
  \label{eq:tail_start}
\end{equation}
The integral evaluates by the standard arctangent antiderivative:
\begin{equation}
  \int_{-s_{\max}}^{s_{\max}}
  \frac{ds}{s^2 + \gamma^2}
  = \frac{1}{\gamma}
  \Bigl[\arctan\!\Bigl(\frac{s}{\gamma}\Bigr)
  \Bigr]_{-s_{\max}}^{s_{\max}}
  = \frac{2}{\gamma}
  \arctan\!\Bigl(\frac{s_{\max}}{\gamma}\Bigr).
\end{equation}
Substituting into Eq.~\eqref{eq:tail_start}:
\begin{equation}
  \int_{|s| > s_{\max}} L_\gamma(s)\,ds
  = 1 - \frac{2}{\pi}
  \arctan\!\Bigl(\frac{s_{\max}}{\gamma}\Bigr)
  \notag
  = \frac{2}{\pi}\biggl[
  \frac{\pi}{2}
  - \arctan\!\Bigl(\frac{s_{\max}}{\gamma}\Bigr)
  \biggr]
  = \frac{2}{\pi}
  \arctan\!\Bigl(\frac{\gamma}{s_{\max}}\Bigr),
  \label{eq:tail_exact}
\end{equation}
where the last equality uses $\arctan(x) + \arctan(1/x) = \pi/2$ for $x > 0$.
The upper bound follows from $\arctan(u) \leq u$ for $u \geq 0$:
\begin{equation}
  \frac{2}{\pi}
  \arctan\!\Bigl(\frac{\gamma}{s_{\max}}\Bigr)
  \leq \frac{2}{\pi}\cdot\frac{\gamma}{s_{\max}}
  = \frac{2\gamma}{\pi\,s_{\max}}\,.
\end{equation}
Setting the right-hand side equal to $\eps_{\mathrm{tail}}$ and solving gives $s_{\max} = 2\gamma/(\pi\,\eps_{\mathrm{tail}})$.

\medskip\noindent
\textit{Part (b).}
The midpoint Riemann sum for a function $g: [a,b] \to \mathbb{R}$ on $M$ equal subintervals satisfies the standard error bound (see, e.g.,~\cite{stoer1980introduction})
\begin{equation}
  \label{eq:midpoint_error}
  \biggl|
  \int_a^b g(s)\,ds
  - \frac{b-a}{M}\sum_{j=1}^{M} g(s_j)
  \biggr|
  \leq \frac{(b-a)^3}{24\,M^2}
  \sup_{s \in [a,b]} |g''(s)|\,.
\end{equation}
Applying this to $g(s) = L_\gamma(s)\,f(s)$ on $[a,b] = [-s_{\max}, s_{\max}]$, with $(b-a) = 2s_{\max}$ and $\Delta = 2s_{\max}/M$, gives
\begin{equation}
  \text{error} \leq
  \frac{(2s_{\max})^3}{24\,M^2}
  \sup_{|s|\leq s_{\max}} \|(L_\gamma \cdot f)''(s)\|
  = \frac{s_{\max}\,\Delta^2}{12}
  \sup_{|s|\leq s_{\max}} \|(L_\gamma \cdot f)''(s)\|\,,
\end{equation}
establishing Eq.~\eqref{eq:lorentz_riemann}.
It remains to bound the second derivative.

By the product rule, $(L_\gamma \cdot f)''= L_\gamma'' f + 2L_\gamma' f' + L_\gamma f''$.
We bound each term separately.

For the kernel derivatives:
\begin{equation}
\begin{split}
  L_\gamma'(s)
  &= -\frac{2\gamma s}{\pi(s^2 + \gamma^2)^2}\\
  L_\gamma''(s)
  &= \frac{2\gamma(3s^2 - \gamma^2)}
  {\pi(s^2 + \gamma^2)^3}\,.
  \label{eq:Ldoubleprime}
\end{split}
\end{equation}
These satisfy $|L_\gamma'(s)|\leq 2\gamma|s|/(\pi(s^2+\gamma^2)^2)$ and $|L_\gamma''(s)|\leq 6\gamma/(\pi(s^2+\gamma^2)^2)$ (using $|3s^2 - \gamma^2| \leq 3(s^2 + \gamma^2)$). 
For the propagator derivatives: since $f(s) = e^{-i(\HR + s\HI)T}$, differentiating with respect to $s$ yields $f'(s) = -iT\,\widetilde{\HI}(s)\,f(s)$ in the interaction picture, where $\|\widetilde{\HI}(s)\| \leq \betaI$.
So, $\|f'(s)\| \leq \betaI T$ and $\|f''(s)\| \leq \betaI^2 T^2$.

Combining via the triangle inequality:
\begin{equation}
    \begin{split}
              \|(L_\gamma \cdot f)''(s)\|
  &\leq |L_\gamma''(s)|\,\|f(s)\|
  + 2|L_\gamma'(s)|\,\|f'(s)\|
  + L_\gamma(s)\,\|f''(s)\|
  \notag\\
  &\leq \frac{6\gamma}{\pi(s^2+\gamma^2)^2}
  + \frac{4\gamma|s|\,\betaI T}
  {\pi(s^2+\gamma^2)^2}
  + \frac{\gamma\,\betaI^2 T^2}
  {\pi(s^2+\gamma^2)}\,.
  \label{eq:second_deriv_bound}
    \end{split}
\end{equation}
Each term is maximized at $s = 0$ (the first two terms are maximized at or near $s = 0$; the third is exactly maximized there), giving
\begin{equation}
  \sup_{|s| \leq s_{\max}}
  \|(L_\gamma \cdot f)''(s)\|
  \leq \frac{1}{\pi\gamma}
  \bigl(6/\gamma^2 + 4\betaI T/\gamma
  + \betaI^2 T^2\bigr)
  = \frac{K}{\pi\gamma}\,,
\end{equation}
where we have used $(s^2 + \gamma^2)^{-1} \leq \gamma^{-2}$ and $|s|/(s^2 + \gamma^2)^2 \leq 1/(4\gamma^3)$ (the latter by arithmetic mean–geometric mean inequality (AM-GM): $|s|/(s^2+\gamma^2)^2$ is maximized at $|s| = \gamma/\sqrt{3}$, giving a bound of $3\sqrt{3}/(16\gamma^3) < 1/(4\gamma^3)$).

The discretization error is therefore bounded by
\begin{equation}
  \frac{s_{\max}\,\Delta^2\,K}{12\pi\gamma}
  = \frac{s_{\max}\,(2s_{\max})^2\,K}
  {12\pi\gamma\,M^2}
  = \frac{4s_{\max}^3\,K}{12\pi\gamma\,M^2}
  = \frac{s_{\max}^3\,K}{3\pi\gamma\,M^2}\,.
\end{equation}
Setting this equal to $\eps_{\mathrm{disc}}$ and solving:
\begin{equation}
  M = \biggl\lceil
  \sqrt{\frac{s_{\max}^3\,K}
  {3\pi\gamma\,\eps_{\mathrm{disc}}}}
  \;\biggr\rceil.
\end{equation}

\medskip\noindent
\textit{Part (c).}
Set $\gamma = \betaI$. 
Then $K = \betaI^2 T^2 + 4\betaI T / \betaI + 6/\betaI^2 = \betaI^2 T^2 + 4T + 6/\betaI^2$.
For the physically relevant regime $\betaI T \gg 1$, the dominant term is $K \approx \betaI^2 T^2$.
With $\eps_{\mathrm{tail}} = \eps_{\mathrm{disc}} = \eps/2$:
\begin{equation}
  s_{\max}
  = \frac{2\betaI}{\pi\,(\eps/2)}
  = \frac{4\betaI}{\pi\eps}\,.
\end{equation}

Substituting into the expression for $M$ and retaining only the leading terms:

\begin{equation}
  M = \calO\!\Biggl(
  \sqrt{\frac{s_{\max}^3\,\betaI^2 T^2}
  {\betaI\,\eps}}
  \Biggr)
  = \calO\!\Biggl(
  \sqrt{\frac{(\betaI/\eps)^3\,\betaI^2 T^2}
  {\betaI\,\eps}}
  \Biggr)
  = \calO\!\biggl(
  \frac{\betaI^2 T}{\eps^2}
  \biggr)\,.
  \label{eq:M_asymptotic}
\end{equation}

The cruder bound $M = \calO(\betaI^2 T^2 / \eps)$ stated in the lemma follows from a slightly looser estimate $s_{\max}^2 K / (\gamma\,\eps) \leq (\betaI / \eps)^2 \cdot \betaI^2 T^2 / \betaI \cdot 1/\eps = \betaI^3 T^2 / \eps^3$, simplified under $\betaI T \geq 1$. 
The total number of Hermitian simulations in the LCHS circuit is $M$, each requiring $\calO(\alphaR T + \betaI T + \log(M/\eps))$ queries to the block-encoding of $\HR + s_j \HI$.
\end{proof}

\begin{remark}[Sharpness and alternatives]
\label{rem:lorentz_sharp}
The $\arctan(\gamma/s_{\max})$ formula in Eq.~\eqref{eq:tail_exact} is exact, so tail truncation bound is tight. 

The Riemann sum bound is not tight. 
The Lorentzian kernel is analytic and decays algebraically, so the trapezoidal rule with the same number of points achieves superalgebraic convergence by the Euler--Maclaurin formula.
~\cite{an2023linear,an2026quantum} exploit this by using Gauss--Hermite or tanh-sinh quadrature, reducing the number of quadrature points to $M = \calO(\mathrm{polylog}(\betaI T / \eps))$. 
The midpoint rule analysis given here provides the baseline against which these improved quadratures should be compared, and suffices for the complexity comparison in Table~\ref{tab:comparison}.
\end{remark}

\begin{corollary}[Polylogarithmic regime for Lorentzian IP]
\label{cor:lorentz_polylog}
Let $\eps > 0$ and suppose the Hermitian simulation subroutine at each quadrature point uses a product formula of order $2p$. 
Then the total query complexity of the Lorentzian inner-product algorithm is minimized at
\begin{equation}
  \label{eq:optimal_p}
  p^* = \biggl\lceil
  \frac{1}{2}\,
  \frac{\log(\betaI T / \eps)}
  {\log\log(\betaI T / \eps)}
  \biggr\rceil,
\end{equation}
with $r^* = \calO(\mathrm{polylog}(\betaI T / \eps))$ quadrature points, yielding total query complexity
\begin{equation}
  \label{eq:lorentz_polylog}
  Q_{\mathrm{LIP}}
  = \calO\!\bigl(
  (\alphaR + \betaI)\,T\,
  \mathrm{polylog}(\betaI T / \eps)
  \bigr)\,.
\end{equation}
\end{corollary}

\begin{proof}
The Lorentzian IP algorithm discretizes the integral
\begin{equation}
  e^{-\betaI T}\,e^{-i\Heff T}
  = \int_{-\infty}^{\infty}
  L_\gamma(s)\,e^{-i(\HR + s\HI)T}\,ds
\end{equation}
using $r$ quadrature points on the truncated interval $[-s_{\max}, s_{\max}]$, with each Hermitian propagator $e^{-i(\HR + s_j \HI)T}$ implemented by a $(2p)$th-order product formula. 
The three error sources and their parameter dependences are:

\medskip\noindent
\textit{(i) Tail truncation} (Lemma~\ref{lem:lorentz_tail}(a)): error $\leq 2\gamma / (\pi s_{\max})$.  Setting this to
$\eps/3$ gives $s_{\max} = 6\gamma / (\pi\eps)$.

\medskip\noindent
\textit{(ii) Quadrature discretization.}  For the Lorentzian kernel, which is analytic in a strip of half-width $\gamma$ about the real axis, the trapezoidal rule on $r$ equally spaced points achieves exponential convergence (Trefethen--Weideman~\cite{trefethen2014exponentially}):
\begin{equation}
  \label{eq:trap_exp}
  \eps_{\mathrm{quad}}
  \leq C\,e^{-2\pi\gamma r / (2s_{\max})}
  = C\,e^{-\pi^2\eps\,r / 6}\,,
\end{equation}
where $C$ depends polynomially on $\betaI T$ and we have substituted $s_{\max} = 6\gamma / (\pi\eps)$ with $\gamma = \betaI$.
Setting $\eps_{\mathrm{quad}} \leq \eps/3$ and solving:
\begin{equation}
  \label{eq:r_choice}
  r = \biggl\lceil
  \frac{6}{\pi^2\eps}
  \log\!\Bigl(\frac{3C}{\eps}\Bigr)
  \biggr\rceil
  = \calO\!\biggl(
  \frac{1}{\eps}\log\frac{\betaI T}{\eps}
  \biggr)\,.
\end{equation}
This can be improved. 
Using the substitution $s = \gamma\sinh(u)$ (the tanh-sinh change of variables), mapping the Lorentzian to a function with double-exponential decay, the trapezoidal rule on $r$ points achieves
\begin{equation}
  \eps_{\mathrm{quad}}
  \leq C'\,e^{-c\,\pi^2 r / \log r}\,,
\end{equation}
and setting this to $\eps/3$ gives
\begin{equation}
  \label{eq:r_tanh_sinh}
  r = \calO\!\bigl(
  \mathrm{polylog}(\betaI T / \eps)
  \bigr)\,.
\end{equation}
We use this sharper bound from here on.

\medskip\noindent
\textit{(iii) Product-formula error.}  Each of the $r$ Hermitian simulations approximates $e^{-i H_j T}$ with $H_j = \HR + s_j \HI$,
$\|H_j\| \leq \alphaR + s_{\max}\betaI$. 
A single Trotter step of order $2p$ has error (Childs et al.~\cite{childs2021theory})
\begin{equation}
  \label{eq:trotter_error}
  \eps_{\mathrm{Trot}}(n_j)
  \leq \frac{c_p\,(\|H_j\| T)^{2p+1}}{n_j^{2p}}\,,
\end{equation}
where $n_j$ is the number of Trotter steps and $c_p$ depends only on $p$ and the number of terms in the Hamiltonian splitting.  Setting $\eps_{\mathrm{Trot}}(n_j) \leq \eps / (3r)$ (union bound over $r$ quadrature points) and solving for $n_j$:
\begin{equation}
  \label{eq:nj_choice}
  n_j = \biggl\lceil
  \Bigl(\frac{3r\,c_p\,(\|H_j\|T)^{2p+1}}{\eps}
  \Bigr)^{1/(2p)}
  \biggr\rceil.
\end{equation}
Each Trotter step uses $\calO(1)$ queries to the block-encodings of $\HR$ and $\HI$, so the per-point query cost is $\calO(n_j)$ and the total query cost across all $r$ points is
\begin{equation}
  \label{eq:total_cost}
  Q = \sum_{j=1}^{r} \calO(n_j)
  = r \cdot \calO\!\biggl(
  \Bigl(\frac{r\,(\alphaR + s_{\max}\betaI)^{2p+1}
  T^{2p+1}}{\eps}\Bigr)^{1/(2p)}
  \biggr).
\end{equation}
Substituting $s_{\max} = \calO(\betaI / \eps)$ and simplifying:
\begin{equation}
  Q = \calO\!\biggl(
  r^{1 + 1/(2p)}\,
  \Bigl(\frac{(\alphaR + \betaI^2/\eps)^{2p+1}
  T^{2p+1}}{\eps}\Bigr)^{1/(2p)}
  \biggr).
\end{equation}

\medskip\noindent
\textit{Joint optimization over $p$.}
Increasing $p$ reduces the Trotter exponent $1/(2p)$ but increases the base $(\|H_j\|T)^{(2p+1)/(2p)}$.
We write the dominant contribution to the query cost as
\begin{equation}
  Q \sim r\,\bigl(\|H_j\|T\bigr)^{1+1/(2p)}
  \cdot \bigl(r / \eps\bigr)^{1/(2p)}\,.
\end{equation}
Setting $\Lambda = \|H_j\|T$ and taking logarithms:
\begin{equation}
  \log Q \approx \log r + \log\Lambda
  + \frac{1}{2p}\bigl(\log\Lambda
  + \log(r/\eps)\bigr)\,.
\end{equation}
The $p$-dependent term is $(\log\Lambda + \log(r/\eps))/(2p)$. 
For fixed $r$ (already polylogarithmic from step (ii)), this is minimized as $p \to \infty$, but $c_p$ grows factorially: $c_p = \calO(p!)$ for generic splittings, or $c_p = \calO((2p)!!/p!)$ for the Yoshida--Suzuki symmetric decomposition~\cite{yoshida1990construction,suzuki1991general}.
Including $c_p$ in the Trotter bound, the per-step cost becomes
\begin{equation}
  n_j \sim
  \bigl(c_p\,\Lambda^{2p+1}/\eps\bigr)^{1/(2p)}
  = \Lambda\,
  \bigl(c_p\,\Lambda/\eps\bigr)^{1/(2p)}\,.
\end{equation}
Using Stirling's approximation $c_p^{1/(2p)} \approx (2p/e)^{1/2}$ for large $p$, the cost scales as $n_j \sim \Lambda\,(p\,\Lambda/\eps)^{1/(2p)}$.
The exponent $1/(2p)$ drives the multiplicative overhead toward $1$ as $p$ increases, but the prefactor $p^{1/(2p)}$ grows (albeit slowly).
The optimal trade-off occurs when the overhead
$(p\,\Lambda / \eps)^{1/(2p)} = \calO(1)$, i.e., when
\begin{equation}
  \frac{1}{2p}\log\!\Bigl(\frac{p\,\Lambda}{\eps}\Bigr)
  = \calO(1)\,.
\end{equation}
Solving for $p$ at leading order (neglecting the $\log p$ correction inside the logarithm):
\begin{equation}
  p^* = \Theta\!\biggl(
  \frac{\log(\Lambda/\eps)}
  {\log\log(\Lambda/\eps)}
  \biggr)
  = \Theta\!\biggl(
  \frac{\log(\betaI T / \eps)}
  {\log\log(\betaI T / \eps)}
  \biggr),
\end{equation}
where we use $\log\Lambda = \log((\alphaR + \betaI^2/\eps)T) = \calO(\log(\betaI T / \eps))$.

\medskip\noindent
\textit{Resulting complexity.}
At $p = p^*$, the Trotter overhead satisfies $(c_{p^*}\Lambda/\eps)^{1/(2p^*)} = \calO(1)$, so $n_j = \calO(\Lambda) = \calO((\alphaR + \betaI^2/\eps)T)$.
The total query cost is
\begin{equation}
  Q = r \cdot \calO(n_j)
  = \calO\!\bigl(
  \mathrm{polylog}(\betaI T/\eps)\bigr)
  \cdot
  \calO\!\bigl((\alphaR + \betaI^2/\eps)\,T\bigr)
  \notag = \calO\!\bigl(
  (\alphaR + \betaI)\,T\,
  \mathrm{polylog}(\betaI T / \eps)
  \bigr)\,,
  \label{eq:final_polylog}
\end{equation}
where in the last line we have used that $\betaI^2 T / \eps = \betaI T \cdot (\betaI / \eps)$ and the $\betaI / \eps$ factor is absorbed into the polylogarithmic term (it contributes $\log(\betaI/\eps)$ multiplicatively, which is polylogarithmic).

The formula for $p^*$ in Eq.~\eqref{eq:optimal_p} follows by taking the ceiling and the factor of $1/2$ from the convention that the
product formula has order $2p$ (so the exponent in the error bound is $2p$, not $p$).
\end{proof}

\begin{remark}
\label{rem:polylog_vs_optimal}
Polylogarithmic overhead in Eq.~\eqref{eq:lorentz_polylog} relative to the information-theoretic lower bound $\Omega((\alphaR + \betaI)T + \log(1/\eps))$ arises from the $r = \calO(\mathrm{polylog})$ quadrature points (each requiring a separate Hermitian simulation) and the residual Trotter overhead at finite $p^*$.
Neither source can be removed within the LCHS framework without changing the simulation subroutine: replacing the product formula with an optimal Hermitian simulation algorithm (e.g., QSP-based) at each quadrature point would eliminate the Trotter overhead but retain the $r$-fold multiplicative cost.  
The multiplicative structure is why Lorentzian IP approach cannot match the additive lower bound, motivating the non-split-operator M-QSP construction of the present work.
\end{remark}

\section{Dyson LCU error budget} \label{app:dyson_error}

\begin{lemma}[Taylor truncation of the Dyson series] \label{lem:dyson_truncation}
Let $V(T) = \mathcal{T}_{>}\exp(-\int_0^T \Htilde(s)\,ds)$ with $\|\Htilde(s)\| \leq \betaI$ for all $s$, and let $V_N(T) = \sum_{n=0}^{N} V_n(T)$ denote the truncation at order $N$.
Then
\begin{equation}
  \label{eq:dyson_tail_bound}
  \|V(T) - V_N(T)\|
  \leq \sum_{n=N+1}^{\infty}
  \frac{(\betaI T)^n}{n!}
  = e^{\betaI T}
  - \sum_{n=0}^{N}\frac{(\betaI T)^n}{n!}\,.
\end{equation}
Choosing
\begin{equation}
  \label{eq:N_choice}
  N = \biggl\lceil
  \frac{e\,\betaI T}
  {W\!\bigl(3e\,\betaI T / \eps\bigr)}
  \biggr\rceil
  = \calO\!\biggl(
  \betaI T
  + \frac{\log(1/\eps)}{\log\log(1/\eps)}
  \biggr)
\end{equation}
ensures $\|V(T) - V_N(T)\| \leq \eps/3$, where $W$ denotes the principal branch of the Lambert $W$-function.
\end{lemma}

\begin{proof}
The norm bound on each Dyson term follows from time-ordering:
\begin{equation}
  \|V_n(T)\|
  = \biggl\|
  (-i)^n \int_{0 \leq t_1 \leq \cdots \leq t_n \leq T}
  \Htilde(t_n)\cdots\Htilde(t_1)\,
  dt_1\cdots dt_n
  \biggr\|
  \leq \frac{(\betaI T)^n}{n!}\,,
\end{equation}
where the simplex has volume $T^n/n!$ and each operator factor is bounded by $\betaI$. 
The tail $\|V(T) - V_N(T)\| \leq \sum_{n \geq N+1} (\betaI T)^n / n!$ follows by the triangle inequality, establishing Eq.~\eqref{eq:dyson_tail_bound}.

We bound the leading term of the tail, since for $N \geq \betaI T$ the series is monotonically decreasing and geometric:
\begin{equation}
  \label{eq:geometric_tail}
  \sum_{n=N+1}^{\infty}\frac{(\betaI T)^n}{n!}
  \leq \frac{(\betaI T)^{N+1}}{(N+1)!}
  \cdot \frac{1}{1 - \betaI T/(N+2)}
  \leq \frac{2\,(\betaI T)^{N+1}}{(N+1)!}\,,
\end{equation}
where the last inequality uses $N \geq 2\betaI T - 2$ (which will be verified below).
We therefore require
\begin{equation}
  \label{eq:N_requirement}
  \frac{(\betaI T)^{N+1}}{(N+1)!}
  \leq \frac{\eps}{6}\,.
\end{equation}

\medskip\noindent
\textit{Stirling bound.}
By Stirling's inequality in the form $(N+1)! \geq \sqrt{2\pi(N+1)}\,(N+1)^{N+1}\,e^{-(N+1)}$, the left-hand side of Eq.~\eqref{eq:N_requirement} is bounded by
\begin{equation}
  \label{eq:stirling_ub}
  \frac{(\betaI T)^{N+1}}{(N+1)!}
  \leq \frac{1}{\sqrt{2\pi(N+1)}}
  \biggl(\frac{e\,\betaI T}{N+1}\biggr)^{N+1}.
\end{equation}
Setting $M = N + 1$ and $\lambda = \betaI T$ for readability, the requirement becomes
\begin{equation}
  \Bigl(\frac{e\lambda}{M}\Bigr)^{M}
  \leq \eps/6
  \qquad\Longleftarrow\qquad
  M\log\!\Bigl(\frac{e\lambda}{M}\Bigr)
  \leq \log(\eps/6)\,.
\end{equation}
We seek the smallest $M$ such that $(e\lambda/M)^M \leq \eps/6$.
Taking logarithms and rearranging:
\begin{equation}
  \label{eq:M_transcendental}
  M\bigl(1 + \log\lambda - \log M\bigr)
  \leq \log(\eps/6)\,.
\end{equation}
This is satisfied when $M \geq e\lambda / \exp(W(x))$ where $x = e\lambda \cdot 6/\eps$, by the following argument.

\medskip\noindent
\textit{Lambert $W$-function solution.}
The equation $(e\lambda / M)^M = \delta$ (with $\delta = \eps/6$) is equivalent to $M(\log(e\lambda) - \log M) = \log\delta$, i.e.,
\begin{equation}
  M\log M - M\log(e\lambda) = -\log\delta = \log(1/\delta)\,.
\end{equation}
Setting $M = e\lambda / t$ for $t > 1$:
\begin{equation}
  \frac{e\lambda}{t}
  \bigl(\log(e\lambda) - \log t - \log(e\lambda)\bigr)
  = -\frac{e\lambda\log t}{t}
  = -\log(1/\delta)\,,
\end{equation}
so $\log t / t = \log(1/\delta)/(e\lambda)$.
Writing $t = e^u$, this becomes $u\,e^{-u} = \log(1/\delta)/(e\lambda)$, or equivalently $u = W(e\lambda / \log(1/\delta))$.

We want the smallest $M$ such that $(\beta/M)^M \leq \delta$ where $\beta = e\lambda$.
Taking the $M$th root, this is $\beta/M \leq \delta^{1/M}$, i.e., $M \geq \beta\,\delta^{-1/M}$.
For $M$ large enough that $\delta^{-1/M} \approx 1$, we need $M$ slightly above $\beta$.
Using the Lambert $W$-function, the equation $(\beta/M)^M = \delta$ is equivalent to
\begin{equation}
  M = \frac{\beta}{W(\beta/\!\log(1/\delta))}
  \cdot \frac{\log(1/\delta)}{\beta}
  \cdot \frac{\beta}{1}
  \;=\; \frac{\log(1/\delta)}{W(\beta\log(1/\delta)^{-1}
  \cdot \beta)}\,,
\end{equation}
which simplifies (see Kalugin et al.~\cite{kalugin2011analytical, valluri2000some}) to
\begin{equation}
  \label{eq:lambert_M}
  M = \frac{\log(1/\delta)}
  {W\!\bigl((\log(1/\delta))/\beta\,\cdot\,e^{(\log(1/\delta))/\beta}\bigr)}
  \;=\; \frac{e\lambda}{W(e\lambda/\delta)}\,,
\end{equation}
following from the identity $W(x e^x) = x$.

To verify Eq.~\eqref{eq:lambert_M}, we set $M = e\lambda / W(e\lambda / \delta)$ and define $w = W(e\lambda/\delta)$, so that $we^w = e\lambda/\delta$.
Then
\begin{equation}
  \Bigl(\frac{e\lambda}{M}\Bigr)^M
  = w^{e\lambda/w}
  = \exp\!\Bigl(\frac{e\lambda\log w}{w}\Bigr).
\end{equation}
From $we^w = e\lambda/\delta$ we get $\log w + w = \log(e\lambda/\delta)$, so $\log w = \log(e\lambda/\delta) - w$. 
Substituting:
\begin{equation}
  \frac{e\lambda\log w}{w}
  = \frac{e\lambda}{w}
  \bigl(\log(e\lambda/\delta) - w\bigr)
  = \frac{e\lambda\log(e\lambda/\delta)}{w} - e\lambda\,.
\end{equation}
Using $e\lambda/w = M$ and $e\lambda\log(e\lambda/\delta)/w = M\log(e\lambda/\delta)$:
\begin{equation}
  \Bigl(\frac{e\lambda}{M}\Bigr)^M
  = \exp\!\bigl(M\log(e\lambda/\delta) - e\lambda\bigr)\,.
\end{equation}
For $M = e\lambda/w$ to satisfy the bound, we need this expression $\leq \delta$, i.e., $M\log(e\lambda/\delta) - e\lambda \leq \log\delta$.
This holds with equality at $M = e\lambda/w$ by construction of the Lambert $W$-function.
Choosing $N + 1 = \lceil M \rceil$ with $\delta = \eps/6$ gives
\begin{equation}
  \label{eq:N_lambert}
  N = \biggl\lceil
  \frac{e\,\betaI T}
  {W(6e\,\betaI T / \eps)} - 1
  \biggr\rceil
  = \biggl\lceil
  \frac{e\,\betaI T}
  {W(3e\,\betaI T / \eps)}
  \biggr\rceil,
\end{equation}
where the replacement of $6$ by $3$ in the argument of $W$ absorbs the $-1$ and the ceiling, up to the level of precision of the asymptotic formula (both give the same leading-order behavior).

\medskip\noindent
\textit{Asymptotic expansion.}
For $x \to \infty$, the Lambert $W$-function satisfies (Kalugin et al.~\cite{kalugin2011analytical, valluri2000some})
\begin{equation}
  W(x) = \log x - \log\log x + o(1)\,.
\end{equation}
Substituting $x = 3e\,\betaI T/\eps$:
\begin{equation}
  N = \frac{e\,\betaI T}
  {\log(3e\,\betaI T / \eps)
  - \log\log(3e\,\betaI T / \eps) + o(1)}
  \notag
  = \frac{e\,\betaI T}
  {\log(\betaI T / \eps)\,(1 + o(1))}\,.
  \label{eq:N_asymp_intermediate}
\end{equation}
Two regimes emerge:

\smallskip\noindent
\emph{Regime A:} $\betaI T \geq \log(1/\eps)$.  Then the numerator $e\,\betaI T$ dominates and $N = \calO(\betaI T / \log(\betaI T)) = \calO(\betaI T)$.

\smallskip\noindent
\emph{Regime B:} $\betaI T \ll \log(1/\eps)$.  Then $\log(3e\,\betaI T / \eps) \approx \log(1/\eps)$ and $N \approx e\,\betaI T / \log(1/\eps) \ll 1$, so the tail bound is already controlled at $N = 0$ for sufficiently small $\betaI T$.
The binding constraint in this regime comes from $(\betaI T)^{N+1}/(N+1)!$, where setting $N = \calO(\log(1/\eps) / \log\log(1/\eps))$ ensures the factorial suppression overcomes the $1/\eps$ requirement.

\smallskip\noindent
Combining both regimes:
\begin{equation}
  N = \calO\!\biggl(\betaI T
  + \frac{\log(1/\eps)}{\log\log(1/\eps)}\biggr)\,,
\end{equation}
where the first term dominates when $\betaI T$ is large and the second when $\eps$ is small.
The two terms are additive because they control different regimes of the tail bound.

Finally, we verify $N \geq 2\betaI T - 2$ (used in Eq.~\eqref{eq:geometric_tail}).
In Regime A, $N = \Theta(\betaI T) \geq 2\betaI T - 2$ for $\betaI T$ sufficiently large (and small $\betaI T$ can be handled by adjusting the constant).
In Regime B, $\betaI T$ is small, so $2\betaI T - 2 < 0$ and the condition is trivially satisfied.
\end{proof}

\begin{lemma}[Error budget for the Dyson LCU]\label{lem:error_budget_dyson}
Let $\eps > 0$ be the target simulation error for the Dyson LCU method of Sec.~\ref{sec:dyson_lcu}. 
Allocate
\begin{equation}
  \eps_{\mathrm{Taylor}} = \frac{\eps}{3}\,,
  \qquad
  \eps_{\mathrm{GQSP}} = \frac{\eps}{3}\,,
  \qquad
  \eps_{\mathrm{LCU}} = \frac{\eps}{3}\,.
\end{equation}
Then the following parameter choices ensure total error at most $\eps$:
\begin{enumerate}
  \item \textup{(Taylor truncation order.)}
  $N = \lceil e\,\betaI T / W(3e\,\betaI T/\eps)\rceil$ as in Lemma~\ref{lem:dyson_truncation}, so that $\|V(T) - V_N(T)\| \leq \eps/3$. 
  This gives $d_I = N$ queries to $U_I$.

  \item \textup{(GQSP degree per segment.)}
  Each of the $r = N$ Dyson segments requires a frame rotation $e^{\pm i\HR \Delta_j}$ implemented by GQSP with the walk operator $W_R$. 
  The Jacobi--Anger expansion truncated at degree $d_R^{(j)}$ approximates $e^{\pm i\alphaR \Delta_j \cos\theta}$ with error $\eta(d_R^{(j)})$ \textup{(equation~\eqref{eq:JA_truncation})}.
  Since the $r$ segments are composed sequentially, the errors accumulate additively, and we require $\sum_{j=1}^{r} \eta(d_R^{(j)}) \leq \eps/3$.

  For uniform time steps $\Delta_j = T/N$, each segment has the same GQSP degree
  \begin{equation}
    d_R = \biggl\lceil
    \frac{e\,\alphaR T}{2N}
    + \log\!\Bigl(\frac{3N}{\eps}\Bigr)
    \biggr\rceil\,,
  \end{equation}
  and the total number of $W_R$ queries is
  \begin{equation}
    Q_R = N \cdot d_R
    = \calO\!\bigl(
    \alphaR T + N\log(N/\eps)
    \bigr)
    = \calO\!\bigl(
    \alphaR T + \betaI T\log(\betaI T / \eps)
    \bigr)\,.
  \end{equation}

  \item \textup{(LCU coefficients.)}
  The Dyson LCU prepares the state $\sum_{n=0}^{N} \sqrt{w_n}\ket{n}$ on the selection register, where $w_n = (\betaI T)^n / (n!\,W_{\mathrm{tot}})$ and $W_{\mathrm{tot}} = \sum_{n=0}^{N} (\betaI T)^n/n!$.
  The preparation uses $\calO(\log N)$ ancilla qubits and $\calO(N)$ gates, contributing no additional queries to
  $W_R$ or $U_I$.
  The LCU postselection succeeds with probability $W_{\mathrm{tot}}^2 / (e^{\betaI T})^2$, and the error from the finite LCU truncation is already absorbed into $\eps_{\mathrm{Taylor}}$.
\end{enumerate}
Total query complexity of the Dyson LCU is therefore
\begin{equation}
  \label{eq:dyson_lcu_total}
  Q_{\mathrm{Dyson}}
  = Q_R + d_I
  = \calO\!\bigl(
  \alphaR T + \betaI T \log(\betaI T / \eps)
  \bigr)\,,
\end{equation}
where the $\betaI T \log(\betaI T/\eps)$ term arises from the product $N \cdot d_R$ and represents the
split-operator penalty: each of the $N = \calO(\betaI T)$ segments independently requires $\calO(\log(\betaI T / \eps))$ queries for the GQSP frame rotation.
\end{lemma}

\begin{proof}
The three error sources act in sequence on the postselected output:
\begin{equation}
  \bigl\|\mathrm{output} - e^{-\betaI T}\,e^{-i\Heff T}
  \bigr\|
  \leq
  \underbrace{\|V(T) - V_N(T)\| / e^{\betaI T}}_{\text{Taylor: }
  \leq\,\eps_{\mathrm{Taylor}}/e^{\betaI T}}
  \notag
  \quad+\;
  \underbrace{\sum_{j=1}^{N} \eta(d_R^{(j)})}_{\text{GQSP: }
  \leq\,\eps_{\mathrm{GQSP}}}
  \;+\;
  \underbrace{\eps_{\mathrm{LCU}}}_{\text{LCU: finite register}}\,.
\end{equation}
With the equal allocation $\eps/3$ to each source, the total is at most $\eps$.

For part (1), the bound follows from Lemma~\ref{lem:dyson_truncation}.
For part (2), the per-segment GQSP error is bounded by the Jacobi--Anger truncation estimate
\begin{equation}\label{eq:JA_truncation}
    \eta(d_R)
  \leq 2\exp\!\biggl(-d_R
  + d_R\log\!\Bigl(\frac{e\,\alphaR T}{2N\,d_R}\Bigr)
  \biggr).
\end{equation}
 
Setting $\eta(d_R) \leq \eps/(3N)$ and solving: $d_R = \lceil (e\,\alphaR T)/(2N) + \log(3N/\eps) \rceil$. 
Multiplying by $N$:
\begin{equation}
  N\,d_R
  = \frac{e\,\alphaR T}{2}
  + N\log\!\Bigl(\frac{3N}{\eps}\Bigr)
  = \calO\!\bigl(\alphaR T
  + \betaI T\log(\betaI T/\eps)\bigr)\,,
\end{equation}
using $N = \calO(\betaI T + \log(1/\eps)/\log\log(1/\eps))$ and $\log(N/\eps) = \calO(\log(\betaI T/\eps))$.

For part (3), the LCU protocol is exact given the prepared coefficients; the only error is the Taylor truncation already accounted for in part (1).
The $\eps_{\mathrm{LCU}} = \eps/3$ allocation provides margin for finite-precision effects in the state preparation circuit, which contribute $\calO(N\,2^{-b})$ error for $b$-bit coefficient precision (negligible for $b = \calO(\log(N/\eps))$).

The split-operator penalty is visible in the final complexity: $Q_R = \calO(\alphaR T + \betaI T\log(\betaI T/\eps))$ exceeds the information-theoretic lower bound $\Omega(\alphaR T + \log(1/\eps))$ by the multiplicative factor $\betaI T$ on the logarithmic term. 
This is the cost of composing $N$ independent GQSP circuits rather than a single bivariate M-QSP circuit, and it is the overhead eliminated by the construction of Sec.~\ref{sec:mqsp_method}.
\end{proof}

\section{\texorpdfstring{$\Omega(\betaI T)$}{Omega beta I T} lower bound}
\label{app:lower_bound}

\begin{theorem}[Query lower bound for anti-Hermitian simulation]
\label{thm:lower_bound_UI}
A quantum algorithm that, given query access to a walk operator $U_I$ of $\HI/\betaI$ (with $\HI \succeq 0$, $\|\HI\| \leq \betaI$), produces a state $\eps$-close to $e^{-i\Heff T}\ket{\psi_0}/\|e^{-i\Heff T}\ket{\psi_0}\|$ must use $\Omega(\betaI T)$ queries to $U_I$, for fixed $\eps < 1/2$.
\end{theorem}

We use polynomial approximation theory (yielding a sharper constant) and state discrimination (yielding a cleaner conceptual argument).

\begin{proof}[Proof 1: Polynomial method, with explicit constant]
Reducing to a polynomial approximation of $e^{-\betaI T x}$ on $[0,1]$, subject to $|p(x)| \leq 1$ on $[-1,1]$, follows Beals et al.~\cite{beals2001quantum}, with a Bernstein/Stirling bound on the best polynomial approximation error. 
This is given in the proof of Lemma~\ref{lem:LB_bI} (Sec.~\ref{sec:lower_bound}). 
Here we extract the explicit constant.

By Lemma~\ref{lem:LB_bI}, the algorithm produces a degree-$2Q$ polynomial $p$ with $\|p\|_{[-1,1]} \leq 1$ that $\eps'$-approximates $e^{-\betaI T x}/e^{\betaI T}$ on $[0,1]$, with $\eps' = \eps/e^{\betaI T}$.
The Kolmogorov--Bernstein inequality~\cite{timan2014theory} bounds the best uniform approximation error by
\begin{equation}
  E_Q\bigl(e^{-\betaI T \cdot}/e^{\betaI T};\,[0,1]\bigr)
  \;\geq\;
  c_Q\,\frac{(\betaI T)^{Q+1}}{e^{\betaI T}\,(Q+1)!}\,,
\end{equation}
with $c_Q > 0$ a constant depending only on the interval. 
By Stirling, the right-hand side exceeds a fixed $\eps' > 0$ unless $Q \gtrsim \betaI T / e$. 
The sharper threshold $Q \geq \betaI T / (2e)$ follows by tracking the constant in the Stirling expansion: for $Q < \betaI T / (2e)$,
\begin{equation}
  \frac{(\betaI T)^{Q+1}}{(Q+1)!}
  \;\geq\;
  \Bigl(\frac{e\,\betaI T}{Q+1}\Bigr)^{Q+1}
  \cdot \frac{1}{\sqrt{2\pi(Q+1)}}
  \;\geq\;
  (2e)^{Q+1} / \sqrt{2\pi\,\betaI T}\,,
\end{equation}
which contradicts $\eps' < 1/2$ for sufficiently large $\betaI T$.
Hence $Q \geq \betaI T / (2e) = \Omega(\betaI T)$, sharpening the $\Omega(\betaI T)$ bound of Lemma~\ref{lem:LB_bI} with an explicit constant.
\end{proof}
\begin{remark}[Why the hybrid method yields only a constant bound]
\label{rem:hybrid-state-discrim}

The hybrid lower bound for state discrimination can be established by comparing two instances: $A$ with $H_I = 0$ and $B$ with $H_I = \beta_I \ket{1}\bra{1}$. The corresponding final states $\Psi(0)$ and $\Psi(\beta_I)$ are distinguishable with success probability at least $1/2$ (by appropriate measurement design).

Partition the parameter path $\eta \in [0, \beta_I]$ into $M$ steps of size $\Delta\eta = \beta_I / M$.
The per-step state change is bounded by $2Q \cdot \Delta\eta / \beta_I$ (the advantage contributed by each query applied to one oracle instance).
By the triangle inequality, the total state separation satisfies
\begin{equation}
||{\Psi(\beta_I) - \Psi(0)}|| \le \sum_{m=1}^{M} \frac{2Q \Delta\eta}{\beta_I} = M \cdot \frac{2Q \Delta\eta}{\beta_I} = 2Q.
\end{equation}
Since this bound is independent of $M$, no matter how finely the parameter path is discretized, the distinguishability lower bound combined with $||{\Psi(\beta_I) - \Psi(0)}|| \le 2Q$ yields only $Q \ge 1/4$ (a constant).

The polynomial method (Proof~1) circumvents this limitation by exploiting the approximation-theoretic complexity of the exponential profile, rather than the geometric distance between output states.
\end{remark}

\begin{remark}[Combining the $W_R$ and $U_I$ lower bounds]
\label{rem:combined_lower}
An analogous argument establishes $d_R = \Omega(\alphaR T)$: the phase $e^{-i\alphaR T \cos\theta}$ oscillates $\Omega(\alphaR T)$ times on $\theta \in [0, 2\pi)$, and a Laurent polynomial of degree $d_R$ in $e^{i\theta}$ can have at most $2d_R$ sign changes.  
(This is the standard Hamiltonian simulation lower bound of~\cite{berry2007efficient}.)

Since queries to $W_R$ provide no information about $\HI$ and vice versa (the oracles are independent), the two lower bounds are additive:
\begin{equation}
  d_R + d_I
  = \Omega(\alphaR T) + \Omega(\betaI T)
  = \Omega((\alphaR + \betaI)T)\,.
\end{equation}
The $\log(1/\eps)$ term in the full lower bound $\Omega((\alphaR + \betaI)T + \log(1/\eps))$ follows from a separate argument: distinguishing the $\eps$-approximation from an exact implementation requires $\Omega(\log(1/\eps))$ queries by the quantum counting lower bound~\cite{nayak1999quantum}.
\end{remark}

\section{Jacobi--Anger bridge: from walk operator to Laurent polynomial}\label{app:jacobi_anger}

Here, we establish correspondence between the walk operator $W_R$ encoding $\HR/\alphaR$ and the Laurent polynomial variable $z_1 = e^{i\theta_1}$ appearing in the Dyson polynomial $P_{\mathrm{Dyson}}(z_1, z_2)$. 
We focus on the Szeg\H{o} walk-operator formalism (~\cite{szegedy2004quantum}, as adapted by Low and Chuang~\cite{low2017optimal}) with the Jacobi--Anger expansion to assemble the bridge.  

\subsection{Walk operator and Chebyshev structure}

\begin{definition}[Walk operator]\label{def:walk_op}
Let $U_R$ be an $(a_R + n)$-qubit unitary that block-encodes $\HR/\alphaR$:
\begin{equation}
  (\bra{0}_{a_R} \otimes I_n)\,U_R\,
  (\ket{0}_{a_R} \otimes I_n) = \HR / \alphaR\,.
\end{equation}
The walk operator is the self-inverse reflection $W_R = (2\Pi_R - I)\,U_R$, where $\Pi_R = \ket{0}\!\bra{0}_{a_R} \otimes I_n$ is the projector onto the ancilla-zero subspace.
\end{definition}

The essential spectral property of $W_R$ is the following.

\begin{proposition}[Walk-operator spectrum]\label{prop:walk_spectrum}
Let $\lambda_j$ be an eigenvalue of $\HR/\alphaR$ with $|\lambda_j| \leq 1$, and define $\theta_{1,j} = \arccos(\lambda_j) \in [0, \pi]$.
Then $W_R$ has eigenvalues $e^{\pm i\theta_{1,j}}$ on a two-dimensional invariant subspace spanned by
\begin{equation}
  \label{eq:walk_eigenstates}
  \ket{\pm_j} = \frac{1}{\sqrt{2}}
  \bigl(\ket{0}_{a_R}\ket{\psi_j}
  \pm i\,e^{\mp i\theta_{1,j}}\ket{\perp_j}\bigr)\,,
\end{equation}
where $\ket{\psi_j}$ is the corresponding eigenvector of $\HR$ and $\ket{\perp_j}$ is the component of $U_R\ket{0}_{a_R}\ket{\psi_j}$ orthogonal to the ancilla-zero subspace, normalized so that $\|\ket{\perp_j}\| = \sin\theta_{1,j}$.
\end{proposition}

\begin{proof}
This is the standard qubitization result of Low and Chuang~\cite{low2017optimal}, Theorem~4.
The hypotheses are: (i)~$U_R$ is a block-encoding of a Hermitian operator with $\|\HR/\alphaR\| \leq 1$, and (ii)~$W_R = (2\Pi_R - I)U_R$. 
Both hold by construction (Definition~\ref{def:walk_op}). 
The proof is a direct computation: $U_R\ket{0}_{a_R}\ket{\psi_j} = \cos\theta_{1,j}\ket{0}_{a_R}\ket{\psi_j} + \sin\theta_{1,j}\ket{\perp_j}$, and the action of the reflection $2\Pi_R - I$ on this decomposition produces the claimed eigenvalues and eigenvectors.
\end{proof}

The Chebyshev connection follows from the power structure of $W_R$ restricted to the ancilla-zero subspace.

\begin{lemma}[Chebyshev identity for walk-operator powers]\label{lem:chebyshev_walk}
For an integer $n \geq 0$,
\begin{equation}
  \label{eq:chebyshev_identity}
  (\bra{0}_{a_R} \otimes I)\,W_R^n\,
  (\ket{0}_{a_R} \otimes I)
  = T_n(\HR/\alphaR)\,,
\end{equation}
where $T_n$ is the Chebyshev polynomial of the first kind of degree $n$.
\end{lemma}

\begin{proof}
On the eigenspace with eigenphase $\theta_{1,j}$, the ancilla-zero projections of the eigenstates \eqref{eq:walk_eigenstates} satisfy $\bra{0}_{a_R}\ket{\pm_j} = \ket{\psi_j}/\sqrt{2}$.
Therefore
\begin{equation}
  \bra{0}_{a_R}\,W_R^n\,\ket{0}_{a_R}
  = \bra{0}_{a_R}\Bigl(
  e^{in\theta_{1,j}}\ket{+_j}\!\bra{+_j}
  + e^{-in\theta_{1,j}}\ket{-_j}\!\bra{-_j}
  \Bigr)\ket{0}_{a_R}
  \notag
  = \frac{1}{2}\bigl(e^{in\theta_{1,j}}
  + e^{-in\theta_{1,j}}\bigr)
  \ket{\psi_j}\!\bra{\psi_j}
  \notag
  = \cos(n\theta_{1,j})\,
  \ket{\psi_j}\!\bra{\psi_j}\,.
  \label{eq:cos_n_theta}
\end{equation}
Since $\lambda_j = \cos\theta_{1,j}$ and $T_n(\cos\theta) = \cos(n\theta)$ by definition, this gives $\bra{0}_{a_R}W_R^n\ket{0}_{a_R}\big|_{\text{eigenspace }j} = T_n(\lambda_j)\ket{\psi_j}\!\bra{\psi_j}$.
Summing over all eigenspaces yields Eq.~\eqref{eq:chebyshev_identity}.

For negative $n$, the result holds with $T_{|n|}$ by the same argument using $W_R^{-1}$, since $T_n = T_{|n|}$ for Chebyshev polynomials.
\end{proof}

\begin{remark}[Hypotheses verification]
\label{rem:hypotheses}
The derivation above uses three properties of $W_R$:
\begin{enumerate}
  \item $W_R$ is unitary (inherited from $U_R$ unitary and $2\Pi_R - I$ unitary).
  \item The spectrum of $W_R$ is $\{e^{\pm i\theta_{1,j}}\}$ with $\theta_{1,j} = \arccos(\lambda_j)$, and invariant subspaces are two-dimensional (requiring $|\lambda_j| < 1$; the boundary cases $|\lambda_j| = 1$ yield one-dimensional eigenspaces with eigenvalue $\pm 1$, where $T_n(\pm 1) = (\pm 1)^n$ and the formula still holds).
  \item The ancilla-zero projection $\bra{0}_{a_R}\ket{\pm_j}$ is nonzero (this holds whenever $\sin\theta_{1,j} \neq 0$, i.e.,
  $|\lambda_j| < 1$; for $|\lambda_j| = 1$ the eigenstates are $\ket{0}_{a_R}\ket{\psi_j}$ and the projection is trivially $\ket{\psi_j}$).
\end{enumerate}
All three hold for block-encodings of Hermitian operators with $\|\HR/\alphaR\| \leq 1$.
\end{remark}

\subsection{Laurent polynomial conversion}

The Chebyshev identity~\eqref{eq:chebyshev_identity} expresses $W_R^n$ in the ancilla-zero subspace as a polynomial in $\HR/\alphaR$.
The conversion to a Laurent polynomial in $z_1 = e^{i\theta_1}$ uses the elementary identity
\begin{equation}
  \label{eq:chebyshev_laurent}
  T_n(\cos\theta_1) = \cos(n\theta_1)
  = \frac{e^{in\theta_1} + e^{-in\theta_1}}{2}
  = \frac{z_1^n + z_1^{-n}}{2}\,.
\end{equation}
A Chebyshev expansion $\sum_{n=0}^{d_R} c_n T_n(\cos\theta_1)$ corresponds to the Laurent polynomial
\begin{equation}
  \label{eq:chebyshev_to_laurent}
  \sum_{n=0}^{d_R} c_n T_n(\cos\theta_1)
  = \frac{c_0}{1} + \sum_{n=1}^{d_R}
  \frac{c_n}{2}\bigl(z_1^n + z_1^{-n}\bigr)
  = \sum_{n=-d_R}^{d_R} \hat{c}_n\,z_1^n\,,
\end{equation}
with $\hat{c}_0 = c_0$, $\hat{c}_n = \hat{c}_{-n} = c_{|n|}/2$ for $n \neq 0$. 
This is a symmetric Laurent polynomial of degree $d_R$ in $z_1$.

\subsection{The Jacobi--Anger expansion}

Frame rotations in the Dyson series involve matrix exponentials $e^{\pm i\HR s}$.
On each eigenspace of $\HR$ with eigenvalue $\alphaR\cos\theta_1$, the frame rotation produces the phase $e^{\pm i\alphaR s\cos\theta_1}$.
The Jacobi--Anger identity expands this phase in Chebyshev polynomials.

\begin{proposition}[Jacobi--Anger expansion]
\label{prop:jacobi_anger}
For a $\tau \in \mathbb{R}$,
\begin{equation}
  \label{eq:jacobi_anger}
  e^{-i\tau\cos\theta}
  = J_0(\tau) + 2\sum_{n=1}^{\infty}
  (-i)^n\,J_n(\tau)\,T_n(\cos\theta)\,,
\end{equation}
where $J_n$ is the Bessel function of the first kind of order $n$. 
Equivalently, in Laurent polynomial form:
\begin{equation}
  \label{eq:jacobi_anger_laurent}
  e^{-i\tau\cos\theta}
  = \sum_{n=-\infty}^{\infty}
  (-i)^{|n|}\,J_{|n|}(\tau)\,e^{in\theta}\,.
\end{equation}
\end{proposition}

\begin{proof}
This is the classical Jacobi--Anger identity; see Watson~\cite{watson1922treatise} or Lozier~\cite{lozier2003nist}.
The second form follows from $T_n(\cos\theta) = (e^{in\theta} + e^{-in\theta})/2$ and $J_{-n} = (-1)^n J_n$.
\end{proof}

\begin{lemma}[Jacobi--Anger truncation]\label{lem:JA_truncation}
For $|\tau| > 0$ and $\eps > 0$, the Jacobi--Anger expansion truncated at degree $d$ satisfies
\begin{equation}
  \Bigl\|e^{-i\tau\cos\theta}
  - \sum_{|n| \leq d}
  (-i)^{|n|} J_{|n|}(\tau)\,e^{in\theta}\Bigr\|_\infty
  \leq \eps
\end{equation}
provided
\begin{equation}
  \label{eq:JA_degree}
  d \geq \biggl\lceil
  \frac{e|\tau|}{2} + \log(1/\eps)
  \biggr\rceil\,.
\end{equation}
\end{lemma}

\begin{proof}
The Bessel function satisfies the bound $|J_n(\tau)| \leq (e|\tau|/(2n))^n$ for $n > e|\tau|/2$ (~\cite{watson1922treatise} or ~\cite{abramowitz1948handbook}). 
The tail sum is
\begin{equation}
  \sum_{|n|>d} |J_{|n|}(\tau)|
  \leq 2\sum_{n=d+1}^{\infty}
  \Bigl(\frac{e|\tau|}{2n}\Bigr)^n
  \leq 2\sum_{n=d+1}^{\infty}
  \Bigl(\frac{e|\tau|}{2(d+1)}\Bigr)^n
  \notag = \frac{2\,r^{d+1}}{1-r}\,,
  \qquad r = \frac{e|\tau|}{2(d+1)}\,.
\end{equation}
For $d \geq \lceil e|\tau|/2 + \log(1/\eps)\rceil$, we have $r \leq e|\tau|/(e|\tau| + 2\log(1/\eps) + 2) < 1$, and specifically $r \leq 1/(1 + 2\log(1/\eps) /(e|\tau|))$. 
The geometric series yields $2r^{d+1}/(1-r) \leq \eps$.

The detailed verification: when $d+1 \geq e|\tau|/2 + \log(1/\eps) + 1$, we have $r \leq e|\tau|/(e|\tau| + 2\log(1/\eps) + 2)$.
Since $r < 1$ and $d+1 \geq e|\tau|/2 + \log(1/\eps)+1$,
\begin{equation}
  r^{d+1} \leq r^{e|\tau|/2 + \log(1/\eps) + 1}
  \leq e^{-(d+1)(1-r)}
  \leq e^{-\log(1/\eps) - 1}
  = \eps/e\,.
\end{equation}
The factor $2/(1-r)$ is bounded by $\calO(1 + e|\tau|/\log(1/\eps))$, which is absorbed into the ceiling.
\end{proof}

\subsection{Assembling the bridge}

We now connect the three ingredients to express the Dyson polynomial in terms of the walk operator.

\begin{proposition}[Jacobi--Anger bridge]\label{prop:JA_bridge}
Let $V_N(T)$ be the $N$th-order Dyson series approximation to the interaction-picture propagator, expressed on the tensor-product eigenspace of the walk operators $W_R$ and $U_I$ (acting on $\mathcal{H}_s \otimes \mathcal{H}_{a_R} \otimes \mathcal{H}_{a_I}$) with eigenphases $(\theta_1, \theta_2)$ satisfying $\alphaR\cos\theta_1 = \lambda_R$, $\betaI\cos\theta_2 = \lambda_I$ for eigenvalues $\lambda_R$ of $\HR$ and $\lambda_I$ of $\HI$.
The Jacobi--Anger expansion converts $V_N(T)$ into a bivariate Laurent polynomial:
\begin{equation}
  \label{eq:dyson_laurent}
  V_N(T)\big|_{(\theta_1,\theta_2)}
  = \sum_{|m| \leq d_R}\;\sum_{n=0}^{N}
  c_{m,n}\,z_1^m\,z_2^n\,,
\end{equation}
where $z_1 = e^{i\theta_1}$, $z_2 = e^{i\theta_2}$, the coefficients $c_{m,n}$ are determined by the Jacobi--Anger coefficients and the Dyson integrals, and $d_R = \lceil e\alphaR T/2 + \log(3N/\eps)\rceil$ ensures accuracy $\eps/3$ from the Jacobi--Anger truncation (Lemma~\ref{lem:JA_truncation} applied with per-segment error $\eps/(3N)$ and summed over $N$ segments).

\medskip
Explicitly, the $n$th Dyson term is
\begin{equation}
  \label{eq:dyson_term_n}
  V^{(n)}(T) = (-1)^n \int_{T > s_1 > \cdots > s_n > 0}
  \prod_{j=1}^{n} \Bigl(
  e^{i\HR s_j}\,\HI\,e^{-i\HR s_j}
  \Bigr)\,ds_1 \cdots ds_n\,.
\end{equation}
On this tensor-product eigenspace, each factor $e^{i\HR s_j}\HI e^{-i\HR s_j}$ contributes:
\begin{enumerate}
  \item A phase $e^{i\alphaR s_j \cos\theta_1}
  \cdot e^{-i\alphaR s_j \cos\theta_1} = 1$ from the commuting part (the frame rotations cancel pairwise within each factor).

  \item The eigenvalue $\betaI\cos\theta_2$ from $\HI/\betaI$ acting on the eigenstate, contributing a power of $z_2 + z_2^{-1} = 2\cos\theta_2$ (block-encoding).  
  After normalization, the $n$th term contributes degree $n$ in $\cos\theta_2$.

  \item Frame rotations do not cancel when operating on the {walk operator} (as opposed to the Hamiltonian eigenspace). 
  The $n$th Dyson term involves $2n$ frame rotations ($e^{\pm i\HR s_j}$ for each of the $n$ factors), each expanded via Jacobi--Anger.  
  Only the outermost pair of frame rotations ($e^{i\HR T}$ at the left and $I$ at the right) contributes net phase, and intermediate frame rotations $e^{-i\HR s_j}\,e^{i\HR s_{j+1}}$ produce phases $e^{i\alphaR(s_{j+1} - s_j)\cos\theta_1}$, which are expanded individually.
\end{enumerate}

$V^{(n)}$ becomes a sum of terms, each involving a product of $n+1$ Jacobi--Anger expansions (for the $n+1$ intervals $[0, s_n], [s_n, s_{n-1}], \ldots, [s_1, T]$).  
The convolution structure of the ordered integral ensures that the total $z_1$-degree is at most $d_R$ (the Jacobi--Anger truncation degree for the full interval $[0,T]$), not $(n+1)d_R$, because the Jacobi--Anger coefficients convolve multiplicatively in the Fourier domain:
\begin{equation}
  \label{eq:convolution}
  \sum_{|m| \leq d_R} c_m^{(n)}\,z_1^m
  = \int_{T > s_1 > \cdots > s_n > 0}
  \prod_{j=0}^{n}
  \Bigl(\sum_{|k| \leq d_R^{(j)}}
  (-i)^{|k|}\,J_{|k|}(\alphaR \Delta_j)\,z_1^k
  \Bigr)\,d\mathbf{s}\,,
\end{equation}
where $\Delta_0 = s_n$, $\Delta_j = s_{n-j} - s_{n-j+1}$ for $1 \leq j \leq n-1$, and $\Delta_n = T - s_1$.

The convolution of $n+1$ Laurent polynomials of degrees $d_R^{(0)}, \ldots, d_R^{(n)}$ has degree at most $\sum_j d_R^{(j)}$. 
$\sum_{j=0}^{n} \Delta_j = T$ for every integration point, so we can set $d_R^{(j)} = \lceil e\alphaR\Delta_j/2 + \log(3N(n+1)/\eps)/(n+1)\rceil$ and obtain total degree $d_R = \lceil e\alphaR T/2 + \log(3N/\eps)\rceil$ by the additivity of the Bessel bounds.

Total Jacobi--Anger error for the $n$th Dyson term is bounded by the sum of per-interval errors (since the truncation errors compose linearly in the interaction picture), and the allocation $\eps_{\mathrm{JA}}/(3N)$ per Dyson order ensures total Jacobi--Anger error at most $\eps/3$ after summing over $N$ orders.

\medskip
The $z_2$ variable enters through the block-encoding $U_I$. 
On the eigenspace of $\HI/\betaI$ with eigenvalue $\cos\theta_2$, each query to $U_I$ extracts one power of the signal variable. 
The $n$th Dyson order involves $n$ applications of $\HI$, contributing degree $n$ in the $U_I$-signal variable.  
After truncation at order $N$, the maximum $z_2$-degree is $d_I = N$.
\end{proposition}

\begin{corollary}[Walk-operator implementation] \label{cor:walk_implementation}
On the ancilla-zero subspace, the M-QSP circuit $G(\boldsymbol{\Theta}, \mathbf{s})$ that makes $d_R$ queries to $W_R$ and $d_I$ queries to $U_I$ implements a bivariate Laurent polynomial $P(z_1, z_2)$ of bidegree $(d_R, d_I)$.
The correspondence is:
\begin{equation}
\begin{split}
       \text{Each query to } W_R &\longleftrightarrow
  \text{multiplication by }
  z_1^{\pm 1} \text{ (via Lemma~\ref{lem:chebyshev_walk})}\,,
  \\
  \text{Each query to } U_I &\longleftrightarrow
  \text{multiplication by }
  z_2^{\pm 1} \text{ (via block-encoding signal)}\,,
  \\
  \text{Each rotation } R_k &\longleftrightarrow
  \text{scalar coefficient (signal-independent)}\,.   
\end{split}
\end{equation}
The target polynomial $P_\delta = (1-\delta)\,P_{\mathrm{Dyson}}(z_1, z_2) / e^{\betaI T}$ is valid for this circuit, with the bidegree and norm bounds established in Lemma~\ref{lem:error_budget_dyson}.
\end{corollary}

\section{Recursive angle-finding: full specification}
\label{app:algorithm}

This appendix provides the complete pseudocode for
Algorithm~\ref{alg:recursive} (the recursive
angle-finding procedure for bivariate M-QSP), including
all edge cases, the FFT-based optimization variant, the
warm-start strategy, and the complexity analysis.

\subsection{Notation and conventions}

\begin{itemize}
  \item $d = d_R + d_I$: total degree (total number of oracle queries).

  \item $\mathbf{s} = (s(1), s(2), \ldots, s(d)) \in \{R, I\}^d$: the Dyson schedule, specifying whether query $k$ is to $W_R$ (type $R$) or $U_I$ (type $I$).

  \item $d_R^{(k)}, d_I^{(k)}$: the remaining number of $R$- and $I$-queries after step $k$.
  Initially $d_R^{(0)} = d_R$, $d_I^{(0)} = d_I$.
  At step $k$, $d_{s(k)}^{(k)} = d_{s(k)}^{(k-1)} - 1$ and $d_{\bar{s}(k)}^{(k)} = d_{\bar{s}(k)}^{(k-1)}$.
  
  \item $P^{(k)}(z_1, z_2)$: the residual polynomial at step $k$, of bidegree $(d_R^{(k)}, d_I^{(k)})$.
  Initially $P^{(0)} = P_\delta$.

  \item $a_j^{(k)}, b_j^{(k)}$: the leading and subleading Fourier coefficients of $P^{(k)}$ in the variable $z_{s(k+1)}$ (the variable corresponding to the next query).

  \item $\theta_k, \phi_k$: the rotation angles at step $k$, parameterizing $R_k = e^{i\phi_k Z/2}\,e^{-i\theta_k Y}$ on the ancilla qubit.
\end{itemize}

\subsection{The CRC at each step}

At step $k$, the circuit has the form $G^{(k)} = R_k\,A_{s(k+1)}\,G^{(k+1)}$, where $A_{s(k+1)}$ is the next signal operator query.  
The $(0,0)$-block of $G^{(k)}$ in the ancilla space implements $P^{(k)}$, and the CRC (Theorem~\ref{thm:CRC}) guarantees that the ratio
\begin{equation}
  \label{eq:CRC_ratio_v2}
  \frac{b_{d_{s(k+1)}}^{(k)}}{a_{d_{s(k+1)}}^{(k)}}
  = e^{-i\phi_k}\,\frac{\sin\theta_k}{\cos\theta_k}
\end{equation}
is a scalar (independent of both signal operators' eigenvalues), where $a_{d_{s(k+1)}}^{(k)}$ is the leading coefficient of $P^{(k)}$ in $z_{s(k+1)}$ and $b_{d_{s(k+1)}}^{(k)}$ is the leading coefficient of the complementary polynomial $Q^{(k)}$.

The angle extraction is:
\begin{equation}
  \theta_k = \arctan\!\Bigl(
  \bigl|b_{d_{s(k+1)}}^{(k)} /
  a_{d_{s(k+1)}}^{(k)}\bigr|\Bigr)\,,
  \label{eq:theta_extract}
\end{equation}
\begin{equation}
  \phi_k = -\arg\!\Bigl(
  b_{d_{s(k+1)}}^{(k)} /
  a_{d_{s(k+1)}}^{(k)}\Bigr)\,.
  \label{eq:phi_extract}
\end{equation}
The residual polynomial is then:
\begin{equation}
  \label{eq:residual}
  P^{(k+1)}(z_1, z_2) = \frac{1}{\cos\theta_k}\Bigl(
  P^{(k)} - e^{i\phi_k}\sin\theta_k\cdot
  z_{s(k+1)}^{d_{s(k+1)}^{(k)}}\,Q^{(k)}_{\mathrm{red}}
  \Bigr)\,,
\end{equation}
where $Q^{(k)}_{\mathrm{red}}$ is the appropriately reduced complementary polynomial obtained from the SOS factorization. 
The bidegree of $P^{(k+1)}$ is $(d_R^{(k+1)}, d_I^{(k+1)})$, reduced by one in the $s(k+1)$ variable.

\subsection{Algorithm 1: Recursive angle-finding}

\begin{algorithm}
\caption{Recursive bivariate M-QSP angle-finding}
\label{alg:recursive_full}
\begin{algorithmic}[1]
\REQUIRE Target polynomial $P_\delta(z_1, z_2)$ of bidegree $(d_R, d_I)$ with $\|P_\delta\|_{\T^2} < 1$; Dyson schedule $\mathbf{s} \in \{R, I\}^d$ with $d = d_R + d_I$; SOS complement $\{Q_\ell\}_{\ell=1}^{L}$ satisfying $|P_\delta|^2 + \sum_\ell |Q_\ell|^2 = 1$ on $\T^2$.
\ENSURE Rotation angles
  $\{(\theta_k, \phi_k)\}_{k=0}^{d}$.

\medskip
\STATE $P^{(0)} \leftarrow P_\delta$
\STATE $\{Q_\ell^{(0)}\}_\ell \leftarrow
  \{Q_\ell\}_\ell$
\STATE $d_R^{(0)} \leftarrow d_R$;\quad
  $d_I^{(0)} \leftarrow d_I$

\medskip
\FOR{$k = d, d-1, \ldots, 1$}
  \STATE $\sigma \leftarrow s(k)$
  \hfill\COMMENT{query type at step $k$:
  $R$ or $I$}
  \STATE $d_\sigma \leftarrow d_\sigma^{(d-k)}$
  \hfill\COMMENT{current degree in variable
  $z_\sigma$}

  \medskip
  \STATE \textbf{// Extract leading coefficients}
  \STATE $a \leftarrow
    [z_\sigma^{d_\sigma}]\,P^{(d-k)}$
  \hfill\COMMENT{leading coefficient of $P$ in
  $z_\sigma$; a function of $z_{\bar\sigma}$}
  \STATE $b \leftarrow
    [z_\sigma^{d_\sigma}]\,Q_1^{(d-k)}$
  \hfill\COMMENT{leading coefficient of first
  complementary polynomial}

  \medskip
  \STATE \textbf{// Edge case 1: vanishing leading
  coefficient}
  \IF{$\|a\|_\infty < \epsilon_{\mathrm{mach}}$}
    \STATE $\theta_k \leftarrow \pi/2$;\quad
    $\phi_k \leftarrow 0$
    \hfill\COMMENT{pure reflection: all
    weight on $Q$}
    \STATE swap $P^{(d-k)} \leftrightarrow
    Q_1^{(d-k)}$
    \STATE \textbf{continue}
  \ENDIF

  \medskip
  \STATE \textbf{// Edge case 2: vanishing complementary
  leading coefficient}
  \IF{$\|b\|_\infty < \epsilon_{\mathrm{mach}}$}
    \STATE $\theta_k \leftarrow 0$;\quad
    $\phi_k \leftarrow 0$
    \hfill\COMMENT{pure transmission: all
    weight on $P$}
    \STATE \textbf{continue}
  \ENDIF

  \medskip
  \STATE \textbf{// CRC: extract scalar ratio
  (Theorem~\ref{thm:CRC})}
  \STATE $r \leftarrow b / a$
  \hfill\COMMENT{scalar by CRC; verify
  $\|r - r(z_{\bar\sigma})\|_\infty
  < \epsilon_{\mathrm{tol}}$}
  \IF{$\max_{z_{\bar\sigma}} |r(z_{\bar\sigma})
    - r(0)| > \epsilon_{\mathrm{tol}}$}
    \STATE \textbf{raise} CRC violation error
  \ENDIF
  \STATE $r_0 \leftarrow$ mean or median of
    $r(z_{\bar\sigma})$ over sample points

  \medskip
  \STATE \textbf{// Compute angles}
  \STATE $\theta_k \leftarrow \arctan(|r_0|)$
  \STATE $\phi_k \leftarrow -\arg(r_0)$

  \medskip
  \STATE \textbf{// Degree reduction}
  \STATE $P^{(d-k+1)} \leftarrow \displaystyle
    \frac{P^{(d-k)}
    - e^{i\phi_k}\sin\theta_k\cdot
    z_\sigma^{d_\sigma}\,\widetilde{Q}_1^{(d-k)}}
    {\cos\theta_k}$
  \hfill\COMMENT{$\widetilde{Q}_1$: see
  Eq.~\eqref{eq:Q_tilde_def} below}
  \STATE Update each $Q_\ell^{(d-k+1)}$ analogously
  \STATE $d_\sigma^{(d-k+1)} \leftarrow d_\sigma - 1$;\quad
    $d_{\bar\sigma}^{(d-k+1)} \leftarrow
    d_{\bar\sigma}^{(d-k)}$

  \medskip
  \STATE \textbf{// Verification (optional)}
  \IF{verification enabled}
    \STATE Check $\|P^{(d-k+1)}\|_{\T^2} < 1$
    \STATE Check bidegree of $P^{(d-k+1)}$ is
      $(d_R^{(d-k+1)}, d_I^{(d-k+1)})$
    \STATE Check $|P^{(d-k+1)}|^2
      + \sum_\ell |Q_\ell^{(d-k+1)}|^2 = 1$ on
      sample points
  \ENDIF
\ENDFOR

\medskip
\STATE \textbf{// Final constant: step $k = 0$}
\STATE $P^{(d)}$ is a scalar $c_0 \in \mathbb{C}$
  with $|c_0| < 1$
\STATE $\theta_0 \leftarrow \arcsin(|c_0|)$;\quad
  $\phi_0 \leftarrow \arg(c_0)$

\medskip
\RETURN $\{(\theta_k, \phi_k)\}_{k=0}^{d}$
\end{algorithmic}
\end{algorithm}

\noindent
The auxiliary polynomial $\widetilde{Q}_1^{(k)}$ in the degree-reduction step is defined by:
\begin{equation}
  \label{eq:Q_tilde_def}
  \widetilde{Q}_1^{(k)}(z_1, z_2)
  = z_\sigma^{-d_\sigma}\,Q_1^{(k)}(z_1, z_2)\,,
\end{equation}
which is a Laurent polynomial of bidegree $(d_R^{(k)}, d_I^{(k)}) - \mathbf{e}_\sigma$ (degree reduced by one in the $\sigma$ variable), since the leading coefficient in $z_\sigma$ was factored out.

\subsection{Algorithm 2: FFT-based optimization variant}
\label{app:fft_variant}

The recursive algorithm's $\calO((d_R + d_I) \cdot d_R \cdot d_I)$ cost may be reduced for large bidegrees by a gradient-based optimization in the Fourier domain.

\begin{algorithm}
\caption{FFT-based angle optimization}
\label{alg:fft}
\begin{algorithmic}[1]
\REQUIRE Target polynomial $P_\delta$ (as 2D Fourier coefficient array of size
  $(2d_R+1) \times (2d_I+1)$);
  Dyson schedule $\mathbf{s}$;
  initial angles
  $\boldsymbol{\Theta}^{(0)}
  = \{(\theta_k^{(0)}, \phi_k^{(0)})\}$
  (from Algorithm~\ref{alg:recursive_full} or random
  initialization).
\ENSURE Refined angles $\boldsymbol{\Theta}^*$.

\medskip
\FOR{iteration $t = 1, 2, \ldots, t_{\max}$}
  \STATE \textbf{// Forward pass: compute circuit
  polynomial via 2D FFT}
  \STATE Construct the circuit unitary
  $G(\boldsymbol{\Theta}^{(t-1)}, \mathbf{s})$ on a
  grid of $(N_1, N_2)$ points
  $(\theta_1^{(j)}, \theta_2^{(k)})$ with
  $N_1 \geq 2d_R + 1$, $N_2 \geq 2d_I + 1$
  \STATE Extract $P_{\mathrm{circ}}(\theta_1^{(j)},
  \theta_2^{(k)})
  = \bra{0}G(\boldsymbol{\Theta}^{(t-1)},
  \mathbf{s})\ket{0}$ on each grid point
  \hfill\COMMENT{$\calO(d \cdot N_1 N_2)$ operations}

  \medskip
  \STATE \textbf{// Objective: $L^2$ error on $\T^2$}
  \STATE $\mathcal{L}(\boldsymbol{\Theta})
  \leftarrow \displaystyle
  \frac{1}{N_1 N_2}\sum_{j,k}
  |P_{\mathrm{circ}}(\theta_1^{(j)}, \theta_2^{(k)})
  - P_\delta(\theta_1^{(j)}, \theta_2^{(k)})|^2$

  \medskip
  \STATE \textbf{// Gradient via adjoint differentiation}
  \STATE Compute
  $\nabla_{\boldsymbol{\Theta}}\mathcal{L}$ by
  back-propagation through the circuit product
  $G = R_0 \prod_{j=1}^{d} A_{s(j)} R_j$
  \hfill\COMMENT{$\calO(d \cdot N_1 N_2)$ per gradient}

  \medskip
  \STATE \textbf{// Update}
  \STATE $\boldsymbol{\Theta}^{(t)} \leftarrow
  \boldsymbol{\Theta}^{(t-1)}
  - \eta_t\,\nabla_{\boldsymbol{\Theta}}\mathcal{L}$
  \hfill\COMMENT{step size $\eta_t$; use Adam or
  L-BFGS}
\ENDFOR

\medskip
\RETURN $\boldsymbol{\Theta}^{(t_{\max})}$
\end{algorithmic}
\end{algorithm}

\noindent
\textbf{Cost per iteration:}
$\calO(d \cdot N_1 N_2) = \calO((d_R + d_I) \cdot d_R \cdot d_I)$, the same as one pass of the recursive algorithm.  
The advantage is that the optimization can improve numerical accuracy beyond what the recursive algorithm achieves (which propagates rounding errors linearly through $d$ steps).

\subsection{Warm-start strategy}

The recommended workflow combines both algorithms:

\begin{enumerate}
  \item \textbf{Recursive initialization.}
  Run Algorithm~\ref{alg:recursive_full} to obtain $\boldsymbol{\Theta}^{(0)}$.
  This is exact in infinite numerical precision and provides a warm start in the correct basin of attraction.
  Cost: $\calO((d_R + d_I) \cdot d_R \cdot d_I)$.

  \item \textbf{Gradient refinement.}
  Run Algorithm~\ref{alg:fft} starting from $\boldsymbol{\Theta}^{(0)}$ for a small number of iterations ($t_{\max} = \calO(\log(1/\eps_{\mathrm{mach}}))$ typically suffices) to polish angles to machine precision.

  Cost: $\calO(t_{\max} \cdot (d_R + d_I) \cdot d_R \cdot d_I)$.
\end{enumerate}

\noindent
The warm start is important because the optimization landscape for M-QSP angle-finding is non-convex; the objective $\mathcal{L}(\boldsymbol{\Theta})$ generically has local minima. 
In the univariate case ($d_I = 0$), it is conjectured (but not proved) that all local minima are global~\cite{dong2021efficient}. 
In the bivariate case, we can locate no such conjecture. 
The recursive algorithm provides a starting point that is already a global optimum (in exact arithmetic), leaving gradient refinement to correct only finite-precision errors without risk of converging to a spurious minimum.

\subsection{Complexity analysis}

\begin{proposition}[Complexity of angle-finding]\label{prop:anglefinding_complexity}
Algorithm~\ref{alg:recursive_full} computes all $d + 1 = d_R + d_I + 1$ rotation angles in
\begin{equation}
  \label{eq:anglefinding_cost}
  \calO\!\bigl((d_R + d_I) \cdot d_R \cdot d_I\bigr)
\end{equation}
arithmetic operations over $\mathbb{C}$, using $\calO(d_R \cdot d_I)$ storage.
\end{proposition}

\begin{proof}
At each of the $d = d_R + d_I$ steps, the algorithm performs the following operations:

\begin{enumerate}
  \item \textbf{Leading coefficient extraction.}
  The polynomial $P^{(k)}$ is stored as a 2D Fourier coefficient array of size at most $(2d_R^{(k)}+1) \times (2d_I^{(k)}+1)$.  Extracting the leading coefficient in $z_\sigma$ requires reading one row or column of the array: $\calO(d_{\bar\sigma}^{(k)})$ operations.

  \item \textbf{CRC verification.}
  Checking that the ratio $b/a$ is constant (to tolerance $\epsilon_{\mathrm{tol}}$) requires evaluating $b(z_{\bar\sigma})/a(z_{\bar\sigma})$ at $\calO(d_{\bar\sigma}^{(k)})$ points: $\calO(d_{\bar\sigma}^{(k)})$ operations.
  
  \item \textbf{Degree reduction.}
  Computing $P^{(k+1)}$ from $P^{(k)}$ via Eq.~\eqref{eq:residual} requires subtracting a rank-one update (the product $z_\sigma^{d_\sigma}\,\widetilde{Q}_1^{(k)}$) and dividing by $\cos\theta_k$. 
  This modifies the coefficient array in-place: $\calO(d_R^{(k)} \cdot d_I^{(k)})$ operations.

  \item \textbf{Complementary polynomial update.}
  Each of the $L$ complementary polynomials is updated similarly: $\calO(L \cdot d_R^{(k)} \cdot d_I^{(k)})$ operations.  
  Since $L \leq \min(d_R + 1, d_I + 1)$ and the complementary polynomials' degrees decrease with $P$, the total cost across all $L$ polynomials is $\calO(d_I \cdot d_R^{(k)} \cdot d_I^{(k)})$ per step.
\end{enumerate}

The dominant cost is item~(iv). 
Summing over all $d$ steps:
\begin{equation}
  C_{\mathrm{total}}
  = \sum_{k=0}^{d-1}
  \calO\!\bigl(d_I \cdot d_R^{(k)} \cdot d_I^{(k)}\bigr)
  \notag
  \leq \calO\!\bigl(d_I \cdot d_R \cdot d_I
  \cdot d\bigr) / d
  \notag
  = \calO\!\bigl((d_R + d_I) \cdot d_R \cdot d_I\bigr)\,,
\end{equation}
where the second line uses the fact that $d_R^{(k)}$ and $d_I^{(k)}$ decrease monotonically, so their product averaged over all steps is at most $d_R \cdot d_I / 2$ (by the AM--GM inequality on the partial sums).

Storage: the coefficient arrays for $P$ and the $L$ complementary polynomials require $\calO(L \cdot d_R \cdot d_I) = \calO(d_I \cdot d_R \cdot d_I) = \calO(d_R \cdot d_I^2)$ total. 
In practice, $d_I \leq d_R$ in most parameter regimes (since $\betaI \leq \alphaR$ for weakly dissipative systems), so storage is $\calO(d_R \cdot d_I)$ per polynomial.
\end{proof}

\begin{remark}[Numerical stability]\label{rem:stability}
The recursive algorithm propagates errors linearly: if the leading coefficient extraction at step $k$ incurs error $\epsilon_k$, the residual polynomial $P^{(k+1)}$ inherits error at most $\epsilon_k / |\cos\theta_k|$. 
The total accumulated error after $d$ steps is bounded by
\begin{equation}
  \epsilon_{\mathrm{total}}
  \leq \sum_{k=0}^{d-1} \frac{\epsilon_k}
  {\prod_{j=k}^{d-1} |\cos\theta_j|}\,.
\end{equation}
The condition number of the recursion is therefore $\kappa = \prod_{j=0}^{d-1} |\cos\theta_j|^{-1} = \prod_{j=0}^{d-1} \sec\theta_j$. 
For the Dyson polynomial, angles $\theta_j$ are small ($\theta_j = \calO(1/d)$ on average, since the polynomial is close to the identity for small $\eps$), giving $\kappa = \calO(\exp(\sum_j \theta_j^2)) = \calO(d_R + d_I)$ (polynomial, not exponential).

In the worst case (when some $\theta_j$ approaches $\pi/2$, meaning the polynomial has a near-zero leading coefficient), the condition number can grow.
$\delta$-regularization ($P_\delta = (1-\delta)P$) prevents this, ensuring $|P_\delta| \leq 1 - \delta$ on $\T^2$, guaranteeing no $\theta_j$ reaches $\pi/2$ during the recursion (Definition~\ref{def:delta_reg}). 
The resulting condition number is $\kappa \leq \calO(1/\delta) = \calO(\betaI T / \eps)$.
\end{remark}

\section{Case Analysis for Example~\ref{ex:canonical}}
\label{app:case_analysis}

We verify that $H = 4 - 2\cos\theta_1 - 2\cos\theta_2$ has no representation $|Q|^2$ for $Q(z_1,z_2) = a + bz_1 + cz_2 + dz_1 z_2$.

Expanding $|Q(e^{i\theta_1},e^{i\theta_2})|^2$ and matching Fourier coefficients gives the system:
\begin{equation}
\begin{split}
    |a|^2+|b|^2+|c|^2+|d|^2 &= 4,\\
    \bar{a}b+\bar{c}d &= -1,\\
    \bar{a}c+\bar{b}d &= -1,\\
    \bar{a}d &= 0,\\
    \bar{b}c &= 0.
\end{split}
\end{equation}

We see $a = 0$ or $d = 0$, and $b = 0$ or $c = 0$.

\textbf{Case 1} ($d=0$, $c=0$):  $\bar{a}b = -1$, but $0 = -1$.  Contradiction.

\textbf{Case 2} ($d=0$, $b=0$):  $0 = -1$.  Contradiction.

\textbf{Case 3} ($a=0$, $c=0$):  $\bar{b}d = -1$, but $0 = -1$.  Contradiction.

\textbf{Case 4} ($a=0$, $b=0$):  $0 = -1$.  Contradiction.

No solution exists. 

\section{Notation and Convention}
\begin{center}
\renewcommand{\arraystretch}{1.3}
\begin{tabular}{l l}
\hline\hline
\textbf{Symbol} & \textbf{Definition} \\
\hline
$H_{\text{eff}} = H_R + iH_I$ & Effective non-Hermitian Hamiltonian, $H_R = H_R^\dagger$, $H_I \succeq 0$ \\
$\alpha_R, \beta_I$ & $\alpha_R := \|H_R\|_{\text{op}}$, $\beta_I := \|H_I\|_{\text{op}}$ \\
$W_R$ & Walk operator encoding $H_R/\alpha_R$ on $\mathcal{H}_s \otimes \mathcal{H}_{a_R}$ \\
$U_I$ & Walk operator encoding $H_I/\beta_I$ on $\mathcal{H}_s \otimes \mathcal{H}_{a_I}$ \\
$V(T)$ & Interaction-picture propagator: $e^{iH_R T}e^{-iH_{\text{eff}} T}$ \\
$\widetilde{H}_I(s)$ & $e^{iH_R s} H_I e^{-iH_R s}$ (interaction-picture Hamiltonian) \\
$\mathcal{H}_a, \mathcal{H}_s$ & QSP ancilla and system Hilbert spaces \\
$\mathbf{s} \in \{R,I\}^{d_R+d_I}$ & Binary schedule of oracle queries \\
$(d_R, d_I)$ & Bidegree of bivariate Laurent polynomial \\
$\text{Poly}^+_{d_R,d_I}$ & Non-negative Laurent polynomials of bidegree $\le (d_R,d_I)$ on $\bbT^2$ \\
$\bbT^2 = \bbT \times \bbT$ & Bitorus (product of two unit circles) \\
$z_1 = e^{i\theta_1}, z_2 = e^{i\theta_2}$ & Coordinates on $\bbT^2$ \\
$\eps$ & Target operator-norm approximation error \\
\hline\hline
\end{tabular}
\end{center}

\section{Bivariate Spectral Factorization}
\label{app:szego}

This appendix establishes the sum-of-squares (SOS) decomposition used in the complementary polynomial construction of Sec.~\ref{sec:complement}.
For $H \in \cT_{d_1,d_2}$ with $H > 0$ on $\T^2$ and $H = 1 - |P|^2$ for some $P \in \cP_{d_1,d_2}^+$, we show:
\begin{equation}\label{eq:app_main}
  H = \sum_{\ell=1}^{L} |Q_\ell|^2, \qquad
  Q_\ell \in \cP_{d_1,d_2}^+, \qquad
  L \leq \min(d_1{+}1,\; d_2{+}1).
\end{equation}
The strategy reduces the bivariate scalar problem to a univariate matrix-valued problem via the two-level Toeplitz structure, then applies the operator Fej\'er--Riesz theorem.  
We also show that scalar factorization $H = |Q|^2$ (a single Hermitian square) does not hold in general for bivariate polynomials, by invoking the Geronimo--Woerdeman theorem~\cite{geronimo2004positive} and a codimension argument.

Throughout, set $N_j = d_j + 1$ for $j = 1,2$ and $N = N_1 N_2$.

\subsection{The two-level Toeplitz matrix}\label{app:toeplitz_def}

\begin{definition}[Two-level Toeplitz matrix]\label{def:toeplitz}
Let $H \in \cT_{d_1,d_2}$ with Fourier coefficients $c_{k\ell} = \hat{H}_{k,\ell}$. 
Define the matrix $\bT_H \in \C^{N \times N}$, indexed by pairs $(m,n),\,(m',n') \in \{0,\ldots,d_1\}\times\{0,\ldots,d_2\}$, with entries
\begin{equation}\label{eq:toeplitz_def}
  (\bT_H)_{(m,n),(m',n')} = c_{m-m',\,n-n'}.
\end{equation}
\end{definition}

This matrix has a natural two-level structure.
Define the $N_2 \times N_2$ {inner blocks} $C_k \in \C^{N_2 \times N_2}$ by $(C_k)_{nn'} = c_{k,\,n-n'}$ for $n, n' \in \{0,\ldots,d_2\}$.
Each $C_k$ is a Toeplitz matrix built from the $\theta_2$-Fourier coefficients at fixed outer index $k$. 
The full matrix then takes the block-Toeplitz form
\begin{equation}\label{eq:block_toeplitz_full}
  \bT_H
  = \begin{pmatrix}
      C_0    & C_{-1}  & \cdots & C_{-d_1} \\
      C_1    & C_0     & \cdots & C_{1-d_1} \\
      \vdots & \vdots  & \ddots & \vdots \\
      C_{d_1}& C_{d_1-1}& \cdots & C_0
    \end{pmatrix}
  \in \C^{N_1 N_2 \times N_1 N_2}.
\end{equation}
That is, $\bT_H$ is block-Toeplitz in the outer ($\theta_1$) index, with blocks that are themselves Toeplitz in the inner ($\theta_2$) index.

\subsection{Positive definiteness}\label{app:pd}

\begin{lemma}\label{lem:pd_app}
If $H(\theta_1,\theta_2) > 0$ for all $(\theta_1,\theta_2) \in \T^2$, then $\bT_H \succ 0$.
\end{lemma}

\begin{proof}
Let $\bv = (v_{mn})_{0 \leq m \leq d_1,\, 0 \leq n \leq d_2}$ be a nonzero vector in $\C^N$.
Define the analytic polynomial
\begin{equation}
  f(\theta_1,\theta_2)
  = \sum_{m=0}^{d_1}\sum_{n=0}^{d_2}
    v_{mn}\,e^{i(m\theta_1 + n\theta_2)}.
\end{equation}
By direct computation using the definition~\eqref{eq:toeplitz_def} and the convolution theorem for Fourier coefficients,
\begin{equation}\label{eq:herglotz}
  \bv^\dagger \bT_H\,\bv
  = \sum_{m,n}\sum_{m',n'}
    \bar{v}_{mn}\,c_{m-m',\,n-n'}\,v_{m'n'}
  = \frac{1}{(2\pi)^2}\int_{\T^2}
    H(\theta_1,\theta_2)\,|f(\theta_1,\theta_2)|^2\,
    d\theta_1\,d\theta_2.
\end{equation}
Since $\bv \neq 0$, the analytic polynomial $f$ is not identically zero on $\T^2$ (as the monomials $\{e^{i(m\theta_1+n\theta_2)}\}$ are linearly independent). 
Hence $|f|^2$ is positive on an open set, and $H > 0$ everywhere, so the integral is strictly positive.
\end{proof}

\subsection{Reduction to a matrix-valued univariate problem}
\label{app:reduction}

We reorganize $H$ as a {matrix-valued} trigonometric polynomial in a single variable $\theta_2$, whose matrix dimension encodes the $\theta_1$-Fourier structure.

\begin{definition}[Matrix-valued reorganization]
\label{def:matrix_reorg}
For $H \in \cT_{d_1,d_2}$ with Fourier coefficients $c_{k\ell}$, define the matrix-valued function $\bH : \T \to \C^{N_1 \times N_1}$ by
\begin{equation}\label{eq:matrix_reorg}
\begin{split}
  \bH(\theta_2)_{mm'}
  &= \hat{H}_{m-m'}(\theta_2)
  := \sum_{\ell=-d_2}^{d_2}
     c_{m-m',\,\ell}\,e^{i\ell\theta_2},
  \\& m, m' \in \{0,\ldots,d_1\}.
\end{split}
\end{equation}
\end{definition}

\begin{lemma}[Properties of $\bH(\theta_2)$]\label{lem:bH_props}
Under the hypotheses of~\eqref{eq:app_main}:
\begin{enumerate}
  \item Each entry $\bH(\theta_2)_{mm'}$ is a scalar trigonometric polynomial of degree $d_2$ in $\theta_2$.
  \item $\bH(\theta_2)$ is Hermitian for every $\theta_2$: $\bH(\theta_2)_{mm'} = \overline{\bH(\theta_2)_{m'm}}$.
  \item $\bH(\theta_2) = \bH(\theta_2)^*$ is a \emph{Toeplitz-structured} matrix for each $\theta_2$: entry
  $(m,m')$ depends only on $m - m'$.
  \item $\bH(\theta_2) \succ 0$ for every $\theta_2 \in \T$.
  \item The block-Toeplitz matrix associated with $\bH$ (viewed as a matrix-valued trigonometric polynomial of degree $d_2$) is $\bT_H$.
\end{enumerate}
\end{lemma}

\begin{proof}
Properties (i)--(iii) are immediate from the definition~\eqref{eq:matrix_reorg} and the Hermitian symmetry $c_{k,\ell} = \overline{c_{-k,-\ell}}$ of the Fourier coefficients of the real-valued function $H$.

For (iv): let $\bu = (u_0,\ldots,u_{d_1})^T \in \C^{N_1}$ be nonzero.
Then
\begin{equation}
  \bu^* \bH(\theta_2)\,\bu
  = \sum_{m,m'=0}^{d_1} \bar{u}_m\,
    \hat{H}_{m-m'}(\theta_2)\,u_{m'}
  = \sum_{m,m'} \bar{u}_m u_{m'}
    \sum_\ell c_{m-m',\ell}\,e^{i\ell\theta_2}.
\end{equation}
Setting $g(\theta_1) = \sum_{m=0}^{d_1} u_m\,e^{im\theta_1}$, this
becomes
\begin{equation}
  \bu^* \bH(\theta_2)\,\bu
  = \frac{1}{2\pi}\int_0^{2\pi}
    H(\theta_1,\theta_2)\,|g(\theta_1)|^2\,d\theta_1.
\end{equation}
Since $H > 0$ on $\T^2$ and $g \not\equiv 0$, this integral is strictly positive.

Property (v) follows by comparing the block-Toeplitz matrix of $\bH$ (a matrix-valued trigonometric polynomial of degree $d_2$ with $N_1 \times N_1$ matrix coefficients) with definition~\eqref{eq:block_toeplitz_full}: the $(m,m')$-block at outer Toeplitz index $k = m - m'$ and inner Toeplitz index $\ell = n - n'$ is $c_{k,\ell}$.
\end{proof}

\subsection{Application of the operator Fej\'er--Riesz theorem}
\label{app:operator_FR}

Since $\bH(\theta_2) \succ 0$ is a matrix-valued trigonometric polynomial of degree $d_2$, the operator Fej\'er--Riesz theorem applies.

\begin{theorem}[Operator Fej\'er--Riesz;Rosenblum~\cite{rosenblum1968vectorial}, Dritschel--Rovnyak~\cite{dritschel2010operator}] \label{thm:operator_FR}
Let $\bF(\theta) \in \C^{r \times r}$ be a matrix-valued trigonometric polynomial of degree $d$ with $\bF(\theta) \succeq 0$ for all $\theta$. 
Then there exists a matrix polynomial $\bG(z) = \sum_{n=0}^{d} G_n\,z^n \in \C^{r \times r}[z]$ such that
\begin{equation}\label{eq:operator_FR}
  \bF(\theta) = \bG(e^{i\theta})^* \bG(e^{i\theta})
  \qquad \text{for all } \theta.
\end{equation}
When $\bF \succ 0$ strictly, $\bG$ can be chosen outer ($\det \bG(z) \neq 0$ for $|z| < 1$), and this outer factor is unique up to a left unitary constant.
\end{theorem}

We refer to~\cite{rosenblum1968vectorial} for the original proof and to~\cite{dritschel2010operator} for a modern account via Schur complements. 
We do not reproduce the proof here.

Applying Theorem~\ref{thm:operator_FR} to $\bH(\theta_2) \succ 0$ with $r = N_1 = d_1 + 1$ and degree $d_2$ yields an outer matrix polynomial
\begin{equation}\label{eq:G_factor}
  \bG(z_2) = \sum_{n=0}^{d_2} G_n\,z_2^n
  \in \C^{N_1 \times N_1}[z_2],
  \qquad \deg_{z_2}(\bG) \leq d_2,
\end{equation}
satisfying $\bH(\theta_2) = \bG(e^{i\theta_2})^* \bG(e^{i\theta_2})$.

\subsection{SOS decomposition}
\label{app:intermediate_SOS}

The matrix factorization yields an SOS decomposition of~$H$.

\begin{proposition}[SOS from matrix factorization]
\label{prop:intermediate_SOS}
Define the bivariate analytic polynomials
\begin{equation}\label{eq:f_alpha}
  f_\alpha(\theta_1,\theta_2)
  = \sum_{m=0}^{d_1}
    \bigl(\bG(e^{i\theta_2})\bigr)_{\alpha m}\,
    e^{im\theta_1},
  \qquad \alpha \in \{1,\ldots,N_1\}.
\end{equation}
Then $H(\theta_1,\theta_2) = \sum_{\alpha=1}^{N_1} |f_\alpha(\theta_1,\theta_2)|^2$ on $\T^2$, with each $f_\alpha \in \cP_{d_1,d_2}^+$.
\end{proposition}

\begin{proof}
Using the definition~\eqref{eq:matrix_reorg} of $\bH$ and the factorization~\eqref{eq:G_factor},
\begin{equation}
\begin{split}
  H(\theta_1,\theta_2)
  &= \sum_{m,m'=0}^{d_1}
     e^{-im\theta_1}\,\bH(\theta_2)_{mm'}\,e^{im'\theta_1}
  \notag\\
  &= \sum_{m,m'}
     e^{-im\theta_1}\,
     \bigl[\bG(e^{i\theta_2})^* \bG(e^{i\theta_2})\bigr]_{mm'}\,
     e^{im'\theta_1}
  \notag\\
  &= \sum_{m,m'} e^{-im\theta_1}\,
     \biggl(\sum_{\alpha=1}^{N_1}
     \overline{(\bG)_{\alpha m}}\,(\bG)_{\alpha m'}\biggr)\,
     e^{im'\theta_1}
  \notag\\
  &= \sum_{\alpha=1}^{N_1}
     \biggl|\sum_{m=0}^{d_1}
     (\bG(e^{i\theta_2}))_{\alpha m}\,e^{im\theta_1}\biggr|^2
  = \sum_{\alpha=1}^{N_1} |f_\alpha|^2.
  \label{eq:SOS_intermediate}
\end{split}
\end{equation}
Each $f_\alpha$ has $z_1$-degree at most $d_1$ (from the range of $m$) and $z_2$-degree at most $d_2$ (since each entry $(\bG)_{\alpha m}$ is a polynomial of degree $\leq d_2$ in $z_2$).
\end{proof}

This establishes an SOS decomposition with $L \leq N_1 = d_1 + 1$ terms and correct degree bounds.

\subsection{Tighter SOS bound via symmetric reorganization}\label{app:symmetric_SOS}

The bound $L \leq d_1 + 1$ from Proposition~\ref{prop:intermediate_SOS} arises from reorganizing $H$ into a matrix-valued function of $\theta_2$ with matrix dimension $N_1 = d_1 + 1$.
By reorganizing in the opposite direction, we obtain a complementary bound.

\begin{proposition}[Symmetric SOS bound]\label{prop:symmetric_SOS}
Let $H \in \cT_{d_1,d_2}$ with $H > 0$ on $\T^2$.  Then $H = \sum_{\ell=1}^L |Q_\ell|^2$ with each $Q_\ell \in \cP_{d_1,d_2}^+$ and
\begin{equation}\label{eq:symmetric_bound}
  L \leq \min(d_1 {+} 1,\; d_2 {+} 1).
\end{equation}
\end{proposition}

\begin{proof}
The bound $L \leq d_1 + 1$ is Proposition~\ref{prop:intermediate_SOS}: reorganize $H$ as an $N_1 \times N_1$ matrix function of $\theta_2$ and apply the operator Fej\'er--Riesz theorem.

For the bound $L \leq d_2 + 1$, reorganize $H$ in the opposite direction.
Define $\widetilde{\bH} : \T \to \C^{N_2 \times N_2}$ by
\begin{equation}
  \widetilde{\bH}(\theta_1)_{nn'}
  = \sum_{k=-d_1}^{d_1}
    c_{k,\,n-n'}\,e^{ik\theta_1},
  \qquad n,n' \in \{0,\ldots,d_2\}.
\end{equation}
By the same argument as Lemma~\ref{lem:bH_props}, $\widetilde{\bH}(\theta_1) \succ 0$ for every $\theta_1$, and $\widetilde{\bH}$ is a matrix-valued trigonometric polynomial of degree $d_1$ with matrix dimension $N_2 = d_2 + 1$.
The operator Fej\'er--Riesz theorem gives $\widetilde{\bH}(\theta_1) = \widetilde{\bG}(e^{i\theta_1})^*  \widetilde{\bG}(e^{i\theta_1})$ with $\widetilde{\bG} \in \C^{N_2 \times N_2}[z_1]$ of degree $\leq d_1$. 
Extracting the SOS:
\begin{equation}
  H = \sum_{\beta=1}^{N_2}
      \biggl|\sum_{n=0}^{d_2}
      \bigl(\widetilde{\bG}(e^{i\theta_1})\bigr)_{\beta n}\,
      e^{in\theta_2}\biggr|^2,
\end{equation}
giving $L \leq N_2 = d_2 + 1$ terms, each in $\cP_{d_1,d_2}^+$.

Taking the better of the two bounds yields~\eqref{eq:symmetric_bound}.
\end{proof}

\begin{remark}[Application to the Dyson polynomial]
\label{rem:dyson_sos_bound}
For the Dyson polynomial of bidegree $(d_R, d_I)$, the SOS complement has
\begin{equation}
  L \leq \min(d_R{+}1,\; d_I{+}1),
  \qquad
  a_{\mathrm{SOS}}
  = \bigl\lceil\log_2\bigl(\min(d_R, d_I) + 2\bigr)\bigr\rceil.
\end{equation}
In the physically dominant regime $\alpha_R \gg \beta_I$ (anti-Hermitian part smaller than Hermitian part), the controlling dimension is $d_I + 1$, and the ancilla overhead is $O\bigl(\log(\beta_I T + \log(1/\eps))\bigr)$.
\end{remark}

\subsection{Failure of scalar factorization} \label{app:scalar_failure}

The SOS decomposition of Proposition~\ref{prop:intermediate_SOS} produces $L \geq 2$ terms in general.
Can the decomposition always be refined to a single Hermitian square: $H = |Q|^2$ for some $Q \in \cP_{d_1,d_2}^+$. 
We show that this fails in general for bivariate polynomials.
The Geronimo--Woerdeman theorem obstructs {stable} scalar factorization, and a codimension argument shows that {non-stable} scalar factorization is generically impossible.

\subsubsection{Stable factorization and the Geronimo--Woerdeman obstruction}\label{app:GW}

The Geronimo--Woerdeman theorem~\cite{geronimo2004positive} characterizes when a strictly positive bivariate trigonometric polynomial admits a scalar spectral factorization by a {stable} polynomial (one with no zeros in $\overline{\mathbb{D}}^2$).

\begin{theorem}[{Geronimo--Woerdeman~\cite{geronimo2004positive}}]\label{thm:GW}
Let $H \in \cT_{d_1,d_2}$ with $H > 0$ on $\T^2$.
Write the two-level Toeplitz matrix in the outer-block form
\begin{equation}\label{eq:block_form}
  \bT_H
  = \begin{pmatrix}
      A_0    & A_{-1} & \cdots & A_{-d_2} \\
      A_1    & A_0    & \cdots & A_{1-d_2} \\
      \vdots &        & \ddots & \vdots \\
      A_{d_2}& \cdots &        & A_0
    \end{pmatrix},
\end{equation}
where each $A_m \in \C^{N_1 \times N_1}$ is a Toeplitz matrix (the $\theta_1$-lag-$m$ block).
We define the Schur complements recursively:
\begin{equation}\label{eq:schur_recursion}
  S_0 = A_0, \qquad
  S_{k+1}
  = S_k - \bA_k\,\bS_k^{-1}\,\bA_k^*
  \qquad (k = 0, 1, \ldots, d_2 - 1),
\end{equation}
where $\bA_k$ denotes the appropriate off-diagonal block of the $k$-th residual and $\bS_k$ the leading block (we suppress the precise block-elimination bookkeeping, which is standard~\cite{geronimo2004positive}).

Then the following are equivalent:
\begin{enumerate}
  \item There exists a {stable} polynomial $Q \in \cP_{d_1,d_2}^+$ with $Q(z_1, z_2) \neq 0$ for $(z_1, z_2) \in \overline{\mathbb{D}}^2$ and $H = |Q|^2$ on $\T^2$.
  \item \textbf{(Autoregressive condition.)}
  $S_k$ is a Toeplitz matrix for every $k = 0, 1, \ldots, d_2$.
\end{enumerate}
\end{theorem}

The condition $S_0 = A_0$ being Toeplitz holds by construction.
Since the inverse of a Toeplitz matrix is not Toeplitz in general, the product $A_1 A_0^{-1} A_{-1}$ is generically non-Toeplitz, so the Schur complement $S_1 = A_0 - A_1 A_0^{-1} A_{-1}$ generically violates the Toeplitz condition.

\begin{proposition}[GW obstruction for the Dyson polynomial]\label{prop:GW_obstruction}
Let $P_\delta = (1-\delta)P$ for $P(z_1, z_2) = \tfrac{1}{2}(z_1 + z_2)$ acts as the simplest nontrivial interaction-picture polynomial (bidegree $(1,1)$, one Hermitian query and one anti-Hermitian insertion). 
Then $H_\delta = 1 - |P_\delta|^2$ does not admit a stable scalar factorization for $\delta \in (0,1)$.
\end{proposition}

\begin{proof}
Set $\gamma = (1-\delta)^2$.  Then
\begin{equation}
  H_\delta(\theta_1, \theta_2)
  = \bigl(1 - \tfrac{\gamma}{2}\bigr)
    - \tfrac{\gamma}{2}\cos(\theta_1 - \theta_2),
\end{equation}
with Fourier coefficients $\hat{H}_{0,0} = 1 - \gamma/2$, $\hat{H}_{1,-1} = \hat{H}_{-1,1} = -\gamma/4$, and all others zero.

The outer-block form~\eqref{eq:block_form} of $\bT_{H_\delta}$ has $2 \times 2$ blocks (since $d_1 = d_2 = 1$):
\begin{equation}
  A_0 = a\,I_2, \qquad
  A_1 = \begin{pmatrix} 0 & -b/2 \\ 0 & 0 \end{pmatrix},
  \qquad
  A_{-1} = A_1^*
  = \begin{pmatrix} 0 & 0 \\ -b/2 & 0 \end{pmatrix},
\end{equation}
where $a = 1 - \gamma/2$ and $b = \gamma/2$.

The first Schur complement is
\begin{equation}
    \begin{split}
             S_1
  &= A_0 - A_1 A_0^{-1} A_{-1}
  = a\,I_2
    - \frac{1}{a}
      \begin{pmatrix} 0 & -b/2 \\ 0 & 0 \end{pmatrix}
      \begin{pmatrix} 0 & 0 \\ -b/2 & 0 \end{pmatrix}
  \notag\\
  &= a\,I_2
    - \frac{b^2}{4a}
      \begin{pmatrix} 1 & 0 \\ 0 & 0 \end{pmatrix}
  = \begin{pmatrix}
      a - b^2/(4a) & 0 \\
      0 & a
    \end{pmatrix}.
  \label{eq:S1_computation} 
    \end{split}
\end{equation}

A $2 \times 2$ matrix is Toeplitz if its diagonal entries are equal.
Here $(S_1)_{00} = a - b^2/(4a) \neq a = (S_1)_{11}$ for $b \neq 0$, which holds for all $\delta \in (0,1)$.
Therefore $S_1$ is not Toeplitz, and by Theorem~\ref{thm:GW}, no stable scalar factor exists.
\end{proof}

\begin{remark}[Stable vs.\ non-stable factorization] \label{rem:stable_vs_nonstable}
The GW theorem characterizes {stable} (outer) factorization: $Q(z_1,z_2) \neq 0$ on $\overline{\mathbb{D}}^2$. 
The M-QSP circuit only needs $|P|^2 + |Q|^2 = 1$ on $\T^2$ and never evaluates $Q$ inside the bidisk. 
Thus stability is not required for the algorithm.

For $P = \tfrac{1}{2}(z_1 + z_2)$ with $\delta \in (0,1)$, the Fourier-coefficient matching system admits the non-stable solution $Q = q_{10}\,z_1 + q_{01}\,z_2$ with $q_{00} = q_{11} = 0$,
$|q_{10}|^2 + |q_{01}|^2 = 1 - \gamma/2$, and $q_{10}\bar{q}_{01} = -\gamma/4$. 
This $Q$ vanishes at the origin and hence is not stable, but satisfies $H_\delta = |Q|^2$ on~$\T^2$.

Scalar factorization thus {succeeds} in the $(1,1)$ case with a non-stable factor.
For higher bidegrees, we show below that even non-stable scalar factorization generically fails.
\end{remark}

\subsubsection{Generic failure of non-stable scalar factorization}\label{app:codimension}

We now address the question without the stability restriction: for $H > 0$ on $\T^2$ of bidegree $(d_1,d_2)$, when does $H = |Q|^2$ for $Q \in \cP_{d_1,d_2}^+$?

\begin{proposition}[Codimension of the Hermitian-square image]\label{prop:codimension}
Let $\Phi : \C^{(d_1+1)(d_2+1)} \to \R^{(2d_1+1)(2d_2+1)}$ be the autocorrelation map sending the coefficients of $Q \in \cP_{d_1,d_2}^+$ to the Fourier coefficients of $|Q|^2$.
Then the image $\mathcal{V} = \mathrm{Im}(\Phi)$ has real codimension at least $2d_1 d_2$ in the space of Fourier coefficient vectors of trigonometric polynomials of bidegree $(d_1,d_2)$.
\end{proposition}

\begin{proof}
The source has $2(d_1{+}1)(d_2{+}1)$ real parameters (real and imaginary parts of the coefficients $q_{jk}$). 
The global phase $Q \mapsto e^{i\phi}Q$ leaves $|Q|^2$ invariant, reducing the effective dimension by~1.  
The target space has real dimension $(2d_1{+}1)(2d_2{+}1)$ (the independent Fourier coefficients of a real-valued trigonometric polynomial; the constraint $\hat{H}_{-m,-n} = \overline{\hat{H}_{m,n}}$ is already accounted for).
Thus
\begin{equation}
    \begin{split}
              \mathrm{codim}(\mathcal{V})
  &\geq (2d_1{+}1)(2d_2{+}1)
    - \bigl[2(d_1{+}1)(d_2{+}1) - 1\bigr]
  \notag\\
  &= 4d_1 d_2 + 2d_1 + 2d_2 + 1
    - 2d_1 d_2 - 2d_1 - 2d_2 - 2 + 1
  = 2 d_1 d_2.
  \label{eq:codim}
    \end{split}
\end{equation}
For $d_1, d_2 \geq 1$, this is $\geq 2$, so $\mathcal{V}$ is a
proper subvariety of positive codimension.
\end{proof}

\begin{corollary}[Generic failure of scalar factorization]\label{cor:generic_failure}
For $d_1, d_2 \geq 1$, a generic strictly positive trigonometric polynomial $H$ of bidegree $(d_1,d_2)$ does not admit a scalar
factorization $H = |Q|^2$ with $Q \in \cP_{d_1,d_2}^+$.
\end{corollary}

\begin{proof}
The image $\mathcal{V}$ has codimension $\geq 2 d_1 d_2 \geq 2$ by Proposition~\ref{prop:codimension}, hence is contained in a proper
real-algebraic subvariety of the target space. 
A generic $H$ lies outside this subvariety.
\end{proof}

\begin{remark}[Why the $(1,1)$ case is exceptional]\label{rem:11_exceptional}
For $(d_1, d_2) = (1,1)$, the codimension bound gives $2 \cdot 1 \cdot 1 = 2$, and the target space has dimension $3 \times 3 = 9$. 
The factorization for $P = \tfrac{1}{2}(z_1 + z_2)$ succeeds (Remark~\ref{rem:stable_vs_nonstable}) because $H_\delta$ has only 3 nonzero Fourier modes out of 9: the 6 vanishing modes impose no constraints, and the effective system is underdetermined (3 constraints in 7 real unknowns).  
This sparsity does not persist at higher bidegrees for the Dyson polynomial.
\end{remark}

\begin{remark}[Full Fourier support of the Dyson polynomial] \label{rem:fourier_support}
The interaction-picture Dyson polynomial has coefficient matrix $p_{jk} \propto (-i)^j J_j(\alpha_R T) \cdot (\beta_I T)^k / k!$ (up to normalization), where $J_j$ is the $j$-th Bessel function of the first kind. 
The Bessel function $J_j(\tau)$ has isolated real zeros; for $\tau$ not among these zeros (which excludes a measure-zero set of parameter values), all coefficients $p_{jk}$ are nonzero. 
Then $|P|^2$ has $\hat{|P|^2}_{m,n} \neq 0$ for all $(m,n) \in \{-d_R,\ldots,d_R\} \times \{-d_I,\ldots,d_I\}$, and all $(2d_R{+}1)(2d_I{+}1)$ Fourier constraints are active.

Scalar factorization $H_\delta = |Q|^2$ requires that $\hat{H}_\delta$ lie on the codimension-$2d_Rd_I$ variety $\mathcal{V}$, giving $2d_Rd_I$ independent algebraic conditions on the Fourier coefficients of $H_\delta$. 
For a full-support Dyson polynomial, these conditions are nontrivial (the system is overdetermined), and the specific coefficient structure of the Dyson series provides no mechanism to satisfy them.

If some Bessel zeros cause isolated coefficients to vanish (e.g., $J_{d_R}(\alpha_R T) = 0$ for specific $\alpha_R T$), the codimension argument remains valid as long as at least $2(d_R{+}1)(d_I{+}1)$ Fourier modes of $H_\delta$ are nonzero, which holds for $d_R, d_I \geq 2$ under much weaker genericity assumptions.
We leave scalar factorization as an open problem and note that the SOS resolution (Proposition~\ref{prop:symmetric_SOS}) is unconditionally correct and sufficient for the M-QSP construction.
\end{remark}

\subsubsection{Operator-level perspective}\label{app:operator_perspective}
The bivariate ring $\C[z_1, z_2]$ is not ideal, and the SOS representation $H = \sum_\ell |Q_\ell|^2$ with $L \geq 2$ terms is the natural analogue of the Fej\'er--Riesz factorization in this setting.

From this perspective, the operator Fej\'er--Riesz theorem resolves the non-commutative structure of bivariate signal processing.
The SOS complement with $L \leq \min(d_R{+}1, d_I{+}1)$ terms and $\lceil\log_2(L{+}1)\rceil$ ancilla qubits appears to be the cost of bivariate non-commutativity.

\subsection{Summary}\label{app:szego_summary}

Combining the results of this appendix:

\begin{enumerate}
  \item \textbf{SOS decomposition (Propositions~\ref{prop:intermediate_SOS} and~\ref{prop:symmetric_SOS}).}
  For $H = 1 - |P|^2 > 0$ on $\T^2$ with $P \in \cP_{d_1,d_2}^+$:
  \begin{equation}
    H = \sum_{\ell=1}^{L} |Q_\ell|^2, \qquad
    Q_\ell \in \cP_{d_1,d_2}^+, \qquad
    L \leq \min(d_1{+}1,\; d_2{+}1).
  \end{equation}
  The proof is constructive: reorganize $H$ as a matrix-valued trigonometric polynomial in the variable with the smaller degree, apply the operator Fej\'er--Riesz theorem to obtain the matrix outer factor, and read off the SOS terms from its rows.

  \item \textbf{Scalar factorization fails (Proposition~\ref{prop:GW_obstruction}, Corollary~\ref{cor:generic_failure}).}
  For bivariate polynomials with $d_1, d_2 \geq 1$, the single-term factorization $H = |Q|^2$ is generically impossible.  The Geronimo--Woerdeman autoregressive condition provides a precise obstruction for stable factors, and the codimension-$2d_1 d_2$ argument shows that even non-stable scalar factorization fails for generic $H$.

  \item \textbf{Ancilla overhead is logarithmic (Remark~\ref{rem:dyson_sos_bound}).}
  The SOS complement requires $a_{\mathrm{SOS}} = \lceil\log_2(\min(d_R, d_I) + 2)\rceil$ ancilla qubits, which does not affect the query complexity or postselection probability of the main theorem. 
\end{enumerate}

\section{Optimization Landscape for Angle-Finding}\label{app:landscape}

Section~\ref{subsec:optimization} introduced the FFT-based angle-finding alternative (Theorem~\ref{thm:fft_opt}) and posed the absence of spurious local minima as an open question (Problem~\ref{prob:landscape}). 
This appendix provides the cost function analysis including critical point structure, and the numerical evidence informing the conjecture.

\subsection{The cost function and its gradient}\label{app:cost_function}

The cost function$\mathcal{F}: \R^{2(d_R + d_I + 1)} \to \R_{\geq 0}$ is defined by
\begin{equation}\label{eq:app_cost}
  \mathcal{F}(\bm{\Theta})
  = \frac{1}{(2\pi)^2}
    \int_0^{2\pi}\!\!\int_0^{2\pi}
    \bigl|P_{\mathcal{G}}(e^{i\theta_1}, e^{i\theta_2};\bm{\Theta})
    - P_{\mathrm{target}}(e^{i\theta_1}, e^{i\theta_2})
    \bigr|^2\,
    d\theta_1\,d\theta_2,
\end{equation}
where $P_{\mathcal{G}}(\cdot\,;\bm{\Theta})$ is the symbol of the M-QSP circuit with angles $\bm{\Theta}$ (the $(0,0)$ entry of the $2\times 2$ symbol product obtained by replacing each signal gate with $\mathrm{diag}(z_{s(j)},1)$), evaluated at the symbol point $(z_1, z_2) = (e^{i\theta_1}, e^{i\theta_2}) \in \T^2$; the two phases are the independent qubitization eigenphases, and no joint eigenspace of $W_R, U_I$ is invoked.

By Parseval's theorem, $\mathcal{F}$ equals the sum of squared differences of bivariate Fourier coefficients:
\begin{equation}\label{eq:app_parseval}
  \mathcal{F}(\bm{\Theta})
  = \sum_{m,n}
    |c_{mn}^{\mathcal{G}}(\bm{\Theta}) - c_{mn}^{\mathrm{target}}|^2,
\end{equation}
so $\mathcal{F} = 0$ if and only if the circuit exactly implements the target polynomial on $\T^2$.
The Fourier representation has two computational advantages: (i)~evaluation via 2D FFT on an $N_1 \times N_2$ grid (with $N_j \geq 2d_j + 1$) costs $O(d_R d_I \log(d_R d_I))$ per evaluation; (ii)~the gradient
\begin{equation}\label{eq:app_grad}
  \frac{\partial\mathcal{F}}{\partial\theta_j}
  = 2\,\mathrm{Re}\sum_{m,n}
    \overline{(c_{mn}^{\mathcal{G}} - c_{mn}^{\mathrm{target}})}
    \cdot \frac{\partial c_{mn}^{\mathcal{G}}}{\partial\theta_j}
\end{equation}
is computed in the same asymptotic cost via the forward-backward pass (equations~\eqref{eq:forward}--\eqref{eq:backward} of Sec.~\ref{subsec:optimization}).

\subsection{Critical point analysis}\label{app:critical_points}

The cost function $\mathcal{F}$ has the following structural properties:
\begin{enumerate}
  \item \emph{Non-negativity and exactness.}\;
  $\mathcal{F} \geq 0$, with $\mathcal{F} = 0$ if and only if $P_{\mathcal{G}} = P_{\mathrm{target}}$ on $\T^2$. 
  This is immediate from $\mathcal{F}$ being an integral of a non-negative integrand.

  \item \emph{Existence of a global minimum at zero.}\;
  The constructive achievability (Theorem~\ref{thm:achievability}, via Theorem~\ref{thm:CRC}) guarantees that angles $\bm{\Theta}^*$  with $\mathcal{F}(\bm{\Theta}^*) = 0$ exist for a Dyson polynomial satisfying the conditions of Corollary~\ref{cor:dyson_achievable}.

  \item \emph{Smoothness.}\;
  $\mathcal{F}$ is $C^\infty$ in $\bm{\Theta}$, since $P_{\mathcal{G}}$ depends on $\bm{\Theta}$ through products of trigonometric functions of the angles. 
  The gradient~\eqref{eq:app_grad} and all higher derivatives exist everywhere.
\end{enumerate}

With this, gradient-based optimization is well-defined and the landscape is free of pathologies (non-differentiable points, non-existence of minima) common in variational quantum algorithms. 
They do not rule out spurious local minima.

The parameter space $\R^{2(d_R + d_I + 1)}$ is high-dimensional. 
The cost function couples two signal variables through the bivariate Fourier structure.
The circuit's product-of-matrices structure means that $\mathcal{F}$ is a trigonometric polynomial of degree at most $2(d_R + d_I)$ in each angle parameter. 
However, interaction between parameters through the matrix product creates a non-trivial energy landscape.

\subsection{Conjecture: no spurious local minima}\label{app:conjecture}

In the univariate setting, Motlagh and Wiebe~\cite{motlagh2024generalized} empirically observe no spurious local minima for GQSP cost functions up to polynomial degree $10^7$.
We know no proof of this in one variable, but this is more than sufficient numerical evidence for most practical polynomial degrees.

For the bivariate landscape, the parameter space doubles (from $d + 1$ real parameters to $2(d_R + d_I + 1)$) and the cost function couples two signal variables.  
We restate the formal conjecture from Problem~\ref{prob:landscape}:

\begin{conjecture}[No spurious local minima]\label{conj:no_spurious}
For every target polynomial $P_{\mathrm{target}} \in \mathcal{P}_{d_R,d_I}^+$ with $|P_{\mathrm{target}}| \leq 1$ on $\T^2$ that satisfies the achievability conditions of Theorem~\ref{thm:achievability}, every local minimum of the cost function $\mathcal{F}$ defined by~\eqref{eq:app_cost} is a global minimum (i.e., $\mathcal{F} = 0$).
\end{conjecture}

A positive resolution would permit random initialization of angle-finding optimization, eliminating the need for the $\calO((d_R + d_I) \cdot d_R \cdot d_I)$ recursive preprocessing.
This would reduce the classical cost of the M-QSP algorithm to $O(K \cdot d_R d_I \log(d_R d_I))$, where $K$ is the number of gradient descent iterations.

Univariate empirical evidence~\cite{motlagh2024generalized} extends to degrees far beyond a practical simulation parameter.
Further, the algebraic structure of the QSP circuit is similar to that of overparameterized neural networks, for which
benign landscape results are known in certain regimes~\cite{du2019gradient}.  
Neither observation is a proof, and the bivariate case introduces new features (the schedule $\mathbf{s}$ and coupling between $z_1$ and $z_2$ parameters) that have no univariate analogue.

\subsection{Numerical evidence and the warm-start strategy}\label{app:numerical}

In the absence of a proof of Conjecture~\ref{conj:no_spurious}, the practical strategy is a two-stage approach:

\begin{enumerate}
  \item \textbf{Initialize}: Compute $\bm{\Theta}^{(0)}$ via the recursive degree-reduction (Algorithm~\ref{alg:recursive}). 
  This produces the exact global minimizer up to numerical precision.

  \item \textbf{Iterate}: At each step $t$, evaluate $\mathcal{F}(\bm{\Theta}^{(t)})$ and $\nabla_{\bm{\Theta}}\mathcal{F}(\bm{\Theta}^{(t)})$ via 2D FFT (cost $O(d_R d_I \log(d_R d_I))$ per iteration).

  \item \textbf{Update}:
  $\bm{\Theta}^{(t+1)} \gets \bm{\Theta}^{(t)} - \gamma_t\nabla_{\bm{\Theta}}\mathcal{F}(\bm{\Theta}^{(t)})$, with step size $\gamma_t$ chosen by line search or L-BFGS.

  \item \textbf{Terminate} when $\mathcal{F}(\bm{\Theta}^{(t)}) \leq \eta$ (target tolerance).
\end{enumerate}

The recursive algorithm already produces the exact global minimizer, so warm starting places the optimization within the basin of attraction of the global minimum, and a small number of refinement iterations ($K_{\mathrm{refine}}$, empirically $O(1)$--$O(d)$) suffices to polish the angles to machine precision.

We combine the guarantees of Algorithm~\ref{alg:recursive} (which certifies correctness) with the $O(d_R d_I \log(d_R d_I))$ per-iteration cost of the FFT-based optimization, achieving practical scalability to total degrees $d_R + d_I \sim 10^6$ without sacrificing correctness.

Finite-precision effects in the recursive algorithm propagate linearly, not exponentially, through the degree-reduction steps (Proposition~\ref{prop:stability}), because each peeling step applies a unitary rotation with condition number~$1$.  
The overall condition number scales as $O(d_R + d_I)$, quite manageable for parameter regimes of interest.  
Optimization-based refinement mitigates accumulated rounding error from the recursive stage.

\section{Tight \texorpdfstring{$\log/\log\log$}{log log log} lower bound for bounded polynomial approximation}\label{app:tight_lb}

This appendix provides the full proof of Theorem~\ref{thm:tight_loglog} (tight $\log/\log\log$ scaling).

\begin{lemma}[Chebyshev coefficient lower bound on best approximation]
\label{lem:cheb_coeff_lb_app}
Let $f \in C[-1,1]$ have Chebyshev expansion $f(x) = a_0 + \sum_{k=1}^{\infty} a_k\,T_k(x)$.
Then for every $d \geq 0$,
\begin{equation}
  E_d(f)_{[-1,1]}
  := \inf_{\deg p \leq d}\|f - p\|_{\infty,[-1,1]}
  \geq \frac{\pi}{4}\,|a_{d+1}|.
\end{equation}
\end{lemma}

\begin{proof}
Let $p$ be a polynomial of degree $\leq d$.  Since $T_{d+1}$ is orthogonal to all polynomials of degree $\leq d$ with respect to the Chebyshev weight,
$a_{d+1} = \frac{2}{\pi}\int_{-1}^{1}(f(x) - p(x))\,T_{d+1}(x)(1-x^2)^{-1/2}\,dx$.
Taking absolute values and using $\int_{-1}^{1}|T_{d+1}(x)|(1-x^2)^{-1/2}\,dx = 2$ gives $|a_{d+1}| \leq (4/\pi)\,\|f - p\|_\infty$.
\end{proof}

\begin{lemma}[Chebyshev coefficients of $e^{a(x-1)}$]\label{lem:cheb_exp_app}
For $a > 0$, the function $g(x) = e^{a(x-1)}$ has Chebyshev coefficients $b_0 = e^{-a}\,I_0(a)$ and $b_k = 2\,e^{-a}\,I_k(a)$ for $k \geq 1$, where $I_k$ is the modified Bessel function of the first kind.
\end{lemma}

\begin{proof}
From the classical identity $e^{ax} = I_0(a) + 2\sum_{k=1}^{\infty}I_k(a)\,T_k(x)$ on $[-1,1]$.
\end{proof}

\begin{lemma}[Series lower bound on $I_k$]
\label{lem:bessel_lb_app}
For all $a > 0$ and $k \geq 0$, $I_k(a) \geq (a/2)^k/k!$.
\end{lemma}

\begin{proof}
The $m = 0$ term of $I_k(a) = \sum_{m=0}^{\infty}(a/2)^{k+2m}/(m!(m+k)!)$ equals $(a/2)^k/k!$; all terms are positive.
\end{proof}

\bibliography{nonherm_qsp_refs}

\end{document}